\documentclass[11pt,a4paper]{article}

\usepackage{hyperref}
\usepackage{color}
\usepackage{amsthm}
\usepackage{mathtools}
\usepackage{graphicx}
\usepackage{amssymb,latexsym,cite}
\usepackage{amsmath}
\usepackage{amsfonts}
\usepackage{mathrsfs}
\usepackage{bbm}
\usepackage{bm}
\usepackage[T1]{fontenc}

\usepackage{tikz-cd}
\usepackage[matrix,arrow,color]{xy}

\newcommand{\dwedge}{\curlywedge}
\newcommand{\om}{\omega}
\newcommand{\epsi}{\epsilon}
\newcommand{\nn}{\nonumber}

\newcommand{\scrF}{\mathscr{F}}

\newcommand{\midwedge}{\text{\Large$\wedge$}}
\newcommand{\midfwedge}{\text{\large$\wedge$}}

\newcommand{\dsf}{{\mathsf{d}}}

\def\DD{{\rm D}}

\def\slasha#1{\setbox0=\hbox{$#1$}#1\hskip-\wd0\hbox to\wd0{\hss\sl/\/\hss}}

\def\periodb#1{\setbox0=\hbox{$#1$}#1\hskip-\wd0\hbox to\wd0{-}}

\newcommand{\CA}{\mathcal{A}}    			

\newcommand{\CB}{\mathcal{B}}

\newcommand{\CC}{\mathcal{C}}

\newcommand{\CCD}{\mathscr{D}}

\newcommand{\CF}{\mathcal{F}}

\newcommand{\CJ}{\mathcal{J}}

\newcommand{\CCM}{\mathscr{M}}

\newcommand{\CP}{\mathcal{P}}

\newcommand{\CCP}{\mathscr{P}}

\newcommand{\CT}{\mathcal{T}}

\newcommand{\CV}{\mathcal{V}}
\newcommand{\CCV}{\mathscr{V}}

\newcommand{\CCW}{\mathscr{W}}
\newcommand{\CX}{\mathcal{X}}

\newcommand{\CZ}{\mathcal{Z}}

\newcommand{\CE}{\mathcal{E}}
\newcommand{\frg}{\mathfrak{g}}				

\newcommand{\FR}{\mathbbm{R}}     			
\newcommand{\RZ}{\mathbbm{Z}}     			

\newcommand{\dd}{\mathrm{d}}     			

\newcommand{\LL}{\mathrm{L}}     			

\newcommand{\sG}{\mathsf{G}}

\newcommand{\sSO}{\mathsf{SO}}

\newcommand{\comment}[1]{}     				
\def\tyng(#1){\hbox{\tiny$\yng(#1)$}}			
\def\tyoung(#1){\hbox{\tiny$\young(#1)$}}			

\newcommand{\beq}{\begin{eqnarray}}
\newcommand{\eeq}{\end{eqnarray}}

\definecolor{outrageousorange}{rgb}{1.0, 0.43, 0.29}

\newcommand{\Tr}{\mathrm{Tr}}

\theoremstyle{plain}

\newtheorem{lemma}[equation]{Lemma}

\theoremstyle{definition}

\theoremstyle{remark}

\def\beq{\begin{equation}}
\def\bee{\begin{equation}}
\def\eeq{\end{equation}}
\def\bea{\begin{eqnarray}}
\def\eea{\end{eqnarray}}
\def\ba{\begin{align}}
\def\ea{\end{align}}

\numberwithin{equation}{section}
\catcode`@=12

\begin{document}

\begin{titlepage}

\renewcommand{\thefootnote}{\fnsymbol{footnote}}

\begin{flushright}
\small
{\sf EMPG--20--03}
\end{flushright}

\begin{center}

\vspace{1cm}

\baselineskip=24pt

{\Large\bf $\boldsymbol{L_{\infty}}$-Algebras of Einstein--Cartan--Palatini Gravity}

\baselineskip=14pt

\vspace{1cm}

{\bf Marija Dimitrijevi\'c \'Ciri\'c}${}^{\,(a)\,,\,}$\footnote{Email: \ {\tt
    dmarija@ipb.ac.rs}} \ \ \ \ \ {\bf Grigorios Giotopoulos}${}^{\,(b)\,,\,}$\footnote{Email: \ {\tt
    gg42@hw.ac.uk}} \\[2mm] {\bf Voja Radovanovi\'c}${}^{\,(a)\,,\,}$\footnote{Email: \ {\tt
    rvoja@ipb.ac.rs}} \ \ \ \ \ {\bf Richard
  J. Szabo}${}^{\,(b)\,,\,}$\footnote{Email: \ {\tt R.J.Szabo@hw.ac.uk}}
\\[6mm]

\noindent  ${}^{(a)}$ {\it Faculty of Physics, University of
  Belgrade}\\ {\it Studentski trg 12, 11000 Beograd, Serbia}
\\[3mm]

\noindent  ${}^{(b)}$ {\it Department of Mathematics, Heriot-Watt University\\ Colin Maclaurin Building,
  Riccarton, Edinburgh EH14 4AS, U.K.}\\ and {\it Maxwell Institute for
Mathematical Sciences, Edinburgh, U.K.} \\ and {\it Higgs Centre
for Theoretical Physics, Edinburgh, U.K.}
\\[30mm]

\end{center}

\begin{abstract}
\noindent
We give a detailed account of the cyclic $L_\infty$-algebra
formulation of general relativity with cosmological constant in the
Einstein--Cartan--Palatini formalism on
spacetimes of arbitrary dimension and signature, which encompasses all
symmetries, field equations and Noether identities of gravity without matter fields. We present a local
formulation as
well as a global covariant framework, and  an
explicit isomorphism between the two $L_\infty$-algebras in the case
of parallelizable spacetimes. By  duality, we show that our $L_\infty$-algebras
describe the complete BV--BRST formulation of
Einstein--Cartan--Palatini gravity. We give a general description of
how to extend on-shell redundant symmetries in topological gauge theories
to off-shell correspondences between symmetries in terms of
quasi-isomorphisms of $L_\infty$-algebras. We use this to extend the on-shell
equivalence between gravity and
Chern--Simons theory in three dimensions to an explicit 
$L_\infty$-quasi-isomorphism between differential graded Lie
algebras which applies off-shell and for degenerate dynamical metrics. In contrast, we show that there is no morphism between the
$L_\infty$-algebra underlying
gravity and the differential graded Lie algebra
governing $BF$ theory in four dimensions. 
\end{abstract}

\end{titlepage}

{\baselineskip=10pt
\tableofcontents
}

\setcounter{footnote}{0}
\renewcommand{\thefootnote}{\arabic{footnote}}

\newpage

\section{Introduction}

Recent developments in string theory have suggested that the
low-energy effective dynamics of closed strings in non-geometric flux
compactifications may be governed by noncommutative and even
nonassociative deformations of
gravity~\cite{Blumenhagen:2010hj,Lust:2010iy,Blumenhagen:2011ph,Mylonas:2012pg,Blumenhagen:2013zpa,Mylonas:2013jha}. The
framework of nonassociative differential geometry was developed in
this context
in~\cite{Aschieri:2015roa,Blumenhagen:2016vpb,Aschieri:2017sug}. However,
the metric aspects of the theory have proved more difficult to develop
fully; in particular, a suitable generalization of the
Einstein--Hilbert action is not known. On the other hand,
noncommutative and nonassociative deformations of gravity are possible
to study in the Einstein--Cartan
formulation~\cite{Barnes:2015uxa,Barnes:2016cjm}, and this is the main
motivation behind the present paper: We wish to understand the
symmetries of these deformed gravity theories, their field equations
and their Lagrangian formulations.

One possible
path towards systematically understanding the symmetries and dynamics of noncommutative and nonassociative gravity is through the language of $L_\infty$-algebras.
$L_\infty$-algebras are generalizations of differential graded Lie algebras with infinitely-many graded antisymmetric
brackets, related to each other by higher homotopy versions of the
Jacobi identity. Their first appearence in the physics literature can be traced
back to higher spin gauge theories, where closure of the gauge algebra
necessitates using field dependent gauge parameters~\cite{Berends:1984rq}.
They first appeared systematically
in closed bosonic string field theory, where they govern the
``generalized'' gauge symmetries and
dynamics of the theory~\cite{Zwiebach:1992ie}: both the gauge
transformations and field equations of the theory involve
infinitely-many higher brackets of a cyclic $L_\infty$-algebra. They were systematically treated in the mathematics literature~\cite{LadaStasheff92}, where
they were shown to be dual to differential graded commutative
algebras. In~\cite{Linfty} it was suggested that the complete data of
classical field theories fit into truncated versions of $L_\infty$-algebras
with finitely many non-vanishing brackets, again encoding both gauge
transformations and dynamics; the prototypical examples are pure gauge
theories such as Yang--Mills theory and Chern--Simons theory. This was shown much earlier by~\cite{Stasheff97} to be a consequence of the duality with
the BV--BRST formalism, the details of which were further explained recently by~\cite{BVChristian}. Indeed, the BV--BRST complex is the familiar
physics incarnation of the duality between differential graded commutative algebras and
$L_\infty$-algebras: One may directly convert the BV complex of a classical field theory to
an $L_\infty$-algebra, and \emph{vice versa}. This explains why the
``generalized'' gauge symmetries and dynamics of every classical
perturbative field theory are organised by an underlying
$L_\infty$-algebra structure.\footnote{This holds only for polynomial field theories whose underlying
spaces of fields are vector spaces or affine spaces. Otherwise, one is restricted to the pertubation theory around a classical solution.}

The framework of $L_\infty$-algebras naturally seems to allow the
possibility for encoding noncommutativity and
nonassociativity, particularly the covariance of curvature fields and closure of
the gauge algebras which no longer follow the usual classical rules in
general. Indeed, $L_\infty$-algebras are known to play a crucial role in deformation theory, a famous example being Kontsevich's
formality theorem in deformation quantization whose proof is based on
$L_\infty$-quasi-isomorphisms of differential graded Lie algebras~\cite{Kontsevich:1997vb}. 
It was shown in~\cite{Blumenhagen:2018kwq,Kupriyanov:2019cug} that
noncommutative and nonassociative versions of the standard gauge
theories fit the same prescription as their classical counterparts,
though typically with infinitely-many brackets. This suggests the possibility of encoding nonassociative gravity in
the language of $L_\infty$-algebras, a setting where the symmetries
and dynamics would become more transparent and perhaps
an action functional could even be identified. However, there is
even more power in the approach: The $L_\infty$-algebra formulation can also
treat non-Lagrangian field theories, which have no action principle.

Classical general relativity on a $d$-dimensional manifold $M$ with
metric $g$ and dynamics governed by the
Einstein--Hilbert action functional 
\begin{align} \label{eq:EHaction}
S_{\textrm{\tiny EH}}(g) = \frac1{2\,\kappa^2} \, \int_M\, 
  (R-2\,\Lambda) \, \sqrt{-g} \ \dd^dx
\end{align}
cannot be directly interpreted as a gauge theory of principal bundle
connections. Of course, it does define a `generalized' gauge theory in
a certain extended sense whose symmetries are diffeomorphisms of $M$, and the corresponding BV formalism can be
developed as in e.g.~\cite{Cattaneo:2015xca}. However, this requires
working on the space of non-degenerate metric tensors (as does the very
definition \eqref{eq:EHaction}), which is an open
subset of the vector space of all symmetric rank~$2$ tensors
on $M$, and so does not fit naturally into an
$L_\infty$-algebra framework wherein the space of dynamical fields is
required to be a vector space. The way around this is to remember that
the $L_\infty$-algebra formulation works at the perturbative level and
to linearize the theory by expanding the metric $g$ around a chosen background. The space of metric
fluctuations $h$ is now a vector space, but the $L_\infty$-algebra
formulation will involve infinitely-many brackets of $h$ from the expansion
of the Einstein equations coming from \eqref{eq:EHaction} about the fixed background, as described
in~\cite{Linfty}; the $L_\infty$-algebra approach to general
relativity in this linearized setting is also discussed
by~\cite{Reiterer:2018wcb,Nutzi:2018vkl}. 

To avoid the introduction of an infinity of brackets, one may instead appeal to the Einstein--Cartan formulation which
rewrites general relativity as a gauge
theory on a principal bundle over $M$; the corresponding action is the
Palatini action whose definition allows for degenerate configurations, hence having a linear space of fields. This is the theory that we shall work with in this
paper; we call it the Einstein--Cartan--Palatini (ECP)
formulation of gravity.\footnote{It is also (perhaps more correctly) known as the
Einstein--Cartan--Sciama--Kibble theory.} The purpose of this paper is
then twofold. Firstly, we will express ECP gravity in the framework of
$L_\infty$-algebras, which should be familiar to experts in the BV
formalism, but perhaps not to a broader audience; the BV--BRST
formulation in four dimensions is developed in~\cite{ECBV} and the $L_{\infty}$-algebras we present in $d=4$ may also be obtained by direct dualization, one of which we show explicitly. Working in the  framework of $L_{\infty}$-algebras,  including degenerate metrics, allows for inversion up to homotopy of the relevant morphisms, a fact we make use in the main text. The physical (or otherwise) nature of degenerate solutions has been long studied, see for instance~\cite{degenerated'AuriaRegge}~\cite{degenerateTseytlin}~\cite{degenerateGiddings}, but we shall not concern ourselves with this matter in this paper. Secondly,
the formalism of this paper is the first step towards setting up a
framework for a new approach to noncommutative and nonassociative
deformations of gravity, and a larger class of theories satisfying a certain module property. It is in the course of thinking about this
latter problem that we realised a complete and explicit account of the first order
formalism for general relativity in the framework of
$L_\infty$-algebras does not seem to be available in the literature,
and in the following we focus on the classical case; its
noncommutative and nonassociative deformations will be treated in
subsequent papers, see \cite{NCProc} for a glimpse of the advantages of the approach. We translate the ECP formulation of
gravity in arbitrary spacetime dimension $d$ 
into the $L_\infty$-algebra framework, including the dynamics, the
local gauge and diffeomorphism symmetries, and the corresponding Noether identities; as we work on a Lie algebraic
level, the signature of spacetime has no bearing on our calculations
and we shall usually keep it arbitrary. Our formalism sets the stage
for many future investigations into the role played by
$L_\infty$-algebras in classical general relativity. For instance, deformations of the ECP functional and equivalences to other field theories via $L_{\infty}$-quasi-isomorphisms could be investigated. Furthermore, one can reinterpret and construct on-shell tree-level graviton scattering
amplitudes as brackets for the minimal model corresponding to our ECP
$L_\infty$-algebra, along the lines of~\cite{Nutzi:2018vkl}. Moreover, quasi-isomorphisms have been recently shown to underlie spontaneous symmetry breaking in terms of perturbation theory in gauge theories \cite{SymmBreaking}, that is when the underlying $L_{\infty}$-algebras are interepreted as (derived) tangent complexes around classical solutions. Identifying such equivalences for perturbations of gravity around classical backgrounds, breaking full diffeomorphism symmetry, is certainly of interest.

From the perspective of non-geometric string theory, it is natural to
restrict the spacetime manifold $M$ to flat space $\FR^d$, hence all
bundles considered are trivial. In this paper we will focus mostly on
non-covariant (local) gauge transformations, which are applicable only
on parallelizable manifolds $M$. This enables us to make most contact
with the existing literature on the ECP approach to general
relativity, where global issues are typically ignored. Under some
constraints, this is not a big restriction. For example, all
orientable three-manifolds are parallelizable. In the initial value
problem in general relativity, one usually assumes that the spacetime $M$
is globally hyperbolic. In four dimensions, global hyperbolicity
implies $M\simeq \FR\times N$, so that if the Cauchy surface $N$ is orientable
then the spacetime $M$ is also parallelizable. However, one can also
consider more general spacetime manifolds where this approach is
insufficient. We shall address this point in the following and
demonstrate how to modify the $L_\infty$-algebra formulation to a
global covariant structure, which illustrates the power of working in
the full $L_\infty$-algebra picture; for example, while the local
dynamics of three-dimensional general relativity can be formulated entirely in terms of
differential graded Lie algebras, the covariant framework requires
extending to the larger category of $L_\infty$-algebras as higher
brackets are introduced. The covariant structure we describe is essentially dual to the BV-BRST formalism developed in \cite{ECBV}, and in this sense should be viewed as a review, however in the $L_{\infty}$-algebra picture we clarify the geometric meaning behind the properties of the covariant Lie derivatives and different terms appearing in the BV differential.

Nevertheless, the non-covariant $L_\infty$-algebra structure for
parallelizable spacetimes will avoid the global issues involved when
twisting the framework to noncommutative principal
bundles~\cite{NCProc}. The global approach that we present in this
paper would be much more involved from the perspective of twist deformation quantization,
where one would have to deal with finite noncommutative gauge
transformations as putative ``transition functions''. Furthermore,
the noncommutative theory developed in~\cite{NCProc} and subsequent papers will be viewed as
a low-energy effective theory of gravity which encodes
corrections due to noncommutativity. Since this asymptotic expansion
of the deformation
quantization is formal, and usually only makes sense on affine spaces, global issues are ignored from the outset so
that the extraction of explicit first order corrections is immediate. 

\subsubsection*{Summary of results and outline}

The main results and outline of the remainder of this paper are as
follows. Sections~\ref{sec:Linftyreview} and~\ref{sec:ECPgravity} 
review the main background material needed in the rest of the
paper. In Section~\ref{sec:Linftyreview} we introduce the basic
notions surrounding cyclic $L_\infty$-algebras and their morphisms
that we will use, and how they completely determine the symmetries and
dynamics of classical perturbative field theories with generalized
gauge symmetries; we further describe the duality with differential
graded commutative algebras which connects the $L_\infty$-algebra formalism
with the BV--BRST formalism. In Section~\ref{sec:ECPgravity} we
introduce the geometric formulation of Einstein--Cartan--Palatini
gravity with cosmological constant $\Lambda$ in arbitrary spacetime dimension $d$ and signature, including
a description of its local gauge and diffeomorphism symmetries, its
action functional and field equations, and the
corresponding Noether identities.

In Section~\ref{sec:TQFT} we consider the $L_\infty$-algebra
formulation of a simple class of topological field theories that are
related to our gravity theories, and which have gauge and shift
symmetries. The $L_\infty$-algebras in these instances are
differential graded Lie algebras. By virtue of their topological
character, these theories are diffeomorphism-invariant, but
diffeomorphisms are redundant symmetries: they are equivalent on-shell
to local gauge and shift transformations with field dependent
parameters. We demonstrate how to 
extend this correspondence between symmetries to an \emph{off-shell}
equivalence by explicitly constructing a
quasi-isomorphism to an extended $L_\infty$-algebra. Working in the
larger category of $L_\infty$-algebras is crucial for this equivalence, as the
morphism is not a quasi-isomorphism of differential graded Lie algebras. In dimensions $d\geq4$ these theories
also exhibit a simple instance of on-shell higher gauge symmetries. These
observations are used in our later attempts to connect ECP theories
with topological gauge theories in the $L_\infty$-algebra framework. 

In Section~\ref{sec:ECPLinftyalg}, we present our main construction of
the cyclic $L_\infty$-algebra determining ECP gravity, in any dimension $d$
and spacetime signature. We give an explicit construction of the
brackets and the cyclic pairing underlying both local gauge and
diffeomorphism symmetries, the field equations and corresponding Noether identities, and the action
functional. The structure of the $L_\infty$-algebra is largely
dependent on the spacetime dimension $d$. For instance, only for $d=3$
is the $L_\infty$-algebra a differential graded Lie algebra. We
proceed in Section~\ref{sec:BV-BRST} to review the BV--BRST
formalism for ECP gravity in arbitrary dimension and signature, developed recently in \cite{ECBV} for d=4, and
show that it is dual to our formulation in terms of $L_\infty$-algebras.

The
construction works either locally on $M$, or globally if the spacetime
$M$ is parallelizable. We discuss the incompatibility issues with this
construction in the case of non-parallelizable spacetimes where the
underlying bundles need not be trivial, and more generally (even when
all bundles are trivial) the incompatibility of infinitesimal
diffeomorphisms of the physical spacetime with finite gauge transformations. We rectify the
problem by defining a `covariant' version of our $L_\infty$-algebra, dual to the construction of \cite{ECBV},
which encompasses the symmetries and dynamics of ECP gravity on
general spacetimes with general bundles, and which is compatible with finite gauge transformations. We review the geometric meaning of the ``covariant'' Lie derivative appearing, which has already been  widely used in the first order literature and beyond, and clarify its relation with the new brackets. The covariant framework
illustrates the necessity of working in the category of $L_\infty$-algebras: As it
always introduces higher brackets, the covariant $L_\infty$-algebra is
never a differential graded Lie algebra.  It also demonstrates the
importance of including Noether identities into the underlying cochain
complex in order to capture the covariance of the Euler--Lagrange
derivatives of the theory. For parallelizable spacetimes, we present
a (strict) isomorphism between the local and covariant
$L_\infty$-algebras, dual to the symplectomorphism for four-dimensional
gravity in the BV formalism developed in~\cite{ECBV}.

We conclude by applying our constructions to some lower-dimensional cases
of particular interest. In Section~\ref{sec:3dgrav} we look at
three-dimensional gravity, whose underlying ECP $L_\infty$-algebra is
a differential graded Lie algebra. In this case, gravity is known to
be equivalent to a Chern--Simons gauge theory, where the
diffeomorphism symmetry is recovered on-shell by the gauge symmetries
of Chern--Simons theory. Using the framework we develop in
Section~\ref{sec:TQFT}, we construct an explicit 
$L_\infty$-quasi-isomorphism between the differential graded Lie algebras of the
two theories, which extends the equivalence both off-shell and to
degenerate metrics. This problem is also addressed within the strictly non-degenerate setting using the BV
formalism in~\cite{Cattaneo:2017ztd} by providing a (different) symplectomorphism to $BF$ theory. Our result may be viewed as an extension to the degenerate sector. The equivalence is also addressed in \cite{Berktav:2019sgl} using stacks. Finally, in
Section~\ref{sec:4dgrav} we consider four-dimensional gravity, whose
underlying ECP $L_\infty$-algebra is no longer a differential graded
Lie algebra. Abstract strictification theorems~\cite{Berger2007} imply
that any $L_\infty$-algebra is quasi-isomorphic to a differential
graded Lie algebra, though in practise the construction of the
quasi-isomorphism is very
difficult and not very convenient to make explicit. Applying this to four-dimensional
gravity, the strictification of its ECP $L_\infty$-algebra could possibly
correspond to some deformation of $BF$ theory in four
dimensions, as reviewed in~\cite{Freidel:2012np}. Indeed, we show that there cannot exist any
$L_\infty$-morphism between the $L_\infty$-algebra underlying
four-dimensional gravity and the differential graded Lie algebra of
$BF$ theory.

Two appendices at the end of the paper contain illustrative examples
of the details involved in the long cumbersome calculations required
to check the various homotopy relations for the ECP
$L_\infty$-algebras and the $L_\infty$-morphisms that we present in
the main text, and to establish the duality with the BV--BRST
complex of ECP gravity. In view of the duality, one may start from the BV-BRST framework and derive the algebras we present. However, the $L_\infty$-algebras presented in the text have been constructed using the bootstrap approach, unless otherwise explicitly noted. The duality then served as an additional consistency check with known results. Indeed, this work serves as a starting point in twisting the corresponding $L_{\infty}$-algebras and one needs not be fluent in BV--BRST to follow the procedure~\cite{NCProc}, which applies to a general class of theories whose $L_{\infty}$-algebras are modules of a certain Hopf algebra. For this reason, we present the material starting directly from the $L_{\infty}$-algebra picture while making contact with BV--BRST later on. We hope these detailed calculations are
instructive and provide a useful reference for further investigations
into the $L_\infty$-algebra picture of gravity. They are utilized
extensively in~\cite{NCProc} and subsequent papers.

\section{${L_{\infty}}$-algebras, differential graded algebras and
  field theory}
\label{sec:Linftyreview}

In this section we review the required algebraic constructions, and their applications to classical field theories, that are needed in this paper. 

\subsection{$L_{\infty}$-algebras}
\label{sec:Linfty}

We start by introducing notions of $L_\infty$-algebras, which form the central concept in this paper.

\subsubsection*{Brackets and homotopy relations}

An $L_\infty$-algebra is a $\RZ$-graded vector space $V=\bigoplus_{k\in \RZ}\, V_{k}$ equipped with graded antisymmetric multilinear maps
\begin{align*}
\ell_n: \midwedge^n V \longrightarrow V \ , \quad  v_1\wedge \dots\wedge v_n \longmapsto \ell_n (v_1,\dots,v_n)
\end{align*}
for each $n\geq1$, which we call $n$-brackets. The graded antisymmetry translates to 
\begin{equation}\label{eq:gradedantisym}
\ell_n (\dots, v,v',\dots) = -(-1)^{|v|\,|v'|}\, \ell_n (\dots, v',v,\dots) \ ,
\end{equation}
where we denote the degree of a homogeneous element $v\in V$ by $|v|$.
The $n$-bracket is a map of degree $|\ell_n|=2-n$, that is
\begin{equation*}
\big|\ell_n(v_{1}, \dots ,v_{n})\big| = 2-n +\sum_{j=1}^n \, |v_j| \ .
\end{equation*}

The $n$-brackets $\ell_n$ are required to fulfill infinitely many
identities ${\cal J}_n(v_1,\dots,v_n)=0$ for each $n\geq1$, called homotopy relations, with
\begin{align} \label{eq:calJndef}
{\cal J}_n(v_1,\dots,v_n) := \sum^n_{i=1}\, (-1)^{i\,(n-i)} \
  \sum_{\sigma\in{\rm Sh}_{i,n-i}} \, & \chi(\sigma;v_1,\dots,v_n) \\ &
                                                                      \qquad
                                                                      \times
                                                                      \ell_{n+1-i}\big(
                                                                      \ell_i(v_{\sigma(1)},\dots,
                                                                      v_{\sigma(i)}),
                                                                      v_{\sigma(i+1)},\dots
                                                                      ,v_{\sigma(n)}\big)
                                                                      \
                                                                        , \nn
\end{align}
where, for each $i=1,\dots,n$, the second sum runs over $(i,n-i)$-shuffled permutations
$\sigma\in S_n$ of degree $n$ which are restricted as 
\begin{equation*}
\sigma(1)<\dots<\sigma(i) \qquad \mbox{and} \qquad \sigma(i+1)<\dots < \sigma(n) \ .
\end{equation*}
The Koszul sign $\chi(\sigma;v_1,\dots,v_n)=\pm\, 1$ is determined from the grading by
\begin{align*}
v_{\sigma(1)}\wedge\cdots\wedge v_{\sigma(n)} = \chi(\sigma;v_1,\dots,v_n) \ v_1\wedge\cdots\wedge v_n \ .
\end{align*}

For example, the first three identities are given by
\begin{align}
0&={\cal J}_1(v) = \ell_1\big(\ell_1(v)\big) \ , \nn \\[6pt]
0&={\cal J}_2(v_1,v_2)  = \ell_1\big(\ell_2(v_1,v_2)\big) - \ell_2\big(\ell_1(v_1),v_2\big) - (-1)^{|v_1|}\, \ell_2\big(v_1, \ell_1(v_2)\big) \ , \nn \\[6pt]
0&={\cal J}_3(v_1,v_2,v_3) \label{I3} \\[4pt]
 & = \ell_1\big(\ell_3(v_1,v_2,v_3)\big) \nn \\
& \quad + \ell_3\big(\ell_1(v_1),v_2,v_3\big) + (-1)^{|v_1|}\, \ell_3\big(v_1, \ell_1(v_2), v_3\big) + (-1)^{|v_1|+|v_2|}\, \ell_3\big(v_1,v_2, \ell_1(v_3)\big) \nn \\
& \quad + \ell_2\big(\ell_2(v_1,v_2),v_3\big) + (-1)^{(|v_1|+|v_2|)\,|v_3|}\, \ell_2\big(\ell_2(v_3,v_1),v_2\big) +  (-1)^{(|v_2|+|v_3|)\,|v_1|}\, \ell_2\big(\ell_2(v_2,v_3),v_1\big) \ . \nn
\end{align}
The first identity states that the map $\ell_1:V\to V$ is a
differential making $V$ into a cochain complex
\begin{align*}
\cdots \xrightarrow{ \ \ \ell_1 \ \ } V_k \xrightarrow{ \ \ \ell_1 \ \ }
  V_{k+1} \xrightarrow{ \ \ \ell_1 \ \ } \cdots \ .
\end{align*}
The second identity states that $\ell_1$ is a graded
derivation with respect to the $2$-bracket $\ell_2$, that is,
$\ell_2:V_k\wedge V_l\to V_{k+l}$ is a cochain map. For $\ell_3=0$
the third identity is just the graded Jacobi identity for the
$2$-bracket $\ell_2$, while for $\ell_2=0$ it gives the graded Leibniz
rule for the differential $\ell_1$ with respect to the $3$-bracket
$\ell_3$; in general it expresses the coherence condition that makes the Jacobiator for $\ell_2$ on $V_k\wedge V_l\wedge
V_m\to V_{k+l+m-1}$ a
cochain homotopy, that is, the Jacobi identity is violated by
a homotopy. In this sense $L_\infty$-algebras are (strong) homotopy
deformations of differential graded Lie algebras which are the special
cases where the ternary and all higher brackets vanish: $\ell_n=0$ for
all~$n\geq3$. In general, the homotopy relations for $n\geq3$ are
generalized Jacobi identities; for later use, we note that the identity $\CJ_4=0$ is
given by
\begin{align}
\CJ_4(v_1,v_2,v_3, v_4)  
& = \ell_1\big(\ell_4(v_1,v_2,v_3,v_4)\big) \nn \\
& \quad - \ell_4\big(\ell_1(v_1),v_2,v_3, v_4\big) - (-1)^{|v_1|}\,
  \ell_4\big(v_1, \ell_1(v_2), v_3, v_4\big) \nn \\ 
&\quad -(-1)^{|v_1|+|v_2|}\, \ell_4\big(v_1,v_2, \ell_1(v_3), v_4 \big) 
- (-1)^{|v_1|+|v_2|+|v_3|}\, \ell_4\big(v_1,v_2, v_3, \ell_1(v_4)
  \big) \nn \\
& \quad - \ell_2\big(\ell_3(v_1,v_2, v_3),v_4\big) +
  (-1)^{|v_3|\,|v_4|}\, \ell_2\big(\ell_3(v_1, v_2, v_4),v_3\big) \nn \\
& \quad +  (-1)^{(1+|v_1|)\,|v_2|}\, \ell_2\big(v_2, \ell_3(v_1,v_3, v_4)\big) 
-  (-1)^{|v_1|}\, \ell_2\big(v_1, \ell_3(v_2,v_3, v_4)\big) \nn \\
& \quad + \ell_3 \big(\ell_2(v_1,v_2),v_3, v_4 \big) - (-1)^{|v_2|\,|v_3|}\, \ell_3 \big(\ell_2(v_1,v_3),v_2, v_4 
\big) \nn \\
& \quad + (-1)^{(|v_2|+|v_3|)\,|v_4|}\,\ell_3 \big(\ell_2(v_1,v_4),v_2, v_3 \big) 
+ \ell_3\big(v_1, \ell_2(v_2,v_3),v_4 \big) \nn \\
& \quad + (-1)^{|v_3|\,|v_4|}\,\ell_3 \big(v_1, \ell_2(v_2,v_4), v_3 \big) 
+ \ell_3\big(v_1, v_2, \ell_2(v_3,v_4)\big) \ . \label{I4}
\end{align}

\subsubsection*{$L_\infty$-morphisms}

The natural notion of a homomorphism from an $L_\infty$-algebra
$(V,\{\ell_n\})$ to another $L_\infty$-algebra $(V',\{\ell_n'\})$
consists of a collection of multilinear graded
antisymmetric maps
\begin{align*}
\psi_n:\midwedge^n V \longrightarrow V' \ , \quad  v_1\wedge \dots\wedge v_n \longmapsto \psi_n (v_1,\dots,v_n)
\end{align*}
of degree $|\psi_n|=1-n$ for each $n\geq1$, which satisfies an
appropriate identity intertwining the two sets of brackets. The
required identity is specified through
 the somewhat cumbersome relations
\begin{align}
& \sum_{i=1}^{n}\, (-1)^{i\,(n-i)} \ \sum_{\sigma\in{\rm Sh}_{i,n-i}}\,
  \chi(\sigma;v_1,\dots,v_n) \
  \psi_{n+1-i}\big(\ell_i(v_{\sigma(1)},\dots,v_{\sigma(i)}),v_{\sigma(i+1)},
  \dots, v_{\sigma(n)}\big) \nn \\[4pt]
&\hspace{2cm} =\sum_{k=1}^n\,\frac{1}{k!} \, (-1)^{\frac12\,k\,(k-1)} \ \sum_{i_1+\cdots+ i_k=n} \ \sum_{\sigma\in
  {\rm Sh}_{i_1,\dots,i_k}} \, (-1)^{\CZ(\sigma;v_1,\dots,v_n)} \,
  \chi(\sigma;v_1,\dots,v_n) \label{eq:morphismrels} \\
& \hspace{6cm} \times
  \ell_k'\big(\psi_{i_1}(v_{\sigma(1)},\dots,v_{\sigma(i_1)}), \dots,
  \psi_{i_k}(v_{\sigma(n-i_{k}+1)},\dots,v_{\sigma(n)})\big)
  \ , \nn
\end{align}
where, for each $k=1,\dots,n$, the fifth
sum runs over $(i_1,\dots,i_k)$-shuffled permutations $\sigma\in S_n$
which preserve the ordering within each block of length
$i_1,\dots,i_k$ of the partition of $n=i_1+\cdots+i_k$. The
additional sign factor is given by
\begin{align*}
\CZ(\sigma;v_1,\dots,v_n) = \sum_{j=1}^{k-1}\,(k-j)\,i_j +
  \sum_{j=2}^{k}\, (1-i_j) \ \sum_{l=1}^{i_1+\cdots+i_{j-1}} \, \big|v_{\sigma(l)}\big|
\end{align*}
for $\sigma\in {\rm Sh}_{i_1,\dots,i_k}$. Such a homomorphism is
called an $L_\infty$-morphism. Notice that the left-hand side of
\eqref{eq:morphismrels} is formally identical to the homotopy
relations \eqref{eq:calJndef} with $\ell_{n+1-i}$ replaced by
$\psi_{n+1-i}$.

The first condition for $n=1$ (internal
degree $1$) is given by
\begin{align*}
\psi_1\big(\ell_1(v)\big) = \ell_1'\big(\psi_1(v)\big) \ , 
\end{align*}
which states that the map $\psi_1:V\to V'$ is a 
cochain
map with respect to the $n=1$ differentials, that is, it defines a map of
the cochain complexes underlying the $L_\infty$-algebras that acts
degreewise as
$$
\xymatrix{
\cdots \ar[rr]^{\ell_1} & & V_k \ar[rr]^{\!\!\!\ell_1} \ar[d]_{\psi_1}
& & V_{k+1}
\ar[rr]^{\ell_1} \ar[d]_{\psi_1} & & \cdots \\
\cdots \ar[rr]^{\ell_1'} & & V_k' \ar[rr]^{\!\!\!\ell_1'} & & V_{k+1}'
\ar[rr]^{\ell_1'} & & \cdots
}
$$
and so descends to a homomorphism of the corresponding cohomology
groups
\begin{align} \label{eq:cohmor}
\psi_{1\ast}: H^{\bullet}(V,\ell_{1}) \longrightarrow
  H^{\bullet}(V',\ell_{1}') \ .
\end{align}
The second condition for $n=2$ (internal
degree $0$) reads
\begin{align*}
\psi_1\big(\ell_2(v_1,v_2)\big)-\ell_2'\big(\psi_1(v_1),\psi_1(v_2)\big)
                          &= \ell_1'\big(\psi_2(v_1,v_2)\big) + \psi_2\big(\ell_1(v_1),v_2\big)
  + (-1)^{|v_1|} \, \psi_2\big(v_1,\ell_1(v_2)\big) \ , 
\end{align*}
which means that $\psi_1$ preserves the $2$-brackets up
to a homotopy given by $\psi_2$. In particular, if $\psi_n=0$ for all
$n\geq2$, then $\psi_1$ generalizes a homomorphism of differential
graded Lie algebras. On the other hand, even if the underlying
$L_\infty$-algebras are differential graded Lie algebras, an
$L_\infty$-morphism is not generally a morphism of differential graded
Lie algebras.
The third condition for $n=3$ (internal
degree $-1$) is
\begin{align*}
& \psi_{3}\big(\ell_{1}(v_{1}),v_{2},v_{3}\big) +
  (-1)^{|v_{1}|}\,\psi_{3}\big(v_{1},\ell_{1}(v_{2}),v_{3}\big) +
  (-1)^{|v_{1}|+|v_{2}|}\, \psi_{3}(v_{1},v_{2},\ell_{1}(v_{3})\big) -
  \ell_{1}'\big(\psi_{3}(v_{1},v_{2},v_{3})\big) \\
& \hspace{1cm} + \psi_{1}\big(\ell_{3}(v_{1},v_{2},v_{3})\big) -
  \ell_{3}'\big(\psi_{1}(v_{1}),\psi_{1}(v_{2}),\psi_{1}(v_{3})\big) \\[4pt]
& \ =
  (-1)^{|v_{1}|}\,\ell_{2}'\big(\psi_{1}(v_{1}),\psi_{2}(v_{2},v_{3})\big)
  -(-1)^{|v_{2}|\,(1+|v_{1}|)}\,
  \ell_{2}'\big(\psi_{1}(v_{2}),\psi_{2}(v_{1},v_{3})\big) \\
& \hspace{3cm} +(-1)^{|v_{3}|\,(1+|v_{1}|+|v_{2}|)}\,
  \ell_{2}'\big(\psi_{1}(v_{3}),\psi_{2}(v_{1},v_{2})\big) \\
& \quad \ - \psi_{2}\big(\ell_{2}(v_{1},v_{2}),v_{3}\big) -
  (-1)^{(|v_{1}|+|v_{2}|)\,|v_{3}|}\,
  \psi_{2}\big(\ell_{2}(v_{3},v_{1}),v_{2})\big) -
  (-1)^{(|v_{2}|+|v_{3}|)\,|v_{1}|}\,\psi_{2}\big(\ell_{2}(v_{2},v_{3}),v_{1}\big)
  \ .
\end{align*}
In
general, an $L_\infty$-morphism preserves the $n$-brackets up to
homotopy.

From the perspective of the underlying cochain complexes, the natural
notion of isomorphism between $L_\infty$-algebras would be an
$L_\infty$-morphism whose induced map \eqref{eq:cohmor} 
is an isomorphism on the cohomology of
the complexes; in this case we call the collection $\{\psi_n\}$ an
$L_\infty$-quasi-isomorphism and say that the $L_\infty$-algebras are
quasi-isomorphic. Quasi-isomorphism defines an equivalence relation on
the broader set of all $L_\infty$-algebras~\cite{Kontsevich:1997vb}, in contrast to the category of
differential graded Lie algebras where not every quasi-isomorphism has
a homotopy inverse. A stronger notion demands that the
degree~$0$ cochain map $\psi_1:V\to V'$ itself is an isomorphism of
the underlying vector spaces; in this case the collection $\{\psi_n\}$ is called an $L_\infty$-isomorphism and the
$L_\infty$-algebras are said to be isomorphic. From an
$L_\infty$-isomorphism one can reconstruct all brackets of one
$L_\infty$-algebra from the brackets of the other by using the relations
\eqref{eq:morphismrels} if the inverse $\psi_1^{-1}$ is known explicitly. Both notions of
isomorphism between $L_\infty$-algebras will play a role in this
paper.

\subsubsection*{Cyclic pairings}

We will be particularly interested in the case where an
$L_\infty$-algebra $(V,\{\ell_n\})$ is further endowed with a graded symmetric non-degenerate bilinear pairing $\langle-,-\rangle:V\otimes V\to\FR$ which is cyclic in the sense that
\begin{align*}
\langle v_0,\ell_n(v_1,v_2,\dots,v_n)\rangle = (-1)^{n+(|v_0|+|v_n|)\,n+|v_n| \,\sum_{i=0}^{n-1}\,|v_i|} \ \langle v_n,\ell_n(v_0,v_1,\dots,v_{n-1})\rangle
\end{align*}
for all $n\geq1$. This is the natural notion of an inner product on an
$L_\infty$-algebra, and if such a pairing
exists the resulting algebraic structure is called a cyclic
$L_\infty$-algebra. If in addition the pairing is odd, say of degree $-p$, then the only non-vanishing pairings are $\langle-,-\rangle:V_k\otimes V_{p-k}\to \FR$ for $2k<p$ and the cyclicity condition simplifies to
\begin{align*}
\langle\ell_n(v_0,v_1,\dots,v_{n-1}),v_n\rangle = (-1)^{(|v_0|+1)\,n} \ \langle v_0,\ell_n(v_1,\dots,v_n)\rangle \ .
\end{align*}
Given $L_\infty$-algebras $(V,\{\ell_n\})$ and
$(V',\{\ell_n'\})$ which are endowed with cyclic pairings
$\langle-,-\rangle$ and $\langle-,-\rangle'$, then an
$L_\infty$-morphism $\{\psi_n\}$ between them is cyclic if it additionally
preserves the pairing in the sense that
\begin{align*}
\langle\psi_1(v_1),\psi_1(v_2)\rangle' &= \langle v_1,v_2\rangle \ ,
\end{align*}
and
\begin{align}\label{eq:cyclicitycond}
\sum_{i=1}^{n-1} \, (-1)^{i-1+(n-i-1)\,\sum_{j=1}^i\,|v_j|} \,
  \langle\psi_i(v_1,\dots,v_i),\psi_{n-i}(v_{i+1},\dots,v_n)\rangle' &=
                                                                      0
                                                                      \ ,
\end{align}
for all $n\geq3$ and $v_1,\dots,v_n\in V$.

\subsection{Differential graded commutative algebras}
\label{sec:dgca}

There is a duality between semifree differential graded commutative
algebras and $L_{\infty}$-algebras of finite type, which we describe following the well-known mathematical treatment, see e.g.~\cite{LadaStasheff92,LadaMarkl94}. 
A differential graded commutative algebra is a graded commutative
algebra $A=\bigoplus_{k\in \RZ}\, A_{k}$, whose multiplication we
denote by $\,\cdot\,$, which is endowed with a differential of degree~$1$ which
is a graded derivation, that is, a linear map $\dd:A_{k} \rightarrow A_{k+1}$ such that $\dd^{2}=0$ and 
\begin{align*}
\dd(a\cdot b) &= \dd(a)\cdot b + (-1)^{|a|}\, a\cdot \dd(b)
\end{align*}
for all $a,b \in A$ with $a$ homogeneous. 
A graded commutative algebra is said to be of finite type if it is degreewise finite-dimensional; it is called semifree if it is isomorphic to the symmetric tensor algebra of a graded vector space. Semifree differential graded commutative algebras of finite type are in one-to-one correspondence with $L_{\infty}$-algebras of finite type~\cite{LadaStasheff92,LadaMarkl94}. We briefly spell out one direction of this correspondence which will be of use later on.

Let $V=\bigoplus_{k\in \RZ}\, V_{k}$ be a graded vector
space. Consider its suspension which is the degree shifted graded space $\scrF:=V[1]$, that is, $\scrF_{k}=V_{k+1}$, and note that the two are related by the trivial suspension isomorphism which decreases the grading by~$1$:
\begin{align*}
s:V  \longrightarrow \scrF \ , \quad 
v \longmapsto {}^sv
\end{align*}
with $|{}^sv|=|v|-1$.\footnote{One can think of the suspension $\scrF=V[1]$ as the tensor product
$\FR s\otimes V$ with elements
${}^sv=s\otimes v$, where $s$ is a fixed degree~$1$ element which has
a dual $s^{-1}$ with $s^{-1}\,s=1=-s\,s^{-1}$.} This induces an isomorphism of antisymmetric and
symmetric tensor algebras respectively, which is given by
\begin{align*}
s^{\otimes n}: \midwedge^n V \longrightarrow \text{\Large$\odot$}^n \scrF \ , \quad
v_{1} \wedge \cdots \wedge v_{n} \longmapsto (-1)^{\sum_{j=1}^{n-1}\,(n-j)\,|v_{j}|}\ {}^sv_{1}\odot \cdots \odot {}^sv_{n} 
\end{align*}
on each tensor power for $n\geq1$ and on homogeneous elements, and extended linearly. 

Using these trivial isomorphisms, one may equivalently identify an
$L_{\infty}$-algebra structure on $V$ as a coderivation $\DD$ of
degree~$1$ on $\text{\Large$\odot$}^\bullet \scrF$ viewed as a free cocommutative
coalgebra,\footnote{One should actually 
consider the reduced tensor coalgebra, that is, excluding the zeroth
tensor power which is a copy of $\FR$.} such that $\DD^{2}=0$. Explicitly, if $\Delta_\scrF:\text{\Large$\odot$}^\bullet \scrF\to\text{\Large$\odot$}^\bullet \scrF\otimes \text{\Large$\odot$}^\bullet \scrF$ is the free coproduct then
\begin{align*}
\Delta_\scrF\circ\DD = (\DD\otimes 1 + 1\otimes \DD)\circ\Delta_\scrF \ ,
\end{align*}
which implies that the coderivation
$\DD:\text{\Large$\odot$}^\bullet \scrF \rightarrow
\text{\Large$\odot$}^\bullet \scrF$ is completely determined by its
image in $\scrF$ according to the decomposition
$$
{\rm pr}_{\scrF}\circ\DD=\sum_{n=1}^\infty \, \DD_{n}
$$
with degree~$1$ component maps $\DD_{n}: \text{\Large$\odot$}^{n} \scrF
\rightarrow \scrF$, where ${\rm pr}_\scrF:\text{\Large$\odot$}^\bullet
\scrF \rightarrow\scrF$ is the projection to $\scrF$. Then the relation to the graded antisymmetric $n$-brackets defined on $V$, $\ell_{n}: \midwedge^{n} V \to V$, is given by 
\begin{align} \label{eq:ellndef}
\ell_{n} := s^{-1} \circ \DD_{n} \circ s^{\otimes n}: \midwedge^{n} V \xrightarrow{ \ s^{\otimes n} \ } \text{\Large$\odot$}^{n} \scrF \xrightarrow{ \ \DD_n \ } \scrF \xrightarrow{ \ s^{-1} \ } V \ .
\end{align}
The homotopy relations are then equivalent to $\DD^{2}=0$. 

This coalgebra picture is arguably harder to work with: its main
advantage is that the homotopy relations become simple and natural,
and in principle easier to check. Moreover, when
$V$ is of finite type one may unambiguously pass to the dual algebra
$\text{\Large$\odot$}^\bullet \scrF^{*}$ which is a graded commutative
algebra under the usual symmetric tensor product. The coderivation
then dualizes to a graded derivation 
$$
Q:=\DD^{\ast}:\text{\Large$\odot$}^\bullet \scrF^{*} \longrightarrow \text{\Large$\odot$}^\bullet \scrF^{*}
$$ 
of degree $1$, such that $Q^{2}=0$ if and only if $\DD^{2} =0$. It
follows that an $L_{\infty}$-algebra structure on a graded vector
space of finite type is equivalent to a differential derivation of the
symmetric algebra of its suspended dual vector space.

Another advantage of the coalgebra formulation is that the notion of an
$L_\infty$-morphism becomes more natural and transparent. A
morphism between two $L_\infty$-algebras determined by codifferential
coalgebras $(\text{\Large$\odot$}^\bullet
\scrF,\DD)$ and $(\text{\Large$\odot$}^\bullet
\scrF',\DD')$ is then given by a cohomomorphism of codifferential coalgebras:
$$
\Psi:\big(\text{\Large$\odot$}^\bullet
\scrF,\DD\big) \longrightarrow \big(\text{\Large$\odot$}^\bullet
\scrF',\DD'\big) \ ,
$$ 
that is, a degree~$0$ cohomomorphism
of the underlying free cocommutative coalgebras which intertwines the
codifferentials: $\Psi\circ\DD=\DD'\circ \Psi$; with the same sign
conventions for the components $\psi_n$ of
$\Psi$ as used in \eqref{eq:ellndef}, this single
cohomomorphism is equivalent to the collection of maps $\{\psi_n\}$ defining
the $L_\infty$-morphism. The dual
map $\Psi^*:\text{\Large$\odot$}^\bullet
\scrF^*\to \text{\Large$\odot$}^\bullet
\scrF^{\prime\,\ast}$ is a degree-preserving algebra homomorphism which intertwines the
corresponding derivations: 
$$
Q\circ\Psi^*=\Psi^*\circ Q' \ .
$$
An $L_\infty$-quasi-isomorphism in this picture is then naturally an algebra
homomorphism which is an isomorphism between the degree~$0$ cohomology groups of the
differentials $Q$ and $Q'$, whereas an $L_\infty$-isomorphism is
equivalently a coalgebra isomorphism of the corresponding coalgebras.

Finally, a graded symmetric non-degenerate pairing $\langle -,- \rangle$ on $V$ translates into a graded
antisymmetric pairing on the suspension $\scrF := V[1]$ defined by
$\langle -,- \rangle \circ (s^{-1}\otimes s^{-1})$. Since it is
non-degenerate it qualifies as a graded symplectic pairing on
$\scrF$. This then canonically induces a ``constant'' graded
symplectic two-form $\om \in \Omega^{2}(\scrF)$, which enables one to
thus view $\scrF$ not only as a graded symplectic vector space but
also as a graded symplectic manifold. Cyclicity is then equivalent to
$Q$-invariance of $\om$~\cite{BVChristian}. To get an intuition for
the condition \eqref{eq:cyclicitycond}, which is a shifted version of the condition from \cite{Kajiura:2003ax}, consider the case where $\scrF$
is concentrated in degree~$0$. Then by restricting a coalgebra morphism $\Psi: \text{\Large$\odot$}^\bullet
\scrF \rightarrow \text{\Large$\odot$}^\bullet
\scrF' $ to the diagonal of $\text{\Large$\odot$}^\bullet
\scrF$, the condition that the non-linear smooth map
${\rm pr}_{\scrF'}\circ \Psi_{|_{\scrF}}:\scrF \rightarrow \scrF'$ is a
symplectomorphism (of manifolds) is equivalent to the cyclicity
conditions of Section~\ref{sec:Linfty}.

\subsection{Generalized gauge field theories}
\label{sec:Linftygft}

We will now sketch how the algebraic constructions of this section find natural applications to the  treatment of classical field theories. The description we give is meant to provide a small bridge between the physics community interested in bootstrapping $L_{\infty}$-algebras, the dual BV--BRST construction and the properly rigorous treatment in terms of derived geometry, where the end product is identified as the tangent complex to a derived stack (at the trivial solution)\cite{Costello}. The procedure outlined should be rather viewed as an algorithm in identifying the spaces and brackets of the algebras in question.

\subsubsection*{Geometric formulation}

In the physical applications of relevance to this paper, the kinematical data of a (classical) field
theory with generalized
gauge symmetries on an oriented $d$-dimensional manifold $M$\footnote{The discussion applies for non-orientable manifolds, but the pairings appearing in the following should instead map into the density bundle over $M$.} are
encoded in two vector spaces $V_{0}$ and $V_{1}$, which are respectively the vector spaces of
(infinitesimal) gauge parameters and dynamical
fields, being typically sections of vector bundles over M. That is, we shall assume $V_{0}:=\Omega^{l}(M,\mathcal{C})$ and $V_{1}:=\Omega^{k}(M,\CV)$ where  $\CV,\CC $ are vector bundles over M and $0\leq l,k \leq d$.\footnote{In the case where $V_{1}$ is an affine space, such as a space of connections, the following discussion is trivially modified by fixing a reference element.} The gauge transformations acting on the dynamical fields generate a distribution $\CCD$ on $V_1$, which may not be necessarily involutive on the whole of $V_{1}$. By `generalized' gauge transformations we mean
that we include symmetries which are not restricted to vertical automorphisms
of principal bundles, and so go beyond the usual realm of standard
gauge theory. We shall also supplement the picture with the vector spaces  $V_{2}:= \Omega^{d-k}(M,\check{\CV}\,)$ and $V_{3}:=\Omega^{d-l}(M,\check{\CC}\,)$, where $\check{\CV}, \check{\CC}$ denote the dual bundles. The physical content of these spaces will be explained below. In practice it is often convenient to identify the internal dual bundles $\check{\CV}, \check{\CC}$ with certain isomorphic bundles through internal non-degenerate pairings entering the definition of field theories; we shall see this explicitly in the theories considered in later sections.

The
dynamics of the field theory is specified by an action functional
$S:V_1\to\FR$, which is a local function of the fields and their jets that is gauge-invariant: $\delta_\lambda S=0$ for all $\lambda\in V_0$, where
$\delta_\lambda=\delta\circ\iota_\lambda + \iota_\lambda\circ\delta$
with $\delta$ the exterior derivative on $V_1$ and $\iota_\lambda$ the
contraction with $\lambda\in\Gamma(\CCD)$\footnote{We abuse notation slightly and identify the gauge parameter with the vector field it generates on $V_{1}$.}. Its variation $\delta S$ is a
section of the cotangent bundle of
$V_1$ whose fibers we shall consider of degree 2, that is $\delta S:V_1\rightarrow T^*[-3]V_1$, where $[k]$ denotes a degree shift by $k\in\RZ$. Since $V_1$ is a vector space, we may define its cotangent bundle as the product of $V_{1}$ with its dual, for which the correct model in this infinite-dimensional setting is precisely $V_{2}$ defined above. That is
\begin{align}\label{eq:cotancomplex}
T^*[-3]V_1:=V_1\times V_2 \ .
\end{align}

In more detail, the variation of the action functional $\delta S:V_1\to
T^*[-3]V_1$, restricted to sections of compact support, is a map
$$
\delta S:V_1\longrightarrow V_1\times V_2
$$
acting by
$\delta S_{|A} (\delta A)=\langle \delta A,\CF(A)\rangle$ on tangent
vectors $\delta A \in T_{A} V_{1} \simeq V_{1}$ at a point $A \in
V_{1}$ in the space of fields. The non-degenerate
pairing $\langle -, - \rangle:V_1\times V_2\to\FR$ appearing is the natural pairing, that is, using the duality of the internal bundles and wedging the spacetime form parts to get a top form, following with an integration over $M$, while $\CF(A)\in V_{2}$ denotes the Euler--Lagrange derivatives of the functional $S$. Using the above interpretation, we say $V_{2}$ is the ``space of field equations''. The classical solution space, i.e. ``on-shell'' configurations, is defined as those $A\in V_{1}$ such that $\CF(A)=0$. Thus the field equations are enforced by intersecting
the image of $\delta S$ with the image of the zero section of
$T^*V_1$, which defines the critical Euler--Lagrange locus of the
action functional $S$; the distribution $\CCD$ is
involutive on this locus, that is, the gauge transformations close
on-shell. 

The natural pairing extends similarly to $T^{*}[-3] V_{0}:= V_{0} \times V_{3}$,
where now one may use it to prove Noether's second theorem: For all $\lambda \in V_{0}$, gauge-invariance of the action functional implies 
\begin{align*}
\delta_{\lambda} S = \langle \delta_{\lambda} A, \CF(A) \rangle = 0 \
  ,
\end{align*}
so that taking the Sturm--Liouville adjoint $\dsf_{A}$, with respect
to the pairing, of $\delta_{\lambda}$ viewed as a differential
operator acting on a gauge parameter $\lambda$ of compact support
amounts to $\langle \lambda, \dsf_{A} \CF(A)\rangle=0$. Here
$\dsf_{A}$ acts as a local differential operator on $V_{2}$, which may
depend on the fields, and its
image is valued in $V_{3}$. By non-degeneracy of the extended pairing, the `Bianchi
identities' $\dsf_{A} \CF(A)=0$, expressing local differential
relations among the Euler--Lagrange derivatives for any infinitesimal local
symmetry $\delta_\lambda A$, hold \emph{off-shell}, which is simply a reformulation of Noether's second theorem. Given the above interpretation, we say $V_{3}$ is the ``space of Noether identities''. The converse of Noether's second theorem is a means of recovering gauge symmetries of an action functional $S$ which may be unknown \emph{a priori}.\footnote{See e.g.~\cite{Avery:2015rga} for the proofs and comparisons of
  Noether's first and second theorems, and~\cite{Henneaux:1989jq} for their relation
  to gauge symmetries in physics.} 

From this geometric perspective, a classical generalized gauge theory is completely
determined by the moduli space $\CCM$ of its Euler--Lagrange locus $\CF(A)=0$
modulo gauge transformations. Two gauge theories are then physically
equivalent, in the sense that there is a bijection between their physical states, if the corresponding
moduli spaces of solutions to the field equations are isomorphic: $\CCM\simeq\CCM'$. In the case of reducible symmetries\footnote{For example, in the case where symmetries arise from a group action which is not free, i.e. has non-trivial stabilisers.}, the graded vector space may be extended by adjoining $V_{-1}$ containing the ``higher'' gauge parameters and $V_{4}$ its dual containing the ``higher'' Noether identities in a similar vein as above. Higher level reducibility parameters and their duals may also be added, if they occur. Classically this augmentation does not offer much new information, however it is essential in the dual (quantum) BV-BRST formulation  where one is interested in resolving degeneracies for the purpose of path integral methods.

\subsubsection*{$L_\infty$-algebra formulation}
We can translate this geometric picture into the structure of a
four-term $L_\infty$-algebra by linearizing the gauge transformations,
field equations and Noether identities to obtain
the cochain complex
\begin{align}\label{eq:cochaincomplex}
V_0\xrightarrow{ \ \ \ell_1 \ \ }V_1\xrightarrow{ \ \ \ell_1 \ \ 
  }V_2\xrightarrow{ \ \ \ell_1 \ \ }V_3
\end{align}
corresponding to the underlying graded vector space $V:= T^{\star}[-3] \big(V_{0}\oplus V_{1}\big)$, that is
$$
V=V_{0}\oplus V_{1}\oplus V_{2} \oplus V_3 \ .
$$
We then equip this complex with suitable
higher brackets corresponding to the non-linear parts of the theory subject to the homotopy relations in order the recover
the full symmetries and dynamics of the generalized gauge theory.

Given $\lambda \in V_{0}$ and $A\in V_{1}$, the gauge
variations are encoded as the maps $A\mapsto A+\delta_\lambda A$ where
\begin{align} \label{gaugetransfA}
\delta_{\lambda}A=\sum_{n =0}^\infty \, \frac{1}{n!}\, (-1)^{\frac12\,{n\,(n-1)}}\, \ell_{n+1}(\lambda,A,\dots,A) \ \in \ V_{1} \ ,
\end{align}
where the brackets involve $n$ insertions of the field
$A$. 
The Euler--Lagrange derivatives are encoded as
\begin{align}\label{EOM} 
\mathcal{F}(A)=\sum_{n =1}^\infty \, \frac{1}{n!}\, (-1)^{\frac12\,{n\,(n-1)}}\, \ell_{n}(A,\dots,A) \ \in \ V_{2} \ ,
\end{align}
with the covariant gauge variations
\begin{align} \label{gaugetransfF}
\delta_{\lambda} \mathcal{F}=\sum_{n =0}^\infty \, \frac{1}{n!}\, (-1)^{\frac12\,{n\,(n-1)}}\, \ell_{n+2}(\lambda,\mathcal{F},A,\dots,A) \ \in \ V_{2} \ .
\end{align}
We define successive applications of gauge variations by
\begin{align*}
(\delta_{\lambda_1}\,\delta_{\lambda_2})A := \sum_{n =0}^\infty \,
  \frac{1}{n!}\, (-1)^{\frac12\,{n\,(n+1)}}\,
  \ell_{n+2}(\lambda_2,\delta_{\lambda_1} A,A,\dots,A) \ .
\end{align*}
The closure relation for the gauge algebra then has the form
\begin{align}\label{eq:closure}
[\delta_{\lambda_1},\delta_{\lambda_2}]A
  = \delta_{[\![\lambda_1,\lambda_2]\!]_A}A  + \Delta_{\lambda_1,\lambda_2}A\ ,
\end{align}
where 
\begin{align*}
[\![\lambda_1,\lambda_2]\!]_A = -\sum_{n=0}^\infty \, \frac1{n!}\,
  (-1)^{\frac12\,{n\,(n-1)}}\, \ell_{n+2}(\lambda_1,\lambda_2,A,\dots,A)
  \ \in \ V_0 \ ,
\end{align*}
and
\begin{align*}
\Delta_{\lambda_1,\lambda_2}A = \sum_{n=0}^\infty \, \frac1{n!} \,
  (-1)^{\frac12\,(n-2)\,(n-3)} \,
  \ell_{n+3}(\lambda_1,\lambda_2,\CF,A,\dots,A) \ \in \ V_1 \ .
\end{align*}
The distribution
$\CCD\subset TV_1$ spanned by the gauge parameters is involutive on-shell, that is, when $\CF(A)=0$, and on this locus the gauge algebra generally depends on the fields $A$.
The homotopy relations guarantee that the Jacobi
identity is generally satisfied for any triple of maps $\delta_{\lambda_1}$,
$\delta_{\lambda_2}$ and~$\delta_{\lambda_3}$.
The Noether identities are encoded by
\begin{align}\label{eq:Noether}
\dsf_A\CF = \sum_{n=0}^\infty \, \frac1{n!} \, (-1)^{\frac12\, n\, (n-1)} \, \ell_{n+1}(\CF,A,\dots,A) \ \in \ V_3 \ ,
\end{align}
which vanishes identically as a consequence of the homotopy relations ${\cal J}_n(A,\dots,A)=0$, for all $n\geq1$, of the $L_\infty$-algebra. 

Given two generalized gauge theories with underlying
$L_\infty$-algebras $(V,\{\ell_n\})$ and $(V',\{\ell_n'\})$, an 
$L_\infty$-morphism $\{\psi_n\}$ between them relates their
classical moduli spaces $\CCM$ and $\CCM'$ in the following way~\cite{BVChristian}. A
gauge field $A\in V_1$ is sent by an $L_\infty$-morphism $\{\psi_n\}$ into the
gauge field
\begin{align}\label{eq:fieldmorphism}
A'(A) = \sum_{n=1}^\infty\, \frac1{n!}\, (-1)^{\frac12\,n\,(n-1)}\,
  \psi_n(A,\dots,A) \ \in \ V_1' \ ,
\end{align}
such that the corresponding Euler--Lagrange derivative $\CF(A)\in V_2$
is mapped to
\begin{align} \label{eq:ELmorphism}
\CF'(A') = \CF'(\CF,A) = \sum_{n=0}^\infty\, \frac1{n!}\, (-1)^{\frac12\,n\,(n-1)}\,
  \psi_{n+1}(\CF,A,\dots,A) \ \in \ V_2' \ .
\end{align}
It follows that the Euler--Lagrange locus $\CF(A)=0$ is mapped to the
Euler--Lagrange locus $\CF'(A')=0$. An $L_\infty$-morphism also sends
gauge orbits into gauge orbits: A gauge variation $\delta_\lambda A$
for $\lambda\in V_0$ is mapped by $\{\psi_n\}$ to the gauge variation
$\delta_{\lambda'}A'$ where
\begin{align}\label{eq:parmorphism}
\lambda'(\lambda,A) = \sum_{n=0}^\infty\, \frac1{n!}\, (-1)^{\frac12\,n\,(n-1)}\,
  \psi_{n+1}(\lambda,A,\dots,A) \ \in \ V_0' \ .
\end{align}
It follows that gauge equivalence classes of on-shell solutions
$\CF(A)=0$ are sent to gauge equivalence classes of on-shell solutions
$\CF'(A')=0$:
\begin{align}
A'(A+\delta_\lambda A) &= A'(A)+\delta'_{\lambda'(\lambda,A)}A'(A) \ ,
                         \nn \\[4pt]
\CF'(\CF+\delta_\lambda\CF,A+\delta_\lambda A) &=
                                                 \CF'(\CF,A)+\delta'_{\lambda'(\lambda,A)}\CF'(\CF,A)
                                                 \ ,
\label{eq:gaugemorphism}\end{align}
with the closure relation \eqref{eq:closure} mapping as
\begin{align}\label{eq:closuremorphism}
A'(A+\delta_{[\![\lambda_1,\lambda_2]\!]_A}A+\Delta_{\lambda_1,\lambda_2}A)
  = A'(A) +
  \delta'_{[\![\lambda_1',\lambda_2']\!]'_{A'} +
  \lambda'(\lambda_2,\delta_{\lambda_1}A) -
  \lambda'(\lambda_1,\delta_{\lambda_2}A)}A'(A) +
  \Delta'_{\lambda_1',\lambda_2'}A' \ .
\end{align}
In particular, if $\{\psi_n\}$ is an
$L_\infty$-quasi-isomorphism, then the corresponding moduli spaces
$\CCM$ and $\CCM'$ of physical states are
isomorphic~\cite{Kontsevich:1997vb,Kajiura:2001ng,Fukaya:2001uc}. However,
it is important to note that the converse is not true, and it may
happen that two classical field theories have isomorphic moduli spaces
while their underlying $L_\infty$-algebras are \emph{not}
quasi-isomorphic; we shall see some explicit instances of this later on.
The
transformations \eqref{eq:fieldmorphism}--\eqref{eq:closuremorphism} are interpreted in terms of Seiberg--Witten
maps in~\cite{Blumenhagen:2018shf}.

If the symmetries themselves
  have non-trivial symmetries, that is, there are further gauge
  redundancies in the description and the gauge symmetries are
  reducible, then the cochain complex \eqref{eq:cochaincomplex} may be extended into negative degrees $V_{-k}$ for $k\geq1$, which are the spaces of ``higher gauge transformations'', together with their duals $V_{k+3}$; the space $V_{-1}$ is the vector space of gauge transformations of the gauge parameters, $V_{-2}$ contains gauge variations of the gauge transformations of gauge parameters, and so on. These higher gauge symmetries are encoded as
\begin{align}\label{eq:highergauge}
\delta_{(\lambda_{-k-1},A)}\lambda_{-k}=\sum_{n=0}^\infty\, \frac1{n!} \, (-1)^{\frac12\,n\,(n+1)} \, \ell_{n+1}(A,\dots,A,\lambda_{-k-1}) \ \in \ V_{-k} \ ,
\end{align}
where $\lambda_{-k}\in V_{-k}$ for $k\geq0$.
At the classical level in which we are presently working, their
inclusion is purely algebraic and only serves to alter the cohomology
$H^\bullet(V,\ell_1)$ of the underlying cochain complex at its
extremities, leaving the moduli space $\CCM$ of classical states
unchanged. 

One way to encode the action functional of the gauge field theory is
via a symmetric non-degenerate bilinear pairing $\langle -,-\rangle :
V \otimes V\to\FR$ of degree~$-3$, as described earlier, which makes $V$ into a cyclic $L_\infty$-algebra. More specifically, the pairing is only defined when restricted to compactly supported sections due to the usual underlying integration. Said otherwise, the pairing is fiber-wise non-degenerate when viewed as a map into the density bundle. The only non-trivial pairings are
\begin{align*}
\langle -,-\rangle : V_1 \otimes V_2 \longrightarrow\FR \qquad \mbox{and} \qquad \langle -,-\rangle : V_0 \otimes V_3 \longrightarrow\FR \ ,
\end{align*}
and we shall explicitly make use of the cyclicity properties
\begin{align}\label{cyclicity} 
\langle A_{0},\ell_{n}(A_{1},A_2, \dots ,A_{n})\rangle=\langle A_{1},\ell_{n}(A_{0}, A_2,\dots ,A_{n})\rangle
\end{align}
and
\begin{align} \label{eq:cyclicity2}
\langle \lambda,\ell_{n+1}(E,A_{1}, \dots ,A_{n})\rangle=-\langle E,\ell_{n+1}(\lambda, A_1,\dots ,A_{n})\rangle
\end{align}
for all $\lambda\in V_0$, $A_0,A_1,\dots,A_{n}\in V_1$, $E\in V_2$ of suitable compact support and $n\geq1$.
In this case it is easy to see that the field equations $\mathcal{F}(A) = 0$ follow from varying the action functional defined as 
\begin{align} \label{action}
S(A) := \sum_{n=1}^\infty \, \frac{1}{(n+1)!}\, (-1)^{\frac12\,{n\,(n-1)}}\, \langle A, \ell_{n}(A,\dots,A)\rangle \ ,
\end{align}
since then cyclicity implies $\delta S(A)=\langle\CF,\delta A\rangle$.
Cyclicity also implies
\begin{align*}
\delta_\lambda S(A)=\langle\CF,\delta_\lambda A\rangle=-\langle\dsf_A\CF,\lambda\rangle \ ,
\end{align*}
so that gauge invariance of the action functional $\delta_\lambda
S(A)=0$ is then equivalent to the Noether identities
$\dsf_A\CF(A)=0$. Cyclic $L_\infty$-morphisms relate the action
functionals of generalized gauge theories.

\subsubsection*{BV--BRST formulation}

The convention we adopt here for the grading of our field theory
$L_{\infty}$-algebra is opposite to the convention of \cite{Linfty}. This bears no significant
mathematical effect, except for rendering more direct the duality to the BV--BRST
formalism, which makes precise the relation between the
$L_\infty$-algebra formulation and the geometric formulation based on the cotangent bundle
\eqref{eq:cotancomplex}; we review this framework in Section~\ref{sec:BV-BRST} for the specific gauge field
theories of interest in this paper. In particular, this gives a
rigorous description of the quotient defining
the moduli space $\CCM$ of classical solutions, by combining the Koszul--Tate resolution of the quotient by the ideal of Euler--Lagrange derivatives and the Chevalley--Eilenberg resolution of the quotient by gauge transformations. The duality between the
$L_{\infty}$-algebras for classical gauge field theories and their
BV--BRST formalism is precisely the duality discussed in
Section~\ref{sec:dgca}, see e.g.~\cite{BVChristian} for an extensive
review; our sign conventions also differ
from those of~\cite{BVChristian} where a sign factor
$(-1)^{\frac12\,n\,(n-1)}$ is included in the definition of the $n$-brackets
  \eqref{eq:ellndef} as well as in the definition of the $L_\infty$-morphism
  $\{\psi_n\}$ corresponding to a cohomomorphism $\Psi$. In this dual
  formulation, $L_\infty$-quasi-isomorphisms relate
  physically equivalent generalized gauge field theories at the classical level of moduli spaces and observables, and can also provide useful information in their pertubative quantisation; see e.g. \cite{ChristianLoops} \cite{Mnev:2017notes}. We will point out instances of this throughout the text. Reducible symmetries now become important as gauge parameters are promoted to dynamical fields in the BV--BRST framework, called ``ghosts'', and higher gauge parameters become ``ghosts-for-ghosts'', and so on with the purpose of resolving degeneracies of the corresponding action functional.

Strictly speaking, the $L_{\infty}$-algebras that arise in field
theories are not of finite type over $\FR$. However, the underlying
vector spaces consist of sections of vector bundles, and the brackets
are polydifferential operators of finite degree. Hence the
$L_{\infty}$-algebras are `local' in the sense of
\cite{Costello}. Thus after factoring the brackets through the
appropriate jet bundles, these are all pointwise of finite type over
$\FR$. This will be enough for the formal dualization that is needed
for our purposes in the following.
  
\section{$d$-dimensional gravity in the Einstein--Cartan--Palatini formalism}
\label{sec:ECPgravity}

In this section we introduce the generalized gauge field theories of interest in this paper. We review the formulation of general relativity in arbitrary dimension $d$ and signature within the Einstein--Cartan formulation, which enables us to treat gravity as a generalized gauge theory on a principal bundle. We shall then review the Palatini action functional for Einstein--Cartan gravity and the role played by Noether identities in this theory.

\subsection{Fields}
\label{sec:ECPgaugesym}

The data of the Einstein--Cartan--Palatini (ECP) formulation of
gravity in $d$ dimensions is as follows. The background spacetime consists of a smooth 
$d$-dimensional oriented manifold $M$ that admits a pseudo-Riemannian structure of signature $(p,q)$, with
$$
d= p+q \ ,
$$
where all constructions of this paper hold for any spacetime signature
$(p,q)$. Let $\CCV$ be a vector bundle on
$M$, isomorphic to the tangent bundle $TM$, which is equipped with a
fixed metric $\eta$ of signature $(p,q)$ and an orientation on its
fibers; we shall sometimes refer to $\CCV$ as the ``fake tangent bundle''.
The field content then consists of two
fields which are \emph{a priori} independent: a coframe field and a
spin connection.

The coframe field is an orientation-preserving bundle map $e:TM\to\CCV$, covering the identity,
from the tangent bundle $TM$ to the vector
bundle $\CCV$; it can be regarded as a one-form
on $M$ valued in $\CCV$ and used to pull back $\eta$ to a 
(possibly
degenerate) metric $g=e^*\eta$ on $M$ of indefinite signature
$(p,q)$. When the spacetime $M$ is parallelizable, one can
take $\CCV=M\times\FR^{p,q}$ to be the trivial bundle and regard $e$ as
a globally defined one-form on $M$ with values in $\FR^{p,q}$. In that case we write
$e=e^{a}\,{\tt E}_{a}\in \Omega^{1}(M,\FR^{p,q})$, where
$e^a\in\Omega^1(M)$ can be expanded as $e^a=e^a_\mu\, \dd x^\mu$ in a
local holonomic coframe on $M$. This satisfies\footnote{Here and in
  the following we always use the Einstein summation convention: repeated upper and lower indices are implicitly summed over.}
$e^{a}_{\mu}\,\eta_{ab}\,e^{b}_{\nu}=g_{\mu \nu}$, $\mu,\nu=1,\dots,d$
which gives the components of the (possibly degenerate) dynamical
pseudo-Riemannian metric $g$ in any local
coordinate chart of $M$, where $\eta$ is the standard metric
on $\FR^{p,q}$ and ${{\tt E}_{a}}$, $a=1,\dots,d$ form the canonical
oriented pseudo-orthonormal basis of
$\FR^{p,q}$. When $M$ is not parallelizable the discussion which follows
will only be valid on local
trivialisations of $\CCV$. This subtlety is usually
ignored in treatments of the first-order formalism for general
relativity in the literature, as for physical applications such global
issues are usually irrelevant, and in this paper we will also follow
this convention for the most part. We shall return to this point in Section~\ref{sec:covariant}.

The spin connection $\omega$ is an $\sSO_+(p,q)$-connection on the
principal $\sSO_+(p,q)$-bundle $\CCP\to M$ associated to $\CCV$, so that
$\omega\in\Omega^1(\CCP,\mathfrak{so}(p,q))$; it corresponds to local pseudo-orthogonal transformations of the
coframe field that are connected to the identity, which are
parameterized by maps from $M$ to the connected component $\sSO_+(p,q)$ of the indefinite special orthogonal group $\sSO(p,q)$. With the same caveats as discussed above for the coframe
fields, we shall usually regard it as a one-form $\omega\in
\Omega^{1}(M,\mathfrak{so}(p,q))$ and write it as\footnote{Here and in
  the following the square parantheses always mean antisymmetrization over the enclosed indices.}
$$
\omega=\omega^{a}{}_{b}\,{\tt E}^{b}{}_{a}=\omega^{ab}\,{\tt E}_{ba}=\om^{ab}\,{\tt E}_{[ba]}
$$
with $\omega^{ab}=-\omega^{ba} \in \Omega^{1}(M)$, where indices are
raised and lowered with the metric $\eta$ and ${\tt E}^{b}{}_{a}$ are the
$d\times d$ matrix units with matrix elements
$({\tt E}^b{}_a)_{cd}=\delta_{ac}\,\delta^b{}_d$. 

The covariant derivative of $e$ defines the torsion
$T$ of the $\sSO_+(p,q)$-connection which can be regarded as a two-form
on $M$ valued in $\CCV$, while the
covariant derivative of $\omega$ defines its
curvature $R$ which can be regarded as a two-form on $M$ with values in the
endomorphism bundle of the vector bundle
$\CCV$ with structure group $\sSO_+(p,q)$; equivalently, $R$ can be
regarded as valued in the
second exterior power $\midwedge^2\,\CCV$ of $\CCV$, which is isomorphic to the vector bundle $\CCP\times_{\rm ad}\mathfrak{so}(p,q)$ associated to $\CCP$ by the adjoint representation of the structure group. Explicitly,
\begin{align*}
T:=\dd^{\omega} e &= \dd e + \omega \wedge e \ \in \
  \Omega^{2}(M,\CCV) \ , \\[4pt]
R:=\dd^{\om}\om &= \dd\omega + \tfrac{1}{2}\, [\omega,\omega] \ \in \
  \Omega^{2}\big(M,\CCP\times_{\rm ad}\mathfrak{so}(p,q)\big) \ .
\end{align*}
In a local trivialisation, the wedges here mean matrix multiplication followed by the exterior
product on form entries, that is, $\omega \wedge e=
\omega^{a}{}_{b} \wedge e^{b} \, {\tt E}_{a}$ and $\frac{1}{2}[\om,\om]
=\omega^{a}{}_{b} \wedge \omega^{bc} \, {\tt E}_{[ca]}$. The Bianchi identities are
\begin{align} \label{eq:BianchiTR}
\dd^\om T=R\wedge e \qquad \mbox{and} \qquad \dd^\om R=0 \ .
\end{align}

\subsection{Gauge symmetries}
\label{sec:ECPgauge}

The natural internal symmetries of the fields are the
finite $\sSO_+(p,q)$-transformations $h:M\to \CCP\times_{\rm Ad}\sSO_+(p,q)$ given by
\begin{align*}
e\longmapsto h^{-1}\,e \qquad \mbox{and} \qquad \omega\longmapsto
 h^{-1}\,\omega\, h + h^{-1}\,\dd h \ ,
\end{align*}
corresponding to vertical automorphisms of the principal $\sSO_+(p,q)$-bundle
$\CCP\to M$ which cover the identity diffeomorphism on $M$. There is also the action of finite diffeomorphisms $\phi:M\to M$ given by the
pullbacks 
\begin{align*}
e\longmapsto \phi^*e \qquad \mbox{and} \qquad \omega\longmapsto
  \phi^*\omega \ ,
\end{align*}
which map the fields to sections and connections on the corresponding
pullback bundles. The infinitesimal version of diffeomorphisms is not compatible with the global structure of the fields $(e,\om)$, as we will
discuss further in Section~\ref{sec:covariant}. For the time being we
shall work with the local formulation of the field content $(e,\om)$
of the Einstein--Cartan--Palatini theory and hence consider only the
infinitesimal versions of these symmetries; that is, we consider a
parallelizable spacetime $M$ and identify the fields as globally
defined one-forms $e\in\Omega^1(M,\FR^{p,q})$ and $\om\in\Omega^1(M,\mathfrak{so}(p,q))$.

For any infinitesimal gauge parameter function $\rho:M \rightarrow
\mathfrak{so}(p,q)$, which corresponds to a local pseudo-orthogonal
rotation, we may write $\rho=\rho^{ab}\,{\tt E}_{ba}$ for
$\rho^{ab}=-\rho^{ba}\in C^\infty(M)$. Then the internal (infinitesimal) gauge transformations are given by
\begin{align}\label{eq:gaugerot}
\delta_{\rho}e=-\rho \cdot e\qquad \mbox{and} \qquad
\delta_{\rho}\omega=\dd^\om\rho:=\dd\rho + [\omega,\rho] \ .                           
\end{align}
Since $e$ transforms in the fundamental representation of $\sSO_+(p,q)$,\footnote{In
  gauge theory parlance, the coframe field $e$ could be regarded as a
  matter field, as it is a section of $\CCV\otimes T^*M$ with $\CCV$ the vector bundle associated to $\CCP$ by the fundamental representation, but since $e$ represents a gravitational field we refrain from using this terminology to avoid confusion.} the notation $\rho \cdot e$ literally means matrix multiplication of a vector: $\rho \cdot e= \rho^{a}{}_{b}\,e^{b}\,{\tt E}_{a}$. 

In addition to local pseudo-orthogonal rotations, there is also the standard diffeomorphism gauge symmetry of general relativity. Infinitesimal diffeomorphisms correspond to vector fields on the spacetime $M$. For an infinitesimal diffeomorphism $\xi \in \Gamma(TM)$, the corresponding gauge transformations are given by the action of the Lie derivative $\LL_\xi$ on forms:
\begin{align}\label{eq:gaugediff}
\delta_{\xi}e=\LL_{\xi}e \qquad \mbox{and} \qquad
\delta_{\xi}\omega=\LL_{\xi} \omega \ .                       
\end{align}
The Lie derivatives can be evaluated by using Cartan's `magic formula'
\begin{align} \label{eq:CartanLie}
\LL_\xi = \dd\circ\iota_\xi+\iota_\xi\circ\dd
\end{align}
where $\iota_\xi$
denotes contraction with the vector field $\xi$. 

Altogether, the involutive
symmetry distribution $\CCD$ on the space of fields 
\begin{align} \label{eq:ECPfieldspace}
\Omega^1(M,\FR^{p,q})\times\Omega^1\big(M,\mathfrak{so}(p,q)\big)
\end{align}
is the Lie
algebra of gauge symmetries generated by the action of the semi-direct product 
\begin{align} \label{eq:ECPgaugealg}
\Gamma(TM)\ltimes\Omega^0\big(M,\mathfrak{so}(p,q)\big) \ .
\end{align}
Under the infinitesimal symmetries $(\xi,\rho)\in
\Gamma(TM)\times\Omega^0\big(M,\mathfrak{so}(p,q)\big)$, the torsion
and curvature fields transform as
\begin{align*}
\delta_{(\xi,\rho)} T = \LL_\xi T -\rho\cdot T \qquad \mbox{and} \qquad \delta_{(\xi,\rho)} R = \LL_\xi R + [R,\rho] \ .
\end{align*}

Since $\om^{ab}=-\om^{ba}$, under the isomorphism $\mathfrak{so}(p,q)
\simeq \midwedge^{2}(\FR^{p,q})$ given by $\omega^{ab}\,{\tt E}_{[ba]}
\mapsto \om^{[ab]}\,{\tt E}_{a}\wedge {\tt E}_{b}$, we will consider the connection as an element
$\om=\om^{ab}\,{\tt E}_{a}\wedge {\tt E}_{b}\in \Omega^{1}(M,\midwedge^{2}(\FR^{p,q}))$ and the curvature
as an element $R=R^{ab}\,{\tt E}_{a}\wedge {\tt E}_{b}\in
\Omega^{2}(M,\midwedge^{2}(\FR^{p,q}))$. This identification has the
advantage of making our formulas later on much more compact, avoiding
extensive use of indices. This also shows that there is
an isomorphism of representations of $\sSO_+(p,q)$ given by
\begin{align*}
[\rho,\omega]\longmapsto \rho\cdot (\omega^{ab}\,{\tt E}_a\wedge
  {\tt E}_b)=\rho^a{}_c\,\omega^{cb}\, {\tt E}_a\wedge {\tt E}_b +
  \rho^b{}_c\,\omega^{ac}\, {\tt E}_a\wedge {\tt E}_b \ ,
\end{align*}
where the right-hand side is computed by using the Leibniz rule (or
trivial coproduct) for the action of $\mathfrak{so}(p,q)$ on the
two-vector representation $\midwedge^{2}(\FR^{p,q})$.
This can be used in comparing the actions of a spin connection on
another spin connection in the adjoint representation and on a coframe
field in the vector representation to get the useful identity
\begin{lemma}
If $e_{1}, \dots ,e_{d-2}\in \Omega^1(M,\FR^{p,q})$ and
$\om,\om'\in\Omega^1(M,\mathfrak{so}(p,q))$, then
\begin{align}\label{weirdformula} 
e_{1}\dwedge \cdots \dwedge e_{d-2} \dwedge [\omega ,
  \om'\,]= \sum_{i=1}^{d-2}\, e_{1}\dwedge \cdots
  \widehat{e_{i}} \cdots \dwedge e_{d-2} \dwedge (\om \wedge
  e_{i}) \dwedge \om'
\end{align}
 in $d\geq3$ dimensions, where the $\dwedge$-product
means the exterior products of both the differential form parts
and the internal vector space parts,\footnote{The $\dwedge$-product gives $\Omega^{\bullet}(M,\midfwedge^{\bullet}\FR^{p,q})$ the structure of a graded commutative algebra under the combined form degrees. It does not associate with the $\wedge$-product acting via multi-vector representations, as will be apparent in calculations below.} while $\widehat{e_i}$ means
omission of the $i$-th term in the product. 
\end{lemma}
\begin{proof}
Using the invariance of a top exterior vector in $\FR^{p,q}$ under $\sSO_+(p,q)$-transformations yields
\begin{align*}
0=\omega\, \wedge \, (e_1\dwedge\cdots \dwedge e_{d-2}\dwedge\omega'\,)= \omega\, \wedge \,(e_1\dwedge\cdots \dwedge e_{d-2})\dwedge\omega' + (-1)^{d-2} e_1\dwedge\cdots \dwedge e_{d-2}\dwedge[\omega,\omega'\,] \ ,
\end{align*}
and expanding $\omega\, \wedge \,(e_1\dwedge\cdots \dwedge e_{d-2})$ using
the Leibniz rule gives the sum on the right-hand side of~\eqref{weirdformula}. 
\end{proof}
Throughout the text, when  considering gravity we shall use the $\dwedge$-product to denote the double exterior product, while the $\wedge$-product will be reserved for acting via the multivector representation and wedging the respective spacetime forms, unless otherwise explicitly noted.
\subsection{Field equations}
\label{sec:ECPnd}

Given the field content from Section~\ref{sec:ECPgaugesym}, the action functional for Einstein--Cartan--Palatini gravity in $d>2$
dimensions is then given by
\begin{align} \label{eq:ECPaction}
S_{\textrm{\tiny ECP}}(e,\omega) :=\frac1{2\,\kappa^2} \, \int_M\,
  \Tr\Big(\frac{1}{d-2}\,e^{d-2}\dwedge R +\frac{1}{d}\,\Lambda\,
  e^{d}\Big) \ ,
\end{align} 
where the square of the 
parameter $\kappa\in\FR$ gives the gravitational constant and $\Lambda\in\FR$ is the cosmological constant. The integrands
in \eqref{eq:ECPaction}
are $d$-forms on $M$ with values in $\midwedge^d\,\CCV$, and since $\CCV$
has structure group $\sSO_+(p,q)$ it carries a natural volume form such
that the Hodge duality operator $\Tr:\Omega^d(M,\midwedge^d\,\CCV)\to \Omega^d(M)$
extracts the canonical choice of function defined by a section of
$\midwedge^d\,\CCV$. The powers are taken with respect to the
$\dwedge$-product, so that $ e^{d-2} \dwedge R$ and
$e^{d}$ are elements in $\Omega^{d}(M, \midwedge^{d}\,\CCV)$. In a local trivialization,
the Hodge duality operator $\Tr :
\midwedge^{d}(\FR^{p,q}) \rightarrow \FR$ on the internal vector space has the normalization
\begin{align*}
\Tr({\tt E}_{a_1}\wedge\cdots\wedge {\tt E}_{a_d})=\varepsilon_{a_1\cdots a_d} \ ,
\end{align*}
where $\varepsilon_{a_1\cdots a_d}$ is the Levi--Civita symbol in $d$ dimensions, thus yielding an $\FR$-valued $d$-form on the $d$-dimensional manifold $M$, which we can integrate. The action functional \eqref{eq:ECPaction} is invariant under both finite and infinitesimal gauge symmetries of Section~\ref{sec:ECPgauge}. For $d>3$ one can also add to \eqref{eq:ECPaction} higher curvature terms $e^{d-2\,k}\dwedge R^k$ which give Lovelock theories of gravity, but we will shall not consider such extensions in this paper.

The field equations which follow from varying the action functional
\eqref{eq:ECPaction} with respect to compactly supported $(e,\omega)$ are given by
\begin{align}
\CF_e(e,\omega):=e^{d-3}\dwedge R +\Lambda\, e^{d-1} &=0 \ \in \
                                      \Omega^{d-1}\big(M,\text{\Large$\wedge$}^{d-1}\,\CCV\big) \ , \nonumber \\[4pt] 
\CF_\omega(e,\omega):= e^{d-3}\dwedge T &=0 \ \in \
                               \Omega^{d-1}\big(M,\midwedge^{d-2}
                               \,\CCV\big) 
                                     \ ,
\label{eq:ECPeom}\end{align}
where the first equation comes from varying with respect to $e$ and
the second equation with respect to $\omega$ using the
identity \eqref{weirdformula}. In the 
case of a gravitational field that has no singularities, where the
coframe field $e$ is everywhere invertible and so defines a non-degenerate metric
$g=\eta_{ab}\, e^a\otimes e^b$, this formulation is equivalent
on-shell to the metric formulation with the Einstein--Hilbert action
functional \eqref{eq:EHaction}: In this case the second equation is equivalent to
$T=0$, which is just the
condition that the $\sSO_+(p,q)$-connection is torsion-free. This can be uniquely solved (up to gauge equivalence) to give
$\omega$ in terms of $e$, which may then be identified with the Levi--Civita
connection for the metric $g$; metric compatibility follows from the
fact that $\omega$ is an $\sSO_+(p,q)$-connection. The first
equation is then equivalent to the usual vacuum Einstein field equations with
cosmological constant. 

Thus the Einstein--Cartan--Palatani field theory, including degenerate coframes, only recovers the
standard Einstein--Hilbert metric formulation of general relativity
on-shell as a particular subspace of its Euler--Lagrange locus. The
two theories are generally not equivalent, not even on-shell. However, it is
important in our case to allow for degenerate coframes, so that the
space of fields \eqref{eq:ECPfieldspace} is indeed a vector space,
which is a necessary requirement for the $L_\infty$-algebra
formulation. This is not merely a technical burden, since the resulting module structures of the space of fields and that of the $L_{\infty}$-algebras'\footnote{The structures turn out to be modules of the enveloping algebra of vector fields on $M$.} render the theory very amenable to twisting methods \cite{NCProc}. Furthermore, it allows for the study of invertible morphisms up to homotopy in the $L_{\infty}$-algebras category which have direct physical content. This is in marked contrast to the Einstein--Hilbert
formulation, which necessitates a restriction to non-degenerate
metrics in order to write down the action functional
\eqref{eq:EHaction}; however, the space of non-degenerate metrics is
not a vector space. In the remainder of
this paper, the gravitational constant will not play any role and we
shall normalize the action functional \eqref{eq:ECPaction} so that
$2\,\kappa^2=1$. 

\subsection{Noether identities}

Noether's first theorem, which asserts the existence of a weakly
conserved current for each global symmetry of an action functional, is
not relevant in the present context because the action functional for  general relativity does not
have global symmetries.\footnote{Of course, specific solutions may
  have symmetries generated by Killing vectors that can be used to
  produce conserved quantities, but there are no global symmetries
  \emph{a priori}.} On the other hand, Noether's second theorem, which
relates the Euler--Lagrange derivatives of an action functional
off-shell, applies to the local pseudo-orthogonal and diffeomorphism gauge
symmetries of the Einstein--Cartan--Palatini theory, see
e.g.~\cite{Barnich:2016rwk,Montesinos:2018ujm}. The corresponding
Noether identities follow from gauge-invariance of the action
functional~\eqref{eq:ECPaction}:
\begin{align*}
0=\delta_{(\xi,\rho)}S_{\textrm{\tiny ECP}}(e,\omega) = \int_M\, \Tr\Big(\CF_e(e,\omega)\dwedge\delta_{(\xi,\rho)}e+\CF_\omega(e,\omega)\dwedge\delta_{(\xi,\rho)}\omega\Big) \ ,
\end{align*}
where $\delta_{(\xi,\rho)}(e,\omega):=(\delta_\xi e +
\delta_\rho e,\delta_\xi\omega+\delta_\rho\omega)$, and then varying this equation with respect to compactly supported $(\xi,\rho)$ using the explicit expressions for the gauge transformations
from \eqref{eq:gaugerot} and \eqref{eq:gaugediff}. This leads to a pair of strong differential identities among the Euler--Lagrange derivatives:
\begin{align}\label{eq:ECPNoether}
\dsf_{(e,\om)}(\CF_e,\CF_\omega) = (0,0) \ \in \ 
  \Omega^1\big(M,\Omega^d(M)\big) \times \Omega^d\big(M,\midwedge^{d-2}(\FR^{p,q})\big) \ ,
\end{align}
where
\begin{align}
\dsf_{(e,\om)}(\CF_e,\CF_\omega) 
:= \Big( - \dd x^\mu\otimes\Tr \big( &
                                                \iota_\mu \dd
                                                e \dwedge \CF_{e} +
                                                (-1)^{d-1}\,\iota_\mu \dd\om
                                                \dwedge \CF_{\om}   
                                                -\iota_\mu e
                                                \dwedge \dd \CF_{e} \nn \\ & -
                                                (-1)^{d-1}\,\iota_\mu \om\dwedge
                                                \dd
                                                \CF_{\om}\big) \ , \ -\frac{d-1}2 \, \CF_e\wedge e+(-1)^{d-1}\, \dd^\om\CF_\om \Big) \ , \label{eq:ECPcurrents}
\end{align}
and $\iota_\mu$ denotes the contraction with vectors
$\partial_\mu=\frac\partial{\partial x^\mu}$ of the 
local holonomic frame dual to the basis $\{\dd x^\mu\}$ of one-forms
in a local coordinate chart on $M$. Here
  we identify the vector space of one-forms valued in $d$-forms $\Omega^1\big(M,\Omega^d(M)\big) $ with
  $\Omega^1(M) \otimes\Omega^d(M)$, and
$$
\CF_{e} \wedge e:= \big(\CF_{e}^{a_{1}\cdots a_{d-2}\, k}\wedge
\eta_{kl}\,  e^{l}\big) \, {\tt E}_{a_{1}}\wedge \cdots \wedge {\tt
  E}_{a_{d-2}} \ ,
$$
that is, one uses the flat metric to identify $e$ with an element of
$\Omega^{1}(M, (\FR^{p,q})^{\ast})$ and then contracts with the
multivector part in $\midwedge^{d-1} (\FR^{p,q})$, and takes the exterior
product of the differential forms.\footnote{For $d=3$ this reduces to the action of $\mathfrak{so}(p,q)\simeq \midfwedge^{2}(\FR^{p,q})$ on $\FR^{p,q}$.}
  
The converse of Noether's second theorem can be used to work backwards from this identity to deduce that the action functional \eqref{eq:ECPaction} has local pseudo-orthogonal and diffeomorphism gauge symmetries~\cite{Montesinos:2018ujm}.
The first component is the Noether identity
corresponding to local diffeomorphism invariance
$\delta_\xi S_{\textrm{\tiny ECP}}=0$. The second component gives the Noether
identity corresponding to the local pseudo-orthogonal gauge symmetry
$\delta_\rho S_{\textrm{\tiny ECP}}=0$, which also follows from the first Bianchi identity in \eqref{eq:BianchiTR} by taking the covariant
derivative of the second Euler--Lagrange derivative in \eqref{eq:ECPeom}:
\begin{align*}
\dd^\om\CF_\om &= (d-3)\, \dd^\om e\dwedge
                 e^{d-4}\dwedge T + (-1)^{d-3}\, e^{d-3} \dwedge \dd^\om T \\[4pt]
&= (d-3)\, T\dwedge  e^{d-4} \dwedge T + (-1)^{d-3}\, e^{d-3}\dwedge R\wedge e\\[4pt]
&= (-1)^{d-3}\, \frac{d-1}2\, \big(\CF_e-\Lambda\, e^{d-1}\big)\wedge e \\[4pt]
&= (-1)^{d-3}\, \frac{d-1}2\, \CF_e\wedge e \ ,
\end{align*}
where in the third equality we used $T\dwedge T=0$ and the
first Euler--Lagrange derivative from \eqref{eq:ECPeom}, while in the last
equality we used $(e\dwedge e)\wedge e =
(e^a\wedge\eta_{bc}\,e^b\wedge e^c)\,{\tt E}_a=0$. The overall
prefactor appears due to our convenient conventions on
$\midwedge^{\bullet}(\FR^{p,q})$, whereby the contraction in question
may be checked to act as a derivation with respect to the exterior product up to the overall factors.
For further discussion of 
Noether's second theorem in the first order formalism for gravity,
see~\cite{Barnich:2016rwk}.

\section{Topological gauge theories with redundant symmetries}
\label{sec:TQFT}

The ECP action functional \eqref{eq:ECPaction} bears a striking
similarity with the action functionals of a class of topological field
theories known as $BF$ theories. This observation has been the source
of intense investigations into the formulation of gravity as a
deformation of $BF$ theories, and particularly in certain approaches
to quantum gravity (see e.g.~\cite{Freidel:2012np} for a review). We
will come back to these connections in Sections~\ref{sec:3dgrav} and~\ref{sec:4dgrav}.
Before coming to the $L_\infty$-algebra formulation of the ECP gravity
theory from Section~\ref{sec:ECPgravity}, it is therefore a useful warm-up to look at some simpler examples of Schwarz-type topological gauge theories\footnote{See e.g.~\cite{Birmingham:1991ty} for a review.} whose dynamical $L_\infty$-algebras are straightfoward to formulate, yet they involve many features of the more complicated $L_\infty$-algebras that we look at in later sections. They also serve to illustrate some important points concerning the roles of diffeomorphisms, redundant symmetries, and classical equivalences between field theories, that will be particularly pertinent to the discussions in Sections~\ref{sec:3dgrav} and~\ref{sec:4dgrav}.

\subsection{Chern--Simons theory in the $L_\infty$-algebra formalism}
\label{sec:CStheory}

Chern--Simons theory in three spacetime dimensions provides
the basic example of an $L_\infty$-algebra in gauge theory with a very natural and
simple bracket structure, see e.g.~\cite{Linfty,BVChristian}. For
later use, we shall spell out the details and use them to illustrate
how redundant symmetries of a field theory are naturally handled by the $L_\infty$-algebra
framework. 

Let $\sG$ be a Lie group
whose Lie algebra $\frg$ is endowed with an invariant quadratic form
$\Tr_\frg:\frg\otimes\frg\to\FR$; invariance means
$\Tr_\frg([X,Y]_\frg\otimes Z)=\Tr_\frg(X\otimes[Y,Z]_\frg)$ for all
$X,Y,Z\in\frg$, where $[-,-]_\frg$ is the Lie bracket in $\frg$. Let $\CCP\to M$ be a principal
$\sG$-bundle on an oriented
three-dimensional manifold $M$, which for simplicity we assume to be trivial,
$\CCP=M\times\sG$, so that its connections can be regarded as one-forms
on $M$ with values in $\frg$; this restriction will be enough for our
purposes later on. The Lie bracket on $\frg$ is extended to $\Omega^\bullet(M,\frg):=\Omega^\bullet(M)\otimes\frg$ by
\begin{align*}
[\alpha\otimes X,\beta\otimes Y]_\frg := \alpha\wedge \beta \otimes [X,Y]_\frg \ .
\end{align*}

The Chern--Simons action functional
for a gauge field $A\in\Omega^1(M,\frg)$ is then defined by \footnote{Here and in BF theory the wedge symbol within the Trace pairings denotes simply wedging the spacetime form parts, in contrast to the ECP convention.}
\begin{align} \label{eq:CSaction}
S_{\textrm{\tiny CS}}(A) := \frac12\, \int_M\, \Tr_\frg\Big(A\wedge\dd A +
\frac13\, A\wedge [A, A]_\frg\Big) \ .
\end{align}
This action functional is invariant under the
gauge transformations 
\begin{align} \label{eq:CSgaugetransf}
\delta_\lambda A=\dd^A\lambda:=\dd\lambda+[A,\lambda]_\frg \ ,
\end{align}
where $\lambda\in\Omega^0(M,\frg)$.
The
Chern--Simons field equations state that the $\sG$-connection $A$
is flat, that is, its curvature vanishes:
\begin{align} \label{eq:CSeom}
\CF(A):=F=\dd^A
A=\dd A+\tfrac12\,[A,A]_\frg=0 \ \in \ \Omega^2(M,\frg) \ ,
\end{align}
while the Noether identities are equivalent to the Bianchi
identity\footnote{Note that this is in contrast to what happens in
  Einstein--Cartan--Palatini theory, where the first Bianchi identity
  in \eqref{eq:BianchiTR} implies the Noether
identity for local pseudo-orthogonal rotations, but is not generally equivalent to it.}
\begin{align} \label{eq:CSNoether}
\dsf_A \CF:=\dd^A F=\dd F+ [A, F]_\frg=0 \ \in \ \Omega^3(M,\frg) \ .
\end{align}
The classical moduli space $\CCM_{\textrm{\tiny CS}}$ of physical states of Chern--Simons
theory is thus the moduli space of flat $G$-connections on the
three-manifold $M$.

The cochain complex underlying the $L_\infty$-algebra of Chern--Simons gauge theory is simply the de~Rham
complex in three dimensions with values in the Lie algebra $\frg$:
\begin{align*}
0 \xrightarrow{ \ \ \ \ \ } \Omega^0(M,\frg) \xrightarrow{ \ \ \dd \ \ } \Omega^1(M,\frg)
  \xrightarrow{ \ \ \dd \ \ } \Omega^2(M,\frg) \xrightarrow{ \ \ \dd \
  \ }
  \Omega^3(M,\frg) \xrightarrow{ \ \ \dd \ \ } 0 \ .
\end{align*}
The corresponding graded vector space 
$$
V^{\textrm{\tiny CS}}=\Omega^\bullet(M,\frg) = \Omega^0(M,\frg) \ \oplus \
\Omega^1(M,\frg) \ \oplus \ \Omega^2(M,\frg) \ \oplus \ \Omega^3(M,\frg)
$$
has non-zero graded
components $V^{\textrm{\tiny CS}}_k=\Omega^k(M,\frg)$ for $k=0,1,2,3$, whose elements we
denote respectively by $\lambda$, $A$, $\CA$ and $\mit\Lambda$. This yields a
four-term $L_\infty$-algebra with $1$-bracket defined by the exterior
derivative as
$$
\ell^{\textrm{\tiny CS}}_1(\lambda) = \dd\lambda \ \in \
V_1^{\textrm{\tiny CS}} \ , \quad \ell_1^{\textrm{\tiny CS}}(A) = \dd
A \ \in \ V_2^{\textrm{\tiny CS}}\qquad
\mbox{and} \qquad \ell^{\textrm{\tiny CS}}_1(\CA) = \dd \CA \ \in \
V_3^{\textrm{\tiny CS}} \ .
$$
The $2$-brackets are given by the Lie bracket of $\frg$ as
\begin{align*}
\ell^{\textrm{\tiny CS}}_2(\lambda_1,\lambda_2) &=
                                                  -[\lambda_1,\lambda_2]_\frg
                                                  \ \in \
                                                  V_0^{\textrm{\tiny
                                                  CS}} \ , \\[4pt]
\ell^{\textrm{\tiny CS}}_2(\lambda,A) &= -[\lambda,A]_\frg \ \in \
                                        V_1^{\textrm{\tiny CS}} \ , \\[4pt]
\ell^{\textrm{\tiny CS}}_2(\lambda,\CA) &= -[\lambda,\CA]_\frg \ \in \
                                          V_2^{\textrm{\tiny CS}} \ , \\[4pt]
\ell^{\textrm{\tiny CS}}_2(\lambda,{\mit\Lambda}) &=
                                                    -[\lambda,{\mit\Lambda}]_\frg
                                                    \ \in \
                                                    V_3^{\textrm{\tiny
                                                    CS}} \ ,
  \\[4pt]
\ell^{\textrm{\tiny CS}}_2(A_1,A_2) &= -[A_1, A_2]_\frg \ \in \
                                      V_2^{\textrm{\tiny CS}} \ , \\[4pt]
\ell^{\textrm{\tiny CS}}_2(A,\CA) &= - [A, \CA]_\frg \ \in \
                                    V_3^{\textrm{\tiny CS}} \ ,
\end{align*}
with all other brackets vanishing. Thus the gauge theory is organised by a differential graded Lie algebra: The homotopy relations in this case
easily follow from the nilpotence and Leibniz rule for the exterior
derivative $\dd$, together with the Jacobi identity for the Lie
bracket of $\frg$. One easily verifies the kinematical
and dynamical structure of Chern--Simons gauge theory from these
brackets as designed by the prescription of
Section~\ref{sec:Linftygft}; in particular the gauge
transformations, field equations and Noether identities are encoded as
\begin{align*}
\delta_\lambda A&=\ell_1^{\textrm{\tiny
                  CS}}(\lambda)+\ell_2^{\textrm{\tiny CS}}(\lambda,A)
                  \ , \\[4pt]
F&=\ell_1^{\textrm{\tiny CS}}(A)-\frac12\,\ell_2^{\textrm{\tiny
   CS}}(A,A) \ , \\[4pt]
\delta_\lambda F &= \ell_2^{\textrm{\tiny CS}}(\lambda,F) \ , \\[4pt]
\dd^A F&=\ell_1^{\textrm{\tiny CS}}(F)-\ell_2^{\textrm{\tiny CS}}(A,F)
         \ .
\end{align*}

This can be made into a cyclic
$L_\infty$-algebra by defining the pairing of $\frg$-valued forms in
complementary degrees:
\begin{align} \label{eq:CSpairing}
\langle\alpha,\beta\rangle^{\textrm{\tiny CS}} := \int_M\,
\Tr_\frg(\alpha\wedge\beta) \ ,
\end{align}
for $\alpha\in\Omega^p(M,\frg)$ and $\beta\in\Omega^{3-p}(M,\frg)$
with $p=0,1,2,3$. This defines a cyclic non-degenerate pairing (of
degree~$-3$) on $V_1^{\textrm{\tiny CS}}\otimes V_2^{\textrm{\tiny
    CS}}\to\FR$ and $V_0^{\textrm{\tiny CS}}\otimes V_3^{\textrm{\tiny
    CS}}\to\FR$, where cyclicity follows from the invariance of the
quadratic form on the Lie algebra $\frg$. The adjoint of the exterior
derivative $\dd:\Omega^p(M,\frg)\to\Omega^{p+1}(M,\frg)$ with respect
to this inner product is $(-1)^{p+1}\,\dd$, which yields the
appropriate cyclicity structure from which the Chern--Simons action
functional \eqref{eq:CSaction} is reproduced according to the general
prescription of Section~\ref{sec:Linftygft}:
\begin{align*}
S_{\textrm{\tiny CS}}(A) = \frac12\,\langle A,\ell_1^{\textrm{\tiny
  CS}}(A)\rangle^{\textrm{\tiny CS}} -\frac1{3!}\,\langle
  A,\ell_2^{\textrm{\tiny CS}}(A,A)\rangle^{\textrm{\tiny CS}} \ .
\end{align*}
In this sense, the
$L_\infty$-algebra formulation of field theories is a generalization
of Chern--Simons theory, which is dual to the BV formalism for
Lagrangian field theories; see~\cite{BVChristian} for further details
on this perspective. 

\subsection{Diffeomorphisms as redundant symmetries}
\label{sec:diffredundant}

Chern--Simons theory is also the prototypical example of a topological
field theory: Its field equation simply state that the
$\sG$-connection $A$ is flat. In addition, the theory is background independent:
There is no background structure assumed on the manifold $M$, aside
from its orientation. Thus the action functional is constructed
solely of differential forms built from the dynamical field, and
hence it is automatically invariant under orientation-preserving
diffeomorphisms of $M$.\footnote{Orientation-\emph{reversing} diffeomorphisms
change the sign of the Chern--Simons action functional.} In
particular, the action of any infinitesimal diffeomorphism $\xi \in
\Gamma(TM)$ on a connection $A$ is via the Lie derivative: 
\begin{align}\label{eq:deltaxiA}
\delta_{\xi} A := \LL_{\xi} A \ ,
\end{align}
which leaves the Chern--Simons action functional \eqref{eq:CSaction}
invariant:
\begin{align*}
\delta_{\xi} S_{\textrm{\tiny CS}}=0 \ .
\end{align*} 
The corresponding Noether identity may then be read off by using the Cartan formula \eqref{eq:CartanLie} and integrating by parts to get the strong differential identity
\begin{align*}
\dd x^\mu\otimes\Tr_\frg(\iota_\mu\dd A\wedge F)-\dd x^\mu\otimes\Tr_\frg(\iota_\mu A\otimes\dd F) = 0 \ \in \ \Omega^1\big(M,\Omega^3(M)\big) \ .
\end{align*}

One may wonder how to reconcile this apparently ``extra'' symmetry
with the usual gauge symmetry of Chern--Simons theory, and whether this
can be incorporated in the $L_{\infty}$-algebra framework. Of course,
the answer to the first question is well-known to experts. The crucial
point is that the action of any vector field $\xi \in \Gamma(TM)$ may
be compensated by the action of a specially chosen gauge
transformation $\lambda_{\xi} \in \Omega^{0}(M,\frg)$, up to a term
proportional to the field equations. To see this, note that the
variation \eqref{eq:deltaxiA} can be written as
\begin{align*}
\delta_{\xi} A&= \dd \iota_{\xi} A + \iota_{\xi} \dd A \\[4pt]
&= \dd \iota_{\xi} A + \iota_{\xi} \dd A  +\tfrac{1}{2}\,
  \iota_{\xi}[A,A]_\frg - \tfrac{1}{2}\,\iota_{\xi}[A,A]_\frg \\[4pt]
&=\dd \iota_{\xi} A + [A,\iota_{\xi}A]_\frg +\iota_{\xi} F \\[4pt]
&=: \dd^{A} \lambda_{\xi} +\iota_{\xi} \CF(A)
\end{align*}
 where we defined the ``field dependent'' gauge transformation
 $\lambda_{\xi}:= \iota_{\xi} A \, \in \Omega^{0}(M,\frg)$. Thus the
 gauge orbits of the $\Gamma(TM)$-action are included in the gauge
 orbits of the standard gauge transformations, on-shell. In other
 words, the traditional moduli space of physical states for
 Chern--Simons theory $\CCM_{\textrm{\tiny CS}}$ is unchanged if one
 further quotients by the action of (infinitesimal)
 diffeomorphisms.\footnote{This statement holds at the
   infinitesimal level, but fails for finite diffeomorphisms which
   are not connected to the identity.} One then declares the extra
 symmetries to be \emph{redundant} -- they do not extend the
 distribution $\CCD$ of gauge transformations on the space of fields.
 
We will now demonstrate that the $L_\infty$-algebra framework actually
serves as a useful tool for encoding redundant symmetries and 
their relation to the smaller Lie algebra generating the distribution
$\CCD$ of symmetries of the corresponding generalized gauge
theory. Indeed, the vector space $V^{\textrm{\tiny CS}}$ underlying the $L_{\infty}$-algebra of Chern--Simons theory may be extended as
 \begin{align}\label{eq:CSVext}
 V_{\textrm{\tiny CS}}^{\textrm{ext}}:=
  \Gamma(TM)\times \Omega^{0}(M,\frg) \ \oplus \
   \Omega^1(M,\frg) \ \oplus \ \Omega^2(M,\frg) \ \oplus \
   \Omega^{1}\big(M,\Omega^{3}(M)\big)\times \Omega^{3}(M,\frg) \ .
 \end{align}
That is, we simply extend the space of gauge transformations $V^{\textrm{\tiny CS}}_{0}$ to include the extra diffeomorphism gauge symmetry, with elements $\xi\in\Gamma(TM)$, while also extending the space $V_{3}^{\textrm{\tiny CS}}$ to accommodate for the corresponding Noether identity, with elements $\CX\in\Omega^1\big(M,\Omega^3(M)\big)$. The brackets are then modified to accommodate for the action of diffeomorphisms on the various spaces:
 The $1$-bracket is modified trivially as
\begin{align}\label{eq:CSell1ext}
 \ell^{\textrm{ext}}_1(\xi, \lambda) = \dd\lambda \ , \quad \ell_1^{\textrm{ext}}(A) = \dd A \qquad
 \mbox{and} \qquad \ell^{\textrm{ext}}_1(\CA) = (0,\dd \CA) \ ,
\end{align}
 while the modified $2$-brackets are given by
 \begin{align}
 \ell^{\textrm{ext}}_2\big((\xi_{1},\lambda_1)\,,\,(\xi_{2},\lambda_2)\big) &= \big([\xi_{1},\xi_{2}]\,,\, \LL_{\xi_{1}} \lambda_{2} -\LL_{\xi_{2}} \lambda_{1}-[\lambda_1,\lambda_2]_\frg \big)\ , \nn \\[4pt]
 \ell^{\textrm{ext}}_2\big((\xi,\lambda)\,,\,A\big) &= \LL_{\xi} A-[\lambda,A]_\frg \ , \nn \\[4pt]
 \ell^{\textrm{ext}}_2\big((\xi,\lambda)\,,\,\CA\big) &= \LL_{\xi} \CA-[\lambda,\CA]_\frg \ , \label{eq:CSell2ext} \\[4pt]
 \ell^{\textrm{ext}}_2\big((\xi,\lambda)\,,\,(\CX ,{\mit\Lambda})\big) &=\big
 (\dd x^\mu\otimes\Tr_{\frg}(\iota_\mu \dd
 \lambda\otimes {\mit\Lambda})
 +\LL_{\xi}{\CX}\,,\,-[\lambda,\mit\Lambda]_{\frg} +\LL_{\xi} {\mit\Lambda}\big) \ ,
 \nn \\[4pt]
 \ell^{\textrm{ext}}_2(A_1,A_2) &= -[A_2, A_1]_\frg \ , \nn \\[4pt]
 \ell^{\textrm{ext}}_2(A,\CA) &= \big( \dd x^{\mu}\otimes
                                \Tr_{\frg}(\iota_{\mu} \dd A\wedge
                                \CA)-\dd x^{\mu}\otimes
                                \Tr_{\frg}(\iota_{\mu}A\otimes
                                \dd\CA)\,,\,- [A, \CA]_\frg\big) \ . \nn
 \end{align}
In particular, the first $2$-bracket is the Lie bracket for the
extended semi-direct product gauge algebra
$\Gamma(TM)\rtimes\Omega^0(M,\frg)$. The proof of the homotopy
relations for these brackets is formally identical to the proof we
present for the brackets of three-dimensional gravity in
Appendix~\ref{app:3dhomotopy}. These brackets encode all the dynamical data of
the gauge theory as prescribed in Section~\ref{sec:Linftygft}, now
including the action of diffeomorphisms and the corresponding Noether
identity. 

Since $V_1^{\rm ext}=V_1^{\textrm{\tiny CS}}$ and $V_2^{\rm
  ext}=V_2^{\textrm{\tiny CS}}$, the cyclic pairing
$\langle-,-\rangle^{\rm ext}$ is given by \eqref{eq:CSpairing} on
$V_1^{\rm ext}\otimes V_2^{\rm ext}$. It is further extended on $V^{\textrm{ext}}_{0}\otimes V^{\textrm{ext}}_{3}$ by defining 
$$
\langle (\xi,\lambda)\,,\, (\CX,{\mit\Lambda}) \rangle^{\textrm{ext}}
:= \int_{M}\, \iota_{\xi} \CX + \int_M\, \Tr_\frg(\lambda\otimes{\mit\Lambda}) \ .
$$ 
The extended brackets may then be easily checked to be cyclic with respect to
this pairing as well; the calculation in question is carried out in
the case of ECP gravity later on.

\subsection{$L_\infty$-quasi-isomorphism}
 
Although the moduli spaces of classical
solutions in the two $L_{\infty}$-algebras from
Sections~\ref{sec:CStheory} and~\ref{sec:diffredundant} are the same, the
cohomologies of the cochain complexes generated by the differentials
$\ell_{1}^{\textrm{\tiny CS}}$ and $\ell_{1}^{\textrm{ext}}$ are not isomorphic; this is immediately apparent by
looking at the cohomology in degree~$0$ for the two
$L_\infty$-algebras: 
$$
H^0\big(V ^{\textrm{\tiny CS}},\ell_1 ^{\textrm{\tiny CS}}\big) =
H^0(M,\frg) \qquad \mbox{and} \qquad H^0\big(V_{\textrm{\tiny CS}}^{\rm
  ext},\ell_1 ^{\textrm{ext}}\big) =
\Gamma(TM)\times H^0(M,\frg)  \ .
$$
This means that there cannot exist an
$L_\infty$-quasi-isomorphism between the two
$L_\infty$-algebras. However, there is no contradiction here: While
two field theories with quasi-isomorphic $L_{\infty}$-algebras have
isomorphic classical moduli spaces, the converse need not necessarily
hold.

We can nevertheless describe the classical equivalence between the two $L_\infty$-algebra
formulations of Chern--Simons gauge theory as an $L_\infty$-quasi-isomorphism in the
following way. The redundancy may be encoded in the
$L_{\infty}$-algebra framework by further extending the complex
defined by the vector space $V_{\textrm{\tiny CS}}^{\rm ext}$. This
is done by introducing a copy of the redundant subspace in degree $-1$: 
$$
V_{-1}^{\textrm{ext}}:= \Gamma(TM) \ .
$$
As with the rest of the complex, this should be supplemented with
its dual in degree~$4$: 
$$
V_{4}^{\textrm{ext}}:= \Omega^1\big(M,\Omega^{3}(M)\big) \ .
$$
We denote elements of $V^{\textrm{ext}}_{-1}$ by $\check{\xi}$ and elements of
 $V^{\textrm{ext}}_{4}$ by $\check{\CX}$ to distinguish them from their copies in
 degrees~$0$ and~$3$, respectively. The issue with the cohomology of
 $\ell_1^{\textrm{ext}}$ is then fixed by extending it as the
 inclusion $\ell_1^{\textrm{ext}}:V^{\textrm{ext}}_{-1}\to V^{\textrm{ext}}_0$ 
 and the projection $\ell_1^{\textrm{ext}}:V^{\textrm{ext}}_3\to V^{\textrm{ext}}_4$:
 \begin{align*}
 \ell_{1}^{\textrm{ext}}\big(\check{\xi}\, \big):= \big(\check{\xi},0\big) \ \in \
                                  V_{0}^{\rm ext} \qquad \mbox{and} \qquad
 \ell_{1}^{\textrm{ext}} (\CX,{\mit\Lambda}):= \CX \ \in \ V_{4}^{\rm ext} \ .
 \end{align*}
The differential condition  $\ell_{1}^{\textrm{ext}}\circ
\ell_{1}^{\textrm{ext}}=0$ is still satisfied on the extended
complex, but now the cohomologies agree as expected:
$H^\bullet(V^{\textrm{\tiny CS}},\ell_1^{\textrm{\tiny
    CS}})=H^\bullet(V_{\textrm{\tiny CS}}^{\rm ext},\ell_1^{\rm ext})$.

To complete the $L_{\infty}$-algebra extension, one should extend
the $2$-bracket $\ell_{2}^{\textrm{ext}}$ while still satisfying
the remaining homotopy relations. It is easy to see that the following
definition does the job:
\begin{align*}
\ell_{2}^{\textrm{ext}}\big(\check{\xi}\,,\,(\xi,\lambda)\big)&:=\big[\check{\xi},\xi\big]
                                               \ \in \ V ^{\textrm{ext}}_{-1} \ , \\[4pt]
\ell_{2}^{\textrm{ext}} \big(\check{\xi},A\big)&:=
                                                      \big(0,\iota_{\check{\xi}}
                                                      A\big) \ \in \
                                                      V ^{\textrm{ext}}_{0} \ , \\[4pt]
\ell_{2}^{\textrm{ext}} \big(\check{\xi},\CA\big)&:=
                                                        \iota_{\check{\xi}}
                                                        \CA \ \in \
                                                        V ^{\textrm{ext}}_{1} \ , \\[4pt]
\ell_{2}^{\textrm{ext}}\big(\check{\xi}\,,\,(\CX,{\mit\Lambda})\big)&:=\iota_{\check{\xi}}
                                               {\mit\Lambda} \ \in  \
                                               V ^{\textrm{ext}}_{2} \ , \\[4pt]
\ell_{2}^{\textrm{ext}} \big(\check{\xi}, \check{\CX}\big)&:=
                                                                 \big(\LL_{\check{\xi}}
                                                                 \check{\CX},0\big)
                                                                 \ \in
                                                                 \
                                                                 V^{\rm ext}_{3}
                                                                 \ , \\[4pt]
\ell_{2}^{\textrm{ext}}\big((\xi,\lambda)\,,\,\check{\CX}\big)&:=\LL_{\xi} \check{\CX}
                                               \ \in \ V ^{\textrm{ext}}_{4}  \ , \\[4pt]
\ell_{2}^{\textrm{ext}}\big(A\,,\,(\CX,{\mit\Lambda})\big)&:=\dd
                                                                 x^{\mu}\otimes
                                                                 \Tr_{\frg}(\iota_{\mu}A\otimes{\mit\Lambda})
                                                                 \ \in
                                                                 \
                                                                 V^{\textrm{ext}}_{4}
                                                                 \ , \\[4pt]
\ell_{2}^{\textrm{ext}} (\CA_{1},\CA_{2})&:=\dd x^{\mu}\otimes
                                                \Tr_{\frg}(\iota_{\mu}\CA_{1}\wedge\CA_{2})
                                                \ \in \ V ^{\textrm{ext}}_{4} \ ,
\end{align*}
where also $\ell_{2}^{\textrm{ext}}\big(\check{\xi}_{1},\check{\xi}_{2}\big):=0$ as
this lands in $V ^{\textrm{ext}}_{-2}=0$. The pairing is further
extended to $V ^{\textrm{ext}}_{-1}\otimes V ^{\textrm{ext}}_4$ by 
$$
\big\langle \check{\xi}, \check{\CX} \big\rangle ^{\textrm{ext}} := \int_{M}\,
\iota_{\check{\xi}} \check{\CX} \ , 
$$
and the new brackets are cyclic with respect to this pairing as
well. The proofs of the homotopy relations in this case follow exactly
as for the calculations we spell out for gravity in
Appendix~\ref{app:3dhomotopy}, with the only non-trivial (but straightforward) check
occuring in the graded Jacobi identity on a pair of gauge parameters
$(\xi,\lambda)$, and two Euler--Lagrange derivatives $\CA_1$ and
$\CA_2$.

The redundancy of the space of gauge transformations $V_0^{\rm ext}=\Gamma(TM)\times \Omega^{0}(M,\frg)$
should really be regarded in terms of the ``higher gauge transformations''
\eqref{eq:highergauge} which act on it by the redundant symmetries in $V^{\rm
  ext}_{-1}$: Any $(\xi,\lambda) \in V^{\rm ext}_{0}$ generates the gauge
transformation 
$$
\delta^{\rm ext}_{(\xi,\lambda)} A = \LL_{\xi}A + \dd^{A}\lambda \ ,
$$ 
which is on-shell equivalent to the gauge transformation generated by 
$$
(\xi',\lambda'):=(\xi,\lambda) +\delta^{\rm ext}_{(\check{\xi},A)}(\xi,\lambda)
$$ 
for any $\check{\xi} \in V^{\rm ext}_{-1}$ with 
\begin{align*}
\delta^{\rm ext}_{(\check{\xi},A)} (\xi,\lambda):= \ell ^{\textrm{ext}}_{1}\big(\check{\xi}\, \big)
                                    - \ell ^{\textrm{ext}}_{2}\big(A,\check{\xi}\,
  \big)=\big(\check{\xi},-\iota_{\check{\xi}}A\big) \ \in \
  V^{\textrm{ext}}_{0} \ .
\end{align*}
Given this equivalence, all vector fields are redundant in the
following sense: For any $(\xi,\lambda) \in V^{\rm ext}_{0}$, pick
$\check{\xi}:= -\xi \in V^{\rm ext}_{-1}$. Then 
\begin{align*}
\delta^{\rm ext}_{(\xi',\lambda')}A=\delta^{\rm ext}_{(0,\lambda+\iota_{\xi}A)} A 
\end{align*}
which is equivalent to $\delta^{\rm ext}_{(\xi,\lambda)}A$ up to terms
involving Euler--Lagrange derivatives. 

This on-shell equivalence is made precise in terms of an
\emph{off-shell} $L_\infty$-quasi-isomorphism in the following way.
One can eliminate the redundant gauge symmetries by using a
quasi-isomorphism $\{\psi_n ^{\textrm{\tiny CS}}\}$ from the standard Chern--Simons
$L_{\infty}$-algebra of Section~\ref{sec:CStheory} to the extended
$L_\infty$-algebra here, with the component multilinear graded antisymmetric maps
\begin{align*}
\psi_n ^{\textrm{\tiny CS}}:\midwedge^n V^{\textrm{\tiny
  CS}}\longrightarrow V_{\textrm{\tiny CS}}^{\rm ext} \ ,
\end{align*}
of degree $|\psi_n^{\textrm{\tiny CS}}|=1-n$ for $n\geq1$.
 For $\psi ^{\textrm{\tiny CS}}_{1}:V ^{\textrm{\tiny CS}}\to
 V_{\textrm{\tiny CS}}^{\rm ext}$, one uses the canonical embedding of
 the underlying vector space $V^{\textrm{\tiny CS}}$ into
 $V_{\textrm{\tiny CS}}^{\rm ext}$ to define a map of underlying cochain complexes
$$
\xymatrix{
 & & V_0 ^{\textrm{\tiny CS}}\ar[rr]^{\ell_1^{\textrm{\tiny CS}}} \ar[dd]_{\psi_1^{\textrm{\tiny CS}}}& & V^{\textrm{\tiny CS}}_1 \ar[rr]^{\!\!\!\ell_1^{\textrm{\tiny CS}}} \ar[dd]_{\psi^{\textrm{\tiny CS}}_1}
& & V^{\textrm{\tiny CS}}_{2}
\ar[rr]^{\ell^{\textrm{\tiny CS}}_1} \ar[dd]_{\psi^{\textrm{\tiny CS}}_1} \ar[rr]^{\ell_1^{\textrm{\tiny CS}}}& & V_3^{\textrm{\tiny CS}} \ar[dd]_{\psi_1^{\textrm{\tiny CS}}} & & \\
& & & & & & & & & & \\
V_{-1}^{\rm ext}\ar[rr]^{\ell_1^{\rm ext}} & & V_0^{\rm ext} \ar[rr]^{\ell_1^{\rm ext}} & & V_1^{\rm ext} \ar[rr]^{\!\!\!\ell_1^{\rm ext}} & & V_{2}^{\rm ext} \ar[rr]^{\ell_1^{\rm ext}} & & V_3^{\rm ext}
\ar[rr]^{\ell_1^{\rm ext}} & & V_4^{\rm ext}
}
$$
given by
\begin{align*}
\psi ^{\textrm{\tiny CS}}_{1}(\lambda)=(0,\lambda) \ , \quad \psi_1
  ^{\textrm{\tiny CS}} (A)=A \ , \quad \psi_1 ^{\textrm{\tiny
  CS}}(\CA)=\CA \qquad \mbox{and} \qquad \psi_1 ^{\textrm{\tiny
  CS}}({\mit\Lambda})=(0,{\mit\Lambda}) \ .
\end{align*}
For $\psi_{2}^{\textrm{\tiny CS}}:\midwedge^2 V^{\textrm{\tiny CS}}\to V_{\textrm{\tiny
    CS}}^{\rm ext}$, the only non-vanishing components are
\begin{align*}
\psi_{2}^{\textrm{\tiny CS}} (A,{\mit\Lambda})&:=-\big(\dd x^{\mu}\otimes
                           \Tr_{\frg}(\iota_{\mu} \dd \lambda \otimes
                           {\mit\Lambda})\,,\,0\big) \ \in \ V ^{\textrm{ext}}_{3} \ , \\[4pt]
\psi_{2}^{\textrm{\tiny CS}} (\CA_{1},\CA_{2})&:=-\big(\dd x^{\mu}\otimes
                           \Tr_{\frg}(\iota_{\mu}\CA_{1}\wedge\CA_{2})\,,\,0\big)
                           \ \in \ V ^{\textrm{ext}}_{3} \ .
\end{align*}
We set $\psi^{\textrm{\tiny CS}}_n=0$ for all $n\geq3$.

It is easy to check that this defines an $L_{\infty}$-morphism, as one
sees by writing out the two sides of the relations
\eqref{eq:morphismrels}. Furthermore, it is obviously a
quasi-isomorphism, since $\psi^{\textrm{\tiny CS}}_{1}$ descends to the identity on the
corresponding cohomology groups $H^\bullet(V ^{\textrm{\tiny
    CS}},\ell_1 ^{\textrm{\tiny CS}}) = H^\bullet(V ^{\rm ext},\ell_1 ^{\textrm{ext}})$, and is easily
checked to be cyclic. In particular, this provides a simple example of
an $L_\infty$-morphism between differential graded Lie algebras which
is not exactly a differential graded Lie algebra morphism, but only up
to homotopy proportional to $\psi^{\textrm{\tiny CS}}_{2}$: As a
quasi-isomorphism of differential graded Lie algebras, it has an
inverse only in the category of $L_\infty$-algebras. This quasi-isomorphism provides a further way to see why infinitessimal diffeomorphisms are safely ignored in the pertubative quantisation of Chern-Simons.

\subsection{$BF$ theories in the $L_\infty$-algebra formalism}
\label{sec:BFtheory}

The next primary example of a topological field theory with the same
properties is a $BF$ theory,
and the exact analogue of the story spelled out thus far in this
section can be easily adapted to $BF$ theories in arbitrary dimension $d$ and for
any gauge group $\sG$.
The $BF$ action functional is given by
\begin{align}\label{eq:BFaction}
S_{\textrm{\tiny BF}}(B,A):=\int_{M}\,\Tr_\CCW( B\wedge F)
\end{align}
where $M$ is an oriented $d$-dimensional manifold, $F= \dd^{A} A$ is
the curvature of a connection one-form $A \in \Omega^{1}(M,\frg)$
valued in a Lie algebra $\frg$, and $B$ is a $(d{-}2)$-form valued in
a fixed vector space $\CCW$ which is a $\frg$-module. The pairing $\Tr_\CCW:
\CCW\otimes \frg \rightarrow \FR$ is assumed to be invariant under the $\frg$-action:
\begin{align*}
\Tr_\CCW\big( (X \cdot w)\otimes Y + w\otimes[X, Y]_\frg\big)
  =0 \ ,
\end{align*}
for $w\in \CCW$ and $X,Y\in\frg$. 
In the conventional definition of $BF$ theory~\cite{Birmingham:1991ty}, one usually takes
$\CCW=\frg$, so that both fields are valued in the same Lie algebra, with $\CCW$
regarded as a $\frg$-module under the canonical adjoint action of the
Lie algebra on itself. However, this more general formulation will act
as a nice stepping stone between different theories suitable for our purposes.

The field equations are readily seen to be 
\begin{align*}
\CF_{B}& := F= \dd A+ \tfrac{1}{2}\,[A,A]_\frg = 0 \ \in \
         \Omega^{2}(M,\frg) \ , \nn \\[4pt]
\CF_{A}&:= \dd^{A}B= \dd B + A\wedge B = 0 \ \in \ \Omega^{d-1}(M,\CCW) \ ,
\end{align*}
where $A\wedge B$ computes the exterior product of the form
components while pairing the components in $\frg$ and $\CCW$ via the
$\frg$-action. Thus the Euler--Lagrange locus of $BF$ theory are pairs
of a flat $\sG$-connection on $M$ and a covariantly constant
$(d{-}2)$-form valued in a representation $\CCW$ of $\frg$.

The action functional \eqref{eq:BFaction} is invariant under standard (infinitesimal) gauge transformations $\rho \in \Omega^{0}(M,\frg)$ acting as 
\begin{align*}
\delta_{\rho} (B,A) =\big(-\rho \cdot B\,,\, \dd \rho + [A,\rho]_\frg\big) \ .
\end{align*}
Compared to Chern--Simons theory, however, for $d\geq3$ there is the extra ``shift'' symmetry generated by $(d{-}3)$-forms $\tau \in \Omega^{d-3}(M,\CCW)$ valued in $\CCW$, which act as
\begin{align} \label{eq:shift}
\delta_{\tau}(B,A):= (\dd^{A}\tau, 0)= (\dd\tau + A\wedge \tau, 0) \ .
\end{align}
This symmetry is on-shell reducible in dimensions $d\geq4$. The corresponding Noether identities coincide with the usual ``Bianchi identities''
\begin{align*}
\dsf_{(B,A)}(\CF_B,\CF_A):=\big((-1)^{d-3}\,\dd^{A}\CF_{B}\,,\,\dd^{A}\CF_{A}-\CF_{B} \wedge B\big) = (0,0) \ \in \ \Omega^3(M,\frg) \times \Omega^{d}(M,\CCW) \ ,
\end{align*}
where the overall sign only serves to eliminate signs in the cyclic pairing and brackets below. 

The cyclic $L_{\infty}$-algebra of $BF$ theory in $d$ dimensions is given by the underlying graded vector space 
\begin{align*}
V^{\textrm{\tiny BF}}:= V_{0}^{\textrm{\tiny BF}} \oplus V_{1}^{\textrm{\tiny BF}} \oplus V_{2}^{\textrm{\tiny BF}} \oplus V_3^{\textrm{\tiny BF}}
\end{align*}
where
\begin{align}\label{eq:ndBF}
V_{0}^{\textrm{\tiny BF}}&=\Omega^{d-3}(M,\CCW) \times \Omega^{0}\big(M,\frg\big) \ , \nn
\\[4pt]
V_{1}^{\textrm{\tiny BF}}&= \Omega^{d-2}(M,\CCW) \times
\Omega^{1}\big(M,\frg\big) \ ,  \\[4pt]
V_{2}^{\textrm{\tiny BF}}&=\Omega^{2}\big(M,\frg\big) \times
\Omega^{d-1}(M,\CCW) \ , \nn \\[4pt]
V_3^{\textrm{\tiny BF}}&= \Omega^{3}(M,\frg)\times \Omega^d(M,\CCW)
\ . \nn
\end{align} 
We denote gauge parameters by $(\tau,\rho)\in V_{0}^{\textrm{\tiny BF}}$, dynamical fields by $(B,A)\in V_{1}^{\textrm{\tiny BF}}$, Euler--Lagrange derivatives by $(\CB,\CA)\in V_{2}^{\textrm{\tiny BF}}$, and Noether identities by $(\CT,\CP)\in V_{3}^{\textrm{\tiny BF}}$. The non-trivial brackets are then
\begin{flalign} 
\ell_{1}^{\textrm{\tiny BF}}(\tau,\rho)&=(\dd \tau,\dd\rho) \ \in \ V_{1}^{\textrm{\tiny BF}} \ , \nn \\[4pt]
\ell_{1}^{\textrm{\tiny BF}}(B,A)&=(\dd A,\dd B) \ \in \ V_{2}^{\textrm{\tiny BF}} \ , \nn
\\[4pt]
\ell_{1}^{\textrm{\tiny BF}}(\CB,\CA)&=(\dd \CB,\dd \CA) \ \in \ V_{3}^{\textrm{\tiny BF}} \ , \nn 
\\[6pt]
\ell_{2}^{\textrm{\tiny BF}}\big((\tau_{1},\rho_{1})\,,\,(\tau_{2},\rho_{2})\big)&=\big(-\rho_{1} \cdot \tau_{2} +\rho_{2} \cdot \tau_{1}\,,\,-[\rho_{1},\rho_{2}]_\frg\big) \ \in \ V_{0}^{\textrm{\tiny BF}} \ , \nn
\\[4pt]
\ell_{2}^{\textrm{\tiny BF}}\big((\tau,\rho)\,,\,(B,A)\big)&=\big(-\rho \cdot B
+A\wedge \tau \,,\, -[\rho,A]_\frg\big) \ \in \ V_{1}^{\textrm{\tiny BF}} \
, \label{eq:ndBFbrackets} \\[4pt]
\ell_{2}^{\textrm{\tiny BF}}\big((\tau,\rho) \,,\, (\CB,\CA)\big)&=\big(-
[\rho, \CB]_\frg \,,\, -\rho\cdot \CA + \CB\wedge \tau \big) \ \in
\ V_{2}^{\textrm{\tiny BF}} \ , \nn \\[4pt]
\ell_{2}^{\textrm{\tiny BF}}\big((\tau,\rho)\,,\,(\CT,\CP)\big)&= \big
(-[\rho,\CT]_\frg\,,\,-\rho \cdot \CP + (-1)^{d-3} \, \CT \wedge \tau
\big) \ \in \ V_3^{\textrm{\tiny BF}} \ , \nn \\[4pt] 
\ell_{2}^{\textrm{\tiny BF}}\big((B_{1},A_{1})\,,\,(B_{2},A_{2})\big)&=-\big([A_{1},A_{2}]_\frg\, ,\, A_{1} \wedge B_{2} +A_{2} \wedge B_{1} \big) \ \in \ V_{2}^{\textrm{\tiny BF}} \ , \nn \\[4pt]
\ell_{2}^{\textrm{\tiny BF}}\big((B,A)\,,\,(\CB,\CA)\big)&=-\big([A,\CB]_\frg\,,\, A\wedge \CA - \CB \wedge B\big)  \ \in \ V_3^{\textrm{\tiny BF}} \ , \nn
\end{flalign}
while all remaining brackets vanish. Thus the dynamics of $BF$ theory in any dimension $d$ is also organised by a differential graded Lie algebra.
The cyclic pairing is given naturally as
\begin{align*}
\langle(B,A)\,,\,(\CB,\CA) \rangle^{\textrm{\tiny BF}} &:= \int_{M}\, \Tr_\CCW( B\wedge \CB + \CA\wedge A) \ , \\[4pt]
\langle(\tau,\rho)\,,\,(\CT,\CP) \rangle^{\textrm{\tiny BF}} &:= \int_{M}\, \Tr_\CCW(\tau\wedge \CT + \rho \wedge \CP ) \ .
\end{align*}

As with Chern--Simons theory, BF theory is invariant under infinitesimal diffeomorphisms, parameterized by vector fields in $\Gamma(TM)$, acting via the Lie derivative. However, these are again redundant: for any $\xi \in \Gamma(TM)$ we compute
\begin{align}
\delta_{\xi}(B,A):\!\!&= (\LL_{\xi} B,\LL_{\xi} A) \nn \\[4pt]
&=(\dd \iota_{\xi} B + \iota_{\xi} \dd B, \dd \iota_{\xi} A +\iota_{\xi} \dd A) \nn\\[4pt]
&=\big(\dd \tau_{\xi} +\iota_{\xi} \dd B + \iota_{\xi}(A\wedge B)-\iota_{\xi}(A\wedge B)\, , \, \dd \rho_{\xi}+ \iota_{\xi} \dd A +\iota_{\xi} [A,A]_\frg - \iota_{\xi}[A,A]_\frg\big) \nn \\[4pt]
&=\big(\dd \tau_{\xi} + \iota_{\xi} \dd^{A} B+A\wedge\tau_{\xi} -\rho_{\xi}\cdot B \, , \, \dd\rho_{\xi} +[A,\rho_{\xi}]_\frg +\iota_{\xi}\dd^{A}A\big )\nn \\[4pt]
&=\big( \dd^{A} \tau_{\xi} - \rho_{\xi}\cdot B +\iota_{\xi} \CF_{A}\, , \, \dd^{A}\rho_{\xi} + \iota_{\xi} \CF_{B} \big) \nn \\[4pt]
&=\delta_{(\tau_{\xi},\rho_{\xi})}(B,A) +(\iota_{\xi} \CF_{A}\, , \, \iota_{\xi} \CF_{B}) \label{eq:BFdiff}
\end{align} 
where we defined $(\tau_{\xi},\rho_{\xi}):= (\iota_{\xi}B,\iota_{\xi} A) \, \in V_{0}^{\textrm{\tiny BF}}.$ Thus again the diffeomorphism symmetry is redundant, as the further quotient of the moduli space of classical solutions $\CCM_{\textrm{\tiny BF}}$ has no effect.

Of course, one may now augment the $BF$ $L_\infty$-algebra with this symmetry and its corresponding Noether identity, as we did for Chern--Simons theory. One can then follow the same procedure by adding the redundancy at $V_{-1}^{\textrm{ext}}$ and its dual at $V_4 ^{\textrm{ext}}$, and finally one ends up with a cyclic $L_{\infty}$-algebra that is quasi-isomorphic to the one constructed here, which did not include the redundant symmetries to begin with. One also observes that the Einstein--Cartan--Palatini action functional \eqref{eq:ECPaction} with $\Lambda=0$ is a special instance of the $BF$ action functional \eqref{eq:BFaction} with $\frg=\mathfrak{so}(p,q)$, $\CCW=\midwedge^{d-2}\FR^{p,q}$ and the dynamical fields $(B,A)=(e^{d-2},\om)$; however, this restriction of the field $B$ to diagonal decomposable forms in $\Omega^{d-2}(M,\midwedge^{d-2}\FR^{p,q})$ breaks the shift symmetry of $BF$ theory in $d\geq4$ dimensions. We will see how to interpret this observation in the $L_\infty$-algebra framework later on.

Here we have worked in a local formulation of the field theory. There
are different perspectives on what the global structure of the fields
$(B,A)$ should be. The usual approach is to view $A$ as a connection
on a principal $\sG$-bundle over $M$ and $B$ as a $(d{-}2)$-form
valued in the vector bundle associated to the representation $\CCW$. In
the case of non-trivial bundles, the diffeomorphism invariance has to
be treated in a somewhat different way, and this will be explored in
the special case of ECP gravity in Sections~\ref{sec:covariant}
and~\ref{sec:ECPiso}. There are also other interpretations, such as
that of higher gauge theory which considers the fields as forms valued
in a strict Lie 2-algebra \cite{BFasHigher}, or equivalently in a 2-term $L_\infty$-algebra. We shall not delve into
this interpretation in the present paper, which however is interesting
in view of the $L_{\infty}$-algebra framework under consideration.

\subsection{Higher shift symmetries}
\label{sec:highershift}

Another new feature of $BF$ theories, compared to Chern--Simons theory, is that they possess additional
redundant symmetries, in addition to diffeomorphisms: The shift
symmetry \eqref{eq:shift} is on-shell reducible in dimensions
$d\geq4$. This means that, strictly speaking, we should also include
in \eqref{eq:ndBF} the negatively-graded vector spaces
$V_{-k}^{\textrm{\tiny BF}} = \Omega^{d-3-k}(M,\CCW)$ for
$k=1,\dots,d-3$, which parameterize ``higher gauge transformations'',
together with their duals $V_{k+3}^{\textrm{\tiny
    BF}}=\Omega^{k+3}(M,\frg)$ and the obvious brackets in
\eqref{eq:ndBFbrackets}. We spell out this out explicitly for
$BF$ theories in the simplest case $d=4$.

Any element $\tau \in \Omega^{1}(M,\CCW)$ generates the shift symmetry
\begin{align*}
\delta_{\tau} B= \dd^{A}\tau \ \in \ \Omega^{2}(M,\CCW) \ . 
\end{align*}
Now take $\tau':=\tau + \dd^{A}\epsi\, \in \Omega^{1}(M,\CCW)$ for any $\epsi \in \Omega^{0}(M,\CCW)$. Then this generates the shift symmetry 
\begin{align*}
\delta_{\tau'} B &= \dd^{A}\tau' \\[4pt]
&= \dd^{A}\tau + (\dd^{A})^{2} \epsi \\[4pt]
&= \delta_{\tau}B + F\cdot \epsi \\[4pt]
&= \delta_{\tau}B + \CF_{B} \cdot \epsi \ ,
\end{align*}
and so the two transformations differ by a term proportional to an
Euler--Lagrange derivative, that is, $\tau$ and $\tau'=: \tau +
\delta_{(\epsi,A)}\tau$  generate the same symmetry on-shell. This
leads to a further redundancy in the subspace of $V^{\textrm{\tiny
    BF}}_{0}$ generating the distribution of shift symmetries on
$V^{\textrm{\tiny BF}}_{1}$; the redundancy lives in the subspace of
covariantly exact one-forms valued in $\CCW$. We may parameterize this by
$V^{\textrm{\tiny BF}}_{-1}:= \Omega^{0}(M,\CCW)$; for $d=4$ this is
enough and no further gauge redundancy in the description exists,
while in higher dimensions the form degrees change and similarly higher-to-higher gauge transformations are required, and so on.

Let us now describe the complete extended $BF$ $L_\infty$-algebra. We
extend the cochain complex by introducing 
\begin{align*}
V^{\textrm{\tiny BF}}_{-1}:= \Omega^{0}(M,\CCW) \qquad \text{and} \qquad
  V^{\textrm{\tiny BF}}_{4}:= \Omega^{4}(M,\frg) \ ,
\end{align*}
and denote the corresponding elements by $\epsi \in V^{\textrm{\tiny
    BF}}_{-1}$ and $\CE \in V^{\textrm{\tiny BF}}_{4}$. The brackets
\eqref{eq:ndBFbrackets} are extended as 
\begin{align*}
\ell_{1}^{\textrm{\tiny BF}}(\epsi)=(\dd \epsi,0 )  \ \in \
  V^{\textrm{\tiny BF}}_{0} \qquad \text{and} \qquad
  \ell_{1}^{\textrm{\tiny BF}}(\CT,\CP)=\dd \CT \ \in \
  V^{\textrm{\tiny BF}}_{4} \ ,
\end{align*}
together with
\begin{align*}
\ell_{2}^{\textrm{\tiny BF}}\big(\epsi\,,\,(\tau,\rho)\big)&= \rho
                                                             \cdot
                                                             \epsi \
                                                             \in \
                                                             V^{\textrm{\tiny
                                                             BF}}_{-1}
                                                             \ , \\[4pt]
  \ell_{2}^{\textrm{\tiny BF}}\big(\epsi\,,\,(B,A) \big)&=-(A\cdot
                                                          \epsi,0) \
                                                          \in \
                                                          V^{\textrm{\tiny
                                                          BF}}_{0} \ ,
  \\[4pt]
  \ell_{2}^{\textrm{\tiny BF}}\big(\epsi\,,\,(\CB,\CA)\big)&=-(\CB
                                                             \cdot
                                                             \epsi,0)
                                                             \ \in \
                                                             V^{\textrm{\tiny
                                                             BF}}_{1}
                                                             \ , \\[4pt]
  \ell_{2}^{\textrm{\tiny BF}}\big(\epsi\,,\,
  (\CT,\CP)\big)&=-(0,\CT\cdot \epsi) \ \in \ V^{\textrm{\tiny
                  BF}}_{2} \ , \\[4pt]
  \ell_{2}^{\textrm{\tiny BF}}(\epsi,\CE)&=(0,\CE \cdot \epsi) \ \in \
                                           V^{\textrm{\tiny BF}}_{3} \
  , \\[4pt]
  \ell_{2}^{\textrm{\tiny
  BF}}\big((\tau,\rho)\,,\,\CE\big)&=-[\rho,\CE]_\frg \ \in \
                                     V^{\textrm{\tiny BF}}_{4} \ , \\[4pt]
  \ell_{2}^{\textrm{\tiny
  BF}}\big((B,A)\,,\,(\CT,\CP)\big)&=[A,\CT]_\frg \ \in \
                                     V^{\textrm{\tiny BF}}_{4} \ , \\[4pt]
  \ell_{2}^{\textrm{\tiny
  BF}}\big((\CB_{1},\CA_{1})\,,\,(\CB_{2},\CA_{2})\big)&=[\CB_{1},\CB_{2}]_\frg
                                                         \ \in \
                                                         V^{\textrm{\tiny
                                                         BF}}_{4} \ .
\end{align*}
The pairing extends naturally as 
\begin{align*}
\langle \epsi, \CE \rangle^{\textrm{\tiny BF}} = \int_{M}\,
  \Tr_{\CCW}(\epsi \wedge \CE) \ .
\end{align*}
All checks of the homotopy and cyclicity relations follow as before without any
genuine novelty.

Observe in contrast with the redundancy of diffeomorphisms, one cannot simply pass to the subcomplex by ``deleting'' $V^{\textrm{\tiny BF}}_{-1}$ and $V^{\textrm{\tiny BF}}_{4}$ via a quasi-isomorphism since $H^{-1}(V ^{\textrm{\tiny BF}},\ell_1 ^{\textrm{\tiny 
BF}}) =\FR$ does not vanish. This is because there is no clear splitting of the gauge parameters $V^{\textrm{\tiny BF}}_{0}$ into reducible and irreducible components, as it happens with diffeomorphisms. It is in this sense that the extented BF complex is crucial in the pertubative BV quantisation to fully resolve degeneracies, while diffeomorphisms may be safely ignored.
\subsection{Three-dimensional $BF$ and Chern--Simons theories}
\label{sec:3dBFCS}

$BF$ theories in three dimensions are particularly interesting in the present context. For $d=3$, the formulation of Section~\ref{sec:BFtheory} is equivalent to the formulation of Section~\ref{sec:CStheory} for a specific Chern--Simons gauge theory: one takes the Lie algebra of the Chern--Simons theory to be the semi-direct product 
$$
\hat{\frg}:= \CCW \rtimes \frg \ , 
$$
where we view the vector space $\CCW$ as an abelian Lie algebra and use the action of $\frg$ on $\CCW$ to define the Lie bracket on $\hat{\frg}$. The invariant non-degenerate pairing on $\CCW\otimes \frg$ extends to $\hat{\frg}\otimes \hat{\frg}$, acting trivially on the $\frg \otimes \frg$ and $\CCW\otimes \CCW$ parts, and by symmetry to the rest. Then any gauge field $\hat{A}\in \Omega^{1}(M,\hat{\frg})$ has a unique decomposition 
$$
\hat{A}=(B,A)
$$ 
with $B\in\Omega^1(M,\CCW)$ and $A\in\Omega^1(M,\frg)$. With these choices, some simple algebra shows that the Chern--Simons action functional coincides with the $BF$ action functional in three dimensions:
\begin{align*}
S_{\textrm{\tiny CS}}(\hat{A})&= \frac{1}{2}\,\int_{M}\,\Tr_{\hat\frg}\Big( \hat{A}\wedge\dd \hat{A} +\frac{1}{3}\, \hat{A}\wedge [\hat{A},\hat{A}]_{\hat\frg} \Big) \nn \\[4pt]
&=\frac{1}{2}\,\int_{M}\, \Tr_{\hat\frg}\Big( (B,A)\wedge(\dd B,\dd A) +\frac{1}{3}\,\big(B,A\big)\wedge\big(A\wedge B+ A\wedge B,[A,A]_\frg \big) \Big) \\[4pt]
&= \frac{1}{2}\,\int_{M}\,\Tr_\CCW\Big( 2 \, B\wedge\dd A +\frac{2}{3}\, A\wedge( A\wedge B) +\frac{1}{3} \, B \wedge [A,A]_\frg \Big) \\[4pt]
&= \int_{M}\,\Tr_\CCW\Big(B\wedge \dd A + \frac{1}{2}\, B\wedge [A,A]_\frg\Big) \\[4pt]
&= \int_{M}\, \Tr_\CCW( B\wedge F) \\[4pt] &= S_{\textrm{\tiny BF}}(B,A) \ .
\end{align*}

The gauge symmetries map into each other as expected: the gauge transformation
\begin{align*}
\delta_{(\tau,\rho)}(B,A)=\big(\delta_{\tau} B + \delta_{\rho}B\,,\, \delta_\tau A+ \delta_{\rho} A\big):=\big(\dd \tau + A\cdot  \tau - \rho\cdot B\,,\,\dd \rho + [A,\rho]_\frg \big)
\end{align*}
with $\hat\lambda=(\tau,\rho)\in\Omega^0(M,\hat\frg)$ maps to 
\begin{align*}
\delta_{\hat\lambda} \hat{A} :\!\!&= \dd\hat\lambda +[\hat{A}, \hat\lambda]_{\hat\frg} \\[4pt]
&=(\dd\tau, \dd \rho) +\big[(B,A)\,,\, (\tau,\rho)\big]_{\hat\frg} \\[4pt] 
&=\big(\dd \tau + A\cdot\tau -\rho\cdot B\,,\,\dd\rho +[A,\rho]_\frg\big) \ .
\end{align*}
The two field theories are essentially related by identity type redefinitions. The same holds of course at the level of the underlying $L_{\infty}$-algebras, which are (strictly) isomorphic: the isomorphism is given by $\hat\psi_{1}: V^{\textrm{\tiny BF}} \rightarrow V^{\textrm{\tiny CS}}$ as above, collecting the respective fields in each degree. In this sense, higher-dimensional $BF$ theory is one way of generalizing Chern--Simons gauge theory to higher dimensions.

\section{Einstein--Cartan--Palatini ${L_\infty}$-algebras}
\label{sec:ECPLinftyalg}

We are now ready to move on to the main constructions of this paper: the cyclic $L_\infty$-algebras underlying the ECP formalism from Section~\ref{sec:ECPgravity}.

\subsection{Brackets}
\label{sec:LinftyECP}

We will first write down the general form of the $L_{\infty}$-algebra
structure for ECP gravity in an arbitrary
dimensionality $d>2$ discussed in Section~\ref{sec:ECPgravity} for the case when the spacetime $M$ is parallelizable. Recalling the $d$-dimensional field equations
\eqref{eq:ECPeom} and Noether identities \eqref{eq:ECPNoether}, the vector space $V$ underlying the corresponding
$L_\infty$-algebra is 
\begin{align}\label{eq:ECPvectorspace}
V:= V_{0} \oplus V_{1} \oplus V_{2}\oplus V_3
\end{align}
where 
\begin{align*}
V_{0}&=\Gamma(TM)\times \Omega^{0}\big(M,\mathfrak{so}(p,q)\big) \ ,  \\[4pt] 
V_{1}&= \Omega^{1}(M,\FR^{p,q}) \times
       \Omega^{1}\big(M,\mathfrak{so}(p,q)\big) \ , \\[4pt]
V_{2}&=\Omega^{d-1}\big(M,\midwedge^{d-1}(\FR^{p,q})\big)\times
       \Omega^{d-1}\big(M,\midwedge^{d-2}(\FR^{p,q})\big) \ , \\[4pt]
V_3&= \Omega^{1}\big(M,\Omega^{d}(M)\big)
                                    \times
                                    \Omega^d\big (M,\midwedge^{d-2}(\FR^{p,q})\big)
       \ .
\end{align*} 
In the following we denote gauge parameters by $(\xi,\rho)\in V_0$,
dynamical fields by $(e,\omega)\in V_1$, Euler--Lagrange derivatives by
$(E,{\mit\Omega})\in V_2$, and Noether identities by $({\CX},{\CP})\in V_3$.

The brackets on $V$ may then be given as follows. 
The non-vanishing $1$-brackets are defined by 
\begin{align*}
	\ell_{1}(\xi,\rho)=(0,\dd\rho) \ \in \ V_{1} \  \qquad \mbox{and} \qquad \ell_{1}(E,{\mit\Omega})=(0,(-1)^{d-1}\, \dd {\mit\Omega}) \ \in \ V_{3} \ ,
\end{align*}
while the non-vanishing $2$-brackets are defined by
\begin{align}
	\ell_{2}\big((\xi_{1},\rho_{1})\,,\,(\xi_{2},\rho_{2})\big)&=\big(
[\xi_{1},\xi_{2}]\,,\,-[\rho_{1},\rho_{2}]+{\xi_{1}}(\rho_{2})
- {\xi_{2}}(\rho_{1})\big)  \ \in \ V_{0} \ , \nn \\[4pt]
	\ell_{2}\big((\xi,\rho)\,,\,(e,\omega)\big)&=\big(-\rho
        \cdot e +\LL_{\xi}e\,,\, -[\rho,\omega]+\LL_{\xi}\om\big) \
        \in \ V_{1}\ , \nn \\[4pt]
	\ell_{2}\big((\xi,\rho)\,,\,(E,{\mit\Omega})\big)&=\big(- \rho\cdot
        E+\LL_{\xi}E\,,\, -\rho
        \cdot {\mit\Omega} +\LL_{\xi}{\mit\Omega}\big) \ \in \ V_{2} \
                                                              , \nn \\[4pt]
\ell_{2}\big((\xi,\rho)\,,\,({\CX},{\CP})\big)&=\big
                                                                     (\dd
                                                                     x^\mu\otimes\Tr(\iota_\mu \dd
                                                                     \rho
                                                                     \dwedge
                                                                     {\CP})
                                                                     +\LL_{\xi}{\CX}\,,\,-\rho
                                                                     \cdot
                                                                     {\CP}
                                                                     +\LL_{\xi}
                                                                     {\CP}\big)
                                                                     \
                                                                     \in
                                                                     \
                                                                     V_3
                                                                     \
                                                                     ,
                                                                     \label{eq:ell2gaugefield}
  \\[4pt]
\ell_{2}\big((e,\om)\,,\,(E,{\mit\Omega})\big)&=\Big(\dd x^\mu\otimes\Tr \big( \iota_\mu \dd e \dwedge E + (-1)^{d-1}\, \iota_\mu \dd\om \dwedge {\mit\Omega} \nn   -\iota_\mu e \dwedge \dd E - (-1)^{d-1}\, \iota_\mu \om\dwedge \dd {\mit\Omega}\big) \,, \nn \\ 
&\phantom{{}\big(\Tr ( \iota_\mu \dd e \dwedge E +{}}
  \frac{d-1}2 \, E\wedge e - (-1)^{d-1}\, \omega \wedge {\mit\Omega} \Big) \ \in \ V_3 \ . \nn
\end{align}
The first bracket of \eqref{eq:ell2gaugefield} is simply the Lie
bracket of the gauge algebra \eqref{eq:ECPgaugealg}.
The multivector action of a gauge transformation
$\rho\in\Omega^0(M,\mathfrak{so}(p,q))$ on the Euler--Lagrange
derivatives $(E,{\mit\Omega})$ and on the Noether identity $\CP$ is via the trivial coproduct $\Delta_0(\rho)=\rho\otimes 1+
1\otimes\rho$ on the
$\dwedge$-products of fields, in the fundamental representation for $e$
and in the adjoint
representation for $\omega$. For example, on $ e_{1} \dwedge
\dd e_{2} \in \Omega^{3}(M,\midwedge^{2}(\FR^{p,q}))$ the action is
$$
\rho \cdot( e_{1} \dwedge \dd e_{2})= (\rho \cdot
e_{1})\dwedge \dd e_{2} + e_{1} \dwedge(\rho \cdot \dd e_{2}) 
$$
while on $ e \dwedge \dd \omega\in\Omega^3(M,\midwedge^3(\FR^{p,q}))$ the action is
$$
\rho \cdot (e\dwedge \dd \omega)= (\rho \cdot e)
\dwedge \dd\omega+ e \dwedge [\rho, \dd \omega] \ .
$$
Similarly, the Lie derivative $\LL_{\xi}$ acts on
${\CX}\in\Omega^{1}(M)\otimes \Omega^{d}(M)$ via the Leibniz rule,
that is,
the trivial coproduct $\Delta_0(\LL_\xi)=\LL_\xi\otimes1+1\otimes\LL_\xi$.

Next, consider the brackets involving only dynamical fields. The $d{-}2$-bracket is defined by
\begin{align*}
& \ell_{d-2}\big((e_{1},\omega_{1})\,,\, \dots
  \,,\,(e_{d-2},\om_{d-2})\big) \\[4pt]
& \hspace{2cm} = (-1)^{\frac{1}{2}\,(d-2)\,(d-3)}\, \sum_{ \sigma\in
  S_{d-2}}\, \big ( e_{\sigma(1)} \dwedge e_{\sigma(2)} \dwedge  \cdots  \dwedge e_{\sigma(d-3)}\dwedge \dd \omega_{\sigma(d-2)} \,,\, \\
& \hspace{8cm}  e_{\sigma(1)} \dwedge e_{\sigma(2)} \dwedge
  \cdots \dwedge \dd e_{\sigma(d-2)}\big) \\[4pt]
& \hspace{2cm} =	(-1)^{\frac{1}{2}\,(d-2)\,(d-3)}\, (d-3)!\,
  \sum_{i=1}^{d-2}\, \big( e_{1}\dwedge  \cdots \widehat{e_{i}} \cdots  \dwedge
   e_{d-2} \dwedge \dd \om_{i} \,,\, \\
& \hspace{8cm} 
  e_{1}\dwedge \cdots 
  \widehat{e_{i}} \cdots \dwedge  e_{d-2} \dwedge \dd e_{i} \big) \ \in \ V_{2} \ .
\end{align*}
The $d{-}1$-bracket is defined by
\begin{align}
& \ell_{d-1}\big((e_{1},\omega_{1})\,,\,
  \dots\,,\,(e_{d-1},\om_{d-1})\big) \nn \\[4pt]
& \quad =(-1)^{\frac{1}{2}\,(d-1)\,(d-2)}\, \sum_{ \sigma \in
  S_{d-1}}\, \big(e_{\sigma(1)} \dwedge \cdots e_{\sigma(d-3)}\dwedge
\tfrac12\,[\omega_{\sigma(d-2)}, \om_{\sigma(d-1)}] + \Lambda \,
  e_{\sigma(1)}\dwedge \cdots \dwedge  e_{\sigma(d-1)} \,,\, \nn \\ 
& \hspace{6cm} e_{\sigma(1)} \dwedge  \cdots \dwedge e_{\sigma(d-3)} \dwedge (\om_{\sigma(d-2)}\wedge e_{\sigma(d-1)}) \big) \nn \\[4pt]
& \quad =(-1)^{\frac{1}{2}\,(d-1)\,(d-2)}\, (d-3)!\,
\sum_{\stackrel{\scriptstyle i,j=1}{\scriptstyle i\neq j}}^{d-1}\,
\big(e_{1}\dwedge \cdots  \widehat{e_{i}\dwedge e_{j}} \cdots \dwedge
  e_{d-1} \dwedge \tfrac12\,[\om_{i}, \om_{j}] \,,\,  \label{eq:ell1def} \\
& \hspace{7cm} e_1\dwedge \cdots
  \widehat{e_{i}\dwedge e_{j}}  \cdots  \dwedge e_{d-1} \dwedge (\om_{i} \wedge e_{j})\big) \nn \\
&\hspace{2cm} +(-1)^{\frac{1}{2}\,(d-1)\,(d-2)}\, (d-1)!\, \Lambda \,
(e_{1}\dwedge \cdots \dwedge e_{d-1},0) \ \in \ V_{2} \ . \nn
\end{align} 
The first expressions for the brackets $\ell_{d-2}$ and $\ell_{d-1}$
show that they are manifestly symmetric on $V_{1}$ as required. The
second simplified expressions are obtained by noting that $e_{i}
\dwedge e_{j}=e_{j}\dwedge e_{i}$. For $d>3$ the $1$-bracket on fields
is given by
\bea\nn
\ell_{1}(e,\omega)=(0,0) \ \in \ V_{2} \ ,
\eea
while for $d>4$ the $2$-bracket on fields is
\bea\nn
\ell_{2}\big((e_{1},\omega_{1})\,,\,(e_{2},\omega_{2})\big)=(0,0)
        \ \in \ 
        V_{2} \ .
\eea
The remaining brackets are all
identically zero for all $d\geq3$.

By construction, these brackets encode the gauge transformations,
field equations and Noether identities of pure Einstein--Cartan--Palatini gravity in
dimensions $d \geq 3$ as given in
\eqref{gaugetransfA}--\eqref{eq:Noether}. For example, the gauge
transformations are encoded by
\begin{align*} 
\delta_{(\xi,\rho)}(e,\omega)
&=\big(\delta_{(\xi,\rho)}e\,,\,\delta_{(\xi,\rho)}\omega\big)
                                  \\[4pt]
&=\big(-\rho\cdot e+\LL_{\xi}e\,,\, \dd\rho - [\rho,\omega] +
  \LL_{\xi}\om\big) \\[4pt]
&=\ell_{1}(\xi,\rho) + \ell_{2}\big((\xi,\rho) \,,\,
  (e,\omega)\big) \ \in \ V_1 \ ,
\end{align*}
with the closure relation
\begin{align} \label{eq:gaugealgebra}
\big[\delta_{(\xi_1,\rho_1)}\,,\,\delta_{(\xi_2,\rho_2)}\big](e,\omega)
  =\delta_{-\ell_2((\xi_1,\rho_1),(\xi_2,\rho_2))}(e,\om) 
\end{align}
reflecting the module structure of the space of fields
\eqref{eq:ECPfieldspace} for the Lie algebra of gauge transformations
\eqref{eq:ECPgaugealg}. Similarly, one also easily verifies 
$$
\delta_{(\xi,\rho)}(\CF_e,\CF_\om) =
\ell_2\big((\xi,\rho)\,,\,(\CF_e,\CF_\om)\big) \ \in \ V_2 \ ,
$$
and the Noether identities are encoded
through
\begin{align*}
\dsf_{(e,\om)}(\CF_e,\CF_\om) =
  \ell_1(\CF_e,\CF_\om)+\ell_2\big((\CF_e,\CF_\om)\,,\,(e,\om)\big) \
  \in \ V_3 \ .
\end{align*}

The quickest and most economical way to obtain these brackets is by bootstrapping
\cite{Linfty}. One writes out the gauge transformations, Euler--Lagrange derivatives
and Noether identities, and separates the orders of fields within the pairing as in Section~\ref{sec:Linftygft}. Then we demand that they are
equal to a specific expansion in terms of linear, bilinear, trilinear,
etc. brackets as in \eqref{gaugetransfA}, \eqref{EOM} and \eqref{eq:Noether}, and one
reads off the brackets by direct comparison together with the demand of cyclicity \eqref{action}. The extra non-zero
brackets may be obtained by demanding that the homotopy relations
hold. For example, the extra $2$-bracket
$\ell_{2}((\xi_{1},\rho_{1}),(\xi_{2},\rho_{2}))$ carries
information about the Lie algebra structure of the gauge
transformations, and of the action of $\Gamma(TM)$ on
$\Omega^{0}(M,\mathfrak{so}(p,q))$.  Indeed, this is the way we recovered the $L_{\infty}$-algebra. The ``disadvantage'' of the approach is that one needs to check the homotopy relations explicitly.  Another way to get the brackets is by developing and then dualizing the BV--BRST
complex of the field theory~\cite{BVChristian}, which has the advantage of automatically
guaranteeing the homotopy relations, but at the cost of lengthy and cumbersome
dualization calculations. The proof of the homotopy relations is tedious, but largely independent of the spacetime dimension
$d$. We will explicitly prove the homotopy
relations in Appendix~\ref{app:3dhomotopy} in the simplest case $d=3$, the proof for higher
dimensions being similar but requiring special care of the extra coframe field
factors which are manifested in the form of higher brackets. We review and dualise the BV--BRST formalism of ECP  developed by~\cite{ECBV} for $d=4$ in Section~\ref{sec:BV-BRST}, confirming the above $L_{\infty}$-algebra structure.

\subsection{Cyclic pairing}
\label{sec:cyclicpairing}

Given the brackets of Section~\ref{sec:LinftyECP}, we wish to write
the Einstein--Cartan--Palatini action functional \eqref{eq:ECPaction} as in
\eqref{action}. For this, we need a suitable non-degenerate bilinear pairing $\langle -,-
\rangle:V_{1} \otimes V_{2} \rightarrow \FR$, which we shall show is
given by
\begin{align} \label{eq:ECPpairing}
\langle (e,\om) \,,\, (E,{\mit\Omega}) \rangle:= 
\int_{M}\, \Tr \big(e\dwedge E+ (-1)^{d-1}\,\om \dwedge {\mit\Omega} \big) = 
\int_{M}\, \Tr \big(e\dwedge E+ {\mit\Omega} \dwedge \om \big) \ .
\end{align}
This can be extended to make \eqref{eq:ECPvectorspace} into a cyclic
$L_\infty$-algebra by introducing an additional pairing
$\langle-,-\rangle:V_0\otimes V_3\to\FR$ given by
\begin{align} \label{eq:ECPpairing2}
\langle(\xi,\rho)\,,\,({\CX},{\CP})\rangle := \int_M\,
  \iota_\xi{\CX} + \int_M\, \Tr\big(\rho\dwedge{\CP}\big) \ .
\end{align}
The most general possible bilinear pairing could in principle include
two arbitrary constants in front of each integrand. However, cyclicity
demands they are set to unity, and we will now show that the pairings \eqref{eq:ECPpairing} and \eqref{eq:ECPpairing2}
indeed have the right cyclicity properties. The only non-trivial checks
required in \eqref{cyclicity} are for the brackets $\ell_{d-2}$ and
$\ell_{d-1}$. The explicit demonstration of the cyclicity \eqref{eq:cyclicity2}
for the pairing \eqref{eq:ECPpairing} is
elementary as it involves only the bracket $\ell_2$ from \eqref{eq:ell2gaugefield}.

Let us first establish cyclicity with respect to the bracket
$\ell_{d-2}$. 
\begin{lemma}
If $(e_i,\omega_i)\in V_1$ for $i=0,1,\dots,d-2$, then
\begin{align*}
\big\langle (e_{0},\om_{0}) \,,\, \ell_{d-2}\big((e_{1},\omega_{1}),
  \dots ,(e_{d-2},\om_{d-2}) \big) \big\rangle =\big\langle
  (e_{1},\om_{1}) \,,\, \ell_{d-2}\big((e_{0},\omega_{0}),
  (e_2,\omega_2),\dots ,(e_{d-2},\om_{d-2}) \big) \big\rangle \ .
\end{align*}
\end{lemma}
\begin{proof}
We shall ignore the overall constants since they are the same on both
sides of this equality. Thus writing out the left-hand side we obtain
\begin{flalign}
&\int_M\, \Tr \bigg( \sum_{i=1}^{d-2}\, \big(e_{0}\dwedge e_{1}
\dwedge  \cdots  \widehat{e_{i}}  \cdots \dwedge e_{d-2}\dwedge \dd
\om_{i} + (-1)^{d-1}\, \om_0 \dwedge e_{1}\dwedge  \cdots
\widehat{e_{i}} \cdots  \dwedge e_{d-2} \dwedge \dd e_{i} \big) \bigg) \nn \\
&\hspace{1cm} = \int_M\, \Tr \bigg( e_{0}\dwedge e_{2} \dwedge  \cdots
\dwedge e_{d-2} \dwedge \dd \om_{1} + \sum_{i =2}^{d-2}\,
\big(e_{0}\dwedge e_{1}\dwedge  \cdots  \widehat{e_{i}} \cdots \dwedge
e_{d-2} \dwedge \dd \om_{i} \big) \label{eq:LHSelln-2cyclic} \\
&\hspace{1cm} \phantom{{}= \int_M\, \Tr \bigg( {}}  +
(-1)^{d-1}\, \om_0\dwedge e_{2}\dwedge \cdots \dwedge e_{d-2} \dwedge
\dd e_{1} \nn \\
&\hspace{1cm} \hspace{1cm} \phantom{{}= \int_M\, \Tr \bigg( {}}
+ (-1)^{d-1}\, \sum_{i=2}^{d-2}\, \big(\omega_0 \dwedge e_{1} \cdots
\widehat{e_{i}}  \cdots  \dwedge e_{d-2} \dwedge \dd e_{i} \big) \bigg) \ . \nn
\end{flalign}
Integrating by parts on the first and
third terms, and dropping exact forms since we only consider coframes with compact support, we get
\begin{flalign*}
&\int_M\, \Tr \bigg ( (-1)^{d-3}\,\om_1 \dwedge e_{2}\dwedge
\cdots  \dwedge e_{d-2} \dwedge \dd e_{0} -\om_1\dwedge e_{0}\dwedge \dd (e_{2} \dwedge  \cdots  \dwedge e_{d-2}) \\
&\phantom{{}= \int_M\, \Tr \bigg( {}} + \sum_{i=2}^{d-2}\,
\big(e_{0}\dwedge e_{1}\dwedge \cdots  \widehat{e_{i}} \cdots  \dwedge
e_{d-2} \dwedge \dd \om_{i} \big) + \om_0\dwedge e_{1} \dwedge \dd (e_{2} \dwedge  \cdots  \dwedge e_{d-2}) \\
&\phantom{{}= \int_M\, \Tr \bigg( {}}+e_{1} \dwedge e_{2}\dwedge
\cdots \dwedge e_{d-2} \dwedge \dd \om_{0} + (-1)^{d-1}\, \sum_{i=2}
^{d-2}\, \big(\om_0\dwedge e_{1}\dwedge \cdots \widehat{e_{i}} \cdots
\dwedge e_{d-2} \dwedge \dd e_{i}\big) \bigg) \ .
\end{flalign*}
Now using
\begin{flalign*}
\om_0\dwedge e_{1}\dwedge \dd( e_{2} \dwedge  \cdots  \dwedge e_{d-2}) 
=-(-1)^{d-3}\, \sum_{i=2}^{d-2}\, \big(\om_0 \dwedge e_{1}\dwedge \cdots
\widehat{e_{i}} \cdots \dwedge e_{d-2} \dwedge \dd e_{i} \big) \nn
\end{flalign*} 
we see that the fourth and sixth terms cancel. Substituting similarly for the second term, we get
\begin{flalign*}
&\int_M\, \Tr \bigg ((-1)^{d-1}\, \om_1 \dwedge e_{2}\dwedge
\cdots  \dwedge e_{d-2} \dwedge \dd e_{0} + 
(-1)^{d-1}\, \sum_{i=2}^{d-2}\, \big(  \om_1 \dwedge e_{0}\dwedge \cdots
\widehat{e_{i}} \cdots  \dwedge e_{d-2} \dwedge \dd e_{i} \big) \\
&\phantom{{}= \int_M\, \Tr \bigg( {}} + \sum_{i=2}^{d-2}\,
\big(e_{1}\dwedge e_{0}\dwedge \cdots  \widehat{e_{i}} \cdots  \dwedge
e_{d-2} \dwedge \dd \om_{i} \big)+ e_{1} \dwedge e_{2} \dwedge \cdots
\dwedge e_{d-2} \dwedge \dd \om_{0}  \bigg) \ .
\end{flalign*} 
This is just the equality \eqref{eq:LHSelln-2cyclic} with the indices
$1$ and $0$ interchanged, showing that the pairing is indeed cyclic under
$\ell_{d-2}$ as claimed.
\end{proof}

Next we establish cyclicity with respect to the bracket
$\ell_{d-1}$. 
\begin{lemma}
If $(e_i,\omega_i)\in V_1$ for $i=0,1,\dots,d-1$, then
\begin{align*}
\big\langle (e_{0},\om_{0}) \,,\, \ell_{d-1}\big((e_{1},\omega_{1}),
  \dots ,(e_{d-1},\om_{d-1}) \big) \big\rangle =\big\langle
  (e_{1},\om_{1}) \,,\, \ell_{d-1}\big((e_{0},\omega_{0}),
  (e_2,\omega_2),\dots ,(e_{d-1},\om_{d-1}) \big) \big\rangle \ .
\end{align*}
\end{lemma}
\begin{proof}
Ignoring again overall prefactors, we compute
\begin{flalign*}
& \big\langle (e_{0},\om_{0}) \,,\, \ell_{d-1}\big((e_{1},\omega_{1}),
\dots ,(e_{d-1},\om_{d-1}) \big) \big\rangle \\
& \hspace{3cm} =\nn  \int_M\, \Tr \bigg(\,
\sum_{\stackrel{\scriptstyle i,j=1}{\scriptstyle i\neq j}}^{d-1}\, \big(e_{0}\dwedge e_{1}\dwedge \cdots  \widehat{e_{i}\dwedge e_{j}} \cdots  \dwedge e_{d-1}\dwedge \tfrac12\,[\om_{i}, \om_{j}] \big) 
\\  & \hspace{5cm} + (-1)^{d-1}\, \sum_{\stackrel{\scriptstyle
    i,j=1}{\scriptstyle i\neq j}}^{d-1}\, \big(\om_0 \dwedge e_{1} \dwedge  \cdots \widehat{e_{i}\dwedge e_{j}} \cdots  \dwedge
e_{d-1}\dwedge (\om_{i}\wedge e_{j}) \big) \bigg) \ .
\end{flalign*}
Here we dropped the cosmological constant term as it is easily seen to
contribute cyclically to this pairing, since $e_0\dwedge
e_1=e_1\dwedge e_0$.
We wish to write this in a manifestly symmetric form under the
exchange of the indices $1$ and $0$. Since $[\om_{i},
\om_{j}]=[\om_{j}, \om_{i}]$, the first term may be rewritten as 
\begin{flalign*}
& \sum_{\stackrel{\scriptstyle i,j=1}{\scriptstyle i\neq j}}^{d-1}\,
\big(e_{0}\dwedge e_{1}\dwedge \cdots  \widehat{e_{i}\dwedge e_{j}}
\cdots  \dwedge e_{d-1}\dwedge \tfrac12\,[\om_{i}, \om_{j}] \big) \\
& \hspace{4cm} = 2\, \sum_{\stackrel{\scriptstyle i,j=1}{\scriptstyle
    i< j}}^{d-1}\, \big( e_{0} \dwedge e_{1}\dwedge \cdots \widehat{e_{i}\dwedge e_{j}} \cdots  \dwedge e_{d-1}\dwedge \tfrac12\,[\om_{i}, \om_{j}] \big)\\[4pt]
& \hspace{6cm} =\sum_{\stackrel{\scriptstyle i,j=2}{\scriptstyle
    i< j}}^{d-1}\, \big( e_{0} \dwedge e_{1}\dwedge \cdots \widehat{e_{i}\dwedge e_{j}} \cdots  \dwedge e_{d-1}\dwedge [\om_{i}, \om_{j}] \big) \\
&  \hspace{7cm} + \sum_{j=2}^{d-1}\, \big( e_{0} \dwedge
\widehat{e_{1}} \cdots  \widehat{e_{j}} \cdots  \dwedge e_{d-1}\dwedge
[\om_{1}, \om_{j}] \big) \ .
\end{flalign*}
For the second term, we use $\om \dwedge e = -e \dwedge \om $ to get
\begin{flalign*}
(-1)^{d-1}\, (-1)^{d-3}\, \sum_{\stackrel{\scriptstyle i,j=1}{\scriptstyle i\neq
    j}}^{d-1}\, \big( e_{1}\dwedge  \cdots
\widehat{e_{i}\dwedge e_{j}}& \cdots  \dwedge e_{d-1}\dwedge (\om_{i}\wedge e_{j})\dwedge \om_{0} \big) \\ & =
\sum_{\stackrel{\scriptstyle i,j=1}{\scriptstyle
    i\neq j}}^{d-1}\, \big( e_{1}\dwedge \cdots  \widehat{e_{i}\dwedge
  e_{j}} \cdots  \dwedge e_{d-1} \dwedge (\om_{i}\wedge e_{j}) \dwedge
\om_0 \big) \ .
\end{flalign*}
Next we use the identity \eqref{weirdformula} to write this as
\begin{flalign*}
\sum_{i=1}^{d-1}\, \big(e_{1}\dwedge  \cdots \widehat{e_{i}}
\cdots  \dwedge e_{d-1} \dwedge [\om_{i} , \om_{0}]\big)&= \sum_{i=2}^{d-1}\, \big(e_{1}\dwedge  \cdots \widehat{e_{i}} \cdots \dwedge e_{d-1}\dwedge[\om_{i}, \om_{0}]\big) \\
& \hspace{1cm} + e_{2}\dwedge  \cdots  \dwedge e_{d-1}\dwedge
[\om_{1} , \om_{0}] \ .
\end{flalign*}
Finally collecting everything together, we get
\begin{flalign*}
\big\langle (e_{0},\om_{0}) \,,\,& \ell_{d-1}\big((e_{1},\omega_{1}),
\dots ,(e_{d-1},\om_{d-1}) \big) \big\rangle \\ 
&=\int_M\, \Tr \bigg(\, \sum_{\stackrel{\scriptstyle
    i,j=2}{\scriptstyle i< j}}^{d-1}\, \big( e_{0} \dwedge e_{1}\dwedge \cdots \widehat{e_{i}\dwedge e_{j}} \cdots  \dwedge e_{d-1}\dwedge [\om_{i}, \om_{j}] \big) \\
&\phantom{{}=\int_M\, \Tr \bigg( {}} + \sum_{j =2}^{d-1}\, \big( e_{0} \dwedge \widehat{e_{1}} \cdots  \widehat{e_{j}} \cdots  \dwedge e_{d-1}\dwedge [\om_{1}, \om_{j}] \big) \\
&\phantom{{}=\int_M\, \Tr \bigg( {}}+\sum_{i=2}^{d-1}\,
\big(e_{1}\dwedge  \cdots \widehat{e_{i}} \cdots \dwedge
e_{d-1}\dwedge[\om_{i}, \om_{0}]\big)+ e_{2}\dwedge  \cdots
\dwedge e_{d-1}\dwedge [\om_{1} , \om_{0}]\bigg ) \ .
\end{flalign*}
This expression is manifestly symmetric under exchange of the
indices $1$ and $0$: The first term is invariant since
$e_{1}\dwedge e_{0}=e_{0}\dwedge e_{1}$, the last term is unchanged
since $[\om_{1}, \om_{0}]=[\om_{0},\om_{1}]$, and the remaining
terms map into each other under the exchange. This completes the
proof of cyclicity of the pairing \eqref{eq:ECPpairing}.
\end{proof}

Now we establish cyclicity with respect to the bracket
$\ell_{2}$. We only exhibit the non-trivial check, with the rest following similarly.
\begin{lemma}
	If $(\xi_i,\rho_i)\in V_0$ for $i=0,1$ and $(\CX,\CP) \in V_{3}$, then
	\begin{align*}
	\big\langle (\xi_{0},\rho_{0}) \,,\, \ell_{2}\big((\xi_{1},\rho_{1}),
	(\CX,\CP) \big) \big\rangle =-\big\langle
	(\xi_{1},\rho_{1}) \,,\, \ell_{2}\big((\xi_{0},\rho_{0}),
	(\CX,\CP) \big) \big\rangle =\big\langle
	(\CX,\CP) \,,\, \ell_{2}\big((\xi_{0},\rho_{0}),
	(\xi_{1},\rho_{1}) \big) \big\rangle \ .
	\end{align*}
\end{lemma}
\begin{proof} By linearity, it suffices to prove the result for
  decomposable $\CX= \CX_{1}\otimes \CX_{d} \in
  \Omega^{1}(M,\Omega^{d}(M))$ with $\CX_{1}$ a one-form and $\CX_{d}$ a
  $d$-form on $M$. Then
	\begin{align*}
	\int_{M}\, \iota_{\xi_{0}}\LL_{\xi_{1}}\CX &= \int_{M}\, \big( \iota_{\xi_{0}}\LL_{\xi_{1}}\CX_{1} \, \CX_{d}+\iota_{\xi_{0}}\CX_{1} \, \dd \iota_{\xi_{1}} \CX_{d}\big) \\[4pt]
	&= \int_{M}\,\big(\iota_{\xi_{0}}\LL_{\xi_{1}}\CX_{1} \, \CX_{d} -\dd \iota_{\xi_{0}} \CX_{1}\wedge \iota_{\xi_{1}} \CX_{d}\big) \\[4pt]
	&= \int_{M}\, \big( \iota_{\xi_{0}} \LL_{\xi_{1}} \CX_{1} \, \CX_{d} -\iota_{\xi_{1}} \dd \iota_{\xi_{0}} \CX_{1} \, \CX_{d}\big) \\[4pt]
	&= \int_{M}\, \big(\iota_{\xi_{0}}\LL_{\xi_{1}}\CX_{1}\, \CX_{d} - \LL_{\xi_{1}} \iota_{\xi_{0}} \CX_{1}\, \CX_{d}\big) \\[4pt]
	&= \int_{M}\, \iota_{[\xi_{0},\xi_{1}]} \CX \ ,
	\end{align*}
where we firstly used the trivial coproduct to distribute the Lie derivative, then applied Cartan's magic formula, integrated by parts and used the derivation property of the contraction. Lastly we used the Cartan identity
\begin{align}\label{eq:Cartanidiota}
\iota_{[\xi_1,\xi_0]} = \LL_{\xi_1}\circ\iota_{\xi_0} -
  \iota_{\xi_0}\circ\LL_{\xi_1} \ .
\end{align}
The final equality says that the initial quantity on the
left-hand side is antisymmetric under the exchange of the vector
fields $\xi_{0}$ and $\xi_{1}$. Using this, the left-hand side of the
cyclicity identity expands as
\begin{align*}
\big\langle (\xi_{0},\rho_{0}) \,,\, \ell_{2}\big((\xi_{1},\rho_{1}),
(\CX,\CP) \big) \big\rangle &=\big\langle (\xi_{0},\rho_{0})\,,\,
                              \big(\dd x^{\mu}\otimes \Tr(\iota_{\mu}
                              \dd \rho_{1} \dwedge
                              \CP)+\LL_{\xi_{1}}\CX, -\rho_{1}\cdot
                              \CP +\LL_{\xi_{1}} \CP\big)\big\rangle \\[4pt]
&=\int_M\, \iota_{[\xi_{0},\xi_{1}]}\CX + \int_{M}\, \Tr \big( \iota_{\xi_{0}} \dd \rho_{1} \dwedge
  \CP - \rho_{0} \dwedge  \rho_{1}\cdot
  \CP + \rho_{0} \dwedge \LL_{\xi_{1}} \CP\big) \\[4pt]
&=\int_M\, \iota_{[\xi_{0},\xi_{1}]}\CX + \int_{M}\, \Tr \big(-\rho_{1} \dwedge \LL_{\xi_{0}} \CP + [\rho_{1},\rho_{0}]\dwedge \CP +\rho_{0}\dwedge \LL_{\xi_{1}}\CP \big)
\end{align*}
where we used used the derivation propery of the contraction,
integrated by parts and used the invariance of a top exterior vector
in $\FR^{p,q}$ under $\mathfrak{so}(p,q)$ rotations. This is
manifestly antisymmetric under the exchange of indices $0$ and $1$,
thus proving the first cyclicity identity. Further manipulating, noting that $\iota_{\xi_{0}}\dd \rho_{1} = \LL_{\xi_{0}} \rho_{1}$ since $\rho_{1}$ is a zero-form, this is also equal to
\begin{align*}
\int_M\, \iota_{[\xi_{0},\xi_{1}]} \CX + \int_{M}\, \Tr \big((\LL_{\xi_{0}}\rho_{1}-\LL_{\xi_{1}}\rho_{0})\dwedge \CP - [\rho_{0},\rho_{1}]\dwedge \CP \big)
\end{align*}
which gives the final cyclicity identity.
\end{proof}

Finally, the proof of the cyclicity of $\big\langle (\xi,\rho) \,,\, \ell_{2}\big((e,\om),
(E,\Omega) \big) \big\rangle
$ contains essentially no new ideas, apart from the fact that
$\rho\cdot e \dwedge E = - \frac{d-1}{2}\, \rho \dwedge (E\wedge e)$,
which follows from $0=(\rho\dwedge E)\wedge e$ since the internal
vector space exterior products combine to a $d{+}1$-form in $d$ dimensions, and then distributing the contractions.
To summarise, we have completely encoded the dynamical content
of pure ECP gravity in $d>2$ dimensions in terms of
an $L_{\infty}$-algebra equipped with a suitable cyclic pairing. In
Sections~\ref{sec:3dgrav} and~\ref{sec:4dgrav} we will consider two physically
relevant examples explicitly to illustrate
this structure.

\subsection{Covariant $L_\infty$-algebra}
\label{sec:covariant}

We will now discuss the covariance problems associated with our
$L_\infty$-algebra for the ECP theory, and how
to deal with them. This is done using the well known covariant Lie derivative~\cite{Jackiw1980}~\cite{SUGRAbook}~\cite{Prabhu2017}, whose geometric meaning we review and meet its avatar in the resulting brackets. The resulting $L_{\infty}$-algebra is completely dual to the covariant version of the BV-BRST complex developed in \cite{ECBV} for $d=4$.

\subsubsection*{Finite gauge transformations}

Recall from Sections~\ref{sec:ECPgaugesym}
and~\ref{sec:ECPgauge} that in
order for these formulas to make sense, one firstly has to consider a
parallelizable spacetime manifold $M$ and {fix} the `fake tangent
bundle' $\CCV = M \times \FR^{p,q}$.\footnote{Alternatively,
  they make sense in any fixed local trivialization of $\CCV$,
  but then one has to face the problem of patching together the locally defined
$L_\infty$-algebras in a suitable way to an `$L_\infty$-algebroid stack' on the
  spacetime $M$. In the present paper we instead follow a more concrete
  approach to this problem, which is detailed in the following.} Under this choice, the
coframe field $e$ may be viewed globally as a one-form on $M$ valued in
$\FR^{p,q}$ and the connection $\omega$ as a one-form valued in $\mathfrak{so}(p,q)$.  For our original motivation on parallelisable manifolds and noncommutative or nonassociative deformations this suffices~\cite{NCProc}, however one runs into issues if finite gauge transformations and the possible non-parallelisability of spacetime are taken into account. The coframe is globally encoded as a one-form $\tilde{e} \in \Omega^{1}(\CCP,\FR^{p,q})$ and the connection is globally
encoded as a one-form
$\tilde\omega\in\Omega^1(\CCP,\mathfrak{so}(p,q))$, on the associated principal $\sSO_+(p,q)$-bundle
$\CCP\rightarrow M$. Given a local trivialisation of $\CCP$, or
equivalently a local section $s:U \rightarrow \CCP$ for $U\subset M$, one defines the
gauge field $\om:=s^{*}\tilde{\om} \in
\Omega^{1}(U,\mathfrak{so}(p,q))$. Given another local section $s':U\to\CCP$
with $\omega':=s^{\prime\ast}\tilde\omega$, the two pullbacks are 
related by 
$$
\om'= h^{-1} \, \om \, h +h^{-1}\, \dd h
$$
where $h:U \rightarrow \sSO_+(p,q)$ is the finite gauge transformation
defined by $s'=s \, h$. 

The ``problem'' arises with the diffeomorphism symmetry of the theory:
given an infinitessimal diffeomorphism of the base manifold $M$,
parameterized by a vector field $\xi \in \Gamma(TM)$, it is clear that 
$$
\om' + \LL_{\xi} \om' \neq {h}^{-1}\, (\om + \LL_{\xi} \om) \, {h} + {h}^{-1}\,
\dd{h}
$$ 
for any ${h}:U \rightarrow \sSO_+(p,q)$, and furthermore one cannot
identify any section of $ \CCP$ which pulls back $\tilde\omega$ to $\om +\LL_{\xi} \om$. That is, the expressions $\om+ \LL_{\xi}\om$ and $\om' + \LL_{\xi}\om'$ no longer define a
  connection on the same bundle $\CCP$. This is apparent if one uses
  finite diffeomorphisms $\phi:M\to M$ of the base, where the pullbacks of the
  gauge fields $\phi^*\omega$ and $\phi^*\omega'$ do define a
  connection, but on the pullback principal bundle $\phi^*\CCP$ over
  $M$ instead. Thus, strictly speaking, our approach applies only when
  one {completely disregards global structures} and thus also finite
  gauge transformations, viewing the fields as globally defined
  defined one-forms on the base space $\om \in
  \Omega^{1}(M,\mathfrak{so}(p,q))$ and $ e\in \Omega^{1}(M,\FR^{p,q})$
  which transform only under infinitesimal gauge transformations as
  expected.
\subsubsection*{Covariantization of diffeomorphisms}

We have seen that a Lie derivative of the gauge field on the base space no longer defines a connection on the same principal
bundle, even if the bundle is trivial, and similarly for the coframe
field when viewed as a section of $T^{\ast}M\otimes \CCV$. An equivalent way to spot the issue is from the fact that the action of the Lie derivative
does not commute with the action of a finite gauge transformation
$h:U\rightarrow \sSO_+(p,q)$; for example, on the coframe field 
\begin{align*}
\LL_{\xi} (h^{-1} \, e) \neq  h^{-1}\, \LL_{\xi} e \ .
\end{align*}
Equivalently, for an infinitesimal pseudo-orthogonal rotation $\rho:
M\rightarrow \mathfrak{so}(p,q)$, 
$$
[\delta_{\xi},\delta_{\rho}] \neq 0 \ ,
$$
as was already implied by the semi-direct product structure of the gauge
algebra \eqref{eq:ECPgaugealg}.

The resolution comes by identifying the correct way to act directly on the global fields $\tilde{e}$, $\tilde{\om}$ living on $\CCP$, where they appear as genuine one-forms valued in fixed vector spaces. We briefly describe this and skip the straightforward differential geometrical calculations, since they are well known and not relevant for the rest of the section. Each connection identifies a horizontal distribution $\mathrm{Hor}(\CCP)\subset T\CCP$ via its kernel, splitting the tangent bundle $T\CCP =\mathrm{Vert}(\CCP)\oplus  \mathrm{Hor}(\CCP)$ such that $\mathrm{Hor}(\CCP)\cong TM$ via the differential of the bundle projection $\pi:\CCP\rightarrow M$ and $\mathrm{Vert}(\CCP)\equiv \ker(\dd \pi)$. Using this identification, $\Gamma(TM)\cong \Gamma(\mathrm{Hor}(\CCP))$ as vector spaces, and so we may act with the unique lift $\tilde{\xi} \in \Gamma(\mathrm{Hor}(\CCP))$ of any $\xi \in \Gamma(TM)$, via the Lie derivative of the total space instead. Notice although the connection gives a lift $\Gamma(TM)\cong \Gamma(\mathrm{Hor}(\CCP))$, this is not a Lie algebra morphism, that is \begin{align}\label{nonint} [\tilde{\xi_{1}},\tilde{\xi_{2}}]=\widetilde{[\xi_{1},\xi_{2}]} + \iota_{\tilde{\xi_{2}}} \iota_{\tilde{\xi_{1}}} \tilde{R}
\end{align} expressing that the non-integrability of the horizontal distribution is controlled by the curvature of the connection. Acting on the global fields $(\tilde{e}, \, \tilde{\om})$ using the equivariance, horizontal and vertical properties of the fields along with Cartan calculus on $\CCP$, \begin{align*}
	(\LL_{\tilde{\xi}}\, \tilde{e},\LL_{\tilde{\xi}}\,  \tilde{\om})= \big((\dd^{\tilde{\om}}\circ \iota_{\tilde{\xi}} + \iota_{\tilde{\xi}} \circ \dd^{\tilde{\om}})\, \tilde{e}\,, \, \iota_{\tilde{\xi}} \, \tilde{R} \big)
	\end{align*} 
where the right hand side is manifestly horizontal and equivariant. Hence $(\tilde{e},\tilde{\om})+ (\LL_{\tilde{\xi}}\, \tilde{e},\LL_{\tilde{\xi}}\,  \tilde{\om})$ define a coframe and a connection as expected. Using a local section $s:U \rightarrow \CCP$ and equivariance, the expressions pull down to define the infinitessimal action we are after
\begin{align} \label{covgauge}
\delta_{\xi}^{\rm cov}(e,\om):=\big(\LL_{\xi}^{\om}e\,,\, \iota_{\xi}R\big)
=\big(\LL_{\xi} e + \iota_{\xi}\om \cdot e\,,\, \iota_{\xi}\dd \om +
[\iota_{\xi}\om,\om]\big) \ ,
\end{align}
where the {\emph{covariant Lie derivative}}$$\LL_{\xi}^{\om}:= \dd^{\om} \circ \iota_{\xi} +\iota_\xi \circ
\dd^{\om}$$ is defined on the spacetime as an appropriate modification of the Cartan formula \eqref{eq:CartanLie}. Notice, due to $\eqref{nonint}$ this does not form a Lie algebra action of $\Gamma(TM)$ on the space of fields; however, we will see the $L_{\infty}$-algebra framework is sufficient to accomodate such situations. As a further check, one may confirm by working directly on the base spacetime manifold that the above infinitesimal action is indeed covariant, i.e. it commutes with local pseudo-orthogonal
rotations:
$$
[\delta_{\xi}^{\rm cov},
\delta_{\rho}]
(e,\om)=(0,0) \ .
$$
In particular, they commute with finite pseudo-orthogonal rotations which
are connected to the identity, that is, the fields
$(e,\om)+\delta^{\rm cov}_{\xi}(e,\om)$ then do form a proper section
of a vector bundle and a connection on its associated principal
bundle, as expected from the total space formulation. Obviously, the above discussion applies for gauge fields with any internal group $G$ and for matter fields valued in any $G$-representation. As such, covariant Lie derivatives (also known as `covariant general coordinate
transformations') have appeared in various contexts, to name a few: Studying symmetries and conserved quantities of gauge theories on fixed background spacetimes~\cite{Jackiw1980}; More specifically, these produce the correct symmetric and gauge invariant energy-momentum tensor in Minkowski spacetime, avoiding the ``adhoc'' Belinfante procedure. They also appear necessarily in the closure of local supersymmetry transformations in supergravity~\cite{SUGRAbook}. Furthermore, they have been recently used to study black hole thermodynamic laws in the case where non-trivial bundle topologies underlie the dynamical fields \cite{Prabhu2017}. 

The action functional \eqref{eq:ECPaction} is indeed invariant under
the new covariant diffeomorphisms. For example, one can check 
\begin{align*}
\delta^{\rm cov}_{\xi} (e^{d-2}\dwedge R)= \LL_{\xi}(e^{d-2}\dwedge R) + \iota_{\xi}\om \cdot (e^{d-2}\dwedge R)
\end{align*}
where the first term vanishes upon integration over $M$ by the usual
diffeomorphism invariance, and the second term vanishes by invariance under local pseudo-orthogonal transformations. Both at the field transformation and at the action functional level we see the two diffeomorphism actions differ by a local rotation, thus they are equivalent. We will see this equivalence translates to the corresponding cyclic $L_{\infty}$-algebras being isomorphic.
Expressing $\delta^{\rm cov}_\xi S_{\textrm{\tiny ECP}} =0$ through 
\begin{align*}
\delta^{\rm cov}_{\xi}S_{\textrm{\tiny ECP}}= \int_{M}\,\Tr\big(\CF_{e} \dwedge \delta_{\xi}^{\rm cov}e + \CF_{\om} \dwedge \delta_{\xi}^{\rm cov}\om\big)
\end{align*}
and isolating $\xi$ as previously, the Noether identity corresponding
to the covariant diffeomorphism modifies only the first component of
\eqref{eq:ECPcurrents} to
\begin{align}
\dsf_{(e,\om)}^{\rm cov}(\CF_e,\CF_\omega) := \bigg( & \dd
                                                       x^\mu\otimes\Tr
                                                       \Big(\iota_{\mu}e\dwedge
                                                       \dd\CF_{e} -
                                                       \iota_{\mu} \dd
                                                       e \dwedge
                                                       \CF_{e} -
                                                       (-1)^{d-1}\,
                                                       \iota_{\mu} \dd
                                                       \om \dwedge
                                                       \CF_{\om} \label{covnoether}\\ &
                                                                      +
                                                                      \iota_{\mu}\om\dwedge\Big[\frac{d-1}2\,
                                                                      \CF_{e}\wedge
                                                                      e-
                                                                      (-1)^{d-1}\,
                                                                      \om \wedge \CF_{\om}\Big]
                                               \Big) \,,\,
                                                              -\frac{d-1}2
                                                                      \,
                                                              \CF_e\wedge
                                                              e+(-1)^{d-1}\,
                                                                      \dd^\om\CF_\om
                                                              \bigg) \
                                                              . \nn
\end{align}

\subsubsection*{Covariant brackets}

We shall now spell out the brackets of the bootstrapped $L_\infty$-algebra
corresponding to the covariant gauge transformations
\eqref{covgauge}. The $L_{\infty}$-algebra we obtain turns out to be dual to the covariant BV differential obtained in \cite{ECBV} for the case of $d=4$. The brackets involving solely dynamical fields and the Euler--Lagrange derivatives are
the same as those of Section~\ref{sec:LinftyECP}, as only the gauge
transformations and Noether identities are affected in the covariant
formulation, but not the dynamics. The underlying vector space is
locally as before, however we now consider the possibly non-trivial
bundle structures properly. For this, we parameterize
$\sSO_+(p,q)$-connections on the principal bundle $\CCP\to M$ by
one-forms $\Omega^{1}\big(M,\CCP \times_{\rm ad}
\mathfrak{so}(p,q)\big)$ on the base $M$ valued in the adjoint bundle of $\CCP$, in the usual way by fixing some reference
connection $\omega_0$.
Then the graded vector space of the covariant $L_\infty$-algebra is
$$
V^{\rm cov}:= V_{0}^{\rm cov} \oplus V_{1}^{\rm cov} \oplus V_{2}^{\rm
  cov} \oplus V_3^{\rm cov}
$$
where
\begin{align*} 
V_{0}^{\rm cov}&=\Gamma(TM)\times \Omega^{0}\big(M,\CCP \times_{\rm ad} \mathfrak{so}(p,q)\big) \ , \nn
\\[4pt]
V_{1}^{\rm cov}&= \Omega^{1}(M,\CCV) \times
\Omega^{1}\big(M,\CCP \times_{\rm ad} \mathfrak{so}(p,q)\big) \ , \\[4pt]
V_{2}^{\rm cov}&=\Omega^{d-1}\big(M,\midwedge^{d-1}\,\CCV\big) \times
\Omega^{d-1}\big(M,\midwedge^{d-2}\,\CCV\big) \ , \nn \\[4pt]
V_3 ^{\rm cov}&= \Omega^1\big(M,\Omega^d(M)\big)\times\Omega^d\big(M,\midwedge^{d-2}\,\CCV\big) 
\ . \nn
\end{align*} 
We denote elements of these vector spaces with the same symbols as
previously.

We will only write out the brackets $\ell_n^{\rm cov}$ which differ from those of the
non-covariant formulation of Section~\ref{sec:LinftyECP}. The only
existing brackets from \eqref{eq:ell2gaugefield} which are modified
are
\begin{align} 
\ell_{2}^{\rm cov}\big((\xi_{1},\rho_{1})\,,\,(\xi_{2},\rho_{2})\big)&=\big([\xi_{1},\xi_{2}]\,,\,-[\rho_{1},\rho_{2}] \big) \ , \nn 
\\[4pt]
\ell_{2}^{\rm cov}\big((\xi,\rho)\,,\,(e,\om)\big)&=
\big(-\rho\cdot e+\LL_\xi e\,,\,-[\rho,\om]+\iota_\xi\dd\om\big) \ ,
                                                    \nn \\[4pt]
\ell_{2}^{\rm cov}\big((\xi,\rho)\,,\,(E,\mit\Omega)\big)&=\big(
                                                           \LL_\xi E
                                                           -\rho\cdot
                                                           E\,,\, \dd
                                                           \iota_{\xi}
                                                           \mit\Omega
                                                           - \rho
                                                           \cdot
                                                           \mit\Omega\big)
                                                           \
                                                           , \label{eq:ell2cov}
  \\[4pt]                                                    
\ell_{2}^{\rm cov}\big((\xi,\rho)\,,\,({\CX},{\CP})\big)&=\big (\LL_{\xi}{\CX}\,,\,-\rho\cdot{\CP}\big) \ , \nn
\\[4pt]
\ell_{2}^{\rm cov}\big((e,\om)\,,\,(E,{\mit\Omega})\big)&=\Big(\dd
                                                          x^\mu\otimes\Tr
                                                          \big(
                                                          \iota_\mu
                                                          \dd e
                                                          \dwedge E +
                                                          (-1)^{d-1}\,
                                                          \iota_\mu
                                                          \dd\om
                                                          \dwedge
                                                          {\mit\Omega}
                                                          \nn
                                                          -\iota_\mu e
                                                          \dwedge \dd
                                                          E) \,,\, \nn
  \\ & \hspace{6cm}
 \frac{d-1}2 \, E\wedge e - (-1)^{d-1}\, \omega \wedge {\mit\Omega} \big) \ . \nn
\end{align}
There are also a number of new higher brackets that emerge. The
non-trivial covariant $3$-brackets are given by
\begin{align}
\ell_3^{\rm
  cov}\big((\xi_{1},\rho_{1})\,,\,(\xi_{2},\rho_{2})\,,\,(e,\om)\big)&=\big(0\,,\,-\iota_{\xi_{1}}\iota_{\xi_{2}}\dd
                                                                             \om\big)
                                                                             \
                                                                             ,
  \nn \\[4pt]
\ell_{3}^{\rm
  cov}\big((\xi_{1},\rho_{1})\,,\,(\xi_{2},\rho_{2})\,,\,(\CX,\CP)\big)&=\big(0\,,\,-\dd
                                                                         \iota_{\xi_{1}}
                                                                         \iota_{\xi_{2}}
                                                                         \CP\big)
                                                                         \
                                                                         ,
                                                                         \nn
  \\[4pt]
\ell_{3}^{\rm cov}\big((\xi,\rho)\,,\,(e_{1},\om_{1})\,,\,(e_{2},\om_{2})\big)&=-\big(\iota_{\xi}\om_{1}\cdot
                                                                  e_{2}
                                                                  +\iota_{\xi}
                                                                  \om_{2}
                                                                  \cdot
                                                                  e_{1}\,,\,
                                                                  [\iota_{\xi}\om_{1},\om_{2}]
                                                                  +[\iota_{\xi}\om_{2},\om_{1}]
                                                                  \big)
                                                                  \ ,
  \nn \\[4pt]
\ell_{3}^{\rm
  cov}\big((\xi,\rho)\,,\,(E,{\mit\Omega})\,,\,(e,\om)\big)&=\Big(\iota_{\xi}\om\cdot E\,,\,
                                                             \om
                                                             \wedge
                                                             \iota_{\xi}{\mit\Omega}
                                                             +(-1)^{d-1}\,
                                                             \frac{d-1}2
                                                             \,\iota_{\xi}(E\wedge e)
                                                             \Big) \
                                                             , \label{eq:ell3cov}
  \\[4pt]
\ell_{3}^{\rm
  cov}\big((\xi,\rho)\,,\,(e,\om)\,,\,(\CX,\CP)\big)&= \big( \dd x^{\mu} \otimes \Tr(\iota_{\mu} \iota_{\xi} \dd \om \dwedge \CP) \,,\, 0\big) \
  , \nn \\[4pt]
\ell_{3}^{\rm
  cov}\big((E,{\mit\Omega})\,,\,(e_{1},\om_{1})\,,\,(e_{2},\om_{2})\big)&=-\Big(\frac{d-1}2\,
                                                                          \dd
                                                                          x^{\mu}\otimes \Tr\big(\iota_{\mu}\om_{1}\dwedge(E\wedge e_{2})+\iota_{\mu} \om_{2} \dwedge (E\wedge e_{1})\big) \nn \\
& \hspace{1cm} - (-1)^{d-1}\, \dd x^\mu\otimes\Tr\big(\iota_{\mu} \om_{1}\dwedge(\om_{2}\wedge {\mit\Omega})
  + \iota_{\mu} \om_{2}\dwedge(\om_{1}\wedge {\mit\Omega} )
  \big)\,,\,0\Big) \nn
\end{align}
while the non-trivial covariant $4$-brackets are given by
\begin{align}
\ell_{4}^{\rm cov}\big((\xi_{1},\rho_{1})\,,\,(\xi_{2},\rho_{2})\,,
  \,(e_{1},\om_{1})\,,\,(e_{2},\om_{2})\big) &=
  \big(0\,,\,\iota_{\xi_{1}}\iota_{\xi_{2}}[\om_{1},\om_{2}]\big) \ ,
  \nn \\[4pt]
\ell_{4}^{\rm cov}\big((\xi_{1},\rho_{1})\,,\,(\xi_{2},\rho_{2})\,,
  \,(e,\om)\,,\,(\CX,\CP)\big) &= \big(0\,,\,(-1)^{d-1}\, \om \wedge \iota_{\xi_{1}}\iota_{\xi_{2}} \CP \big)
  \ , \label{eq:ell4cov} \\[4pt]
\ell_{4}^{\rm cov}\big((\xi,\rho)\,,\,(e_{1},\om_{1})\,,
  \,(e_{2},\om_{2})\,,\,(\CX,\CP)\big) &= \big(\dd x^{\mu} \otimes
                                         \Tr(\iota_{\mu}\iota_{\xi}[\om_{1},\om_{2}]\dwedge
                                         \CP)\,,\,0 \big) \ . \nn
\end{align}
On all other fields and in all other degrees, the covariant brackets
coincide with the brackets of Section~\ref{sec:LinftyECP}:
$\ell_n^{\rm cov}=\ell_n$ otherwise. The proof of the homotopy
relations for these covariant brackets is discussed in
Appendix~\ref{app:3dcovhomrels}.

It is straightforward to check that the covariant brackets continue to
encode all kinematical and dynamical information about 
Einstein--Cartan--Palatini gravity,
now incorporating the covariant infinitesimal action of
diffeomorphisms discussed previously. Strictly speaking, the brackets
written above only make sense on local trivializations of the
underlying vector bundles. However, the $\ell^{\rm cov}_{n}$-polynomial
expressions of physical interest patch up to global objects, by
covariance. For example, the gauge transformations of the dynamical
fields are now given by
\begin{align*}
\delta^{\rm cov}_{(\xi,\rho)}(e,\om)=\ell^{\rm cov}_{1}(\xi,\rho)+
  \ell^{\rm cov}_{2}\big((\xi,\rho)\,,\,(e,\om)\big) -\frac{1}{2}\,
  \ell^{\rm cov}_{3}\big((\xi,\rho)\,,\,(e,\om)\,,\,(e,\om)\big) \ \in
  \ V^{\rm cov}_{1} \ ,
\end{align*}
where a non-trivial $3$-bracket $\ell^{\rm cov}_{3}$ arises because
\eqref{covgauge} now involves third degree polynomial combinations.
Similarly, one may read off the brackets from the polynomial
expression for the Noether identities, which now includes a non-trivial
$3$-bracket $\ell^{\rm cov}_{3}$ by \eqref{covnoether}, so that
\begin{align*}
\dsf^{\rm cov}_{(e,\om)}(\CF_{e},\CF_{\om})=\ell^{\rm
  cov}_{1}(\CF_{e},\CF_{\om})+\ell^{\rm
  cov}_{2}\big((\CF_{e},\CF_{\om})\,,\,(e,\om)\big) -
  \frac{1}{2}\,\ell^{\rm
  cov}_{3}\big((\CF_{e},\CF_{\om})\,,\,(e,\om)\,,\,(e,\om)\big) \ \in
  \ V^{\rm cov}_{3} \ .
\end{align*}

However, it is not immediately obvious that these new brackets encode
the expected covariance of the Euler--Lagrange derivatives, that is,
\begin{align*}
\delta^{\rm cov}_{(\xi,\rho)}(\CF_{e},\CF_{\om})=\ell^{\rm cov}_{2}\big((\xi,\rho)\,,\,(\CF_{e},\CF_{\om})\big) +\ell^{\rm cov}_{3}\big((\xi,\rho)\,,\,(\CF_{e},\CF_{\om})\,,\,(e,\om)\big) \ .
\end{align*}
The part concerning local pseudo-orthogonal rotations is immediate, so that expanding the right-hand side for $\rho=0$ we confirm
\begin{align*}
\big(\LL_{\xi}\CF_{e}\,,\, \dd \iota_{\xi} \CF_{\om}\big)&+
                                                           \Big(\iota_{\xi}\om\cdot
                                                           \CF_{e}\,,\,\om
                                                           \wedge
                                                           \iota_{\xi}
                                                           \CF_{\om}+(-1)^{d-1}\,
                                                           \frac{d-1}2
                                                           \,\iota_{\xi}(\CF_{e}\wedge
                                                           e)\Big)\\[4pt]&=\Big(\LL^{\om}_{\xi}
                                                                           \CF_{e}\,,\,
                                                                           \LL_{\xi}
                                                                           \CF_{\om}
                                                                           -\iota_{\xi}
                                                                           \dd
                                                                           \CF_{\om}-\iota_{\xi}(\om
                                                                           \wedge
                                                                           \CF_{\om})
                                                                           +
                                                                           \iota_{\xi}\om
                                                                           \cdot
                                                                           \CF_{\om}
                                                                           +(-1)^{d-1}\,
                                                                           \frac{d-1}2\,
                                                                           \iota_{\xi}(\CF_{e}
                                                                           \wedge
                                                                           e)\Big)
  \\[4pt]&=\bigg( \LL_{\xi}^{\om} \CF_{e}\,,\, \LL_{\xi}^{\om}
           \CF_{\om}
           -\iota_{\xi}\Big(\dd^{\om}\CF_{\om}-(-1)^{d-1}\,
           \frac{d-1}2 \,\CF_{e}\wedge
           e\Big)\bigg) \\[4pt]
&=\big(\LL_{\xi}^{\om} \CF_{e}\,,\, \LL_{\xi}^{\om} \CF_{\om}\big) \\[4pt]
&=\delta^{\rm cov}_{(\xi, 0)}(\CF_{e},\CF_{\om}) \ ,
\end{align*}
where in the first equality we used $\LL_{\xi}= \iota_{\xi}\circ\dd
+\dd\circ \iota_{\xi}$ and $\iota_{\xi}(\om\wedge
\CF_{\om})=\iota_{\xi}\omega\cdot \CF_{\om} - \om \wedge \iota_{\xi} \CF_{\om}$,
together with the definition of $\LL_{\xi}^{\om}$ acting on the
Euler--Lagrange derivatives which are forms valued in vector bundles
associated to multivector representations of ${\sf SO}_+(p,q)$. In
the second equality we used again the definition of $\LL_{\xi}^{\om}$
together with $\dd^{\om}$, while in the fourth equality we used the
Noether identity corresponding to invariance under local
pseudo-orthogonal rotations. From this perspective, the input of the Noether identities
is crucial. The naive bootstrap method, excluding the demand of cyclicity, would
result in a simpler version of the brackets avoiding the use of the
Noether identities. However, the resulting $L_{\infty}$-algebra would
not be cyclic with respect to the natural pairing introduced in
Section~\ref{sec:cyclicpairing}: The requirement of cyclicity modifies
the brackets via the application of the Noether identities.

Another new feature which appears here is in the closure of the gauge
transformations. The covariant brackets also encode these, but now in
the more general sense \eqref{eq:closure} where the
bracket of the gauge algebra is field-dependent (but closure still
holds off-shell):
\begin{align*}
\big[\delta^{\rm cov}_{(\xi_{1},\rho_{1})} \,,\, \delta^{\rm cov}_{(\xi_{2},\rho_{2})}\big](e,\om)
= \delta^{\rm cov}_{[\![(\xi_{1},\rho_{1}),(\xi_{2},\rho_{2})]\!]^{\rm cov}_{(e,\om)}}(e,\om) \ ,
\end{align*}
where 
\begin{align}
[\![(\xi_{1},\rho_{1})\,,\,(\xi_{2},\rho_{2})]\!]^{\rm cov}_{(e,\om)}
  &= -\,\ell_{2}^{\rm
    cov}\big((\xi_{1},\rho_{1})\,,\,(\xi_{2},\rho_{2})\big) \nn
  \\ & \quad \,-
    \ell_{3}^{\rm
    cov}\big((\xi_{1},\rho_{1})\,,\,(\xi_{2},\rho_{2})\,,\,(e,\om)\big)
  +\frac{1}{2}\,\ell^{\rm
       cov}_{4}\big((\xi_{1},\rho_{1})\,,\,(\xi_{2},\rho_{2})\,,\,(e,\om)\,,\,(e,\om)\big) \nn \\[4pt]
&= \big(-[\xi_1,\xi_2]\,,\, [\rho_1,\rho_2]+\iota_{\xi_1}\iota_{\xi_2}R\big)
       \ . 
\label{eq:covnonLie}
\end{align}
This encodes directly the possible non-integrability of the horizontal lifting corresponding to each connection \eqref{nonint}. In particular, this means that the space of fields $V_1^{\rm cov}$
does \emph{not} form a module over the Lie algebra of gauge transformations
on $V_0^{\rm cov}$, in marked contrast with the non-covariant
approach, see \eqref{eq:gaugealgebra}. This formula is also noted
in~\cite{ECBV}, albeit in the dual and shifted picture in which the Lie derivative is viewed as an odd operator, where it is shown that the usual Cartan identity
\begin{align}\label{eq:CartanidLiecomm}
\LL_{[\xi_1,\xi_2]}=\LL_{\xi_1}\circ\LL_{\xi_2}-\LL_{\xi_2}\circ\LL_{\xi_1}
\end{align}
is violated by the covariant Lie derivative
$\LL^\omega_\xi$ via a term involving the action of the
contracted curvature 
$\iota_{\xi_1} \iota_{\xi_2} R\in \Omega^{0}\big(M,\CCP \times_{\rm ad}
\mathfrak{so}(p,q)\big)$, as in \eqref{eq:covnonLie}. 

\subsection{Cyclic $L_\infty$-isomorphism}
\label{sec:ECPiso}

In this section we have introduced two $L_\infty$-algebra formulations
of ECP gravity, one local and the other capturing the requisite
covariance properties for non-trivial spacetimes $M$. We
would now like to show that these two formulations are physically
equivalent locally, in the sense discussed in
Section~\ref{sec:Linftygft}. In fact, we exhibit a stronger result:
In the case where the underlying manifold $M$ is parallelizable, the two
theories are equivalent in the sense that their underlying
$L_{\infty}$-algebras are isomorphic. Indeed, all vector bundles in
question are then trivial, and so the underlying vector spaces of the
covariant and non-covariant formulations are identical, $V^{\rm cov}=
V$. The $L_{\infty}$-morphism we present here has been constructed partly via the help of dualisation from the symplectomorphism demonstrated in \cite{ECBV} for d=4. In fact since the map does not interact with dynamics, it has formally the same form in any dimension.

Let $\{\psi^{\rm cov}_n\}$ be the collection of multilinear graded antisymmetric
maps 
$$
\psi^{\rm cov}_n: \midwedge^{n} V^{\rm cov} \longrightarrow V \ ,
$$ 
of degree $|\psi^{\rm cov}_n|=1-n$ for $n\geq1$, defined as follows:
$\psi^{\rm cov}_1:V^{\rm cov}\to V$ is the identity map
$$
\psi^{\rm cov}_1(v) = v
$$
for all $v\in V^{\rm cov}$, the map $\psi^{\rm cov}_2:\midwedge^{2} V^{\rm cov}
\rightarrow V$ has only non-trivial components given by
\begin{align*}
\psi^{\rm cov}_2\big( (\xi,\rho)\,,\, (e,\om) \big) &=  \big(0\,,\,-\iota_{\xi}\om\big) \ \in
  \ V_{0} \ , \\[4pt]
\psi^{\rm cov}_2\big((\xi,\rho)\,,\,(\CX,\CP)\big) &=
                                                     \big(0\,,\,-(-1)^{d-1}\,
                                                     \iota_{\xi}\CP\big) \ \in
                                           \ V_{2} \ , \\[4pt]
\psi^{\rm cov}_2\big((e,\om)\,,\,(\CX,\CP)\big) &= \big(-\dd x^{\mu}\otimes \Tr(
                                        \iota_\mu\om \dwedge \CP)\,,\,0\big) \
                                        \in \ V_{3} \ ,
\end{align*}
while $\psi^{\rm cov}_{n}=0$ for all $n\geq 3$. Then $\{\psi^{\rm cov}_n\}$ is a
cyclic $L_\infty$-isomorphism between the cyclic $L_\infty$-algebras
$\big(V^{\rm cov}, \{\ell^{\rm cov}_{n}\}, \langle -, - \rangle\big)$
and $\big(V,\{\ell_{n}\}, \langle-,- \rangle\big)$. One easily
verifies the Seiberg--Witten maps from
\eqref{eq:fieldmorphism}--\eqref{eq:gaugemorphism} in this instance with
\begin{align*}
(e,\om) ^{\rm cov} =(e,\om) \ , \quad (\CF_e,\CF_\om)^{\rm cov} =(\CF_e,\CF_\om) \qquad \mbox{and} \qquad (\xi,\rho)^{\rm cov} =(\xi,\rho-\iota_\xi\om) \ ,
\end{align*}
so that $\delta^{\rm
  cov}_{(\xi,\rho)}(e,\om)=\delta_{(\xi,\rho-\iota_\xi\om)}(e,\om)$
and $\delta^{\rm
  cov}_{(\xi,\rho)}(\CF_e,\CF_\om)=\delta_{(\xi,\rho-\iota_\xi\om)}(\CF_e,\CF_\om)$
with the gauge algebra mapping as
\begin{align*}
[\![(\xi_{1},\rho_{1})\,,\,(\xi_{2},\rho_{2})]\!]^{\rm cov}_{(e,\om)}
  = \ & -\ell_2\big((\xi_1,\rho_1-\iota_{\xi_1}\om) \,,\,
  (\xi_2,\rho_2-\iota_{\xi_2}\om)\big) \nn \\ & + \big(\xi_2-\xi_1\,,\,\rho_2-\rho_1 +
  \iota_{\xi_1}\delta_{(\xi_2,\rho_2-\iota_{\xi_2}\om)}\om -
  \iota_{\xi_2}\delta_{(\xi_1,\rho_1-\iota_{\xi_1}\om)}\om\big) \ .
\end{align*}

Despite their simplicity, showing that the maps $\{\psi^{\rm cov}_n\}$ satisfy the required relations
\eqref{eq:morphismrels} of an $L_\infty$-morphism is, like the proof
of the homotopy relations, a tedious calculation which
largely does not depend on the spacetime dimension $d$. We give the
proof for the case $d=3$ in Appendix~\ref{app:3dmorrels}; the proof is similar for
$d \geq 4$.

Because $\{\psi^{\rm cov}_n\}$ is an $L_\infty$-morphism, since
$\psi^{\rm cov}_1$ is the identity it follows that $\{\psi^{\rm cov}_n\}$ is an
$L_\infty$-isomorphism. To check cyclicity of the map, since the
cyclic pairing is the same on both $L_\infty$-algebras of the gravity
theory, it follows immediately that 
\begin{align*}
\langle\psi^{\rm cov}_{1}(v_{1}), \psi^{\rm cov}_{1}(v_2) \rangle =\langle v_{1}, v_{2}
  \rangle
\end{align*}
for all $v_{1},v_{2} \in V^{\rm cov}$, again because $\psi^{\rm cov}_{1}$ is the identity.
Furthermore, it is a straightforward degreewise calculation to check
that
\begin{align*}
(-1)^{|v_1|}\,\langle \psi^{\rm cov}_{1}(v_{1}),\psi^{\rm cov}_{2}(v_{2},v_{3}) \rangle -
  \langle \psi^{\rm cov}_{2}(v_{1},v_{2}),\psi^{\rm cov}_{1}(v_{3})\rangle =0
\end{align*}
and
\begin{align*}
\langle\psi^{\rm cov}_2(v_1,v_2),\psi^{\rm cov}_2(v_3,v_4)\rangle = 0 \ ,
\end{align*}
for all $v_{1},v_{2},v_{3},v_4 \in V^{\rm cov}$. The remaining
cyclicity relations are all trivial.

\section{BV--BRST formalism for Einstein--Cartan--Palatini gravity}
\label{sec:BV-BRST}

The duality between differential graded commutative algebras and
$L_{\infty}$-algebras of finite type discussed in
Section~\ref{sec:dgca} converts the BV complex of a classical field
theory into an $L_\infty$-algebra as described in \cite{Linfty}, and
\emph{vice versa}.  In this section we shall
explicitly demonstrate this fact in the case of the non-covariant
ECP formalism, after reviewing the BRST complex of ECP gravity~\cite{BB,MSS,Piguet} and its augmented BV-BRST version~\cite{ECBV}, following the conventions of \cite{BVChristian}, where one may find a detailed introduction to the subject. The covariant BV-BRST complex of \cite{ECBV} proceeds analogously and is indeed dual to the covariant $L_\infty$-algebra presented in the last section.

\subsection{BRST complex}
\label{sec:ECPBRST}

The BRST complex for ECP
gravity in $d$ dimensions is obtained as the Chevalley--Eilenberg
resolution for the quotient of the space of fields
\eqref{eq:ECPfieldspace} by the gauge algebra
\eqref{eq:ECPgaugealg}. It has underlying vector space
\begin{align*}
\mathscr{F}_{\textrm{\tiny BRST}}&={\scrF_{\textrm{\tiny BRST}}}\,_0 \ 
                                   \oplus \ {\scrF_{\textrm{\tiny
                                   BRST}}}\,_{-1} \ ,
\end{align*}
where
\begin{align*}
{\scrF_{\textrm{\tiny BRST}}}\,_0 &= \Omega^{1}(M,\FR^{p,q})\times
  \Omega^{1}\big(M,\mathfrak{so}(p,q) \big) \ , \\[4pt]
{\scrF_{\textrm{\tiny BRST}}}\,_{-1} &= \Gamma [1](TM)\times \Omega^{0} [1]\big(M,\mathfrak{so}(p,q)
  \big) \ ,
\end{align*}
so that the dynamical fields are elements $(e,\omega) \in
{\mathscr{F}_{\textrm{\tiny BRST}}}\,_{{0}}$ and the gauge parameters
are elements\footnote{Strictly speaking these should be denoted as
  $({}^s\xi,{}^s\rho)$, where $(\xi,\rho)$ are the gauge parameters
  introduced in Section~\ref{sec:ECPgravity} and $s$ is the suspension isomorphism in \eqref{eq:shiftiso} below, but we do not indicate $s$ explicitly in order to streamline our formulas in the following.}
$(\xi,\rho) \in {\mathscr{F}_{\textrm{\tiny BRST}}}\,_{{-1}}$ with $e=e^a\, {\tt E}_a$,
$\omega=\omega^{ab}\, {\tt E}_{ba}$ and $\rho=\rho^{ab}\,
{\tt E}_{ba}$. In the
language of the BRST formalism, the elements of the odd degree spaces of gauge
parameters, which define sections of a distribution $\CCD\subset T
{\mathscr{F}_{\textrm{\tiny BRST}}}\,_{{0}}$ with a degree shift of~$1$, are called ghosts. 

On a local chart for $M$ with coordinates $x=(x^\mu)$, the fields are
expanded in holonomic bases as $e=e^{a}_{\mu}(x) \, \dd x^{\mu}\,
{\tt E}_{a}$ and $\xi=\xi^\mu(x) \, \partial_\mu$, where
$\partial_\mu=\frac\partial{\partial x^\mu}$, and similarly for the
rest of the fields. Abusing notation slightly, we shall consider the
components $e^{a}_{\mu}$ as elements of the dual space
$\mathscr{F}_{\textrm{\tiny BRST}}^\star$, thus viewing $e^{a}_{\mu}(x)$ as
coordinate functions on the infinite-dimensional vector space
$\Omega^{1}(M,\FR^{p,q})$ via the evaluation map
\begin{align*}  
e^{\prime\,a}_{\mu}{} _{|_{x}}: \Omega^{1}(M,\FR^{p,q}) \longrightarrow \FR \ , \quad
e \longmapsto e^{a}_{\mu}(x) \ ,
\end{align*}
and similarly for the rest of the fields. Abusing notation slightly,
we will sometimes drop the primes in the following.

The BRST differential $Q_{\textrm{\tiny
    BRST}}$ should act on a suitable space of functionals of the
field complex, which we denote by $\mathcal O(\scrF_{\textrm{\tiny BRST}})$. The precise
definition of this space will not be of concern to us, and it is often
different depending on the context and goals. For our purposes, the
following naive description will suffice: Consider
$\scrF_{\textrm{\tiny BRST}}^{\star}:={\sf Hom}(\scrF_{\textrm{\tiny BRST}},\FR)$, the space
of (continuous) $\FR$-linear functionals on $\scrF_{\textrm{\tiny
    BRST}}$; note that these are \emph{not} sections of the dual
  bundles. This space includes the coordinate maps
$e^{\prime\,a}_{\mu}{}_{|_{x}}$ above, as well as maps factoring through
the jet bundles, such as $\partial_{\nu}
e^{\prime\,a}_{\mu}{}_{|_{x}}$, which extract the values of
derivatives of the fields at a point in a specified coordinate chart
of the underlying manifold $M$. We shall take 
$
\mathcal O (\scrF_{\textrm{\tiny BRST}}):=
\text{\Large$\odot$}^{\bullet}_{\FR}\, \scrF_{\textrm{\tiny
    BRST}}^{\star},
$
the symmetric tensor algebra over $\FR$ of the dual of $\mathscr{F}_{\textrm{\tiny BRST}}$, as the space of polynomial functionals on the field complex. The usual subtleties regarding the topology of $\CF_{\textrm{\tiny BRST}}^{\star}$ and in which category the tensor product is taken are treated in detail in \cite{Costello,Costello2}. For our purposes, it will be safe to treat this tensor product formally as the algebraic tensor product.

Then $Q_{\textrm{\tiny
    BRST}}$ is a degree~$1$ derivation 
$Q_{\textrm{\tiny BRST}}: \text{\Large$\odot$}_\FR^\bullet \scrF^{\star}_{\textrm{\tiny BRST}}
\rightarrow \text{\Large$\odot$}_\FR^\bullet \scrF^{\star}_{\textrm{\tiny BRST}}$ such
that $Q_{\textrm{\tiny BRST}}^{2}=0$. By virtue of being a
derivation, $Q_{\textrm{\tiny BRST}}$ is completely determined by its action on the basis elements
of $\scrF^{\star}_{\textrm{\tiny BRST}}$. For Einstein--Cartan--Palatini gravity,
this takes the form~\cite{BB,MSS,Piguet} 
\begin{align}
Q_{\textrm{\tiny BRST}} e^{a}_{\mu}  &= \LL_{\xi} e^{a}_{\mu} -(\rho\cdot
                 e)^{a}_{\mu}  =(\xi^{\nu}\odot \partial_{\nu}
                 e^{a}_{\mu} +
                 e^{a}_{\nu}\odot \partial_{\mu}\xi^{\nu}) -
                 \rho^{a}{}_{b}\odot e^{b}_{\mu} \ , \nn \\[4pt]
Q_{\textrm{\tiny BRST}}\om^{ab}_{\mu} &= \LL_{\xi} \om^{ab}_{\mu} + \dd^\om\rho ^{ab}{}_{\mu}
                  = (\xi^{\nu}\odot \partial_{\nu}
                  \om^{ab}_{\mu}+\om^{ab}_{\nu}\odot \partial_{\mu}
                  \xi^{\nu}) + \partial_{\mu} \rho^{ab}\odot 1 +
                  \om^{ac}_{\mu} \odot \rho_c{}^{b} -
                  \rho^{a}{}_{c}\odot \om^{cb}_{\mu} \ , \nn \\[4pt]
Q_{\textrm{\tiny BRST}}\xi^{\mu}&=\tfrac{1}{2}\, [\xi,\xi]^{\mu} =
            \xi^{\nu}\odot \partial_{\nu}\xi^{\mu} \ , \nn \\[4pt]
Q_{\textrm{\tiny BRST}}\rho^{ab}&= \LL_\xi\rho^{ab}
                                     -\tfrac12\,[\rho,\rho]^{ab} = 
\xi^{\nu} \odot \partial_{\nu} \rho^{ab}
               -\rho^{a}{}_{c} \odot \rho^{cb} \ . \label{eq:BRST}
\end{align}
The BRST operator $Q_{\textrm{\tiny BRST}}$ encodes the symmetries of
ECP theory, that is, physical
(gauge-invariant) states modulo gauge transformations are classes in the degree~$0$ cohomology
of the BRST complex; in particular, the action functional
$S_{\textrm{\tiny ECP}}$ is a cocycle in degree~$0$: $Q_{\textrm{\tiny BRST}}S_{\textrm{\tiny ECP}}=0$. In a non-holonomic basis of vector fields $\{p_{\alpha}\}$ for
$\Gamma(TM)$ with Lie brackets
$[p_{\alpha},p_{\beta}]=f^{\gamma}{}_{\alpha\beta}\, p_{\gamma}$, the
third BRST
transformation takes the form
\begin{align*}
Q_{\textrm{\tiny BRST}}\xi^{\alpha}=\tfrac{1}{2}\,
[\xi,\xi]^{\alpha} = \tfrac{1}{2}\, f^{\alpha}{}_{\beta \gamma}\,
\xi^{\beta}\odot \xi^{\gamma} + \xi^{\beta}\odot
p_{\beta}(\xi^{\alpha}) \ , 
\end{align*}
which illustrates the similarity to the Chevalley--Eilenberg
differential of a finite-dimensional Lie algebra. However, 
this is \emph{not} the Chevalley--Eilenberg dual of $\scrF_{\textrm{\tiny
    BRST}}$ as a $C^\infty(M)$-module; for instance, restricting to $\big(\Gamma(TM), [-,-]\big) $, this would give the de~Rham complex $\big(\Omega^{\bullet}(M), \dd\big)$ instead.

The BRST differential contains all of the kinematical gauge structure of
ECP gravity. For this, we note that the
differential $Q_{\textrm{\tiny BRST}}$ dualizes to a codifferential
$$
\DD_{\textrm{\tiny BRST}}=Q_{\textrm{\tiny BRST}}^{\star}
:\text{\Large$\odot$}^\bullet \scrF_{\textrm{\tiny BRST}}
\longrightarrow \text{\Large$\odot$}^\bullet \scrF_{\textrm{\tiny
    BRST}} \ ,
$$
which may be decomposed as 
$$
{\rm pr}_{\scrF_{\textrm{\tiny BRST}}}\circ \DD_{\textrm{\tiny BRST}}=\sum_{n=1}^\infty \, {\DD_{\textrm{\tiny BRST}}}\,_{n}
$$
where the components are maps ${\DD_{\textrm{\tiny BRST}}}\,_{n}: \text{\Large$\odot$}^{n} \scrF_{\textrm{\tiny BRST}}
\rightarrow \scrF_{\textrm{\tiny BRST}}$. Introduce
the suspension isomorphism $s: \scrF_{\textrm{\tiny BRST}}[-1] \rightarrow
\scrF_{\textrm{\tiny BRST}}$, which induces an isomorphism of graded algebras given by 
\begin{align}\label{eq:shiftiso}
s^{\otimes n}: \midwedge^n \scrF_{\textrm{\tiny BRST}}[-1] & \longrightarrow
             \text{\Large$\odot$}^n \scrF_{\textrm{\tiny BRST}} \ , \nn \\
{}^{s^{-1}}v_{1} \wedge \cdots \wedge {}^{s^{-1}}v_{n} & \longmapsto
               (-1)^{\sum_{j=1}^{n-1}\, (n-j)\, |{}^{s^{-1}}v_{j}|} \,
               v_{1}\odot \cdots \odot v_{n} \ .
\end{align}
Then the graded antisymmetric brackets of the kinematical part of the
$L_{\infty}$-algebra that we defined in Section~\ref{sec:LinftyECP} are
given exactly by 
\begin{align*} 
\ell_{n} := s^{-1} \circ {\DD_{\textrm{\tiny BRST}}}\,_{n} \circ s^{\otimes n}: \midwedge^{n}
  \scrF_{\textrm{\tiny BRST}}[-1] \longrightarrow
  \scrF_{\textrm{\tiny BRST}}[-1] \ ,
\end{align*} 
and nilpotence $Q_{\textrm{\tiny BRST}}^{2}=0$ translates to the
homotopy relations for the brackets \cite{BVChristian}. Dualizing back
and forth in this infinite-dimensional case is a delicate issue;
however, the formal dualization below make sense because one may
interpret the brackets $\ell_{n}$ as maps between respective jet bundles, and then dualize pointwise.  

We will now indicate how to explicitly calculate these brackets. For a diffeomorphism $\xi \in {\scrF_{\textrm{\tiny BRST}}}\,_{-1}$ and a coframe field $e \in {\scrF_{\textrm{\tiny BRST}}}\,_{0}$, using the natural duality pairing $\langle-|-\rangle$ between $\text{\Large$\odot$}^\bullet\scrF_{\textrm{\tiny BRST}}$ and $\text{\Large$\odot$}_\FR^\bullet\scrF_{\textrm{\tiny BRST}}^\star$ we get
\begin{align*}
\langle Q_{\textrm{\tiny BRST}}e^{\prime\,a}_{\mu} | \xi \odot e \rangle &=
                                                                \langle
                                                                \xi^{\prime\,\nu}
                                                                \odot \partial_{\nu}
                                                                e^{\prime\,a}_{\mu}
                                                                +e^{\prime\,a}_{\nu}
                                                                \odot \partial_{\mu}
                                                                \xi^{\prime\,\nu}|
                                                                \xi
                                                                \odot
                                                                e
                                                                \rangle
  \\[4pt] 
& =(-1)^{|e'|\, |\xi'|} \, (\xi^{\nu}\,\partial_{\nu}{e}^{a}_{\mu}+{e}^{a}_{\nu}\,\partial_{\mu}\xi^{\nu}) \\[4pt]
&=\langle e^{\prime\,a}_{\mu}| \LL_{\xi}e \rangle \\[4pt]
&=: \langle e^{\prime\,a}_{\mu}| (-1)^{|Q_{\textrm{\tiny BRST}}|\,|e'|}\,
  {\DD_{\textrm{\tiny BRST}}}\,_{2}(\xi \odot e) \rangle \\[4pt]
&= \langle e^{\prime\,a}_{\mu}|{\DD_{\textrm{\tiny BRST}}}\,_{2}(\xi \odot e) \rangle
\end{align*}
where we used $|Q_{\textrm{\tiny BRST}}|=1$, $|e'|=0$ and $|\xi|=-1$. Hence ${\DD_{\textrm{\tiny BRST}}}\,_{2}(\xi \odot e)=\LL_{\xi}e$, and so\footnote{Yet another
  exterior product appears here: A wedge product inside a bracket
  denotes the exterior product of $\midwedge^\bullet
  \scrF_{\textrm{\tiny BRST}}[-1]$.} 
\begin{align*}
\ell_{2}({}^{s^{-1}}\xi \wedge {}^{s^{-1}} e)&= s^{-1} \circ {\DD_{\textrm{\tiny
                                     BRST}}}\,_{2} \circ (s\otimes
                                     s)({}^{s^{-1}}\xi\wedge {}^{s^{-1}} e)
  \\[4pt]
&= (-1)^{|{}^{s^{-1}}\xi| \, |s|}\, s^{-1} \circ {\DD_{\textrm{\tiny BRST}}}\,_{2} (\xi \odot e ) \\[4pt]
&=s^{-1} \circ {\DD_{\textrm{\tiny BRST}}}\,_{2} (\xi \odot e) \\[4pt]
&= \LL_{{}^{s^{-1}}\xi} {}^{s^{-1}}e \ ,
\end{align*} 
which agrees with \eqref{eq:ell2gaugefield}. Similarly, for a local 
pseudo-orthogonal rotation $\rho\in{\scrF_{\textrm{\tiny BRST}}}\,_{-1}$ we calculate
\begin{align*}
\langle Q_{\textrm{\tiny BRST}}e^{\prime\,a}_{\mu}| \rho \odot e\rangle &= \langle
                                                          -\rho^{\prime\,a}{}_{b}
                                                          \odot
                                                          e^{\prime\,b}_{\mu}
                                                          | \rho
                                                          \odot e
                                                          \rangle
  \\[4pt]
&=-(-1)^{|\rho'|\, |e'|}\,\rho^{a}{}_{b}\, {e}^{b}_{\mu} \\[4pt]
&=-\langle e^{\prime\,a}_{\mu} | \rho \cdot e\rangle \\[4pt]
&=: \langle e^{\prime\,a}_{\mu}|-(-1)^{|Q_{\textrm{\tiny BRST}}|\,|e'|}\,{\DD_{\textrm{\tiny BRST}}}\,_2(\rho \odot e)\rangle \ .
\end{align*} 
Hence ${\DD_{\textrm{\tiny BRST}}}\,_2(\rho \odot e)=- \rho \cdot e$, and so
\begin{align*}
\ell_{2}({}^{s^{-1}}\rho \wedge {}^{s^{-1}} e)&= s^{-1} \circ
                                         {\DD_{\textrm{\tiny BRST}}}\,_{2} \circ (s\otimes
                                         s)({}^{s^{-1}}\rho\wedge {}^{s^{-1}}
                                         e) \\[4pt]
&= (-1)^{|{}^{s^{-1}}\rho| \, |s|}\, s^{-1} \circ {\DD_{\textrm{\tiny BRST}}}\,_{2} (\rho \odot e ) \\[4pt]
&=s^{-1} \circ {\DD_{\textrm{\tiny BRST}}}\,_{2} (\rho \odot e) \\[4pt]
&=- {}^{s^{-1}}\rho \cdot {}^{s^{-1}}e \ ,
\end{align*} 
which also agrees with the corresponding brackets in \eqref{eq:ell2gaugefield}.
In a similar fashion one dualizes the rest of $Q_{\textrm{\tiny BRST}}$ and
recovers all of the brackets containing the kinematical gauge structure of the
field theory, which in this case are guaranteed to satisfy the
corresponding homotopy relations, since $Q_{\textrm{\tiny BRST}}^{2}=0$.

\subsection{BV--BRST complex}
\label{sec:ECPBVBRST}

We now need to augment the kinematical gauge structure provided in the
BRST formalism of Section~\ref{sec:ECPBRST} by the dynamical data
comprising the field equations and the Noether identities. Starting from the classical BRST complex $\scrF_{\textrm{\tiny
    BRST}}$ from Section~\ref{sec:ECPBRST}, the BV complex provides the Koszul--Tate resolution of its degree~$0$ cohomology (gauge
equivalence classes of fields) modulo the ideal of Euler--Lagrange derivatives. It is defined
as the functionals $\mathcal O(\scrF_{\textrm{\tiny BV}})$ on its shifted cotangent bundle:
\begin{align} \label{eq:BVT*}
\scrF_{\textrm{\tiny BV}}:= T^{*}[-1]\scrF_{\textrm{\tiny BRST}} \ .
\end{align}
This takes the form
\begin{align}\label{eq:JBV}
\scrF_{\textrm{\tiny BV}}={\scrF_{\textrm{\tiny BV}}}\,_{-1} \ \oplus
  \ {\scrF_{\textrm{\tiny BV}}}\,_{0} \ \oplus \ 
  {\scrF_{\textrm{\tiny BV}}}\,_{1} \ \oplus \ {\scrF_{\textrm{\tiny BV}}}\,_{2} \ ,
\end{align}
with
\begin{align*}
{\scrF_{\textrm{\tiny BV}}}\,_{-1} &= \Gamma [1](TM)\times \Omega^{0} [1]\big(M,\mathfrak{so}(p,q)
  \big) \ , \\[4pt]
{\scrF_{\textrm{\tiny BV}}}\,_{0} &= \Omega^{1}(M,\FR^{p,q})\times
  \Omega^{1}\big(M,\mathfrak{so}(p,q) \big) \ , \\[4pt]
{\scrF_{\textrm{\tiny BV}}}\,_{1} &=
                                    \Omega^{d-1}[-1]\big(M,\midwedge^{d-1}(\FR^{p,q})
                                    \big) \times \Omega^{d-1}[-1]\big(M,\midwedge^{d-2}
                                                            (\FR^{p,q})\big) \ , \\[4pt] 
{\scrF_{\textrm{\tiny BV}}}\,_{2} &=
                                    \Omega^{1}[-2]\big(M,\Omega^{d}(M)\big)
                                    \times
                                    \Omega^{d}[-2]\big(M,\midwedge^{d-2}(\FR^{p,q})\big)
                                    \ , 
\end{align*}
where elements of the degree~$1$ spaces are called antifields, which we denote by
$(e^{\dagger}, \om^{\dagger})$, while elements of the degree~$2$ spaces are called
antighosts,\footnote{These are not the antighost fields which are sometimes introduced in the BRST formalism, but this terminology is convenient in the present context.} denoted by $(\xi^{\dagger}, \rho^{\dagger})$. The antifields and
antighosts are transversal sections to the gauge orbits in the space of
fields and ghosts (odd gauge
parameters) respectively, and they may be paired with fields and ghosts using the pairing which defines the
action functional $S_{\textrm{\tiny ECP}}$ to give an $\FR$-valued $d$-form on
$M$. 

Explicitly, the pairing in the Einstein--Cartan--Palatini action
functional may be regarded as a non-degenerate bilinear form
\begin{align*}
\Tr(- \dwedge -):
  \Omega^{d-k}\big(M,\midwedge^{d-k} (\FR^{p,q})\big) \otimes
  \Omega^{k}\big(M,\midwedge^{k}(\FR^{p,q})\big) \longrightarrow
  \Omega^{d}(M)
\end{align*}
which defines the BV pairing $\langle-,-\rangle_{\textrm{\tiny BV}}$ and dualizes the fields and ghosts to antifields and antighosts
through the assignments
\begin{align*}
 e\in \Omega^{1}(M,\FR^{p,q}) &\implies e^{\dagger} \in
                              \Omega^{d-1}\big(M,
                              \midwedge^{d-1}(\FR^{p,q}) \big) \ , \\[4pt] 
 \om\in \Omega^{1}\big(M,\midwedge^{2} (\FR^{p,q})\big) &\implies
                                                            \om^{\dagger}\in
                                                            \Omega^{d-1}\big(M,\midwedge^{d-2}
                                                            (\FR^{p,q})\big)
  \ , \\[4pt] 
 \rho \in \Omega^{0}\big(M,\midwedge^{2} (\FR^{p,q})\big)
                            &\implies \rho^{\dagger} \in
                              \Omega^{d}\big(M,\midwedge^{d-2}(\FR^{p,q})\big)
  \ , \\[4pt]
 \xi \in \Gamma(TM) &\implies \xi^{\dagger} \in
                      \Omega^{1}\big(M,\Omega^{d}(M)\big) \ ,
\end{align*}
where the final duality is defined by $\iota_{\xi}\xi^{\dagger} \in
\Omega^{d}(M)$. Here the spaces of antifields and antighosts correspond exactly to
the spaces of Euler--Lagrange derivatives and of Noether identities; indeed, the graded vector space
\eqref{eq:JBV} is just the
vector space \eqref{eq:ECPvectorspace} underlying our
$L_\infty$-algebra with degrees shifted by~$1$:
\begin{align*}
{\scrF_{\textrm{\tiny BV}}}=V[1] \ ,
\end{align*} 
because of the degree shift by $-1$ in the cotangent bundle \eqref{eq:BVT*}.

The space $\scrF_{\textrm{\tiny BV}}$, being a cotangent bundle, further has the structure of a graded
(infinite-dimensional) $(-1)$-symplectic manifold, where the canonical
symplectic two-form is given by 
\begin{align} \label{eq:BVomega}
\boldsymbol\om_{\textrm{\tiny BV}}:= \int_{M}\, \Tr\big(\delta e \dwedge
  \delta e^{\dagger} +\delta \om \dwedge \delta \om^{\dagger} - \delta
  \rho \dwedge \delta \rho^{\dagger}\big) - \int_M\, \iota_{\delta \xi}
  \delta\xi^{\dagger} \ .
\end{align}
As expected, this is the shifted (and dual) version of the cyclic pairing \ref{eq:ECPpairing}. It induces a graded Poisson bracket $\{-,-\}_{\textrm{\tiny BV}}$ of
smooth functions on $\scrF_{\textrm{\tiny BV}}$, that we consider to
be the space $\text{\Large$\odot$}_\FR^\bullet \scrF_{\textrm{\tiny
    BV}}^{\star}$ as in Section~\ref{sec:ECPBRST}, which is called the antibracket. In this sense, fields and their
antifield partners may be regarded as canonically conjugate
variables. 

We would now like to define the BV extension of the action functional $S_{\textrm{\tiny
    ECP}}$ to a local action functional $S_{\textrm{\tiny BV}}$ on
$\scrF_{\textrm{\tiny BV}}$ of degree~$0$ that satisfies the classical master
equation
\begin{align}\label{eq:CME}
\{S_{\textrm{\tiny BV}},S_{\textrm{\tiny BV}}\}_{\textrm{\tiny BV}}=0 \ .
\end{align}
This is always possible for BV complexes which are built out of
underlying BRST complexes~\cite{BV81}, and it automatically extends
the BRST differential $Q_{\textrm{\tiny BRST}}$ on
$\scrF_{\textrm{\tiny BRST}}$ to the BV differential $Q_{\textrm{\tiny
    BV}}$ on $\scrF_{\textrm{\tiny BV}}$ as the cotangent lift
\begin{align*}
Q_{\textrm{\tiny BV}}F:= \{S_{\textrm{\tiny
                           BV}},F\}_{\textrm{\tiny BV}} = -\int_M \ \sum_{A\in\{e,\om,\rho,\xi\}} \, \bigg(\Big\langle\frac{\delta S _{\textrm{\tiny BV}}}{\delta A}\,,\,\frac{\delta F}{\delta A^\dag}\Big\rangle^\star_{\textrm{\tiny BV}} + \Big\langle\frac{\delta S _{\textrm{\tiny BV}}}{\delta A^\dag}\,,\,\frac{\delta F}{\delta A}\Big\rangle^\star_{\textrm{\tiny BV}} \bigg) \ ,
\end{align*}  
where $\langle-,-\rangle^\star_{\textrm{\tiny BV}}$ denotes the dual BV pairing and  $\frac{\delta}{\delta e}:=\frac\delta{\delta e_\mu^a}\otimes{\tt F}^a\,\partial_\mu$ with ${\tt F}^a$ dual to ${\tt E}_a$ and $\partial_\mu$ dual to $\dd x^\mu$, and similarly for the other fields.
The classical master equation \eqref{eq:CME} together with the graded
Jacobi identity for the antibracket guarantee that $Q_{\textrm{\tiny BV}}^{2}=0$. 

A solution to \eqref{eq:CME} is given by the BV pairing as~\cite{BV81}
\begin{align}\label{eq:BVaction}
S_{\textrm{\tiny BV}}=S_{\textrm{\tiny ECP}} + \int_{M}\, \Tr
  \big(Q_{\textrm{\tiny BRST}}e\dwedge e^{\dagger} +Q_{\textrm{\tiny BRST}}\om
  \dwedge \om^{\dagger} - Q_{\textrm{\tiny BRST}}\rho \dwedge
  \rho^{\dagger}\big) - \int_M\, \iota_{Q_{\textrm{\tiny
  BRST}}\xi}\xi^{\dagger} \ ,
\end{align}
where $Q_{\textrm{\tiny BRST}}e:=Q_{\textrm{\tiny BRST}}e^{a}_{\mu}\otimes
{\tt E}_{a}\, \dd x^{\mu}\in\text{\Large$\odot$}_\FR^\bullet \scrF_{\textrm{\tiny BRST}}^{\star}\otimes
\scrF_{\textrm{\tiny BRST}}$, and similarly for the other fields. This defines the \emph{minimal} part of the
BV extension of the Einstein--Cartan--Palatini theory; the
non-minimal part of the action functional is trivial from our
classical perspective as it involves products of auxiliary
fields and antighosts, and so is the $Q_{\textrm{\tiny BV}}$-exact
variation of a gauge fixing fermion. Explicitly, the minimal BV extension of the action functional for $d$-dimensional gravity reads
\begin{align*}
S_{\textrm{\tiny BV}} &= \int_M\, 
  \Tr\Big(\frac{1}{d-2} \,e^{d-2}\dwedge R +\frac{1}{d}\,\Lambda\,
  e^{d} + \big(\LL_\xi e-\rho\cdot e\big)\dwedge e^\dag +
                         \big(\LL_\xi\om+\dd^\om\rho\big)\dwedge\om^\dag
  \Big) \\ & \quad - \frac12\, \int_M \, \Tr\big((2\,\LL_\xi\rho-
             [\rho,\rho])\dwedge\rho^\dag\big) - \frac12\,
             \int_M\, \iota_{[\xi,\xi]}\xi^\dag \ .
\end{align*}

From the form of the BV action
functional \eqref{eq:BVaction}, it immediately follows that the
pullback of the BV differential to the natural Lagrangian
submanifold of $\scrF_{\textrm{\tiny BV}}$ provided by the zero section is given by
\begin{align*}
{Q_{\textrm{\tiny BV}}}_{|_{\scrF_{\textrm{\tiny
  BRST}}}}=Q_{\textrm{\tiny BRST}} \ ,
\end{align*}
and so this part of the BV differential $Q_{\textrm{\tiny BV}}$
includes the kinematical gauge sector of our $L_{\infty}$-algebra via
dualization. The BV transformations of the antifields
$Q_{\textrm{\tiny BV}}e^{\dagger}$ and $Q_{\textrm{\tiny
    BV}}\om^{\dagger}$ incorporate the dynamical brackets of our
$L_{\infty}$-algebra, while $Q_{\textrm{\tiny BV}}\rho^\dag$ and
$Q_{\textrm{\tiny BV}}\xi^\dag$ encode brackets corresponding to the
Noether identities and the action of the gauge parameters on the space
of Noether identities~\cite{Henneaux:1989jq,Dragon:2012au,BVChristian}. This, together with the fact
that $Q_{\textrm{\tiny BV}}^{2}=0$, guarantee that the homotopy
relations are satisfied in any dimension $d\geq3$. The canonical symplectic
two-form \eqref{eq:BVomega} corresponds to the cyclic pairing of the
$L_\infty$-algebra introduced in Section~\ref{sec:cyclicpairing}, with
cyclicity being equivalent to the $Q_{\textrm{\tiny BV}}$-invariance
of $\boldsymbol\om_{\textrm{\tiny BV}}$, and the sign difference comes
from the degree shifting.
Then cyclic $L_\infty$-morphisms correspond to cohomomorphisms of
$\scrF_{\textrm{\tiny BV}}$, in the sense discussed in
Section~\ref{sec:dgca}, whose duals preserve the symplectic
structures. In particular, cyclic quasi-isomorphisms result into equivalent field theories, thus in our case relating equivalent gravity theories.

After some tedious but straightforward calculation, one may arrive at
\begin{align}
Q_{\textrm{\tiny BV}}{e^{\dagger}}^{\,a_{1}\cdots
  a_{d-1}}_{\mu_{1}\cdots \mu_{d-1}}=&-\,
                                       e^{[a_{1}}_{[\mu_{1}}\cdots
                                       e^{a_{d-3}}_{\mu_{d-3}}\,
                                       R^{a_{d-2}a_{d-1}]}_{\mu_{d-2}\mu_{d-1}]}
-\Lambda\, e^{[a_1}_{[\mu_1}\cdots e^{a_{d-1}]}_{\mu_{d-1}]}                                        \nn \\
& +
                                       d\,
                                       \big({e^{\dagger}}^{\,a_{1}\cdots
                                       a_{d-1}}_{[\mu_{1}\cdots
                                       \mu_{d-1}}\, \partial_{\sigma]}\xi^{\sigma}
                                       - {e^{\dagger}}^{[a_{1}\cdots
                                       a_{d-1}}_{\mu_{1}\cdots
                                       \mu_{d-1}}\,
                                       \rho^{b]}{}_{b}\big)+\partial_{\sigma}\big(\xi^{\sigma}\, {e^{\dagger}}^{\,a_{1}\cdots
  a_{d-1}}_{\mu_{1}\cdots \mu_{d-1}} \big) \ , \nn \\[6pt]
Q_{\textrm{\tiny BV}}{\om^{\dagger}}^{\,a_{1}\cdots
  a_{d-2}}_{\mu_{1}\cdots
  \mu_{d-1}}=& -\, e^{[a_1}_{[\mu_1}\cdots e^{a_{d-3}}_{\mu_{d-3}} \, T^{a_{d-2}]}_{\mu_{d-2}\mu_{d-1}]} \nn\\& +
               (d-1)\,\rho^{[b}{}_{b}\,
               {\om^{\dagger}}^{\,a_{1}\cdots a_{d-2}]}_{\mu_{1}\cdots
               \mu_{d-1}} + d\, {\om^{\dagger}}^{\,a_{1}\cdots
                              a_{d-1}}_{[\mu_{1}\cdots
                              \mu_{d-1}}\, \partial_{\sigma]}\xi^{\sigma}
                              + \partial_{\sigma}\big(\xi^{\sigma}\,
                              {e^{\dagger}}^{\,a_{1}\cdots
                              a_{d-1}}_{\mu_{1}\cdots \mu_{d-1}}\big)
                              \ , \nn\\[6pt]
Q_{\textrm{\tiny BV}}{\rho^{\dagger}}^{\,a_{1}\cdots
  a_{d-2}}_{\mu_{1}\cdots \mu_d}=& -\, \frac{d-1}2\, e_{b[\mu_{1}}\,
                                     {e^{\dagger}}^{\,ba_{1}\cdots
                                     a_{d-2}}_{\mu_{2}\cdots \mu_d]}
                                     + (-1)^{d-1}\, \om^{a_{1}}{}_{b[\mu_{1}}{}\,
                                     {\om^{\dagger}} ^{\,ba_{2}\cdots
                                     a_{d-2}}_{\mu_{2}\cdots \mu_d]}
                                     + \partial_{[\mu_{1}}
                                     {\omega^{\dagger}}^{\,a_{1}\cdots
                                     a_{d-2}}_{\mu_{2}\cdots \mu_d]}
                                     \nn \\
                                     & -\rho^{a_{1}}{}_b{}\,{\rho^{\dagger}}^{\,ba_{2}\cdots
                                     a_{d-2}}_{\mu_{1}\cdots \mu_d}
                                     + \partial_{\sigma}\big(\xi^{\sigma}\,
                                       {\rho^{\dagger}}^{\,a_{1}\cdots
                                       a_{d-2}}_{\mu_{1}\cdots
                                       \mu_d}\big) \ , \nn \\[6pt]
Q_{\textrm{\tiny BV}}{\xi_{\nu}^{\dagger}}_{\mu_{1}\cdots
  \mu_d}=&-\varepsilon_{a_{1}\cdots a_d}\,\Big(\partial_{\nu}{e}^{a_{1}}_{[\mu_{1}}\, {e^{\dagger}}^{\,a_{2}\cdots a_d}_{\mu_{2} \cdots \mu_d]}-\partial_{[\mu_{1}}  {e}^{a_{1}}_{|\nu|}\, {e^{\dagger}}^{\,a_{2}\cdots a_d}_{\mu_{2}\cdots \mu_d]} - {e}^{a_{1}}_{|\nu|}\, \partial_{[\mu_{1}} {e^{\dagger}}^{\,a_{2}\cdots a_d}_{\mu_{2}\cdots \mu_d]}
\nn \\&\phantom{{}-\varepsilon_{a_{1}\cdots
        a_d}\,\Big({}}+(-1)^{d-1}\, \partial_{\nu}{\om}^{a_{1}a_{2}}_{[\mu_{1}}\,
        {\om^{\dagger}}^{\,a_{3}\cdots a_d}_{\mu_{2}\cdots
        \mu_d]}-(-1)^{d-1}\, \partial_{[\mu_{1}} {\om}^{a_{1}a_{2}}_{|\nu|}\, {\om^{\dagger}}^{\,a_{3}\cdots a_{d}}_{\mu_{2}\cdots \mu_d]} \nn \\
&\phantom{{}-\varepsilon_{a_{1}\cdots
  a_d}\,\Big({}} - (-1)^{d-1}\, {\om}^{a_{1}a_{2}}_{|\nu|}\, \partial_{[\mu_{1}} {\om^{\dagger}}^{\,a_{3}\cdots a_d}_{\mu_{2}\cdots \mu_d]} -\partial_{\nu}\rho^{a_{1}a_{2}}\,
  {\rho^{\dag}}^{\,a_{3}\cdots a_d}_{\mu_{1}\cdots\mu_d} \Big) \nn \\
&\phantom{{}-\varepsilon_{a_{1}\cdots
  a_d}\,\Big({}}
  + \partial_{\nu}\xi^{\sigma}\, {\xi^{\dag}_{\sigma}}_{\mu_{1}\cdots
  \mu_d} + \partial_{\sigma}\big(\xi^{\sigma}\,
  {\xi_{\nu}^{\dag}}_{\mu_{1}\cdots \mu_d}\big) \ . \label{eq:QBV}
\end{align}
These expressions should be understood as valued in
$\text{\Large$\odot$}_\FR^\bullet \scrF_{\textrm{\tiny BV}}^{\star}$,
but for brevity we do not write the symmetrized tensor products
explicitly. The first two transformations dualize to the brackets involving the
Euler--Lagrange derivatives with respect to the coframe
field $e$ and the spin connection $\omega$ respectively, and the
actions of gauge transformations on them. The last two transformations
dualize to the brackets involving the Noether identities corresponding to local pseudo-orthogonal and diffeomorphism gauge
symmetries, and the action of gauge transformations on them. In
general, the explicit proof of this dualization is a long and
cumbersome calculation, which is also largely dependent on the
spacetime dimension $d$. We illustrate how this works explicitly in
Appendix~\ref{app:4dBVBRST} for the case~$d=4$.

\section{Three-dimensional gravity}
\label{sec:3dgrav}

We will now specialise our discussion to the three-dimensional case, which has special features compared to higher dimensionalities. In particular, in this case the pertinent $L_\infty$-algebra is a differential graded Lie algebra. Using the special feature of general
relativity in three spacetime dimensions, where it defines a
topological field theory in the sense that it is absent of local
propagating degrees of freedom, we can make contact with the
$L_\infty$-algebra formulations of Section~\ref{sec:TQFT} through
explicit $L_\infty$-quasi-isomorphisms. This will show, in particular,
that in the $L_\infty$-algebra framework,
three-dimensional gravity (including degenerate metrics) is
pertubatively \emph{off-shell} equivalent to a Chern--Simons gauge theory, extending
the well-known on-shell equivalence \cite{Witten:1988hc}. This result may be interepreted as an extension of the recent analogue result which applies to the strictly non-degenerate sector \cite{Cattaneo:2017ztd}, phrased in the dual BV-BRST framework.

\subsection{Field equations}
\label{sec:3deom}

Specialising the discussion of Section~\ref{sec:ECPgravity}, the Einstein--Cartan--Palatini action functional for gravity in $d=3$ dimensions including cosmological
constant is given by
\begin{align}
S_{\textrm{\tiny ECP}}(e,\omega):&=\int_M\, \Tr\Big(e \dwedge R + \frac{\Lambda}{3}\,
             e\dwedge e \dwedge e\Big) \nn \\[4pt]
&=\int_M\, \Tr \bigg( \Big(e^{a} \wedge  R^{bc} + \frac{\Lambda}{3}\, e^{a}\wedge e^{b} \wedge e^{c} \Big)\, {\tt E}_{a}\wedge {\tt E}_{b}\wedge {\tt E}_{c}\bigg) \nn \\[4pt]
&= \int_M\, \varepsilon_{abc} \, \Big(e^{a}\wedge R^{bc}+ \frac{\Lambda}{3}\, e^{a}\wedge
  e^{b} \wedge e^{c} \Big) \ ,
\label{eq:ECP3daction}\end{align}
where $\varepsilon_{abc}$ is the Levi--Civita tensor. In this section we shall work in Lorentzian signature $(p,q)=(1,2)$ for definiteness, but our considerations apply equally well in Euclidean signature without substantial change.
The field equations are
\begin{align} \label{eq:eom3d}
R+ \Lambda\, e \dwedge
e=0 \qquad \mbox{and} \qquad T =0 \ .
\end{align}
When $e$ is invertible, these are equivalent to the three-dimensional
Einstein field equations up to local $\sSO_+(1,2)$ Lorentz transformations, whose
solutions are spacetimes of constant curvature. This means that there are no gravitational waves on
three-dimensional spacetimes, but even though all classical spacetimes obtained as solutions to \eqref{eq:eom3d} are
locally gauge equivalent, they can have different topology~\cite{Witten:1988hc}.

\subsection{$L_\infty$-algebra formulation}
\label{sec:3d}

The $L_{\infty}$-algebra corresponding to the gravity theory in $d=3$
dimensions from Section~\ref{sec:3deom} is given by the vector space 
\begin{align}\label{eq:3dV}
V:= V_{0} \oplus V_{1} \oplus V_{2} \oplus V_3
\end{align}
where 
\begin{align} 
V_{0}&=\Gamma(TM)\times \Omega^{0}\big(M,\mathfrak{so}(1,2)\big) \ , \nn
  \\[4pt]
V_{1}&= \Omega^{1}(M,\FR^{1,2}) \times
       \Omega^{1}\big(M,\mathfrak{so}(1,2)\big) \ , \label{eq:3dV012} \\[4pt]
V_{2}&=\Omega^{2}\big(M,\midwedge^{2}(\FR^{1,2})\big) \times
       \Omega^{2}(M,\FR^{1,2}) \ , \nn \\[4pt]
V_3&= \Omega^1\big(M,\Omega^3(M)\big) \times \Omega^3(M,\FR^{1,2})
\ . \nn
\end{align} 
The brackets of Section~\ref{sec:LinftyECP} reduce in this case to
\begin{flalign} 
\ell_{1}(\xi,\rho)&=(0,\dd\rho) \ \in \ V_{1} \ , \nn \\[4pt]
\ell_{1}(e,\omega)&=(\dd \om,\dd e) \ \in \ V_{2} \ , \nn
\\[4pt]
\ell_{1}(E,{\mit\Omega})&=(0,\dd {\mit\Omega}) \ \in \ V_{3} \ , \nn 
\\[6pt]
\ell_{2}\big((\xi_{1},\rho_{1})\,,\,(\xi_{2},\rho_{2})\big)&=\big([\xi_{1},\xi_{2}]\,,\,-[\rho_{1},\rho_{2}]+\xi_1(\rho_{2})
- \xi_{2}(\rho_{1})\big) \ \in \ V_{0} \ , \nn
\\[4pt]
\ell_{2}\big((\xi,\rho)\,,\,(e,\omega)\big)&=\big(-\rho \cdot e
+\LL_{\xi}e \,,\, -[\rho,\omega]+\LL_{\xi}\om\big) \ \in \ V_{1} \
, \label{eq:3dbrackets} \\[4pt]
\ell_{2}\big((\xi,\rho) \,,\, (E,{\mit\Omega})\big)&=\big(-
[\rho, E]+\LL_{\xi}E \,,\, -\rho \cdot
{\mit\Omega}+\LL_{\xi}{\mit\Omega} \big) \ \in
\ V_{2} \ , \nn \\[4pt]
\ell_{2}\big((\xi,\rho)\,,\,({\CX},{\CP})\big)&= \big
(\dd x^\mu\otimes\Tr(\iota_\mu \dd \rho \dwedge
{\CP}) +\LL_{\xi}{\CX}\,,\,-\rho
\cdot {\CP} +\LL_{\xi} {\CP}\big) \ \in \ V_3 \ , \nn \\[4pt] 
\ell_{2}\big((e_{1},\omega_{1})\,,\,(e_{2},\omega_{2})\big)&=-\big([\omega_{1}, \omega_{2}]
+2\,\Lambda\, e_{2}\dwedge e_{1} \,,\, \omega_{1}\wedge
e_{2} +\omega_{2} \wedge e_{1} \big) \ \in \ V_{2} \ , \nn \\[4pt]
\ell_{2}\big((e,\om)\,,\,(E,{\mit\Omega})\big)&=\Big(\dd x^\mu\otimes\Tr \big( \iota_\mu \dd e \dwedge E + \iota_\mu \dd\om \dwedge {\mit\Omega} \nn   -\iota_\mu e \dwedge \dd E - \iota_\mu \om\dwedge \dd {\mit\Omega}\big) \,, \nn \\ 
&\phantom{{}\bigg(\Tr \big( \iota_\mu \dd e \dwedge E +{}}
E\wedge e - \omega \wedge {\mit\Omega} \Big)  \ \in \ V_3 \ , \nn
\end{flalign}
while all the rest of the brackets vanish. Thus three-dimensional gravity is organised by a differential graded Lie algebra. Again the Lie derivative
$\LL_{\xi}$ acts via the Leibniz rule on ${\CX}\in
\Omega^{1}(M)\otimes \Omega^{3}(M)$. The proof of the homotopy
relations in this case is given in Appendix~\ref{app:3dhomotopy}.

As designed by
\eqref{gaugetransfA}--\eqref{eq:Noether}, these encode 
the Euler--Lagrange derivatives
\begin{align*} 
\mathcal{F}(e,\omega) &= (R+\Lambda\, e\dwedge e,T) \\[4pt]
&=(\dd \om,\dd e)+ \frac{1}{2} \, ([\omega,\omega]+2\,\Lambda\, e\dwedge e,2\,\omega\wedge e) \\[4pt]
&=\ell_{1}(e,\omega)-\frac{1}{2}\,
  \ell_{2}\big((e,\omega)\,,\,(e,\omega)\big) \ .
\end{align*}
Moreover, the action functional \eqref{eq:ECP3daction} can be written as in \eqref{action} using the cyclic pairing~\eqref{eq:ECPpairing}:
\begin{align*}
S_{\textrm{\tiny ECP}}(e,\omega) &= \int_M\, \Tr
              \bigg(e\dwedge\Big(\dd\omega+\frac12\, [\omega,\omega]\Big) +
              \frac{\Lambda}3\, e\dwedge e\dwedge e\bigg) \\[4pt]
&= \int_M\, \Tr \Big(\frac12\, (e\dwedge\dd\omega+\omega\dwedge \dd
  e) + \frac1{3!}\, \big(e\dwedge[\omega,\omega]+2\,\omega\dwedge(\omega\wedge e)+2\,\Lambda\,
  e\dwedge e\dwedge e\big) \Big)
  \\[4pt]
&= \frac12\, \big\langle(e,\omega)\,,\,(\dd \om,\dd e) \big\rangle +
  \frac1{3!}\, \big\langle (e,\omega)\,,\,([\omega,\omega] + 2\,
  \Lambda\, e\dwedge e , 2\,\omega\wedge
  e)\big\rangle \\[4pt]
&= \frac12\, \big\langle (e,\omega)\,,\,\ell_1(e,\omega) \big\rangle -
  \frac1{3!}\, \big\langle (e,\omega)\,,\,\ell_2\big(
  (e,\omega)\,,\,(e,\omega)\big) \big\rangle \ , 
\end{align*}
where in the second equality we integrated by parts on
$e\dwedge\dd\omega$ and used $e\dwedge[\omega,\omega] = 
\omega\dwedge (\omega\wedge e)$ which follows from the identity
\eqref{weirdformula}.

\subsection{$L_\infty$-quasi-isomorphism with $BF$ and Chern--Simons formulations}
\label{sec:ECPCS}

Three-dimensional gravity is very similar to $BF$ theory in three dimensions: The action functional \eqref{eq:ECP3daction} is just the action functional \eqref{eq:BFaction} with $\CCW=\FR^{1,2}$ and $\frg=\mathfrak{so}(1,2)$, along with an extra cosmological constant term. The usual shift symmetry \eqref{eq:shift} of $BF$ theory correspondingly contains a slight modification to accomodate for the cosmological constant term: for $\tau \in \Omega^{0}(M,\FR^{1,2})$, the action is invariant under
\begin{align} \label{eq:deltatau}
\delta_\tau e = \dd\tau+\omega\cdot\tau \qquad \mbox{and} \qquad
\delta_\tau\omega =2\, \Lambda \, e\dwedge\tau 
\end{align}
with corresponding Noether identity
\begin{align*}
\dd^{\om}\CF_{e}=\dd^{\om}R+2\,\Lambda\,\dd^{\om}e\dwedge e=2 \, \Lambda\, \CF_{\om}\dwedge e
\end{align*}
in $\Omega^3(M,\midwedge^2(\FR^{1,2}))$, where we wrote $\CF_\om:=T=\dd^\om e$ and $\CF_e:=R+\Lambda\,e\dwedge e$ as
before, and used the second Bianchi identity $\dd^\om R=0$. 
One may readily check the covariance of the Euler--Lagrange derivatives under the new transformation:
\begin{align*}
\delta_{\tau}(\CF_{e},\CF_{\om})=(2\,\Lambda \, \CF_{\om}\dwedge \tau\,,\, \CF_{e} \wedge \tau) \ ,
\end{align*}
so that in this case covariance is preserved but through a mixing of the two Euler--Lagrange derivatives.

Despite the cosmological constant modification, this symmetry may still be used to compensate for the action of any infinitesimal diffeomorphism $\xi\in \Gamma(TM)$. Indeed by choosing $\tau_{\xi}:= \iota_{\xi} e$ and $\rho_{\xi}:= \iota_{\xi} \om$, one can verify 
\begin{align}\label{eq:shiftdiff}
\delta_{\xi}(e,\om):=(\LL_{\xi}e, \LL_{\xi}\om) = \delta_{(\tau_{\xi},\rho_{\xi})}(e,\om) +(\iota_{\xi} \CF_{e}, \iota_{\xi}\CF_{\om})
\end{align}
as in \eqref{eq:BFdiff}.
When $e$ is non-degenerate, that is, it is invertible as a bundle map, one may define the vector field $\xi_{\tau}:= e^{-1}(\tau)$ that generates a shift transformation by any $\tau \in \Omega^{0}(M,\FR^{1,2})$ via \eqref{eq:shiftdiff}.  Thus if one restricts the action functional \eqref{eq:ECP3daction} to non-degenerate coframe fields $e$, then the choice of generating set of gauge transformations is immaterial~\cite{Witten:1988hc,Carlip}.
However, in the $L_{\infty}$-algebra framework we need to allow for degenerate metrics in order for the space of dynamical fields to be a vector space. In this case, the two transformations are \emph{not} equivalent, and indeed the shift symmetry generates a larger symmetry distribution on the space of fields. Hence, with this line of reasoning one should attach the extra symmetries to the cochain complex of the $L_{\infty}$-algebra. Then three-dimensional gravity (including degenerate metrics) is a special case of three-dimensional $BF$ theory, with an additional cosmological constant term.

Extending the three-dimensional ECP complex to include the extra shift symmetry and its corresponding Noether identity leads to the graded vector space
$$
V_{\textrm{\tiny ECP}}^{\rm ext} := V^{\textrm{ext}}_{0} \oplus V^{\textrm{ext}}_{1} \oplus V^{\textrm{ext}}_{2} \oplus V^{\textrm{ext}}_3
$$
where
\begin{align} 
V^{\textrm{ext}}_{0}&=\Gamma(TM)\times \Omega^{0}(M,\FR^{1,2})\times \Omega^{0}\big(M,\mathfrak{so}(1,2)\big) \ , \nn
\\[4pt]
V_1^{\rm ext}=V_{1}&= \Omega^{1}(M,\FR^{1,2}) \times
\Omega^{1}\big(M,\mathfrak{so}(1,2)\big) \ , \label{eq:3dV012ext} \\[4pt]
V_2^{\rm ext}=V_{2}&=\Omega^{2}\big(M,\midwedge^{2}(\FR^{1,2})\big) \times
\Omega^{2}(M,\FR^{1,2}) \ , \nn \\[4pt]
V^\textrm{ext}_{3}&= \Omega^1\big(M,\Omega^3(M)\big)\times \Omega^{3}(M,\midwedge^{2} \FR^{1,2}) \times \Omega^3(M,\FR^{1,2})
\ , \nn
\end{align} 
with the brackets extending as 
\begin{flalign} 
\ell^{\textrm{ext}}_{1}(\xi,\tau,\rho)&=(\dd \tau, \dd\rho) \ \in \ V^{\textrm{ext}}_{1} \ , \nn \\[4pt]
\ell_{1}^{\textrm{ext}}(e,\omega)&=(\dd \om,\dd e) \ \in \ V^{\textrm{ext}}_{2} \ , \nn
\\[4pt]
\ell^{\textrm{ext}}_{1}(E,{\mit\Omega})&=(0,\dd E, \dd {\mit\Omega}) \ \in \ V^{\textrm{ext}}_{3} \ , \nn 
\\[6pt]
\ell^{\textrm{ext}}_{2}\big((\xi_{1},\tau_{1},\rho_{1})\,,\,(\xi_{2},\tau_{2}, \rho_{2})\big)&=\big([\xi_{1},\xi_{2}]\,,\, -\rho_{1}\cdot \tau_{2} + \rho_{2} \cdot \tau_{1} +\xi_{1}(\tau_{2})-\xi_{2}(\tau_{1})\, , \, \nn \\ & \hspace{3cm} -[\rho_{1},\rho_{2}]+\xi_1(\rho_{2})
- \xi_{2}(\rho_{1})\big) \ \in \ V^{\textrm{ext}}_{0} \ , \label{eq:3dbracketsext}
\\[4pt]
\ell^{\textrm{ext}}_{2}\big((\xi,\tau,\rho)\,,\,(e,\omega)\big)&=\big(-\rho \cdot e + \om \cdot \tau
+\LL_{\xi}e \,,\, -[\rho,\omega]+2\,\Lambda \, e \dwedge \tau + \LL_{\xi}\om\big) \ \in \ V^{\textrm{ext}}_{1} \
, \nn \\[4pt]
\ell^{\textrm{ext}}_{2}\big((\xi,\tau, \rho) \,,\, (E,{\mit\Omega})\big)&=\big(-
[\rho, E]+ 2\,\Lambda\, {\mit\Omega}\dwedge \tau+\LL_{\xi}E \,,\, -\rho \cdot
{\mit\Omega}+E\wedge \tau + \LL_{\xi}{\mit\Omega} \big) \ \in
\ V^{\textrm{ext}}_{2} \ , \nn \\[4pt]
\ell_{2}^{\textrm{ext}}\big((\xi,\tau, \rho)\,,\,({\CX},{\CT},{\CP})\big)&= \big
(\dd x^\mu\otimes\Tr(\iota_\mu \dd \rho \dwedge
{\CP}+\iota_{\mu} \dd \tau \dwedge \CT) +\LL_{\xi}{\CX}\,,\, \nn \\
& \hspace{3cm}-[\rho,\CT]+ \LL_{\xi} \CT \, , \,-\rho
\cdot {\CP} +\LL_{\xi} {\CP}\big) \ \in \ V^{\textrm{ext}}_3 \ , \nn \\[4pt] 
\ell^{\textrm{ext}}_{2}\big((e_{1},\omega_{1})\,,\,(e_{2},\omega_{2})\big)&=-\big([\omega_{1},\omega_{2}]
+2\,\Lambda\, e_{2}\dwedge e_{1} \,,\, \omega_{1}\wedge
e_{2} +\omega_{2} \wedge e_{1} \big) \ \in \ V^{\textrm{ext}}_{2} \ , \nn \\[4pt]
\ell_{2}^{\textrm{ext}}\big((e,\om)\,,\,(E,{\mit\Omega})\big)&=\big(\dd x^\mu\otimes\Tr ( \iota_\mu \dd e \dwedge E + \iota_\mu \dd\om \dwedge {\mit\Omega} \nn   -\iota_\mu e \dwedge \dd E - \iota_\mu \om\dwedge \dd {\mit\Omega}) \,, \nn \\ 
& \hspace{3cm}
-[\om, E] + 2\, \Lambda \,{\mit\Omega} \dwedge e\, , \,E\wedge e - \omega \wedge {\mit\Omega} \big)  \ \in \ V^{\textrm{ext}}_3 \ . \nn
\end{flalign}
The extra homotopy relations follow from identical calculations to those of Appendix~\ref{app:3dhomotopy}.
At this point one should further augment the complex by introducing a
copy of the redundant part of the symmetries at degree $-1$ and its
dual at degree $4$. The final extended $L_\infty$-algebra is then
quasi-isomorphic to the $L_\infty$-algebra constructed here but excluding the diffeomorphisms, in exactly the same way as in Section~\ref{sec:BFtheory}. Instead, we will circumvent this step in a more elegant way.

Recalling the equivalence between $BF$ theory in three dimensions and
Chern--Simons theory from Section~\ref{sec:3dBFCS}, it follows that
for $\Lambda=0$ the three-dimensional ECP theory is equivalent to
Chern--Simons theory based on the Lie algebra
$\FR^{1,2}\rtimes\mathfrak{so}(1,2)$. Following the observation
of~\cite{Witten:1988hc}, this equivalence can be extended to
$\Lambda\neq0$, and we shall show that the extended $L_{\infty}$-algebra based on \eqref{eq:3dV012ext} and \eqref{eq:3dbracketsext} is isomorphic to the extended Chern--Simons $L_\infty$-algebra based on \eqref{eq:CSVext}, \eqref{eq:CSell1ext} and \eqref{eq:CSell2ext} for a special choice of Lie algebra. That is, Einstein--Cartan--Palatini theory can be formulated as a
Chern--Simons gauge theory of the sort discussed in
Section~\ref{sec:CStheory} whose gauge group $\sG$ is the isometry
group of the constant curvature three-dimensional spacetime determined by the Einstein equations \eqref{eq:eom3d}: the Poincar\'e group
${\sf ISO}(1,2)=\FR^{1,2}\rtimes\sSO(1,2)$ for vanishing
cosmological constant $\Lambda=0$, the de~Sitter group $\sSO(1,3)$ for $\Lambda>0$, or the anti-de~Sitter group $\sSO(2,2)$ for $\Lambda<0$. The generators ${\tt P}_a$ and ${\tt J}_{ab}=-{\tt J}_{ba}$ of the Lie
algebra $\frg$ of $\sG$, with $a,b=1,2,3$, have Lie brackets
\begin{align*}
[{\tt P}_a,{\tt P}_b]_\frg&=2\,\Lambda\,{\tt J}_{ab} \ , \\[4pt]
[{\tt J}_{ab},{\tt P}_c]_\frg&=\tfrac{1}{2}\,\big(\eta_{bc}\,
{\tt P}_a-\eta_{ac}\,{\tt P}_b\big) \ , \\[4pt] [{\tt J}_{ab},{\tt J}_{cd}]_\frg&=\tfrac{1}{2}\,\big( \eta_{bc}\, {\tt J}_{ad} - \eta_{ac}\, {\tt J}_{bd} + \eta_{ad}\, {\tt J}_{bc} - \eta_{bd}\, {\tt J}_{ac} \big) \ ,
\end{align*}
and there is a natural invariant quadratic form $\Tr_\frg:\frg\otimes\frg\to\FR$
of split signature defined by~\cite{Witten:1988hc}
\begin{align*}
\Tr_\frg({\tt J}_{ab}\otimes {\tt P}_c)= \varepsilon_{abc} \qquad \mbox{and} \qquad
\Tr_\frg({\tt J}_{ab}\otimes {\tt J}_{cd})=0=\Tr_\frg({\tt P}_a\otimes {\tt P}_b) \ .
\end{align*}

By decomposing connections $A\in \Omega^1(M,\frg)$ as
\begin{align} \label{eq:Aeomega}
A=e^a\,{\tt P}_a+\omega^{ab}\,{\tt J}_{ab} \ \in \ \Omega^1(M,\frg) \ ,
\end{align}
a straightforward expansion of the Chern--Simons action
functional \eqref{eq:CSaction} using these commutation relations and invariant pairing shows that it
coincides with the action functional
\eqref{eq:ECP3daction} for Einstein--Cartan--Palatini gravity in three
dimensions:
\begin{align*}
S_{\textrm{\tiny CS}}(A)=S_{\textrm{\tiny ECP}}(e,\om) \ .
\end{align*}
Similarly, decomposing
gauge parameters $\lambda \in \Omega^{0}(M,\frg)$ as 
\begin{align} \label{eq:upstaulambda}
\lambda=\tau^a\,{\tt P}_a + \rho^{ab}\, {\tt J}_{ab} \ \in \ \Omega^0(M,\frg) \ ,
\end{align}
a straightforward expansion again shows that the standard gauge transformations $\delta_{\lambda} A$ in \eqref{eq:CSgaugetransf} 
are equivalent to $\delta_{(\tau,\rho)} (e,\om)$. Completely analogous statements follow for the action of diffeomorphisms, and for the forms of the Euler--Lagrange derivatives and Noether identities.

The precise statement of the equivalence above is that
the two underlying extended cyclic $L_{\infty}$-algebras are (strictly) isomorphic. The isomorphism is given by
\begin{align*}
\psi^{\textrm{\tiny ECP}}_{1}: V_{\textrm{\tiny ECP}}^{\textrm{ext}} &\longrightarrow V_{\textrm{\tiny CS}}^{\textrm{ext}}
\end{align*}
where
\begin{align*}
\psi^{\textrm{\tiny ECP}}_{1}(\xi,\tau,\rho)&= (\xi, \tau^{a}\,{\tt P}_{a}+\rho^{ab}\, {\tt J}_{ab}) \ , \\[4pt]
\psi^{\textrm{\tiny ECP}}_{1}(e,\om)&= e^{a}\,{\tt P}_{a} + \om^{ab}\,{\tt J}_{ab} \ , \\[4pt]
\psi^{\textrm{\tiny ECP}}_{1}(E,{\mit\Omega})&= {\mit\Omega}^{a}\, {\tt P}_{a} + E^{ab}\, {\tt J}_{ab} \ , \\[4pt]
\psi^{\textrm{\tiny ECP}}_{1}(\CX, \CT, \CP)&= (\CX, \CP^{a}\,{\tt
                                              P}_{a} + \CT^{ab}\,{\tt
                                              J}_{ab}) \ ,
\end{align*}
while the remaining maps $\psi^{\textrm{\tiny ECP}}_{n}: \midwedge^{n}
V_{\textrm{\tiny ECP}}^{\textrm{ext}} \rightarrow
V^{\textrm{ext}}_{\textrm{\tiny CS}}$ are set to $0$ for all $n \geq
2$. Since both sides are differential graded Lie algebras, and no
higher morphisms $\psi^{\textrm{\tiny ECP}}_{n}$ arise, the only
relation to check is the condition that the map $\psi^{\textrm{\tiny ECP}}_{1}$ is a morphism of differential graded Lie algebras. The calculations follow by a straightforward expansion of the brackets involved. Furthermore, the map is a cyclic $L_\infty$-morphism, which follows immediately by the definition of the pairing $\Tr_{\frg}$.

Following the discussion of redundant symmetries in Chern--Simons
theory from Section~\ref{sec:TQFT}, one may simply drop the redundant
diffeomorphism symmetries by composing with the quasi-isomorphism on
the Chern--Simons side,\footnote{Strictly speaking we should also
  include vector spaces $V^{\rm ext}_{-1}$ and $V^{\rm ext}_{4}$ for
  this composition to work, but this is easily done and we do not write them explicitly.} with no effect on the moduli space of classical solutions. Then this $L_\infty$-isomorphism shows that three-dimensional ECP theory is perturbatively off-shell equivalent to Chern--Simons theory with the appropriate gauge algebra. 
Note that while on the gravity side the cosmological constant
$\Lambda$ appears explicitly in the dynamical and kinematical
brackets, on the Chern--Simons side of the equivalence $\Lambda$ does
not appear in the definition of the $n$-brackets: it is reinterpreted as part of the structure constants of the chosen Lie algebra, and in this sense it is fully absorbed into the kinematical data instead.

We contrast the description we give here with that of the strong BV equivalence exhibited by \cite{Cattaneo:2017ztd} between three-dimensional non-degenerate ECP gravity and non-degenerate $BF$ theory. The symplectomorphism the authors present uses explicitly the inversion property of non-degenerate coframes, and as such it cannot be straightforwardly extended to the degenerate sector. From the $L_{\infty}$-algebra point of view, it is obvious we cannot intepret their map as an $L_{\infty}$-morphism since restricting to non-degenerate configurations automatically takes us out of the realm of vector spaces and $L_{\infty}$-algebras. Furthermore, our quasi-isomorphism may be dualised to a map that preserves the symplectic 2-form in the corresponding (degenerate) BV-BRST complexes which however will not be a bijection. For the equivalence to be apparent one really needs to work on the underlying $L_{\infty}$-algebras level, where we can invert maps -up to homotopy- which are not necessarily bijective. Another remark is that our maps point out the well known equivalence of $BF$ and certain Chern-Simons models, thus encoding the cosmological constant in the structure Lie algebra of Chern-Simons theory as in \cite{Witten:1988hc}. Indeed our interpretation of redundant symmetries and use of quasi-isomorphisms, including degenerate coframes, comes more closely to that of Witten \cite{Witten:1988hc}.
\section{Four-dimensional gravity}
\label{sec:4dgrav}

In this final section we briefly describe the analogous four-dimensional case, wherein the $L_\infty$-algebra formulation is no longer given by a differential graded Lie algebra.

\subsection{Field equations}

The discussion of Section~\ref{sec:ECPgravity} can also be specialised
to yield the Einstein--Cartan--Palatini action functional for gravity in $d=4$ dimensions including cosmological
constant, which is given by 
\begin{align}
S_{\textrm{\tiny ECP}}(e,\omega):&=\int_M\, \Tr\Big(\frac{1}{2}\, e \dwedge e \dwedge
              R + \frac{\Lambda}{4}\, e\dwedge e \dwedge e\dwedge
              e\Big) \nn \\[4pt] 
&=\int_M\, \Tr \bigg( \Big(\frac{1}{2}\,e^{a}\wedge e^{b} \wedge
  R^{cd} + \frac{\Lambda}{4}\, e^{a}\wedge e^{b} \wedge e^{c} \wedge
  e^{d} \Big)\, {\tt E}_{a}\wedge {\tt E}_{b}\wedge {\tt E}_{c}\wedge
  {\tt E}_{d}\bigg) \nn \\[4pt] \label{eq:ECP4daction}
&= \int_M\, \varepsilon_{abcd}\, \Big(\frac{1}{2}\, e^{a}\wedge
  e^{b}\wedge R^{cd}+ \frac{\Lambda}{4}\, e^{a}\wedge e^{b} \wedge
  e^{c}\wedge e^{d} \Big) \ , 
\end{align}
where again we work in Lorentzian signature $(p,q)=(1,3)$ for definiteness and identify the curvature as $R=R^{ab}\, {\tt E}_{a}\wedge {\tt E}_{b}
\in \Omega^{2}(M,\midwedge^{2}(\FR^{1,3}))$. The field equations are now 
\begin{align*}
e \dwedge R+\Lambda\, e\dwedge e \dwedge e=0 \qquad \mbox{and} \qquad
 e \dwedge T =0 \ .
\end{align*}
In this case, when $e$ is invertible, the second equation is 
equivalent to the torsion-free condition $T=0$, because the map
\begin{align*}
e\dwedge-:\Omega^2\big(M,\FR^{1,3}\big)\longrightarrow \Omega^{3}\big(M,\midwedge^{2}(\FR^{1,3})\big)
\end{align*}
is an isomorphism. This can be used to rewrite the action functional
\eqref{eq:ECP4daction} in terms of the curvature of the Levi--Civita
connection for the metric $g=\eta_{ab}\,e^a\otimes e^b$. Then the first equation can be reduced to the Einstein
equations in four dimensions (up to gauge equivalence).

\subsection{$L_\infty$-algebra formulation} 
\label{sec:Linfty4d}

Next we write out explicitly the $L_\infty$-algebra structure of four-dimensional
gravity from Section~\ref{sec:LinftyECP}, specialised to the case $d=4$. It is given by the vector
space 
$$
V:= V_{0} \oplus V_{1} \oplus V_{2} \oplus V_3
$$
where 
\begin{align*} 
V_{0}&=\Gamma(TM)\times \Omega^{0}\big(M,\mathfrak{so}(1,3)\big) \ ,
       \nn \\[4pt] 
V_{1}&= \Omega^{1}(M,\FR^{1,3}) \times
       \Omega^{1}\big(M,\mathfrak{so}(1,3) \big) \ , \\[4pt]
V_{2}&=\Omega^{3}\big(M,\midwedge^{3}(\FR^{1,3})\big) \times
       \Omega^{3}\big(M,\midwedge^{2}(\FR^{1,3})\big) \ , \\[4pt]
V_3&=\Omega^1\big(M,\Omega^4(M)\big) \times \Omega^4\big(M,\midwedge^2(\FR^{1,3})\big)
     \ .
\nn
\end{align*} 
The non-vanishing brackets may be read
off as follows: The $1$-bracket $\ell_{1}$ is defined by 
\begin{flalign*} 
 \ell_{1}(\xi,\rho)=(0,\dd\rho) \ \in \ V_{1} \ , \quad
 \ell_{1}(e,\omega)=(0,0) \ \in \ V_{2} \  \qquad \mbox{and} \qquad \ell_{1}(E,{\mit\Omega})=(0,-\dd {\mit\Omega}) \ \in \ V_{3} \ .
\end{flalign*}
The $2$-bracket $\ell_2$ is defined by 
\begin{flalign*}
 \ell_{2}\big((\xi_{1},\rho_{1})\,,\,(\xi_{2},\rho_{2})\big)&=
 \big([\xi_{1},\xi_{2}]\,,\,-[\rho_{1},\rho_{2}]+\xi_{1}(\rho_{2})
 - \xi_{2}(\rho_{1}) \big) \ \in \ V_{0} \ , \\[4pt]
 \ell_{2}\big((\xi,\rho)\,,\,(e,\omega)\big)&=(-\rho \cdot e
 +\LL_{\xi}e, -[\rho,\omega]+\LL_{\xi}\om) \ \in \ V_{1} \ , \\[4pt]
 \ell_{2}\big((\xi,\rho)\,,\,(E,{\mit\Omega}) \big)&=(-
 \rho\cdot E+\LL_{\xi}E , -[\rho
 , {\mit\Omega}]+\LL_{\xi}{\mit\Omega}) \
 \in \ V_{2} \ , \nn \\[4pt]
 \ell_{2}\big((\xi,\rho)\,,\,({\CX},{\CP})\big)&=
 \big (\dd x^\mu\otimes\Tr(\iota_\mu \dd \rho \dwedge
 {\CP}) +\LL_{\xi}{\CX}\,,\,-[\rho
 , {\CP}] +\LL_{\xi} {\CP}\big) \ \in \ V_3 \ , \\[4pt]
 \ell_{2}\big((e_{1},\omega_{1})\,,\,(e_{2},\omega_{2})\big)&=-(e_{1} \dwedge \dd \omega_{2}
 + e_{2} \dwedge \dd \omega_{1} , e_{1}
 \dwedge \dd e_{2} + e_{2} \dwedge \dd e_{1}) \ \in \ V_{2} \ , \nn \\[4pt]
 \ell_{2}\big((e,\om)\,,\,(E,{\mit\Omega})\big)&=\Big(\dd x^\mu\otimes\Tr \big( \iota_\mu \dd e \dwedge E - \iota_\mu \dd\om \dwedge {\mit\Omega} \nn   -\iota_\mu e \dwedge \dd E + \iota_\mu \om\dwedge \dd {\mit\Omega} \big) \,, \nn \\ 
 &\phantom{{}\bigg(\dd x^\mu\otimes\Tr \big( \iota_\mu \dd e \dwedge E +{}}
 \tfrac32 \, E\wedge e + [\omega, {\mit\Omega}]\Big) \ \in \ V_3 \ . \nn
\end{flalign*}
The $3$-bracket $\ell_{3}$ is defined by
\begin{align*}
& \ell_{3}\big((e_{1},\omega_{1})\,,\,(e_{2},\omega_{2})\,,\,(e_{3},\omega_{3})
 \big) \\[4pt]
& \hspace{1cm} =-\big( e_{1}\dwedge [\omega_{2} , \omega_{3}] + e_{2} \dwedge [\omega_{1} , \omega_{3}] + e_{3} \dwedge [\omega_{2} , \omega_{1}] +  3!\, \Lambda\, e_{1}\dwedge e_{2}\dwedge
e_{3} \,,\, \\
& \hspace{2cm} e_{1} \dwedge (\omega_{2} \wedge e_{3}) 
 +_{(2\leftrightarrow 3)} + e_{2}\dwedge (\omega_{1} \wedge e_{3}) +_{(1\leftrightarrow 3)} + e_{3}\dwedge (\omega_{2} \wedge e_{1})
 +_{(2\leftrightarrow 1)}  \big) \ \in \ V_{2} \ .
\end{align*} 
The calculations establishing the homotopy relations in this case are
formally identical to those of the three-dimensional case from
Appendix~\ref{app:3dhomotopy}, now including an extra coframe field $e$ where
the higher brackets occur. We do not detail these cumbersome
calculations and instead illustrate how the brackets follow from the dual picture
of the BV--BRST formalism in Appendix~\ref{app:4dBVBRST}. Using the nilpotency of the BV differential this may be seen as an alternative proof of the homotopy relations.

The Euler--Lagrange derivatives are encoded in the
expected way as
\begin{align*} 
\mathcal{F}(e,\omega)&=\big(
                       e\dwedge \dd \omega + e\dwedge \tfrac12\,[\omega ,
                       \omega]+ \Lambda\, e^{3}, e\dwedge \dd e +e\dwedge (\omega\wedge e) \big) \\[4pt] 
&=(0,0)+(e\dwedge \dd \omega, e \dwedge \dd e) + \big( e\dwedge
  \tfrac12\, [\omega , \omega]+\Lambda \, e^{3},e\dwedge (\omega\wedge e) \big) \\[4pt] 
&= \ell_{1}(e,\om)- \frac{1}{2}\,
  \ell_{2}\big((e,\om)\,,\,(e,\om)\big) -\frac{1}{3!} \,
  \ell_{3}\big((e,\om)\,,\,(e,\omega)\,,\,(e,\om) \big) \ .
\end{align*}
The action functional \eqref{eq:ECP4daction} can be written as in \eqref{action} using
the cyclic pairing \eqref{eq:ECPpairing} and the identity~\eqref{weirdformula}:
\begin{align*}
S_{\textrm{\tiny ECP}}(e,\omega) &= \int_M\, \Tr\bigg(\frac12\,
                e^2\dwedge\Big(\dd\omega+\frac12\,[\omega,
                \omega]\Big)+\frac{\Lambda}4\, e^4\bigg) \\[4pt]
&= \int_M\,
  \Tr\Big(\frac1{3!}\,(e^2\dwedge\dd\omega-2\,\omega\dwedge
  e\dwedge
  \dd e)+\frac14\,\big(e^2\dwedge[\omega,\omega]-2\,\omega \dwedge e \dwedge
  (\omega\wedge e) + \Lambda\, e^4\big) \Big) \\[4pt]
&= \frac12\,\big\langle (e,\omega)\,,\,(0,0)\big\rangle + \frac1{3!}\,
  \big\langle (e,\omega)\,,\,(2\,e\dwedge\dd\omega,2\, e\dwedge
  \dd e)\big\rangle \\
& \hspace{5cm} + \frac1{4!}\, \big\langle
  (e,\omega)\,,\, \big(3!\,
  e\dwedge\tfrac12\,[\omega,\omega] + 3!\,\Lambda\, e^3,3!\, e\dwedge (\omega\wedge e) \big)\big\rangle
  \\[4pt]
&= \frac12\, \big\langle (e,\omega)\,,\,\ell_1(e,\omega)\big\rangle -
  \frac1{3!}\, \big\langle
  (e,\omega)\,,\,\ell_2\big((e,\omega)\,,\,(e,\omega)\big) \big\rangle
  \\
& \hspace{5cm} - \frac1{4!}\, \big\langle (e,\omega)\,,\,
  \ell_3\big((e,\omega)\,,\,(e,\omega) \,,\,(e,\omega)\big)
  \big\rangle \ .
\end{align*}

\subsection{Differential graded Lie algebra formulations}

In analogy to what we did in the case of three-dimensional gravity, it
is natural at this point to ask if there is an equivalent formulation
of the four-dimensional ECP theory as a differential graded Lie
algebra, in the sense of a quasi-isomorphism with the
$L_\infty$-algebra of Section~\ref{sec:Linfty4d}; this would
correspond to a strictification of the $L_\infty$-algebra, which is
known to exist on abstract grounds~\cite{Berger2007}. A natural place to
look is again at the $L_\infty$-algebras underlying the $BF$ theories
from Section~\ref{sec:BFtheory}. In particular, one may consider a
four-dimensional $BF$ theory with an additional ``cosmological
constant term'': the action functional is given by\footnote{By
  dimension counting, such cosmological constant deformations of $BF$
  theories are only
  possible in dimensions $d=3,4$.}
\begin{align}\label{eq:BF4d}
S^\Lambda_{\textrm{\tiny BF}}(B,A)= \int_{M}\, \Tr \Big(B \dwedge F +\frac{\Lambda}{2}\, B\dwedge B \Big)
\end{align}
for $B\in \Omega^{2}\big(M,\midwedge^{2}(\FR^{1,3})\big)$ and $A\in
\Omega^{1}\big(M,\midwedge^{2}(\FR^{1,3})\big)$; equivalently, $B$ is
a two-form and $A$ is a connection one-form both valued in the Lie
algebra $\mathfrak{so}(1,3)$. 

As in three spacetime dimensions, the shift
symmetry survives the cosmological constant modification: for $\tau
\in \Omega^{1}(M,\FR^{1,3})$ we define
\begin{align*}
\delta_{\tau}(B,A)=(\dd^{A} \tau, \Lambda\, \tau)
\end{align*}
with the corresponding Noether identity
\begin{align*}
\dd^{A}\CF_{B}-\Lambda \, \CF_{A} = 0 \ .
\end{align*}
Again this symmetry renders diffeomorphisms redundant as in three
dimensions: the shift symmetry is again large enough to kill all
local degrees of freedom, so that the gauge theory is again
topological. However, now the addition of a cosmological constant
$\Lambda\neq0$ breaks the higher shift symmetry discussed in
Section~\ref{sec:highershift}: if $\epsi \in
\Omega^{0}\big(M,\midwedge^{2} (\FR^{1,3})\big)$, then 
$$
\big(\delta_{\tau+\dd^{A}\epsi} - \delta_{\tau}\big) (B,A) =
\big((\CF_{B}-\Lambda\, B)\wedge \epsi, \Lambda \, \dd^{A} \epsi\big)
\ \in \ V^{\textrm{\tiny BF}}_{1} \ ,
$$
and this is not proportional to the Euler--Lagrange derivatives. Of
course, when $\Lambda=0$ the two shift parameters $\tau +\dd^{A}
\epsi$ and $\tau$ induce the same transformation on the space of
fields, up to a term proportional to Euler--Lagrange derivatives, so
that the shift symmetry is reducible as discussed in Section~\ref{sec:highershift}.

The similarity between the four-dimensional ECP action functional
\eqref{eq:ECP4daction} and the action functional \eqref{eq:BF4d} is
evident: the former is given by restricting $B$ to decomposable
diagonal two-forms in
$\Omega^{2}\big(M,\midwedge^{2}(\FR^{1,3})\big)$. Of course, this
restriction breaks the shift symmetry invariance\footnote{If one
  restricts to nondegenerate coframe fields, one may derive a certain
  generalization of the three-dimensional shift symmetry~\cite{Montesinos:2020pxv},
  which is however equivalent to diffeomorphism invariance.} so that
there are 
gravitational waves in four dimensions: four-dimensional gravity
contains local degrees of freedom.
In the $L_{\infty}$-algebra framework, one can see this inequivalence
between the two theories simply by noting that there does not exist a
cyclic morphism between the two $L_{\infty}$-algebras: This can be
immediately seen at the level of the degree-preserving component
$\psi_{1}:V^{\textrm{\tiny BF}}\rightarrow V^{\textrm{\tiny EPC}}$, where the only possible
non-trivial map would be of the form 
\begin{align*}
\psi_{1}\big((\xi,\rho) +(e,\om) +(E,{\mit\Omega}) +(\CX,\CP)\big):=
  (0,\rho)+(0,\om)+(0,{\mit\Omega})+(0,\CP) \ ,
\end{align*}
but this does not commute with the $1$-brackets of the two sides, nor
does it preserve the cyclic pairings.\footnote{The same is true if one
  tries to build a morphism in the opposite direction.} Despite this,
the relation of the Einstein--Cartan--Palatini formulation of gravity
in three and four spacetime dimensions to topological field theories
has prompted studies of deformations of the four-dimensional $BF$
action functional \eqref{eq:BF4d} wherein the constraints above are
implemented dynamically, and thus taking the latter action functional as a candidate for quantization, see e.g. \cite{Freidel:2012np}.
It would be interesting and potentially fruitful to study such
deformations and their explicit relations to four-dimensional gravity
in the off-shell framework of $L_{\infty}$-algebras, which from a
mathematical perspective would give a concrete construction of the
strictification of the ECP $L_\infty$-algebra. 

\paragraph{Acknowledgements.}

We thank Athanasios Chatzistavrakidis, Larisa Jonke, Lukas M\"uller, Christian
S\"amann, Alexander Schenkel, Michele Schiavina, Peter Schupp and
Martin Wolf for helpful discussions
and correspondence.
This work was supported by the Action MP1405 ``Quantum
Structure of Spacetime'' from the European Cooperation in Science and Technology
(COST). The work of {\sc M.D.C.} and {\sc V.R.} is supported by Project
ON171031 of the Serbian Ministry of Education, Science and
Technological Development. The work of {\sc M.D.C.} and {\sc R.J.S.} was supported by the
Croatian Science Foundation Project IP--2018--01--7615 ``New Geometries
for Gravity and Spacetime''. The work of {\sc G.G.} is supported by a
Doctoral Training Grant from the UK Engineering and Physical Sciences
Research Council ({EPSRC}). The work of {\sc R.J.S.} was supported by
the Consolidated Grant ST/P000363/1 ``Particle Theory at the Higgs Centre''
from the UK Science and Technology Facilities Council (STFC).

\paragraph{Data availability.} 

No additional research data beyond the data included and cited in this work are needed to validate the research findings presented.

\appendix

\section{Calculations in three dimensions}

In this appendix we provide various explicit calculations for three-dimensional gravity, illustrating the main steps needed to establish the $L_\infty$-algebra relations of Section~\ref{sec:Linfty}.

\subsection{Homotopy relations}
\label{app:3dhomotopy}

The proof of the homotopy relations for the three-dimensional
Einstein--Cartan--Palatini theory of Section~\ref{sec:3d} is similar in spirit
to the proof of the homotopy relations for three-dimensional
Chern--Simons theory from \cite{Linfty}. In the present setting we must
deal with an extra field in the fundamental representation of the
gauge group and the extra diffeomorphism gauge symmetries. We check
these relations order by order in the total degree of the homotopy,
remembering that $|\ell_{n}|=2-n$. For brevity, in the
calculations below we set the cosmological constant term
to zero, $\Lambda=0$, with the results following straightforwardly for
$\Lambda\neq0$.

\subsubsection*{Differential conditions}

The homotopy relations $\CJ_1=0$ are the differential condition $\ell_1\circ\ell_1=0$. The only
fields we need to check this on is a pair of gauge parameters $(\xi,\rho)$ in degree~$0$:
\begin{align*}
\ell_{1}\big(\ell_{1}(\xi,\rho)\big)= \ell_{1}(0,\dd\rho)=
  (0,\dd^{2}\rho)= (0,0) \ ,
\end{align*}
and a pair of dynamical fields $(e,\om)$ in degree~$1$:
\begin{align*}
	\ell_{1}\big(\ell_{1}(e,\omega)\big)=\ell_{1}(\dd \om, \dd e)=(0,\dd^{2} e)=(0,0) \ .
\end{align*}

\subsubsection*{Leibniz rules}

The homotopy relations $\CJ_2=0$ are the graded Leibniz rule for the differential $\ell_1$ with
respect to the $2$-bracket $\ell_2$.
In this case, we may act non-trivially on fields whose total degrees
are $0$, $1$ and~$2$.

\paragraph{$\underline{{\rm Total~degree}~0:}$}
We act on two pairs of gauge transformations
$(\xi_{1},\rho_{1})$ and $(\xi_{2},\rho_{2})$, and we need to check 
\begin{align*}
\ell_{1}\Big(\ell_{2}\big((\xi_{1},\rho_{1})\,,\,(\xi_{2},\rho_{2})\big)\Big)=
\ell_{2}\big(\ell_{1}(\xi_{1},\rho_{1})\,,\,(\xi_{2},\rho_{2})\big)
  +
  \ell_{2}\big((\xi_{1},\rho_{1})\,,\,\ell_{1}(\xi_{2},\rho_{2})\big)
  \ .
\end{align*}
The left-hand side is
\begin{align*}
\ell_{1}\big([\xi_{1},\xi_{2}]\,,\,-[\rho_{1},\rho_{2}]+\xi_{1}(\rho_{2})
  - \xi_{2}(\rho_{1})
  \big)&=\Big(0,-\dd\big([\rho_{1},\rho_{2}]
         +\xi_{1}(\rho_{2}) - \xi_{2}(\rho_{1})\big) \Big) \\[4pt]
&=-\big(0,[\dd\rho_{1},\rho_{2}]+[\rho_{1},\dd\rho_{2}]+\dd
  \xi_{1}(\rho_{2}) -\dd \xi_{2}(\rho_{1}) \big) \ ,
\end{align*}
while the right-hand side is
\begin{align*} 
&\ell_{2}\big((0,\dd\rho_{1})\,,\,(\xi_{2},\rho_{2})
  \big)+\ell_{2}\big((\xi_{1},\rho_{1})\,,\,(0,\dd \rho_{2})
  \big) \\[4pt]
& \hspace{6cm} =-\big(0,-[\rho_{2},\dd
  \rho_{1}]+\LL_{\xi_{2}}\dd \rho_{1}\big) +
  \big(0,-[\rho_{1},\dd\rho_{2}]+\LL_{\xi_{1}}\dd\rho_{2}\big)
  \\[4pt]
& \hspace{6cm}
  =-\big(0,[\dd\rho_{1},\rho_{2}]+[\rho_{1},\dd\rho_{2}]+\dd
  \xi_{1}(\rho_{2}) -\dd\xi_{2}(\rho_{1}) \big) \ .
\end{align*}

\paragraph{$\underline{{\rm Total~degree}~1:}$}
We act on a pair of gauge transformations $(\xi,\rho)$ and one
pair of dynamical fields $(e,\omega)$, and we need to check
\begin{align*}
\ell_{1}\Big(\ell_{2}\big((e,\om)\,,\,(\xi,\rho)
  \big)\Big)=\ell_{2}\big(\ell_{1}(e,\om)\,,\,(\xi,\rho)
  \big)-\ell_{2}\big((e,\om)\,,\,\ell_{1}(\xi,\rho) \big) \ .
\end{align*}
The left-hand side is
\begin{flalign*}
\ell_{1}(\rho \cdot e - \LL_{\xi}e\,,\,[\rho,\om] -\LL_{\xi}
\om)&=\big(\dd([\rho,\om])-\dd\LL_{\xi}\om\,,\,\dd (\rho \cdot e)
-\dd\LL_{\xi}e \big) \\[4pt]
&= ([\dd
\rho,\om] +[\rho,\dd \om] - \LL_{\xi} \dd \om \,,\, \dd \rho \wedge e + \rho \cdot \dd e - \LL_{\xi}\dd e) 
\end{flalign*}
where we used the Cartan identity
\begin{align*}
\dd\circ\LL_\xi = \LL_\xi\circ\dd \ , 
\end{align*}
while the right-hand side is
\begin{align*}
\ell_{2}\big((\dd \om,\dd e) \,,\, (\xi,\rho)
  \big)-\ell_{2}\big((e,\om) \,,\,(0,\dd \rho) \big)&=([\rho,\dd
                                                         \om] -
                                                         \LL_{\xi}
                                                         \dd\om \,,\, \rho \cdot \dd e - \LL_{\xi} \dd e)
+([\dd \rho,\omega]\,,\, \dd \rho \wedge e) \ .
\end{align*}

\paragraph{$\underline{{\rm Total~degree}~2:}$}
We may act on two pairs of dynamical fields $(e_{1},\om_{1})$ and $(e_{2},\om_{2})$, and we need to check
\begin{align*}
\ell_{1}\Big(\ell_{2}\big((e_{1},\om_{1})\,,\,(e_{2},\om_{2})\big)\Big)=\ell_{2}\big(\ell_{1}(e_{1},\om_{1})\,,\,(e_{2},\om_{2})\big) - \ell_{2}\big((e_{1},\om_{1})\,,\,\ell_{1}(e_{2},\om_{2})\big) \ .
\end{align*}
The left-hand side is 
\begin{align*}
	\ell_{1}\big(-([\om_{2}, \om_{1}], \om_{1}\wedge e_{2} +\om_{2}\wedge e_{1})\big)= (0,-\dd \om_{1} \wedge e_{2} +\om_{1} \wedge \dd e_{2} -\dd \om_{2} \wedge e_{1} + \om_{2} \wedge \dd e_{1}) \ ,
\end{align*}
while the right-hand side is
\begin{align*}
	\ell_{2}\big((\dd \om_{1},\dd e_{1}) \,,\, (e_{2},\om_{2})\big) &- \ell_{2}\big((e_{1},\om_{1})\,,\,(\dd \om_{2}, \dd e_{2})\big) \\[4pt]
	&= - \big(\dd x^\mu\otimes\Tr(\iota_{\mu} \dd e_{2} \dwedge \dd \om_{1} + \iota_{{\mu}}\dd \om_{2} \dwedge \dd e_{1}) \,,\, \dd \om_{1} \wedge e_{2} - \om_{2} \wedge \dd e_{1}\big) \\
	& \quad \, -\big(\dd x^\mu\otimes\Tr(\iota_{\mu} \dd e_{1} \dwedge \dd \om_{2} + \iota_{{\mu}}\dd \om_{1} \dwedge \dd e_{2}) \,,\, \dd \om_{2} \wedge e_{1} - \om_{1} \wedge \dd e_{2}\big) \ .
\end{align*}
The second component here agrees with that of the left-hand side. 
The first component vanishes, because we can write the argument of the Hodge duality operator using the Leibniz rule for the contraction of a spacetime four-form in three dimensions:
\begin{align*}
\iota_{\mu} \dd e_{2} \dwedge \dd \om_{1} + \iota_{{\mu}}\dd \om_{2} \dwedge \dd e_{1} + \dd \om_{2} \dwedge \iota_{\mu} \dd e_{1} + \dd e_{2} \dwedge \iota_{{\mu}}\dd \om_{1}= \iota_{{\mu}} (\dd e_{2} \dwedge \dd \om_{1}+\dd \om_{2} \dwedge \dd e_{1}) =0 \ .
\end{align*}

We may also act on a pair of gauge transformations $(\xi,\rho)$ and a pair of Euler--Lagrange derivatives $(E,{\mit\Omega})$. Then we need to check 
\begin{align*}
\ell_{1}\Big(\ell_{2}\big((\xi,\rho)\,,\,(E,{\mit\Omega})\big)\Big)-\ell_{2}\big(\ell_{1}(\xi,\rho)\,,\,(E,{\mit\Omega})\big) - \ell_{2}\big((\xi,\rho)\,,\,\ell_{1}(E,{\mit\Omega})\big)=(0,0) \ .
\end{align*}
Expanding out the brackets we get
 \begin{align*}
 \ell_{1}(-[\rho,E] + \LL_{\xi}E,-\rho \cdot {\mit\Omega} + \LL_{\xi} {\mit\Omega})& -\ell_{2}\big((0,\dd \rho)\,,\,(E,{\mit\Omega})\big)- \ell_{2}\big((\xi,\rho)\,,\,(0,\dd {\mit\Omega})\big) \\[4pt]
 &=\big(0\,,\,-\dd(\rho\cdot {\mit\Omega})+\LL_{\xi}\dd
   {\mit\Omega} \big)-\big(\dd x^\mu\otimes\Tr(-\iota_{{\mu}}\dd \rho \dwedge \dd {\mit\Omega})\,,\,-\dd \rho \wedge {\mit\Omega}\big) \\
 &\phantom{{}={}}-\big(\dd x^\mu\otimes\Tr(\iota_{{\mu}} \dd  \rho \dwedge \dd {\mit\Omega})\,,\,-\rho \cdot \dd {\mit\Omega} + \LL_{\xi} \dd {\mit\Omega} \big)\\[4pt]
 &=(0,0) \ ,
 \end{align*}
where we used the Leibniz rule for the exterior derivative in the
first bracket.

\subsubsection*{Jacobi identities}

Since $\ell_{3}=0$ by definition, the homotopy relations $\CJ_3=0$ are the graded Jacobi identity
for the $2$-bracket $\ell_2$. In this case, we may act non-trivially
on fields whose total degrees are $0$, $1$, $2$ and $3$. 

\paragraph{$\underline{{\rm Total~degree}~0:}$}
We act on three pairs of gauge parameters, and we need
\begin{align*}
\ell_{2}\Big(\ell_{2}\big((\xi_{1},\rho_{1})\,,\,(\xi_{2},\rho_{2})\big)\,,\, (\xi_{3},\rho_{3})\Big)&+\ell_{2}\Big(\ell_{2}\big((\xi_{3},\rho_{3})\,,\,(\xi_{1},\rho_{1})\big)\,,\,(\xi_{2},\rho_{2})\Big) \nn \\ &\hspace{2cm} + \ell_{2}\Big(\ell_{2}\big((\xi_{2},\rho_{2})\,,\,(\xi_{3},\rho_{3})\big)\,,\,(\xi_{1},\rho_{1})\Big)=(0,0) \ .
\end{align*}
Expanding the first term we get 
\begin{flalign*}
\ell_{2}\Big(\big([\xi_{1},\xi_{2}]&\,,\,-[\rho_{1},\rho_{2}]+\xi_{1}(\rho_{2}) - \xi_{2}(\rho_{1})\big)\,,\,(\xi_{3},\rho_{3})\Big) \\
&=\Big(\big[[\xi_{1},\xi_{2}],\xi_{3}\big]\,,\,\big[[\rho_{1},\rho_{2}], \rho_{3}\big] -[\xi_{1}(\rho_{2}),\rho_{3}]+[\xi_{2}(\rho_{1}),\rho_{3}] +[\xi_{1},\xi_{2}](\rho_{3}) \\ &\hspace{4cm} + \xi_{3}([\rho_{1},\rho_{2}])-\xi_{3}\big(\xi_{1}(\rho_{2})\big) +\xi_{3}\big(\xi_{2}(\rho_{1})\big)\Big) \ .
\end{flalign*}
Permuting the indices $(1,2,3)$ and adding, one sees that the terms
containing a composition of two Lie brackets vanish due to the Jacobi identities for the Lie algebras $\Omega^{0}(M,\mathfrak{so}(1,2))$ and $\Gamma(TM)$. The remaining terms vanish since they form a representation of $\Gamma(TM)$ on $\Omega^{0}(M,\mathfrak{so}(1,2))$ and a derivation with respect to the Lie bracket of $\Omega^{0}(M,\mathfrak{so}(1,2))$. 

\paragraph{$\underline{{\rm Total~degree}~1:}$}
We act on two pairs of gauge parameters and one pair of dynamical fields, and we need 
\begin{align}
\ell_{2}\Big(\ell_{2}\big((\xi_{1},\rho_{1})\,,\,(\xi_{2},\rho_{2})\big)\,,\,(e,\om)\Big)&=\ell_{2}\Big(\ell_{2}\big((\xi_{1},\rho_{1})\,,\,(e,\om)\big)\,,\,(\xi_{2},\rho_{2})\Big) \nn\\ & \hspace{1cm} -\ell_{2}\Big(\ell_{2}\big((\xi_{2},\rho_{2})\,,\,(e,\om)\big)\,,\,(\xi_{1},\rho_{1})\Big) \ .
\label{eq:deg0totdeg1}\end{align}
Expanding the left-hand side we get
\begin{flalign*}
\ell_{2}\Big(\big([\xi_{1},\xi_{2}]&\,,\,-[\rho_{1},\rho_{2}]+\xi_{1}(\rho_{2}) -  \xi_{2}(\rho_{1})\big)\,,\,(e,\om)\Big) \\
&=\big([\rho_{1},\rho_{2}]\cdot e - \xi_{1}(\rho_{2})\cdot e +\xi_{2}(\rho_{1})\cdot e + \LL_{[\xi_{1},\xi_{2}]}e\,,\,\\
&\hspace{2cm} \big[[\rho_{1},\rho_{2}],\omega\big]-[\xi_{1}(\rho_{2}),\om]+[\xi_{2}(\rho_{1}),\om] +\LL_{[\xi_{1},\xi_{2}]}\om\big) \ ,
\end{flalign*}
while the first term in the right-hand side is 
\begin{flalign*}
\ell_{2}\big((-\rho_{1}\cdot e +\LL_{\xi_{1}}e\,,\,&-[\rho_{1},\om]+\LL_{\xi_{1}}\om)\,,\,(\xi_{2},\rho_{2})\big)\\[4pt] &= \big(-\rho_{2}\cdot (\rho_{1}\cdot e)+\rho_{2} \cdot (\LL_{\xi_{1}}e) +\LL_{\xi_{2}}(\rho_{1}\cdot e)-\LL_{\xi_{2}}\LL_{\xi_{1}}e\,,\,\\
&\hspace{2cm}-\big[\rho_{2},[\rho_{1},\om]\big]+[\rho_{2},\LL_{\xi_{1}}\om]-\LL_{\xi_{2}}\LL_{\xi_{1}}\om +\LL_{\xi_{2}}[\rho_{1},\om]\big) \ .
\end{flalign*}
Interchanging the indices $(1,2)$ and subtracting in this last
expression, we see that the two sides of \eqref{eq:deg0totdeg1} match 
for the same representation theoretic reasons as in total
degree~$0$. 

\paragraph{$\underline{{\rm Total~degree}~2:}$}
We may act on collections of one pair of
gauge parameters and two pairs of dynamical fields, or of two pairs of gauge parameters and one pair of Euler--Lagrange derivatives.
For the former case we need
\begin{align*}
\ell_{2}\Big(\ell_{2}\big((e_{1},\om_{1})\,,\,(e_{2},\om_{2}) \big)\,,\,(
\xi,\rho)\Big)&=-\ell_{2}\Big(\ell_{2}\big((
\xi,\rho)\,,\,(e_{1},\om_{1})\big)\,,\,(e_{2},\om_{2})\Big) \\[4pt]
& \hspace{1cm}-\ell_{2}\Big(\ell_{2}\big((
\xi,\rho)\,,\,(e_{2},\om_{2})\big)\,,\,(e_{1},\om_{1})\Big) \ .
\end{align*}
Expanding the left-hand side gives
\begin{align*}
-\ell_{2}\big(([\om_{2}, \om_{1}]&,\omega_{1}\wedge e_{2} +\om_{2}\wedge e_{1})\,,\,(
\xi,\rho)\big) \\[4pt]
&=\big(-[\rho,[\omega_{2},\omega_{1}]]
  +\LL_{\xi}[\omega_{2}, \omega_{1}] \,,\, \\
&\hspace{2cm}-\rho\cdot (\omega_{1}\wedge e_{2}) -\rho \cdot (\omega_{2}\wedge e_{1}) +\LL_{\xi}(\omega_{1}\wedge e_{2})+ \LL_{\xi}(\omega_{2}\wedge e_{1}) \big) \ ,
\end{align*}
while the right-hand side expands into
\begin{align*}
\ell_{2}\big((\rho\cdot e_{1}
  &-\LL_{\xi}e_{1}\,,\,[\rho,\omega_{1}]
    -\LL_{\xi}\om_{1})\,,\,(e_{2},\om_{2})\big)
    +\ell_{2}\big((\rho\cdot e_{2}
    -\LL_{\xi}e_{2}\,,\,[\rho,\omega_{2}]
    -\LL_{\xi}\om_{2})\,,\,(e_{1},\om_{1})\big)\\[4pt] 
&=-\big( [\om_{2},[\rho,\om_{1}]-\LL_{\xi}\om_{1}] \,,\, ([\rho,\om_{1}]-\LL_{\xi}\om_{1})\wedge e_{2} +\om_{2}\wedge(\rho \cdot e_{1} -\LL_{\xi}e_{1})\big)\\
&\hspace{1cm} -\big(
  [\om_{1},[\rho,\om_{2}]-\LL_{\xi}\om_{2}] \,,\,
  ([\rho,\om_{2}]-\LL_{\xi}\om_{2})\wedge e_{1}
  +\om_{1}\wedge(\rho \cdot e_{2} -\LL_{\xi}e_{2}) \big) \\[4pt]
&=\big(-[\rho ,[\omega_{2}, \omega_{1}]] + \LL_{\xi}[\om_{2},\om_{1}] \,,\, \\
&\hspace{1cm} -\rho \cdot (\omega_{1}\wedge e_{2})
  -\rho\cdot(\omega_2\wedge e_1) +\LL_{\xi}(\om_{1}\wedge e_{2})
  +\LL_{\xi}(\om_{2}\wedge e_{1})\big)
\end{align*}
as required, where in the last equality we used the Leibniz rules for the Lie
derivative $\LL_\xi$ and the action of the gauge parameter $\rho$
on the exterior products $\omega\wedge e$ and $[\omega_2,\omega_1]$. The check on collections of fields
involving two pairs of gauge parameters and one pair of Euler--Lagrange derivatives is formally identical to the proof of the total degree~$1$
relation \eqref{eq:deg0totdeg1}, since the dynamical fields and the
Euler--Lagrange derivatives live in the same representations of $\sSO_+(1,2)$ and the bracket $\ell_{2}\big((
\xi,\rho),(e,\omega)\big)$ is formally identical to $\ell_{2}\big((
\xi,\rho),(E,{\mit\Omega}) \big)$. 

\paragraph{$\underline{{\rm Total~degree}~3:}$}
The calculations now become considerably more involved and lengthy, so
we will organise the checks of the graded Jacobi identities in this case
into a sequence of Lemmas. 

\begin{lemma}[Contracted Schouten identity]
If $A,B\in\mathfrak{so}(1,2)\simeq \midwedge^{2}(\FR^{1,2})$ and $v\in\FR^{1,2}$, then
\begin{align} \label{eq:Schouten}
\varepsilon_{abc}\, A^{ab}\, B^c{}_d\,v^d =
  -2\,\varepsilon_{abc}\,B^a{}_d\,A^{db}\, v^c \ .
\end{align}
\end{lemma}
\begin{proof}
We use the three-dimensional Schouten identity\footnote{This identity also holds in Euclidean signature.}
\begin{align*}
\varepsilon_{abc} \, \eta_{dh} + \varepsilon_{ach} \, \eta_{bd}-
  \varepsilon_{bch} \, \eta_{ad} - \varepsilon_{abh} \, \eta_{cd}=0 \ ,
\end{align*}
where $\eta$ is the Minkowski metric on $\FR^{1,2}$.
This identity holds since the left-hand side is antisymmetric in four
indices, whereas the indices vary through $1,2,3$ as we are working in
three dimensions, and hence it vanishes identically. Contracting it
with the components of $A=A^{ab}\,{\tt E}_{ba}$, $B=B^a{}_b\,{\tt
  E}^b{}_a$ and $v=v^a\,{\tt E}_a$ yields \eqref{eq:Schouten}.
\end{proof}

\begin{lemma}
If $(e_{1},\om_{1})$,
$(e_{2},\om_{2})$ and $(e_{3},\om_{3})$ are three pairs of dynamical
fields in $V_1$, then
 \begin{align} \label{eq:deg3eomegaid}
 \ell_{2}\Big(\ell_{2}\big((e_{1},\om_{1})\,,\,(e_{2},\om_{2})\big)\,,\,(e_{3},\om_{3})\Big) &+  \ell_{2}\Big(\ell_{2}\big((e_{3},\om_{3})\,,\,(e_{1},\om_{1})\big)\,,\,(e_{2},\om_{2})\Big) \nn \\
 & \hspace{1cm} +
   \ell_{2}\Big(\ell_{2}\big((e_{2},\om_{2})\,,\,(e_{3},\om_{3})\big)\,,\,(e_{1},\om_{1})\Big)
   = (0,0) \ .
 \end{align}
\end{lemma}
\begin{proof}
Expanding the first term of \eqref{eq:deg3eomegaid}, we get
 \begin{align*}
 \ell_{2}\big((-[\om_{2}, \om_{1}]&, -\om_{1}\wedge e_{2}-\om_{2}\wedge e_{1})\,,\,(e_{3},\om_{3})\big)\\[4pt]
 &=-\Big(\dd x^\mu\otimes\Tr\big(\iota_{{\mu}}\dd e_{3} \dwedge (-[\om_{2}, \om_{1}])+\iota_{{\mu}}\dd\om_{3} \dwedge (-\om_{1}\wedge e_{2} -\om_{2}\wedge e_{1})\\
 &\phantom{{}=-\Big(\dd x^\mu\otimes\Tr\big({}}-\iota_{{\mu}} e_{3} \dwedge \dd(-[\om_{2}, \om_{1}])-\iota_{{\mu}}\om_{3}\dwedge\dd(-\om_{1}\wedge e_{2}-\om_{2}\wedge e_{1})\big) \,,\\
 &\phantom{{}=-\Big({}} \hspace{3cm} -[\om_{2}, \om_{1}] \wedge e_{3} +
   \om_{3}\wedge \om_{1}\wedge e_{2} + \om_{3}\wedge \om_{2} \wedge
   e_{1}\Big) \ .
 \end{align*}
The other two terms are the cyclic permutations among the indices
$(1,2,3)$. Writing these out one sees that the terms in the second
component cancel each other as required, by representation theoretic reasons.

Showing that the first component vanishes requires a bit more work. Writing out the other two permutations among the indices $(1,2,3)$, we collect terms 
involving the fields $\om_{1}, \om_{2}, e_{3}$ and evaluate the Hodge operator explicitly to obtain
\begin{align}
\varepsilon_{abc}\, \big(&2\, \iota_{\mu} \dd {e_{3}}^{a} \wedge {\om_{2}}^{b}{}_{d} \wedge {\om_{1}}^{dc} - 2\, \iota_{\mu} e_{3}{}^a \, \dd({\om_{2}}^{b}{}_{d}\wedge{\om_{1}}^{dc}) + \iota_{\mu} \dd {\om_{1}}^{bc}\wedge {\om_{2}}^{a}{}_{d} \wedge {e_{3}}^{d} \nn \\ &- \iota_{\mu} {\om_{1}}^{bc} \, \dd({\om_{2}}^{a}{}_{d}\wedge {e_{3}}^{d}) + \iota_{\mu} \dd {\om_{2}}^{bc} \wedge {\om_{1}}^{a}{}_{d}\wedge{e_{3}}^{d} -  \iota_{\mu}{\om_{2}}^{bc} \, \dd({\om_{1}}^{a}{}_{d}\wedge {e_{3}}^d) \big) \ . \label{eq:firstslot}
\end{align}	
We then rewrite the fourth term of \eqref{eq:firstslot} as
\begin{align*}
\varepsilon_{abc}\, \iota_{\mu}{\om_{1}}^{bc} \, \dd({\om_{2}}^{a}{}_{d}\wedge {e_{3}}^d)
&=2\,\varepsilon_{bca} \, \iota_{\mu} {\om_{1}}^{bd}\, \dd({\om_{2}}_{d}{}^{c}\wedge {e_{3}}^a) \\[4pt]
&= 2\,\varepsilon_{abc} \, ({e_{3}}^{a}\wedge \iota_{\mu} {\om_{1}}^{bd}\wedge \dd {\om_{2}}_{d}{}^{c} - \dd{e_{3}}^{a}\wedge \iota_{\mu} {\om_{1}}^{bd}\wedge {\om_{2}}_{d}{}^{c})
\end{align*}	
where in the first equality we used the contracted Schouten identity
\eqref{eq:Schouten} and in the last equality the Leibniz rule for the
exterior derivative. 
The sixth term is obtained from this by simply interchanging the indices $(1,2)$. For the third term of \eqref{eq:firstslot}, going through exactly the same steps as for the fourth term allows us to rewrite it as
\begin{align*}
\iota_{\mu} \dd {\om_{1}}^{bc}\wedge {\om_{2}}^{a}{}_{d} \wedge {e_{3}}^{d}= 2\,\varepsilon_{abc} \, {e_{3}}^{a} \wedge \iota_{\mu}\dd{\om_{1}}^{bd} \wedge {\om_{2}}_{d}{}^{c} \ ,
\end{align*}
and the fifth term is simply obtained from this by interchanging the indices $(1,2)$.
Collecting all the terms, the expression \eqref{eq:firstslot} becomes
\begin{align*}
&2\,\varepsilon_{abc} \, \big(\iota_{\mu} \dd {e_{3}}^{a} \wedge
  {\om_{2}}^{b}{}_{d} \wedge {\om_{1}}^{dc} -  \iota_{\mu} e_{3}{}^a
  \, \dd({\om_{2}}^{b}{}_{d}\wedge{\om_{1}}^{dc}) - {e_{3}}^{a}\wedge
  \iota_{\mu}{\om_{1}}^{b}{}_{d}\wedge \dd {\om_{2}}^{dc}
  \\&\phantom{{}2\,\varepsilon_{abc} \, \big({}} \hspace{2cm} + \dd {e_{3}}^{a}\wedge \iota_{\mu} {\om_{1}}^{b}{}_{d}\wedge {\om_{2}}^{dc} + {e_{3}}^{a}\wedge \dd {\om_{1}}^{b}{}_{d}\wedge \iota_{\mu} {\om_{2}}^{dc} - \dd {e_{3}}^a \wedge {\om_{1}}^{b}{}_{d}\wedge \iota_{\mu} {\om_{2}}^{dc}\\
&\phantom{{}2\,\varepsilon_{abc} \, \big({}} \hspace{4cm} +{e_{3}}^{a}\wedge
  \iota_{\mu} \dd {\om_{1}}^{b}{}_{d}\wedge {\om_{2}}^{dc} +
  {e_{3}}^{a} \wedge {\om_{1}}^{b}{}_{d}\wedge \iota_{\mu} \dd
  {\om_{2}}^{dc}\big) \\[4pt]
&\hspace{8cm} =2 \, \varepsilon_{abc} \, \iota_{\mu} \dd
  ({e_{3}}^{a}\wedge {\om_{2}}^{b}{}_{d} \wedge {\om_{1}}^{dc})
  \\[4pt]
&\hspace{10cm}=0
\end{align*}
where we successively used the Leibniz rules for the exterior derivative and the contraction, and the last quantity vanishes because it is the contraction of a four-form in three dimensions.
The remaining terms are simply permutations of the indices $(1,2,3)$,
and so they all vanish as well. This completes the proof of the
homotopy identity \eqref{eq:deg3eomegaid}.
\end{proof}

\begin{lemma}
If $(\xi,\rho)\in V_0$ is a pair of gauge parameters, $(e,\om)\in
V_1$ is a pair of dynamical fields, and 
$(E,{\mit\Omega})\in V_2$ is a pair of Euler--Lagrange derivatives, then
 \begin{align}
 \ell_{2}\Big(\ell_{2}\big((\xi,\rho)\,,\,(e,\om)\big)\,,\,(E,{\mit\Omega})\Big)
   & +
     \ell_{2}\Big(\ell_{2}\big((E,{\mit\Omega})\,,\,(\xi,\rho)\big)\,,\,(e,\om)\Big)
     \nn \\
 & \hspace{3cm}+
   \ell_{2}\Big(\ell_{2}\big((e,\om)\,,\,(E,{\mit\Omega})\big)\,,\,(\xi,\rho)\Big)
   =(0,0) \ . \label{eq:deg3laomOmid}
 \end{align}
\end{lemma}
\begin{proof}
The first term of \eqref{eq:deg3laomOmid} expands as
 \begin{align}
 \ell_{2}\big( (-\rho\cdot e +\LL_{\xi}
   e,&-[\rho,\om]+\LL_{\xi}\om)\,,\,(E,{\mit\Omega})\big) \nn \\[4pt]
 &=\Big(\dd x^\mu\otimes\Tr\big(\iota_{\mu}\dd(-\rho\cdot e +\LL_{\xi}e)\dwedge E
   +\iota_{\mu} \dd(-[\rho,\om]+\LL_\xi \om)\dwedge {\mit\Omega}
   \label{eq:expand1} \\
 &\phantom{{}=\bigg(\dd x^\mu\otimes\Tr\big({}}-\iota_{\mu}(-\rho\cdot e
   +\LL_{\xi}e)\dwedge \dd E -
   \iota_{\mu}(-[\rho,\om]+\LL_{\xi}\om)\dwedge \dd {\mit\Omega}
   \big) \,, \nn \\
 &\phantom{{}=\bigg({}} \hspace{4cm} E\wedge(-\rho \cdot e +
   \LL_{\xi} e) -(-[\rho,\om]+\LL_{\xi}\om)\wedge
   {\mit\Omega} \Big) \ . \nn
 \end{align}
The second term expands as 
 \begin{align}
 -\ell_{2}\big( (-[\rho, E] +\LL_{\xi} E,&-\rho\cdot
                                                {\mit\Omega}+\LL_{\xi}{\mit\Omega})\,,\,(e,\om)\big)
   \nn \\[4pt]
 &=\Big(\dd x^\mu\otimes\Tr\big(\iota_{\mu} \dd e \dwedge (-[\rho , E] +
   \LL_{\xi}E)+\iota_{\mu} \dd \om \dwedge (-\rho \cdot
   {\mit\Omega}+\LL_{\xi} {\mit\Omega}) \label{eq:expand2} \\
 &\phantom{{}=\bigg(\dd x^\mu\otimes\Tr\big({}}-\iota_{\mu} e \dwedge \dd (-[\rho
   ,E] + \LL_{\xi}E) - \iota_{\mu}\om \dwedge \dd (-\rho \cdot
   {\mit\Omega} + \LL_{\xi}{\mit\Omega})\, \big) \,, \nn \\
 &\phantom{{}=\bigg({}} \hspace{4cm}(-[\rho ,E]+
   \LL_{\xi}E)\wedge e - \om \wedge (-\rho \cdot {\mit\Omega}
   +\LL_{\xi}{\mit\Omega})\Big) \ . \nn
 \end{align}
The third term expands as
\begin{align}
\ell_{2}\Big(& \big(\dd x^\mu\otimes\Tr ( \iota_{\mu}\dd e \dwedge E +
               \iota_{\mu}\dd\om \dwedge {\mit\Omega} -\iota_{\mu}e
               \dwedge \dd E - \iota_{\mu}\om\dwedge \dd {\mit\Omega}
               ) \,,\, E\wedge e - \omega \wedge
               {\mit\Omega} \big) \,,\,(\xi,\rho)
               \Big) \label{eq:expand3} \\[4pt]
&=-\Big(\dd x^\mu\otimes\Tr\big(\iota_{\mu}\dd \rho \dwedge (E\wedge e -\om \wedge
  {\mit\Omega})\big) \nn \\ & \hspace{4cm} + \LL_{\xi}\big(\dd x^\mu\otimes(\iota_{\mu} \dd e \dwedge E-\iota_{\mu} e
  \dwedge \dd E + \iota_{\mu}\dd\om \dwedge {\mit\Omega}
  -\iota_{\mu}\om\dwedge \dd {\mit\Omega} ) \big) \,, \nn \\
&\phantom{{}=-\bigg(\Tr\big((\iota_{\mu}\dd \rho \dwedge (E\wedge
  e{}}\hspace{4cm}-\rho\cdot (E\wedge e-\om \wedge {\mit\Omega})
  +\LL_{\xi}(E\wedge e-\om \wedge {\mit\Omega})\Big) \ . \nn
\end{align} 
The second components in all three expanded expressions cancel each
other out, as a consequence of the Leibniz rules for both
gauge transformations in $(\xi,\rho)$. Again, for the first
components we need to work a bit harder.

Firstly, we collect terms in the first components involving the local
Lorentz fields $(\rho,e)$ and $\dd E$. These amount to 
\begin{align*}
\dd x^\mu\otimes\Tr(\rho \cdot \iota_{\mu} e \dwedge \dd E + \iota_{\mu} e\dwedge
  \rho \cdot \dd E) = \dd x^\mu\otimes\Tr \big(\rho \cdot (\iota_{\mu} e \dwedge
  \dd E)\big) =0
\end{align*}
where the vanishing of the last term follows since the local infinitesimal Lorentz transformation
$\rho$ acts on the top exterior vector in $\FR^{1,2}$, hence it is
invariant under finite $\sSO_+(1,2)$ Lorentz transformations and so the infinitesimal
transformation is zero; this is exactly the same argument which shows that
the action
functional \eqref{eq:ECP3daction} is invariant under local
Lorentz transformations. Similar arguments show that the terms involving $(\rho,\om)$ and
$\dd{\mit\Omega}$, $(\rho,\dd e)$ and $E$, and $(\rho,\dd \om)$
and ${\mit\Omega}$ cancel each other out.

Secondly, we collect terms involving $(\dd \rho,e)$ and $E$. These amount to
\begin{align*}
\dd x^\mu\otimes\Tr\big(-\iota_{\mu}(\dd \rho & \wedge e)\dwedge E+\iota_{\mu}
                      e\dwedge (\dd\rho \wedge E)-\iota_{\mu}\dd
                      \rho \dwedge (E\wedge e) \big) \\[4pt]
&=\dd x^\mu\otimes\Tr \big( \dd \rho \wedge (\iota_{\mu}e\dwedge
  E)\big) -\dd x^\mu\otimes\Tr\big((\iota_{\mu}\dd \rho)\wedge e \dwedge E -
  \iota_{{\mu}} \dd \rho \dwedge (E\wedge e)\big ) \ .
\end{align*}
The first term is zero using the $\sSO_+(1,2)$-invariance as before.
For the last two terms, we evaluate the Hodge operator explicitly to
write them as
\begin{align*}
\varepsilon_{abc}\,(\partial_{\mu} \rho^{a}{}_{d} \, e^{d} \wedge
  {E}^{bc} - \partial_{\mu}\rho^{ab} \, {E}^{c}{}_{d}\wedge e^{d})  \ .
\end{align*}
Now using the contracted Schouten identity \eqref{eq:Schouten}, the
first term here becomes
\begin{align*}
\varepsilon_{abc}\,\partial_{\mu} \rho^{a}{}_{d} \, e^{d} \wedge
  {E}^{bc}=-2\, \varepsilon_{abc} \, \partial_{\mu} \rho^{ad}\, {E}_{d}{}^{b}\wedge e^{c}
\end{align*}
and similarly for the second term:
\begin{align*}
-\varepsilon_{abc} \, \partial_{\mu}\rho^{ab} \, {E}^{c}{}_{d}\wedge
  e^{d} = -2 \,
  \varepsilon_{abc}\, \partial_{\mu} \rho^{a}{}_{d}\, {E}^{db}\wedge
  e^{c} \ .
\end{align*}
Hence the two terms cancel when added. Similarly the terms involving
$(\dd \rho,\om)$ and ${\mit\Omega}$ cancel.

Thirdly, we check the terms in the first components involving $(\xi,e)$ and
$E$. For this, we assume that the coframe field $e=e^a\,{\tt E}_a$ is invertible,
so that $\{e^a\}$ forms a basis for one-forms
$\Omega^{1}(M)$.\footnote{If $\{e^a\}$ is degenerate, one can always
  choose another basis and perform similar calculations. Here we make
  this assumption in order to streamline the computation.} The
coframe basis obeys Maurer--Cartan equations
\begin{align*}
\dd e^{c} = \tfrac{1}{2}\,k^{c}{}_{ab}\, e^{a} \wedge e^{b}
\end{align*}
where the local structure functions $k^c{}_{ab}$ are antisymmetric
in their lower indices. We can similarly write $\dd x^\mu\otimes\iota_\mu=e^a\otimes\iota_a$, where $\iota_a$ are the contractions with
vectors in the corresponding dual basis for vector fields
$\Gamma(TM)$, that is, $\iota_a(e^b)=\delta_a^b$. Collecting the
relevant terms from the first expansion \eqref{eq:expand1} and
expressing the arguments of the Hodge operator $\Tr$ in this basis
gives 
\begin{align*}
-\LL_{\xi}e^{a}\otimes ({\tt E}_{a}\dwedge \dd E) +
  (\LL_{\xi}k^{c}{}_{ab})\,e^{a}\otimes(e^{b}\,{\tt E}_{c}\dwedge E) 
                                                                 +k^{c}{}_{ab}\,
                                                                 \LL_{\xi}e^{a}\otimes(e^{b}\,{\tt E}_{c}\dwedge
                                                                 E)                                                                  -
                                                                 k^{c}{}_{ba}\,
                                                                 e^{a}
                                                                 \otimes(\LL_{\xi}
                                                                 e^{b}\,{\tt E}_{c}\dwedge
                                                                 E) \ .
\end{align*} 
Similarly, from the second expansion \eqref{eq:expand2} we get
\begin{align*}
k^{c}{}_{ba}\, e^{b}\otimes(e^{a}\,{\tt E}_{c}\dwedge \LL_{\xi}E)- e^{b}\otimes ({\tt E}_{b}\dwedge \dd \LL_{\xi} E)
\end{align*}
and lastly from the third expansion \eqref{eq:expand3} we get
\begin{align*}
-\LL_{\xi}(k^{c}{}_{ba}\, e^{b})\otimes(e^{a}\,{\tt E}_{c}\dwedge E) -
  k^{c}{}_{ba}\, e^{b}\otimes \LL_{\xi}(e^{a}\, {\tt E}_{c}\dwedge E) +
  \LL_{\xi}e^{b}\otimes({\tt E}_{b}\dwedge \dd E) + e^{b}\otimes(
  {\tt E}_{b}\dwedge \LL_{\xi}\dd E) \ .
\end{align*}
Using the Leibniz rule for the Lie derivative $\LL_\xi$ and the fact
that it commutes with the exterior derivative $\dd$, the third term completely
cancels with the first two terms. One similarly
checks the vanishing of the terms containing $(\xi,\om)$ and
${\mit\Omega}$. This completes the proof of the homotopy
identity~\eqref{eq:deg3laomOmid}. 
\end{proof}

\begin{lemma}
If
$(\xi_{1},\rho_{1})$ and $(\xi_{2},\rho_{2})$ are two pairs of
gauge transformations in $V_0$, and $({\CX},{\CP})\in V_3$
is a pair of Noether identities, then
\begin{align}
\ell_{2}\Big(\ell_{2}\big((\xi_1,\rho_1)\,,\,(\xi_2,\rho_2)\big)\,,\,({\CX},{\CP})\Big)
  &+
    \ell_{2}\Big(\ell_{2}\big(({\CX},{\CP})\,,\,(\xi_1,\rho_1)\big)\,,\,(\xi_2,\rho_2)\Big)
    \nn \\
& \hspace{1cm} +
  \ell_{2}\Big(\ell_{2}\big((\xi_2,\rho_2)\,,\,({\CX},{\CP})\big)\,,\,(\xi_1,\rho_1)\Big)
  = (0,0) \ . \label{eq:deg3laLaid}
\end{align}
\end{lemma}
\begin{proof}
The first term of \eqref{eq:deg3laLaid} expands as
\begin{align*}
\ell_{2}\big(&([\xi_1,\xi_2],-[\rho_1,\rho_2]+\LL_{\xi_{1}}\rho_2
               -\LL_{\xi_{2}}\rho_1)\,,\,({\CX},{\CP})\big)
  \\ & \hspace{2cm}
       =\Big(\dd x^\mu\otimes\Tr\big(\partial_{\mu}(-[\rho_1,\rho_2]+\LL_{\xi_{1}}\rho_2
       -\LL_{\xi_{2}}\rho_1)\dwedge {\CP}\big) + \LL_{[\xi_1,\xi_2]}{\CX} \,, \\
& \hspace{6cm} -(-[\rho_1,\rho_2]+\LL_{\xi_1}\rho_2
  -\LL_{\xi_2}\rho_1)\cdot {\CP} + \LL_{[\xi_1,\xi_2]}
  {\CP}\Big) \ .
\end{align*}
The second term expands as
\begin{align*}
-\ell_{2}\Big(\big(&\dd x^\mu\otimes\Tr(\partial_{\mu}\rho_{1}\dwedge
                     {\CP})+\LL_{\xi_1}{\CX},-\rho_1\cdot
                     {\CP}+\LL_{\xi_1}{\CP}\big)\,,\,(\xi_2,\rho_2)\Big)
  \\
& =\Big( \dd x^\mu\otimes\Tr \big(\partial_{\mu} \rho_{2}\dwedge(-\rho_1 \cdot
  {\CP} +\LL_{\xi_1}{\CP})\big)
  +\LL_{\xi_2}\big(\dd x^\mu\otimes\Tr(\partial_{\mu}\rho_1\dwedge
  {\CP}) \big) + \LL_{\xi_2} \LL_{\xi_1}
  {\CX} \,, \\
&\phantom{{}=\bigg({}} \hspace{4cm} +\rho_2\cdot(\rho_1\cdot
  {\CP})- \rho_2 \cdot \LL_{\xi_1}{\CP} -
  \LL_{\xi_2}(\rho_1 \cdot {\CP}) +
  \LL_{\xi_{2}}\LL_{\xi_{1}} {\CP}\Big) \ .
\end{align*}
The third term is just the negative of the second term with the indices $(1,2)$ interchanged. 
The second component of the identity \eqref{eq:deg3laLaid} is verified
by simply noting that the space $\Omega^{3}(M,\FR^{1,2})$ in which it
lives is a module over the Lie algebra
$\Gamma(TM)\rtimes\Omega^{0}\big(M,\mathfrak{so}(1,2)\big)$, just like
in the proof for two pairs of gauge parameters and one pair of
dynamical fields in total degree~$1$.

For the first component we verify this by collecting similar terms in
turn, starting with the terms involving the action of two Lie derivatives:
\begin{align*}
\LL_{[\xi_1,\xi_2]}{\CX} +
  \LL_{\xi_2}\LL_{\xi_1}{\CX}-\LL_{\xi_1}\LL_{\xi_2}{\CX}=0 \ ,
\end{align*}
where we used the Cartan identity \eqref{eq:CartanidLiecomm}.
Next we collect terms in the first component with one Lie derivative acting:
\begin{align*}
\dd x^\mu\otimes(\partial_{\mu} \LL_{\xi_1}\rho_2\dwedge {\CP} ) - \dd x^\mu\otimes(\partial_{\mu} \LL_{\xi_2}\rho_1\dwedge {\CP} ) + \dd x^\mu\otimes(\partial_\mu \rho_2\dwedge \LL_{\xi_1}{\CP})+\LL_{\xi_2} \big(\dd x^\mu\otimes(\partial_\mu \rho_{1}\dwedge {\CP})\big)\\
- \dd x^\mu\otimes(\partial_\mu \rho_1\dwedge \LL_{\xi_2}{\CP})-\LL_{\xi_1} \big(\dd x^\mu\otimes(\partial_\mu \rho_{2}\dwedge
  {\CP})\big)=0 \ ,
\end{align*}
where the vanishing is easily seen by using the Leibniz rule for the
Lie derivatives acting in the fourth and sixth terms, which is
possible in this case as the Lie derivatives here act on
functions. Lastly we collect terms with no Lie derivatives acting,
which are the leftover terms involving $\rho_1$ and~$\rho_2$:
\begin{align*}
-\dd x^\mu\otimes\Tr\big(\partial_{\mu}([\rho_1,\rho_2])\dwedge
  {\CP}\big)-\dd x^\mu\otimes\Tr\big(\partial_\mu
  \rho_2\dwedge \rho_1 \cdot {\CP} \big)+\dd x^\mu\otimes\Tr\big(\partial_\mu \rho_1\dwedge \rho_2 \cdot
  {\CP} \big) \ .
\end{align*}
Applying the derivatives $\partial_\mu$ using the Leibniz rule, the
terms involving $\partial_\mu \rho_2$ give
\begin{align*}
-\dd x^\mu\otimes\Tr\big([\rho_1,\partial_\mu \rho_2]\dwedge {\CP}\big)
  - \dd x^\mu\otimes\Tr\big(\partial_\mu \rho_2 \dwedge \rho_1
  \cdot {\CP}\big)=-\dd x^\mu\otimes\Tr\big(\rho_1\cdot
  (\partial_{\mu}\rho_{2}\dwedge{\CP})\big) = 0
\end{align*}
since the top exterior vector component is invariant under $\sSO_+(1,2)$-transformations, as before. Similarly the terms involving
$\partial_\mu \rho_1$ vanish, which completes the proof of the
final homotopy identity \eqref{eq:deg3laLaid} in total degree~$3$. 
\end{proof}

As all higher homotopy relations vanish trivially, this completes the
proof that the structure defined by
\eqref{eq:3dV}--\eqref{eq:3dbrackets} is indeed an
$L_{\infty}$-algebra.

\subsection{Covariant homotopy relations}
\label{app:3dcovhomrels}

We shall now provide details of some illustrative checks of the
homotopy relations for the covariant $L_\infty$-algebra of
Section~\ref{sec:covariant} (in the case $d=3$ and
$\Lambda=0$). Since all $1$-brackets are unchanged by the
covariantization, the differential conditions $\CJ^{\rm cov}_1=0$ hold
just like in the non-covariant case. On the other hand, the
$2$-brackets are almost all modified. For example, in the first
bracket of \eqref{eq:ell2cov} the Lie derivatives are removed as they would
otherwise spoil the closure relation for gauge transformations, while
in the sixth bracket the last term from the first component is removed
by the new covariant
Noether identity \eqref{covnoether}. It is then a straightforward proof, very similar to the non-covariant case, that the remaining modifications ensure that altogether the Leibniz rules $\CJ_2^{\rm cov}=0$ hold.

Compared to the non-covariant case of Appendix~\ref{app:3dhomotopy},
the main new features that arise are due to the fact that the covariant
$L_\infty$-algebra for three-dimensional gravity is no longer a
differential graded Lie algebra: there are new non-vanishing higher
brackets \eqref{eq:ell3cov} and \eqref{eq:ell4cov}. These incorporate
the covariant gauge transformations \eqref{covgauge} and their closure,
together with the covariant Noether identity
\eqref{covnoether}. Correspondingly, new higher homotopy identities
should be checked, so we will focus on those.

\subsubsection*{Homotopy Jacobi identities}

The modification of the $2$-brackets leads to the six non-trival
$3$-brackets \eqref{eq:ell3cov}, and these have to be included when proving 
the homotopy relations $\CJ_3^{\rm cov}=0$, which are given in general in \eqref{I3}. The calculations are very similar to the non-covariant case, bearing in mind that 
some of the terms that came from non-covariant $2$-brackets will now appear as a part of the covariant 
$3$-brackets. We illustrate this by explicitly demonstrating the identity
\begin{equation}
\CJ_3^{\rm cov} \big( (\xi_1, \rho_1)\,,\, (\xi_2, \rho_2)\,,\, (e,\omega) \big)   =(0,0) \ .\label{I3Example}
\end{equation}
Firstly, we write out the homotopy relation explicitly to get
\begin{align*}
&\CJ_3^{\rm cov} \big( (\xi_1, \rho_1)\,,\, (\xi_2, \rho_2)\,,\, (e,\omega) \big)  \\[4pt]
& \hspace{1cm} = \ell_{1}^{\rm cov}\Big(\ell_{3}^{\rm
  cov}\big((\xi_1,\rho_1)\,,\, (\xi_2,\rho_2)\,,\, (e,\omega)
  \big)\Big) + \ell_{3}^{\rm cov}\big(\ell_{1}^{\rm
  cov}(\xi_1,\rho_1)\,,\, (\xi_2,\rho_2)\,,\, (e,\omega) \big) \\
& \hspace{2cm} + \ell_{3}^{\rm cov}\big((\xi_1,\rho_1)\,,\, \ell_{1}^{\rm cov}(\xi_2,\rho_2)\,,\, (e,\omega) \big)
+ \ell_{3}^{\rm cov}\big((\xi_1,\rho_1)\,,\, (\xi_2,\rho_2)\,,\, \ell_{1}^{\rm cov}(e,\omega) \big)\\
& \hspace{3cm} + \ell_{2}^{\rm cov}\Big(\ell_{2}^{\rm cov}\big( (\xi_1,\rho_1)\,,\, (\xi_2,\rho_2)\big) \,,\, (e,\omega) \Big)
+ \ell_{2}^{\rm cov}\Big(\ell_{2}^{\rm cov}\big( (e,\omega)\,,\, (\xi_1,\rho_1)\big) \,,\, (\xi_2,\rho_2) \Big)\\
& \hspace{4cm} + \ell_{2}^{\rm cov}\Big(\ell_{2}^{\rm cov}\big( (\xi_2,\rho_2)\,,\, (e,\omega)\big) \,,\, (\xi_1,\rho_1) \Big) \ .
\end{align*}
Then we separately consider terms coming from $3$-brackets and terms coming from $2$-brackets. The 
non-vanishing $3$-bracket terms are given by
\begin{align}
&\ell_{1}^{\rm cov} \big( 0\,,\, -\iota_{\xi_1}\iota_{\xi_2}\dd\omega\big) + \ell_{3}^{\rm cov} \big( (0,\dd\rho_1)\,,\,(\xi_2,\rho_2)\,,\, (e,\omega)   \big)  
+ \ell_{3}^{\rm cov} \big((\xi_1,\rho_1)\,,\, (0,\dd\rho_2)\,,\,
  (e,\omega)   \big) \nn \\[4pt]
& \hspace{3cm} = \big( \iota_{\xi_2}\dd\rho_1 \cdot e - \iota_{\xi_1}\dd\rho_2 \cdot e \,,\, -\dd\iota_{\xi_1}\iota_{\xi_2}\dd\omega 
+ \iota_{\xi_2}[\dd\rho_1,\omega] - \iota_{\xi_1}[\dd\rho_2,\omega]
  \big) \nn \\[4pt]
& \hspace{4cm} = \big( \LL_{\xi_2}\rho_1 \cdot e - \LL_{\xi_1}\rho_2 \cdot e \,,\, -\dd\iota_{\xi_1}\iota_{\xi_2}\dd\omega 
+ \iota_{\xi_2}[\dd\rho_1,\omega] - \iota_{\xi_1}[\dd\rho_2,\omega]
  \big) \ , \label{eq:I}
\end{align}
where in the last line we used $\iota_{\xi}\dd\rho = \LL_\xi \rho$. The $2$-bracket terms are given by
\begin{align}
& \ell_{2}^{\rm cov}\big(( [\xi_1,\xi_2]\,,\, -[\rho_1, \rho_2]) \,,\, (e,\omega) \big)
- \ell_{2}^{\rm cov}\big(( -\rho_1\cdot e + \LL_{\xi_1}e \,,\, -[\rho_1, \omega] + \iota_{\xi_1}\dd\omega ) 
\,,\, 
(\xi_2,\rho_2) \big) \nn \\
& \hspace{4cm} + \ell_{2}^{\rm cov}\big(( -\rho_2\cdot e + \LL_{\xi_2}e \,,\, -[\rho_2, \omega] + \iota_{\xi_2}\dd\omega) 
\,,\, 
(\xi_1,\rho_1) \big) \nn \\[4pt]
& \hspace{2cm} = \big( \LL_{\xi_1}\rho_2 \cdot e - \LL_{\xi_2}\rho_1 \cdot e \,,\, \iota_{[\xi_1,\xi_2]}\dd\omega 
-\iota_{\xi_2}[\dd\rho_1,\omega] +  \iota_{\xi_1}[\dd\rho_2,\omega] + \iota_{\xi_2}\dd\iota_{\xi_1}\dd\omega - 
\iota_{\xi_1}\dd\iota_{\xi_2}\dd\omega \big) \nn \\[4pt]
& \hspace{3cm} = \big( \LL_{\xi_1}\rho_2 \cdot e - \LL_{\xi_2}\rho_1 \cdot e \,,\,  -\dd\iota_{\xi_2}\iota_{\xi_1}\dd\omega 
-\iota_{\xi_2}[\dd\rho_1,\omega] +  \iota_{\xi_1}[\dd\rho_2,\omega]
  \big) \ , \label{eq:II}
\end{align}
where in the last line we used the Cartan identity \eqref{eq:Cartanidiota}.
Adding \eqref{eq:I} and \eqref{eq:II} together then proves
\eqref{I3Example}. With similar techniques, one can show the remaining
homotopy relations $\CJ_3^{\rm cov}=0$.

\subsubsection*{Higher homotopy identities}

The remaining new homotopy relations to consider are $\CJ_4^{\rm
  cov}=0$, which are generally given by \eqref{I4}. 
We will only prove here the homotopy relations involving non-zero
$4$-brackets. Recall that there are three non-zero $4$-brackets given by
\eqref{eq:ell4cov}. 

\paragraph{$\underline{{\rm Total~degree}~2:}$} 
The first non-trival homotopy relation is given by
\begin{align}
&\CJ_4^{\rm cov} \big( (\xi_1,\rho_1)\,,\, (\xi_2,\rho_2)\,,\, (e_1,
  \omega_1)\,,\, (e_2, \omega_2)\big) \nn \\[4pt]
& =   \ell_1^{\rm cov} \Big( \ell_4^{\rm cov} \big( (\xi_1,\rho_1)\,,\, (\xi_2,\rho_2)\,,\, (e_1, \omega_1)\,,\, (e_2, 
\omega_2)\big) \Big) \label{eq:J4deg2}
\\
& - \ell_2^{\rm cov}\Big( \ell_3^{\rm cov} \big( (\xi_1,\rho_1)\,,\, (\xi_2,\rho_2)\,,\, (e_1, \omega_1)\big) \,,\, 
(e_2, \omega_2) \Big)
- \ell_2^{\rm cov}\Big( \ell_3^{\rm cov} \big( (\xi_1,\rho_1)\,,\, (\xi_2,\rho_2)\,,\, (e_2, \omega_2)\big) \,,\, (e_1, 
\omega_1) \Big) \nn \\
& - \ell_2^{\rm cov}\Big((\xi_1,\rho_1)\,,\, \ell_3^{\rm cov} \big( (\xi_2,\rho_2)\,,\, (e_1, \omega_1) \,,\, (e_2, 
\omega_2)\big) \Big)
+ \ell_2^{\rm cov}\Big((\xi_2,\rho_2)\,,\, \ell_3^{\rm cov} \big( (\xi_1,\rho_1)\,,\, (e_1, \omega_1) \,,\, (e_2, 
\omega_2)\big) 
\Big) \nn \\
& + \ell_3^{\rm cov}\Big( \ell_2^{\rm cov} \big( (\xi_1,\rho_1)\,,\, (\xi_2,\rho_2)\big) \,,\, (e_1, \omega_1) \,,\, 
(e_2, \omega_2) \Big)
- \ell_3^{\rm cov}\Big( \ell_2^{\rm cov} \big( (\xi_1,\rho_1)\,,\, (e_1, \omega_1) \big) \,,\, (\xi_2,\rho_2) \,,\, 
(e_2, \omega_2) \Big) \nn \\
& - \ell_3^{\rm cov}\Big( \ell_2^{\rm cov} \big( (\xi_1,\rho_1)\,,\, (e_2, \omega_2) \big) \,,\, (\xi_2,\rho_2) \,,\, 
(e_1, \omega_1) 
\Big) \nn \\
& - \ell_3^{\rm cov}\Big(  (\xi_1,\rho_1)\,,\, \ell_2^{\rm cov} \big( (\xi_2,\rho_2) \,,\, (e_1, \omega_1)\big) \,,\, 
(e_2, \omega_2) \Big)
- \ell_3^{\rm cov}\Big(  (\xi_1,\rho_1)\,,\, \ell_2^{\rm cov} \big( (\xi_2,\rho_2) \,,\, (e_2, \omega_2)\big) \,,\, 
(e_1, \omega_1) \Big) \nn
\end{align}
where we wrote only the non-vanishing brackets. We split this long
equation into three types of contributions:
\begin{align*}
\CJ_4^{\rm cov} \big( (\xi_1,\rho_1)\,,\, (\xi_2,\rho_2)\,,\, (e_1,
  \omega_1)\,,\, (e_2, \omega_2)\big) = \rm{I}^{(2)} (\ell_1^{\rm cov}\circ\ell_4^{\rm cov}) + \rm{II}^{(2)} (\ell_2^{\rm cov}\circ\ell_3^{\rm 
cov}) + {\rm III}^{(2)} (\ell_3^{\rm cov}\circ\ell_2^{\rm cov}) \ . 
\end{align*}
The first term is given by
\begin{equation}
\rm{I}^{(2)} = (0, \dd\iota_{\xi_1}\iota_{\xi_2}[\omega_1,\omega_2]) \ .\nn 
\end{equation}
The second term has four contributions:
\begin{align*}
\rm{II}^{(2)} = \rm{II}^{(2)}_1+ \rm{II}^{(2)}_2 + \rm{II}^{(2)}_3 + \rm{II}^{(2)}_4
\end{align*}
where
\begin{align*}
\rm{II}^{(2)}_1 =& \, (-\iota_{\xi_1}\iota_{\xi_2}\dd\omega_1\cdot e_2,
             -[\iota_{\xi_1}\iota_{\xi_2}\dd\omega_1, \omega_2 ]) \ , \\[4pt]
\rm{II}^{(2)}_2 =& \, {\rm II}^{(2)}_1 \,(e_1\leftrightarrow e_2, \omega_1
             \leftrightarrow \omega_2) \ , \\[4pt]
\rm{II}^{(2)}_3 =& \, \big(-\rho_1\cdot(\iota_{\xi_2}\omega_1\cdot e_2) - \rho_1\cdot(\iota_{\xi_2}\omega_2\cdot e_1) 
+ \LL_{\xi_1}(\iota_{\xi_2}\omega_1\cdot e_2) + \LL_{\xi_1}(\iota_{\xi_2}\omega_2\cdot e_1)  \,,\, \\
& \quad\quad -\big[\rho_1, [\iota_{\xi_2}\omega_1, \omega_2]\big] - \big[\rho_1, [\iota_{\xi_2}\omega_2, \omega_1]\big] 
+ \iota_{\xi_1}\dd[\iota_{\xi_2}\omega_1, \omega_2] +
  \iota_{\xi_1}\dd[\iota_{\xi_2}\omega_2, \omega_1] \big) \ , \\[4pt]
\rm{II}^{(2)}_4 =& \, -\rm{II}^{(2)}_3 \, (\xi_1\leftrightarrow \xi_2, \rho_1 \leftrightarrow \rho_2) \ .
\end{align*}
The third term has five contributions:
\begin{align*}
{\rm III}^{(2)} = {\rm III}^{(2)}_1+{\rm III}^{(2)}_2+{\rm III}^{(2)}_3+{\rm III}^{(2)}_4 + {\rm III}^{(2)}_5
\end{align*}
where
\begin{align*}
{\rm III}^{(2)}_1 =& \, -\big( (\iota_{[\xi_1,\xi_2]}\omega_1) \cdot e_2 + (\iota_{[\xi_1,\xi_2]}\omega_2) \cdot e_1\,,\,  
[\iota_{[\xi_1,\xi_2]}\omega_1 , \omega_2] +
              [\iota_{[\xi_1,\xi_2]}\omega_2, \omega_1] \big) \ , \\[4pt]
{\rm III}^{(2)}_2 =& \, -\big( -[\rho_1,\iota_{\xi_2}\omega_1]\cdot e_2 + \iota_{\xi_2}\iota_{\xi_2}\dd\omega_1\cdot e_2
-\iota_{\xi_2}\omega_2 \cdot (\rho_1\cdot e_1) + \iota_{\xi_2}\omega_2 \cdot \LL_{\xi_1}e_1\,,\, \\
&\qquad\qquad -\big[[\rho_1, \iota_{\xi_2}\omega_1], \omega_2\big] + [\iota_{\xi_2}\iota_{\xi_2}\dd\omega_1 ,\omega_2] 
- \big[\iota_{\xi_2}\omega_2, [\rho_1,\omega_1]\big] +
  [\iota_{\xi_2}\omega_2, \iota_{\xi_1}\dd\omega_1] \big) \ , \\[4pt]
{\rm III}^{(2)}_3 =& \, {\rm III}^{(2)}_2 \, (e_1\leftrightarrow e_2, \omega_1
              \leftrightarrow \omega_2) \ , \\[4pt]
{\rm III}^{(2)}_4 =& \, -{\rm III}^{(2)}_2 \, (\xi_1\leftrightarrow \xi_2, \rho_1
              \leftrightarrow \rho_2) \ , \\[4pt]
{\rm III}^{(2)}_5 =& \, {\rm III}^{(2)}_4 \, (e_1\leftrightarrow e_2, \omega_1 \leftrightarrow \omega_2) \ .
\end{align*}

We now collect all the terms in the first slots of the brackets. Most
of the terms cancel straightforwardly and we are 
left with
\begin{equation}
(\LL_{\xi_1}\iota_{\xi_2}\omega_1 - \iota_{\xi_1}\LL_{\xi_2}\omega_1)\cdot e_2 + (\LL_{\xi_1}\iota_{\xi_2}\omega_1 - 
\iota_{\xi_1}\LL_{\xi_2}\omega_2)\cdot e_1 \ . \nn
\end{equation}
Using the Cartan formula for the Lie derivative and noting that $\iota_{\xi_1}\iota_{\xi_2}\omega= 0$, since 
$\omega$ is a one-form, we see that the remaining terms also
cancel. In this way we have shown that the first slot 
of the brackets in \eqref{eq:J4deg2} is 
equal to zero.
Collecting all the terms in the second slots of the brackets, we notice that the double commutators combine into Jacobi 
identities and thus vanish. Some of the single commutators cancel straightforwardly. The remaining ones also 
cancel, but one has to use the Cartan formula for the Lie derivative
and the identity (\ref{eq:Cartanidiota}). This completes the proof
of the homotopy relation $\CJ_4^{\rm cov}=0$ for
\eqref{eq:J4deg2}.

\paragraph{$\underline{{\rm Total~degree}~4:}$} 
Finally, we have to check that the homotopy relation
\begin{equation}
\CJ_4^{\rm cov} \big( (\xi_1,\rho_1)\,,\, (\xi_2,\rho_2)\,,\, (e, \omega)\,,\, (\CX,\CP) \big)   = (0,0) \label{eq:J4deg4}
\end{equation}
holds. This relation has $13$ non-vanishing terms given by
\begin{align*}
&\CJ_4^{\rm cov} \big( (\xi_1,\rho_1)\,,\, (\xi_2,\rho_2)\,,\, (e, \omega)\,,\, (\CX,\CP)\big) \nn \\[4pt]
& =   \ell_1^{\rm cov} \Big( \ell_4^{\rm cov} \big( (\xi_1,\rho_1)\,,\, (\xi_2,\rho_2)\,,\, (e, \omega)\,,\, 
(\CX,\CP)\big) \Big) 
\\
& \quad - \ell_4^{\rm cov} \big(  \ell_1^{\rm cov}(\xi_1,\rho_1)\,,\, (\xi_2,\rho_2)\,,\, (e, \omega)\,,\, (\CX,\CP) \big) - 
\ell_4^{\rm cov} \big(  
(\xi_1,\rho_1)\,,\, \ell_1^{\rm cov} (\xi_2,\rho_2)\,,\, (e, \omega)\,,\, (\CX,\CP) \big) \nn \\
& \quad - \ell_2^{\rm cov}\Big( \ell_3^{\rm cov} \big( (\xi_1,\rho_1)\,,\, (\xi_2,\rho_2)\,,\, (e, \omega)\big) \,,\, 
(\CX,\CP) \Big)
- \ell_2^{\rm cov}\Big( \ell_3^{\rm cov} \big( (\xi_1,\rho_1)\,,\, (\xi_2,\rho_2)\,,\, (\CX,\CP)\big) \,,\, (e, \omega) 
\Big) \nn \\
& \quad - \ell_2^{\rm cov}\Big((\xi_1,\rho_1)\,,\, \ell_3^{\rm cov} \big( (\xi_2,\rho_2)\,,\, (e, \omega) \,,\, (\CX,\CP) 
\big) \Big)
+ \ell_2^{\rm cov}\Big((\xi_2,\rho_2)\,,\, \ell_3^{\rm cov} \big( (\xi_1,\rho_1)\,,\, (e, \omega) \,,\, (\CX,\CP)\big) 
\Big) \nn \\
& \quad + \ell_3^{\rm cov}\Big( \ell_2^{\rm cov} \big( (\xi_1,\rho_1)\,,\, (\xi_2,\rho_2)\big) \,,\, (e, \omega) \,,\, (\CX, 
\CP) \Big)
- \ell_3^{\rm cov}\Big( \ell_2^{\rm cov} \big( (\xi_1,\rho_1)\,,\, (e, \omega) \big) \,,\, (\xi_2,\rho_2) \,,\, (\CX, 
\CP) \Big) \nn \\
& \quad - \ell_3^{\rm cov}\Big( \ell_2^{\rm cov} \big( (\xi_1,\rho_1)\,,\, (\CX, \CP) \big) \,,\, (\xi_2,\rho_2) \,,\, (e, 
\omega) 
\Big) \nn \\
& \quad - \ell_3^{\rm cov}\Big(  (\xi_1,\rho_1)\,,\, \ell_2^{\rm cov} \big( (\xi_2,\rho_2) \,,\, (e, \omega)\big) \,,\, (\CX, 
\CP) \Big)
- \ell_3^{\rm cov}\Big(  (\xi_1,\rho_1)\,,\, \ell_2^{\rm cov} \big( (\xi_2,\rho_2) \,,\, (\CX, \CP)\big) \,,\, (e, 
\omega) \Big) \ .
\end{align*} 
As previously, we group terms according to the order of brackets as
\begin{align*}
\CJ_4^{\rm cov} \big( (\xi_1,\rho_1)\,,\, (\xi_2,\rho_2)\,,\, (e, \omega)\,,\, (\CX, \CP)\big) = {\rm I}^{(4)} (\ell_1^{\rm cov}\circ\ell_4^{\rm cov}) + {\rm II}^{(4)} (\ell_2^{\rm cov}\circ\ell_3^{\rm cov}) + {\rm III}^{(4)} (\ell_3^{\rm cov}\circ\ell_2^{\rm cov}) \ . 
\end{align*}
The first term has three contributions:
\begin{align*}
{\rm I}^{(4)} = {\rm I}^{(4)}_1 + {\rm I}^{(4)}_2 + {\rm I}^{(4)}_3
\end{align*}
where
\begin{align*}
{\rm I}^{(4)}_1 &= \big( 0\, ,\, \dd\omega\wedge \iota_{\xi_1}\iota_{\xi_2}\CP - \omega\wedge 
\dd\iota_{\xi_1}\iota_{\xi_2}\CP\big) \ , \\[4pt]
{\rm I}^{(4)}_2 &= \big( \dd x^{\mu} \otimes \Tr(\iota_{\mu}\iota_{\xi_2}[\dd\rho_1,\omega]\dwedge \CP)\,,\,0 
\big) \ , \\[4pt]
{\rm I}^{(4)}_3 &= -{\rm I}^{(4)}_2 (\xi_1\leftrightarrow \xi_2, \rho_1 \leftrightarrow \rho_2) \ .
\end{align*}
The second term has four contributions:
\begin{align*}
{\rm II}^{(4)} = {\rm II}^{(4)}_1 + {\rm II}^{(4)}_2 + {\rm II}^{(4)}_3 + {\rm II}^{(4)}_4
\end{align*}
where
\begin{align*}
{\rm II}^{(4)}_1 &= \big( 0\, ,\, -\iota_{\xi_1}\iota_{\xi_2}\dd\omega\cdot \CP \big) \ , \\[4pt]
{\rm II}^{(4)}_2 &= -\big( \dd x^{\mu} \otimes \Tr(\iota_{\mu}\dd\omega\dwedge 
\dd\iota_{\xi_1}\iota_{\xi_2}\CP)\,,\, -\omega\wedge \dd\iota_{\xi_1}\iota_{\xi_2}\CP\big) \ , \\[4pt]
{\rm II}^{(4)}_3 &= \Big( \LL_{\xi_2}\big(\dd x^{\mu} \otimes \Tr(\iota_{\mu}\iota_{\xi_1}\dd\omega\dwedge 
\CP) \big) \,,\, 0\Big) \ , \\[4pt]
{\rm II}^{(4)}_4 &= -{\rm II}^{(4)}_3 (\xi_1\leftrightarrow \xi_2, \rho_1 \leftrightarrow \rho_2) \ .
\end{align*}
Finally, the third term has five contributions:
\begin{align*}
{\rm III}^{(4)} = {\rm III}^{(4)}_1 + {\rm III}^{(4)}_2 + {\rm III}^{(4)}_3 + 
{\rm III}^{(4)}_4 + {\rm III}^{(4)}_5
\end{align*}
where
\begin{align*}
{\rm III}^{(4)}_1 &= \big( \dd x^{\mu} \otimes \Tr(\iota_{\mu}\iota_{[\xi_1,\xi_2]}\dd\omega\dwedge\CP) \, ,\, 0 
\big) \ , \\[4pt]
{\rm III}^{(4)}_2 &= \Big( \dd x^{\mu} \otimes \Tr\big(\iota_{\mu}\iota_{\xi_2}\dd(-[\rho_1,\omega]+ 
\iota_{\xi_1}\dd\omega) \dwedge\CP\big)\,,\, 0\Big) \ , \\[4pt]
{\rm III}^{(4)}_3 &= -\Big( \dd x^{\mu} \otimes \Tr\big(\iota_{\mu}\iota_{\xi_2}\dd\omega\dwedge 
(\rho_1\cdot \CP) \big) \,,\, 0\Big) \ , \\[4pt]
{\rm III}^{(4)}_4 &= -{\rm III}^{(4)}_2 (\xi_1\leftrightarrow \xi_2, \rho_1 \leftrightarrow \rho_2) \ , \\[4pt]
{\rm III}^{(4)}_5 &= -{\rm III}^{(4)}_3 (\xi_1\leftrightarrow \xi_2, \rho_1 \leftrightarrow \rho_2) \ .
\end{align*}

The terms in the second slots combine into
\begin{align*}
\dd\omega\wedge\iota_{\xi_1}\iota_{\xi_2}\CP-\iota_{\xi_1}\iota_{\xi_2}
\dd\omega\cdot
\CP &= \iota_{\xi_1}(\dd\omega\wedge\iota_{\xi_2}\CP)-\iota_{\xi_1}\dd\omega\wedge\iota_{\xi_2}
\CP+\iota_{\xi_2}\iota_{\xi_1}\dd\omega\cdot\CP \\[4pt]
&= \iota_{\xi_1}(\dd\omega\wedge\iota_{\xi_2}\CP)+\iota_{\xi_2}(\iota_{\xi_1}\dd\omega\wedge\CP) \\[4pt]
&= 0 \ ,
\end{align*}
since the remaining two terms are contractions of four-forms, which are identically equal to zero in three dimensions.
For the terms in the first slots, we split them into terms with gauge parameters $\rho$ and terms without them. The terms 
with gauge parameters combine into
\begin{align*}
& \dd x^{\mu} \otimes \Tr \big( -[\rho_1, \iota_{\mu}\iota_{\xi_2}\dd\omega]\dwedge\CP 
-\iota_{\mu}\iota_{\xi_2}\dd\omega \dwedge (\rho_1\cdot\CP) + [\rho_2, \iota_{\mu}\iota_{\xi_1}\dd\omega]\dwedge\CP 
\iota_{\mu}\iota_{\xi_1}\dd\omega \dwedge (\rho_2\cdot\CP) \big) \\[4pt]
& \hspace{3cm} = \dd x^{\mu} \otimes \Tr \big(-\rho_1\cdot (\iota_{\mu}\iota_{\xi_2}\dd\omega\dwedge\CP) + 
\rho_2\cdot(\iota_{\mu}\iota_{\xi_1}\dd\omega\dwedge\CP) \big) \\[4pt]
& \hspace{6cm} = 0 \ ,
\end{align*}
where the vanishing of the last terms follows again from the invariance of a top exterior vector under local Lorentz transformations.
The terms without gauge parameters are given by
\begin{align}
& \!\!\!\! -\dd x^{\mu} \otimes \Tr( \iota_{\mu}\dd\omega\dwedge 
\dd\iota_{\xi_1}\iota_{\xi_2}\CP ) + \LL_{\xi_2}\big(\dd x^{\mu} \otimes \Tr(\iota_{\mu}\iota_{\xi_1}\dd\omega\dwedge 
\CP)\big) - \LL_{\xi_1}\big(\dd x^{\mu} \otimes \Tr(\iota_{\mu}\iota_{\xi_2}\dd\omega\dwedge 
\CP)\big) \label{eq:withoutrho} \\
& \quad + \dd x^{\mu} \otimes \Tr(\iota_{\mu}\iota_{[\xi_1,\xi_2]}\dd\omega\dwedge\CP) + \dd x^{\mu} \otimes \Tr\big( (\iota_{\mu}\iota_{\xi_2}\dd\iota_{\xi_1}\dd\omega) \dwedge\CP\big)
- \dd x^{\mu} \otimes \Tr\big( (\iota_{\mu}\iota_{\xi_1}\dd\iota_{\xi_2}\dd\omega) \dwedge\CP\big) \ . \nn
\end{align}
Using the Cartan identity (\ref{eq:Cartanidiota}), the Leibniz tule for the Lie derivative and noting that
\begin{equation*}
(\LL_{\xi_2}\dd x^{\mu}) \otimes \Tr\big(\iota_{\mu}\iota_{\xi_1}\dd\omega\dwedge 
\CP\big) =   \dd x^{\mu} \otimes \Tr\big( \iota_{\mu}\LL_{\xi_2}\iota_{\xi_1}\dd\omega - 
\LL_{\xi_2}\iota_{\mu}\iota_{\xi_1}\dd\omega \dwedge\CP \big) \ ,
\end{equation*}
the terms \eqref{eq:withoutrho} combine into
\begin{align*}
-\dd x^{\mu} \otimes \Tr\big(& \iota_{\mu}\dd\omega\dwedge \dd\iota_{\xi_1}\iota_{\xi_2}\CP  
+ \iota_{\mu}\iota_{\xi_1}\dd\omega\dwedge \dd\iota_{\xi_2}\CP 
- \iota_{\mu}\iota_{\xi_2}\dd\omega\dwedge \dd\iota_{\xi_1}\CP \\
& + ( \iota_{\mu}\iota_{\xi_2}\dd\iota_{\xi_1}\dd\omega 
- \iota_{\mu}\iota_{\xi_1}\dd\iota_{\xi_2}\dd\omega  - \iota_{\mu}\dd\iota_{\xi_1}\iota_{\xi_2}\dd\omega 
)\dwedge\CP \big) \\[4pt]
& \hspace{1cm} = -\dd x^{\mu} \otimes \Tr\big( \iota_{\mu}\dd\omega\dwedge \iota_{[\xi_2,\xi_1]}\CP - 
\iota_{[\xi_2,\xi_1]}\iota_\mu\dd\omega\dwedge\CP \big)\\[4pt]
& \hspace{2cm} = -\dd x^{\mu} \otimes \Tr\big( - \iota_{[\xi_2,\xi_1]}( \iota_\mu\dd\omega\dwedge\CP) \big) \\[4pt]
& \hspace{3cm} =0 \ ,
\end{align*}
where the final term vanishes since it is a contraction of a four-form, which is identically zero in three dimensions. This completes the proof
of the homotopy relation \eqref{eq:J4deg4}.

\subsection{$L_\infty$-morphism relations}
\label{app:3dmorrels}

We will now prove that the maps $\{\psi^{\rm cov}_n\}$ defined in
Section~\ref{sec:ECPiso} satisfy the $L_\infty$-morphism relations
given by
\eqref{eq:morphismrels} (for the case $d=3$). 
We shall use the facts $V^{\rm cov}=V$, $\psi^{\rm cov}_{1}= {\rm id}_{V}$ and
$\psi^{\rm cov}_{n}=0$ for $n\geq 3$ implicitly with no further mention
below. As with the homotopy relations, we shall proceed by degree $n$,
remembering that $|\psi^{\rm cov}_n|=1-n$. Since $\ell_{n}^{\rm cov}=\ell_{n}$
for $n=1$ and for all $n\geq5$, the relations are
immediate and trivially satisfied in these degrees. As previously, we
set $\Lambda=0$ throughout to simplify the presentation.

\subsubsection*{$\boldsymbol{n=2}$: Internal degree $\boldsymbol{0}$}

To show that the map $\psi^{\rm cov}_1$ preserves the $2$-brackets up to a
homotopy generated by $\psi^{\rm cov}_2$, we may act non-trivially on fields
whose total degrees are $0$, $1$, $2$ and $3$.

\paragraph{$\underline{{\rm Total~degree}~0:}$} 
We act on two pairs of gauge transformations
$(\xi_{1},\rho_{1})$ and $(\xi_{2},\rho_{2})$, and we need to check 
\begin{align*}
-\psi^{\rm cov}_{2}\big(\ell^{\rm cov}_{1}(\xi_{1},\rho_{1})\,,\,
  (\xi_{2},\rho_{2})\big) + \psi^{\rm cov}_{2}\big(\ell_{1}^{\rm cov}&(\xi_{2},\rho_{2})\,,\,(\xi_{1},\rho_{1})\big) + \ell_{2}^{\rm cov}\big((\xi_{1},\rho_{1})\,,\,(\xi_{2},\rho_{2})\big) \\[4pt]
&=
  \ell_{1}\Big(\psi^{\rm cov}_{2}\big((\xi_{1},\rho_{1})\,,\,(\xi_{2},\rho_{2})\big)\Big)
  +\ell_{2}\big((\xi_{1},\rho_{1})\,,\,(\xi_{2},\rho_{2})\big)
  \ .
\end{align*}
Expanding the left-hand side we get
\begin{align*}
(0,-\iota_{\xi_{2}} \dd \rho_{1})+(0,\iota_{\xi_{1}}\dd \rho_{2}) + \big([\xi_{1},\xi_{2}],-[\rho_{1},\rho_{2}]\big)=\big([\xi_{1},\xi_{2}], \LL_{\xi_{1}}\rho_{2}-\LL_{\xi_{2}} \rho_{1} - [\rho_{1},\rho_{2}]\big)
\end{align*}
where we used $\iota_{\xi_{1}} \dd \rho_{2} = \LL_{\xi_{1}} \rho_{2}$
since $\rho_{2}$ is a zero-form, and similarly for $\rho_1$. This is
easily seen to
coincide with the right-hand side, since
$\psi^{\rm cov}_{2}((\xi_{1},\rho_{1}),(\xi_{2},\rho_{2}))=(0,0)$.

\paragraph{$\underline{{\rm Total~degree}~1:}$}
We act on a pair of gauge transformations $(\xi,\rho)$ and one
pair of dynamical fields $(e,\omega)$, and we need to check
\begin{align*}
-\psi^{\rm cov}_{2}\big(\ell_{1}^{\rm cov}(\xi,\rho)\,,\,(e,\om)\big)+\psi^{\rm cov}_{2}\big(\ell_{1}^{\rm cov}(e,\om)\,,\,&(\xi,\rho)\big)+\ell_{2}^{\rm cov}\big((\xi,\rho)\,,\,(e,\om)\big) \\[4pt]
&=\ell_{1}
  \Big(\psi^{\rm cov}_{2}\big((\xi,\rho)\,,\,(e,\om)\big)\Big)+\ell_{2}\big((\xi,\rho)\,,\,(e,\om)\big)
  \ .
\end{align*}
Expanding the left-hand side we get
\begin{align*}
-\psi^{\rm cov}_{2}\big((0,\dd \rho)\,,\, (e, \om)\big)+\psi^{\rm cov}_{2}\big( (\dd
  \om,\dd e)\,,\, (\xi,\rho)\big) +\big(\LL_{\xi} e -\rho \cdot e\,,\,&
  \iota_{\xi} \dd \om - [\rho,\om]\big)\\[4pt] & = \big(\LL_{\xi} e -\rho \cdot e\,,\, \iota_{\xi} \dd \om - [\rho,\om]\big)
\end{align*}
while the right-hand side is given by
\begin{align*}
(0,-\dd \iota_{\xi} \om) + \big(\LL_{\xi} e - \rho\cdot e\,,\, \LL_{\xi}\om -[\rho,\om]\big)=\big(\LL_{\xi} e -\rho \cdot e\,,\, \iota_{\xi} \dd \om - [\rho,\om]\big)
\end{align*}
where we used the Cartan formula \eqref{eq:CartanLie}.

\paragraph{$\underline{{\rm Total~degree}~2:}$}
Checking the $L_\infty$-morphism relation on two pairs of dynamical
fields $(e_{1},\om_{1})$ and $(e_{2},\om_{2})$ is immediate, since all
values of $\psi^{\rm cov}_{2}$ vanish by definition and $\ell_{2}^{\rm
  cov}\big((e_{1},\om_{1})\,,\,(e_{2},\om_{2})\big)
=\ell_{2}\big((e_{1},\om_{1})\,,\,(e_{2},\om_{2})\big).$ 
We may also act on a pair of gauge transformations $(\xi,\rho)$ and a pair of Euler--Lagrange derivatives $(E,{\mit\Omega})$. Then we need to check
\begin{align*}
-\psi^{\rm cov}_{2}\big(\ell_{1}^{\rm cov}(\xi,\rho)\,,\,(E,{\mit\Omega})\big)+\psi^{\rm cov}_{2}\big(\ell_{1}^{\rm cov}(E,{\mit\Omega})\,,\,&(\xi,\rho)\big)+\ell_{2}^{\rm cov}\big((\xi,\rho)\,,\,(E,{\mit\Omega})\big) \\[4pt] 
&=\ell_{1}
  \Big(\psi^{\rm cov}_{2}\big((\xi,\rho)\,,\,(E,{\mit\Omega})\big)\Big)+\ell_{2}\big((\xi,\rho)\,,\,(E,{\mit\Omega})\big)
  \ .
\end{align*}
Expanding the left-hand side gives
\begin{align*}
(0,\iota_{\xi} \dd \mit\Omega) +(\LL_{\xi} E -\rho \cdot E, \dd \iota_{\xi} \mit\Omega -\rho\cdot \mit\Omega) = (\LL_{\xi} E - \rho \cdot E, \LL_{\xi} \mit\Omega - \rho \cdot \mit\Omega)
\end{align*}
where we used the Cartan formula \eqref{eq:CartanLie}, which easily agrees with the
right-hand side.

\paragraph{$\underline{{\rm Total~degree}~3:}$} 
We may act on a pair of gauge transformations $(\xi,\rho)$ and a pair of Noether identities $(\CX,\CP)$. We need to check
\begin{align*}
-\psi^{\rm cov}_{2}\big(\ell_{1}^{\rm cov}(\xi,\rho)\,,\,(\CX,\CP)\big)+\psi^{\rm cov}_{2}\big(\ell_{1}^{\rm cov}(\CX,\CP)\,,\,&(\xi,\rho)\big)+\ell_{2}^{\rm cov}\big((\xi,\rho)\,,\,(\CX,\CP)\big) \\[4pt]
&=\ell_{1}
  \Big(\psi^{\rm cov}_{2}\big((\xi,\rho)\,,\,(\CX,\CP)\big)\Big)+\ell_{2}\big((\xi,\rho)\,,\,(\CX,\CP)\big)
  \ .
\end{align*}
Expanding the left-hand side we get
\begin{align*}
\big(\dd x^{\mu}\otimes \Tr (\iota_{\mu}\dd \rho\dwedge \CP),0\big)+
  (\LL_{\xi}\CX, -\rho \cdot \CP)=\big(\LL_{\xi} \CX +\dd
  x^{\mu}\otimes \Tr (\iota_{\mu}\dd \rho\dwedge \CP)\,,\, -\rho \cdot
  \CP \big) \ ,
\end{align*}
while the right-hand side is
\begin{align*}
(0,-\dd \iota_{\xi} \CP) +\big(\LL_{\xi} \CX + \dd x^{\mu}\otimes \Tr(\iota_{\mu}\dd \rho \dwedge \CP)\,,\,-\rho \cdot \CP + \LL_{\xi} \CP\big)=\big(\LL_{\xi} \CX +\dd x^{\mu}\otimes \Tr (\iota_{\mu}\dd \rho\dwedge \CP)\,,\, -\rho \cdot \CP \big)
\end{align*}
where we used the fact that $\CP$ is a top form, so that $\LL_{\xi} \CP = \dd \iota_{\xi} \CP$.

We may also act on a pair of dynamical fields $(e,\om)$ and a pair of Euler--Lagrange derivatives $(E,{\mit\Omega})$. We need to check
\begin{align*}
-\psi^{\rm cov}_{2}\big(\ell_{1}^{\rm cov}(e,\om)\,,\,(E,{\mit\Omega})\big)+\psi^{\rm cov}_{2}\big(\ell_{1}^{\rm cov}(E,{\mit\Omega})\,,\,&(e,\om)\big)+\ell_{2}^{\rm cov}\big((e,\om)\,,\,(E,{\mit\Omega})\big) \\[4pt] 
&=\ell_{1}
  \Big(\psi^{\rm cov}_{2}\big((e,\om)\,,\,(E,{\mit\Omega})\big)\Big)+\ell_{2}\big((e,\om)\,,\,(E,{\mit\Omega})\big)
  \ .
\end{align*}
Expanding the left-hand side gives
\begin{align*}
\big(\dd x^{\mu}\otimes\Tr(-\iota_{\mu}\dd\om \dwedge {\mit\Omega}+ \iota_{\mu}\dd e \dwedge E+\iota_{\mu} \dd \om \dwedge {\mit\Omega} - \iota_{\mu} e\dwedge \dd E)\,,\, E\wedge e -\om \wedge {\mit\Omega} \big)
\end{align*}
which easily coincides with the right-hand side.

\subsubsection*{$\boldsymbol{n=3}$: Internal degree $\boldsymbol{-1}$}

The non-trivial $L_\infty$-morphism relations in this case comprise
fields whose total degrees are $1$, $2$, $3$ and $4$. We will also
implicitly use the fact that the non-covariant $L_\infty$-algebra in
three dimensions is a differential graded Lie algebra, so the bracket
$\ell_{3}$ vanishes identically.

\paragraph{$\underline{{\rm Total~degree}~1:}$} 
We may act on two pairs of gauge transformations $(\xi_{1},\rho_{1})$
and $(\xi_{2},\rho_{2})$, and a pair of dynamical fields $(e,\om)$. We need to check
\begin{align*}
\psi^{\rm cov}_{2}\Big(\ell_{2}^{\rm
  cov}&\big((\xi_{1},\rho_{1})\,,\,(\xi_{2},\rho_{2})\big)\,,\,(e,\om)\Big)+\psi^{\rm cov}_{2}\Big(\ell_{2}^{\rm
                                    cov}\big((e,\om)\,,\,(\xi_{1},\rho_{1})\big)\,,\,(\xi_{2},\rho_{2})\Big)\\
                                  & \hspace{1cm} +\psi^{\rm cov}_{2}\Big(\ell_{2}^{\rm cov}\big((\xi_{2},\rho_{2})\,,\,(e,\om)\big)\,,\,(\xi_{1},\rho_{1})\Big)+\ell_{3}^{\rm cov}\big((\xi_{1},\rho_{1})\,,\,(\xi_{2},\rho_{2})\,,\,(e,\om)\big) \\[4pt]
&=\ell_{2}\Big(\psi^{\rm cov}_{1}(\xi_{1},\rho_{1})\,,\,\psi^{\rm cov}_{2}\big((\xi_{2},\rho_{2})\,,\,(e,\om)\big)\Big)-\ell_{2}\Big(\psi^{\rm cov}_{1}(\xi_{2},\rho_{2})\,,\,\psi^{\rm cov}_{2}\big((\xi_{1},\rho_{1})\,,\,(e,\om)\big)\Big) \\
&\hspace{0.5cm}-\ell_{2}\Big(\psi^{\rm cov}_{1}(e,\om)\,,\,\psi^{\rm cov}_{2}\big((\xi_{1},\rho_{1})\,,\,(\xi_{2},\rho_{2})\big)\Big)
  \ .
\end{align*}
The left-hand side expands as
\begin{align*}
\big(0,-\iota_{[\xi_{1},\xi_{2}]}\om\big)+\big(0,-\iota_{\xi_{2}}\iota_{\xi_{1}}\dd
                                                   \om
                                                   +[\rho_{1},\iota_{{\xi}_{2}}\om]\big)-\big(0,&-\iota_{\xi_{1}}\iota_{\xi_{2}}
                                                   \dd \om
                                                   +[\rho_{2},\iota_{\xi_{1}}\om]\big)+(0,-\iota_{\xi_{1}}\iota_{\xi_{2}}
                                                   \dd \om\big)
  \\[4pt] 
&=\big(0,\iota_{\xi_{2}}\dd \iota_{\xi_{1}}\om -\iota_{\xi_{1}}\dd \iota_{\xi_{2}}\om+[\rho_{1},\iota_{\xi_{2}}\om]-[\rho_{2},\iota_{\xi_{1}}\om]\big)
\end{align*}
where we used the Cartan identity \eqref{eq:Cartanidiota}.
Expanding the right-hand side yields
\begin{align*}
\big(0,-\LL_{\xi_{1}}\iota_{\xi_{2}} \om +[\rho_{1},\iota_{\xi_{2}}\om]\big)-\big(0,&-\LL_{\xi_{2}}\iota_{\xi_{1}}\om +[\rho_{2},\iota_{\xi_{1}}\om]\big) \\[4pt]
&=\big(0,-\iota_{\xi_{1}}\dd\iota_{\xi_{2}}\om +[\rho_{1},\iota_{\xi_{2}}\om]+\iota_{\xi_{2}}\dd \iota_{\xi_{1}}\om-[\rho_{2},\iota_{\xi_{1}}\om]\big)
\end{align*}
where we used $\LL_{\xi_{1}}\iota_{\xi_{2}}\om=\iota_{\xi_{1}}\dd \iota_{\xi_{2}}\om$, since $\om$
is a one-form.

\paragraph{$\underline{{\rm Total~degree}~2:}$} 
We may act on two pairs of gauge transformations $(\xi_{1},\rho_{1})$
and $(\xi_{2},\rho_{2})$, and a pair of Euler--Lagrange derivatives
$(E,{\mit\Omega})$; in this case all pairings appearing involve
$\psi^{\rm cov}_{2}$ and $\ell_{3}$, which vanish individually by definition.
We may also act on a pair of gauge transformations $(\xi,\rho)$, and
two pairs of dynamical fields $(e_{1},\om_{1})$ and
$(e_{2},\om_{2})$. Then we need to check
\begin{align*}
\psi^{\rm cov}_{2}\Big(\ell_{2}^{\rm
  cov}&\big((\xi,\rho)\,,\,(e_{1},\om_{1})\big)\,,\,(e_{2},\om_{2})\Big)-\psi^{\rm cov}_{2}\Big(\ell_{2}^{\rm
                            cov}\big((e_{2},\om_{2})\,,\,(\xi,\rho)\big)\,,\,(e_{1},\om_{1})\Big)\\
                          & \hspace{1cm} +\psi^{\rm cov}_{2}\Big(\ell_{2}^{\rm cov}\big((e_{1},\om_{1})\,,\,(e_{2},\om_{2})\big)\,,\,(\xi,\rho)\Big)+\ell_{3}^{\rm cov}\big((\xi_{1},\rho_{1})\,,\,(\xi_{2},\rho_{2})\,,\,(e,\om)\big) \\[4pt]
&=\ell_{2}\Big(\psi^{\rm cov}_{1}(\xi,\rho)\,,\,\psi^{\rm cov}_{2}\big((e_{1},\om_{1})\,,\,(e_{2},\om_{2})\big)\Big)+\ell_{2}\Big(\psi^{\rm cov}_{1}(e_{1},\om_{1})\,,\,\psi^{\rm cov}_{2}\big((\xi,\rho)\,,\,(e_{2},\om_{2})\big)\Big) \\
&\hspace{0.5cm}+\ell_{2}\Big(\psi^{\rm cov}_{1}(e_{2},\om_{2})\,,\,\psi^{\rm cov}_{2}\big((\xi,\rho)\,,\,(e_{1},\om_{1})\big)\Big)
  \ .
\end{align*}
Expanding the left-hand side we get
\begin{align*}
-\big(\iota_{\xi}\om_{1} \wedge e_{2}+\iota_{\xi}\om_{2}\wedge
  e_{1}\,,\,[\iota_{\xi}\om_{1},\om_{2}]+[\iota_{\xi}\om_{2},\om_{1}]\big)
  \ ,
\end{align*}
which is equal to the expansion of the right-hand side.

\paragraph{$\underline{{\rm Total~degree}~3:}$} The
$L_\infty$-morphism relation involving three pairs of dynamical fields
$(e_{1},\om_{1})$, $(e_{2},\om_{2})$ and $(e_{3},\om_{3})$ is immediate, since the brackets involved in the dynamics are identical in
both versions of the theory and the map $\psi^{\rm cov}_{2}$ is trivial on the
vectors involved. We may also act on two
pairs of gauge transformations $(\xi_{1},\rho_{1})$ and $(\xi_{2},\rho_{2})$, and a pair of Noether identities $(\CX,\CP)$. We need to check
\begin{align*}
\psi^{\rm cov}_{2}\Big(\ell_{2}^{\rm cov}&\big((\xi_{1},\rho_{1})\,,\,(\xi_{2},\rho_{2})\big)\,,\,(\CX,\CP)\Big)+\psi^{\rm cov}_{2}\Big(\ell_{2}^{\rm cov}\big((\CX,\CP)\,,\,(\xi_{1},\rho_{1})\big)\,,\,(\xi_{2},\rho_{2})\Big)\\&+\psi^{\rm cov}_{2}\Big(\ell_{2}^{\rm cov}\big((\xi_{2},\rho_{2})\,,\,(\CX,\CP)\big)\,,\,(\xi_{1},\rho_{1})\Big)+\ell_{3}^{\rm cov}\big((\xi_{1},\rho_{1})\,,\,(\xi_{2},\rho_{2})\,,\,(\CX,\CP)\big) \\[4pt]
&\hspace{0.5cm} =\ell_{2}\Big(\psi^{\rm cov}_{1}(\xi_{1},\rho_{1})\,,\,\psi^{\rm cov}_{2}\big((\xi_{2},\rho_{2})\,,\,(\CX,\CP)\big)\Big)-\ell_{2}\Big(\psi^{\rm cov}_{1}(\xi_{2},\rho_{2})\,,\,\psi^{\rm cov}_{2}\big((\xi_{1},\rho_{1})\,,\,(\CX,\CP)\big)\Big) \\
&\hspace{2cm}-\ell_{2}\Big(\psi^{\rm cov}_{1}(\CX,\CP)\,,\,\psi^{\rm cov}_{2}\big((\xi_{1},\rho_{1})\,,\,(\xi_{2},\rho_{2})\big)\Big)
  \ .
\end{align*}
Expanding the left-hand side, we get
\begin{align*}
(0,-\iota_{[\xi_{1},\xi_{2}]}\CP)+\big(0,&\,\iota_{\xi_{2}}(\rho_{1}\cdot\CP)\big)-\big(0,\iota_{\xi_1}(\rho_{2}\cdot
                                   \CP)\big)-(0,\dd
                                   \iota_{\xi_{1}}\iota_{\xi_{2}}\CP) \\[4pt]
&=\big(0,-2\,\dd \iota_{\xi_{1}}\iota_{\xi_{2}}\CP-\iota_{\xi_{1}}\dd \iota_{\xi_{2}}\CP + \iota_{\xi_{2}}\dd \iota_{\xi_{1}} \CP+\rho_{1}\cdot \iota_{\xi_{2}}\CP -\rho_{2}\cdot\iota_{\xi_{1}} \CP \big)
\end{align*}
where we used the Cartan identity \eqref{eq:Cartanidiota}.
The right-hand side expands as
\begin{align*}
\big(0,-\rho_{1}\cdot(-\iota_{\xi_{2}}\CP)+\LL_{\xi_{1}}&(-\iota_{\xi_{2}}\CP)\big)-\big(0,-\rho_{2}(-\iota_{\xi_{1}}\CP)-\LL_{\xi_{2}}(-\iota_{\xi_{1}}\CP)\big)
  \\[4pt]
&=\big(0,\rho_{1}\cdot \iota_{\xi_{2}}\CP -\rho_{2}\cdot\iota_{\xi_{1}}\CP-\iota_{\xi_{1}}\dd \iota_{\xi_{2}}\CP + \iota_{\xi_{2}}\dd \iota_{\xi_{1}} \CP-2\,\dd \iota_{\xi_{1}}\iota_{\xi_{2}}\CP\big)
\end{align*}
where we used Cartan's magic formula.

We may further act on a pair of gauge transformations $(\xi,\rho)$, a
pair of dynamical fields $(e,\om)$ and a pair of Euler--Lagrange derivatives 
$(E,{\mit\Omega})$. We need to check
\begin{align*}
\psi^{\rm cov}_{2}\Big(\ell_{2}^{\rm cov}&\big((\xi,\rho)\,,\,(e,\om)\big)\,,\,(E,{\mit\Omega})\Big)+\psi^{\rm cov}_{2}\Big(\ell_{2}^{\rm cov}\big((E,{\mit\Omega})\,,\,(\xi,\rho)\big)\,,\,(e,\om)\Big)\\&+\psi^{\rm cov}_{2}\Big(\ell_{2}^{\rm cov}\big((e,\om)\,,\,(E,{\mit\Omega})\big)\,,\,(\xi,\rho)\Big)+\ell_{3}^{\rm cov}\big((\xi,\rho)\,,\,(e,\om)\,,\,(E,{\mit\Omega})\big) \\[4pt]
&\hspace{1cm} =\ell_{2}\Big(\psi^{\rm cov}_{1}(\xi,\rho)\,,\,\psi^{\rm cov}_{2}\big((e,\om)\,,\,(E,{\mit\Omega})\big)\Big)+\ell_{2}\Big(\psi^{\rm cov}_{1}(e,\om)\,,\,\psi^{\rm cov}_{2}\big((\xi,\rho)\,,\,(E,{\mit\Omega})\big)\Big) \\
&\hspace{2cm}+\ell_{2}\Big(\psi^{\rm cov}_{1}(E,{\mit\Omega})\,,\,\psi^{\rm cov}_{2}\big((\xi,\rho)\,,\,(e,\om)\big)\Big)
  \ .
\end{align*}
Expanding the left-hand side, we have
\begin{align*}
\big(0,\iota_{\xi}(E\wedge e- \om \wedge
  {\mit\Omega})\big)+\big(-\iota_{\xi}\om \cdot E,-\om \wedge
  \iota_{\xi}{\mit\Omega} -\iota_{\xi}(E\wedge e)\big)=\big(-\iota_{\xi}\om
  \cdot E, -\iota_{\xi}\om \cdot {\mit\Omega}\big)
\end{align*}
where we used $\iota_{\xi}(\om\wedge {\mit\Omega})=\iota_{\xi}\om
\cdot {\mit\Omega} - \om \wedge \iota_{\xi}{\mit\Omega}$. This easily
agrees with the right-hand side.

\paragraph{$\underline{{\rm Total~degree}~4:}$} 
The $L_\infty$-morphism relation involving a pair of gauge
transformations $(\xi,\rho)$, and two pairs of Euler--Lagrange
derivatives $(E_{1},{\mit\Omega}_{1})$ and $(E_{2},{\mit\Omega}_{2})$, is trivial since all terms vanish individually, by definition of $\psi^{\rm cov}_{2}$ and $\ell_{3}^{\rm cov}$.
We may also act on a pair
of gauge transformations $(\xi,\rho)$, a pair of dynamical fields $(e,\om)$ and a pair of Noether identities $(\CX,\CP)$. We need to check
\begin{align*}
\psi^{\rm cov}_{2}\Big(\ell_{2}^{\rm cov}&\big((\xi,\rho)\,,\,(e,\om)\big)\,,\,(\CX,\CP)\Big)-\psi^{\rm cov}_{2}\Big(\ell_{2}^{\rm cov}\big((\CX,\CP)\,,\,(\xi,\rho)\big)\,,\,(e,\om)\Big)\\&+\psi^{\rm cov}_{2}\Big(\ell_{2}^{\rm cov}\big((e,\om)\,,\,(\CX,\CP)\big)\,,\,(\xi,\rho)\Big)+\ell_{3}^{\rm cov}\big((\xi,\rho)\,,\,(e,\om)\,,\,(\CX,\CP)\big) \\[4pt]
&\hspace{1cm}=\ell_{2}\Big(\psi^{\rm cov}_{1}(\xi,\rho)\,,\,\psi^{\rm cov}_{2}\big((e,\om)\,,\,(\CX,\CP)\big)\Big)+\ell_{2}\Big(\psi^{\rm cov}_{1}(e,\om)\,,\,\psi^{\rm cov}_{2}\big((\xi,\rho)\,,\,(\CX,\CP)\big)\Big) \\
&\hspace{2cm}+\ell_{2}\Big(\psi^{\rm cov}_{1}(\CX,\CP)\,,\,\psi^{\rm cov}_{2}\big((\xi,\rho)\,,\,(e,\om)\big)\Big)
  \ .
\end{align*}
Expanding the left-hand side, we have
\begin{align*}
\Big(-\dd x^{\mu}\otimes \Tr\big((\iota_{\mu}\iota_{\xi} \dd \om
  -[\rho ,\iota_{\mu}\om])\dwedge\CP\big),&\,0\Big)+\Big(-\dd
                                            x^{\mu}\otimes
                                            \Tr\big(\iota_{\mu}\om
                                            \dwedge(-\rho\cdot\CP)\big),0\Big)\\
                                          &+\Big(\dd x^{\mu}\otimes
                                            \Tr\big(\iota_{\mu}\iota_{\xi}\dd
                                            \om \dwedge
                                            \CP\big),0\Big) \\[4pt]
& \hspace{1cm} =\Big(\dd x^{\mu}\otimes \Tr \big(\rho \cdot
  (\iota_{\mu}\om \dwedge \CP)\big),0\Big) \\[4pt] & \hspace{2cm} = (0,0)
                                                     \ ,
\end{align*}
where we used the Leibniz rule to pull out $\rho$ as an action on a
top exterior vector-valued form, which vanishes by invariance of top
exterior vectors under Lorentz transformations. On the other hand, the right-hand side expands as
\begin{align*}
\Big(-\LL_{\xi}&\,\big(\dd x^{\mu} \otimes \Tr(\iota_{\mu}\om \dwedge \CP) \big)\,,\,0 \Big) - \Big (\dd x^{\mu}\otimes \Tr \big(\iota_{\mu}\dd \om \dwedge \iota_{\xi} \CP - \iota_{\mu}\om \dwedge \dd \iota_{\xi} \CP\big)\,,\, -\om\wedge \iota_{\xi} \CP\Big)\\
&\hspace{4cm}+\Big(\dd x^{\mu}\otimes \Tr \big(\iota_{\mu}\dd \iota_{\xi}\om\dwedge \CP\big)\,,\,-\iota_{\xi}\om \cdot \CP \Big) \\[4pt]
&=\Big(\dd x^{\mu} \otimes \Tr \big((-\iota_{\mu}\iota_{\xi} \dd \om -\iota_{\mu} \dd \iota_{\xi}\om)\dwedge \CP -\iota_{\mu} \om \dwedge \dd \iota_{\xi} \CP -\iota_{\mu}\dd \om \dwedge \iota_{\xi} \CP + \iota_{\mu}\om \dwedge \dd \iota_{\xi} \CP\big) \\
&\hspace{1cm} +\dd  x^\mu \otimes \Tr \big(\iota_{\mu} \dd \iota_{\xi} \om \dwedge \CP\big)\,,\,\iota_{\xi}(\om \wedge \CP) \Big) \\[4pt]
&= \Big( \dd x^{\mu}\otimes \Tr \big(-\iota_{\mu}\iota_{\xi} \dd \om
  \dwedge \CP-\iota_{\mu} \dd \om \dwedge \iota_{\xi}\CP\big)\,,\, 0\Big)
  \\[4pt]
&=\Big( \dd x^{\mu}\otimes \Tr \big(-\iota_{\mu}\iota_{\xi} \dd \om
  \dwedge \CP+\iota_{\mu}\iota_{\xi} \dd \om \dwedge \CP\big)\,,\, 0\Big) \\[4pt]
&= (0,0) \ ,
\end{align*}
where in the first equality we expanded the Lie derivative using
Cartan's magic formula, along with the derivation property of the
contraction, and in the second equality we used $\om\wedge \CP=0$ as it is a four-form in three dimensions. Then using the Cartan identity
\begin{align*}
\iota_{\xi_1}\circ\iota_{\xi_2} = -\iota_{\xi_2}\circ\iota_{\xi_1}
\end{align*}
we wrote $
\iota_{\mu}\dd \om \dwedge \iota_{\xi} \CP =
-\iota_{\xi}(\iota_{\mu}\dd \om \dwedge \CP)+\iota_{\xi}\iota_{\mu}\dd
\om \dwedge \CP=-\iota_{\mu}\iota_{\xi} \dd \om \dwedge \CP
$.

Lastly, we can also act on two pairs of dynamical fields
$(e_{1},\om_{1})$ and $(e_{2},\om_{2})$, and one pair of Euler--Lagrange derivatives $(E,{\mit\Omega})$. We need to check
\begin{align*}
\psi^{\rm cov}_{2}&\Big(\ell_{2}^{\rm
  cov}\big((e_{1},\,\om_{1})\,,\,(e_{2},\om_{2})\big)\,,\,(E,{\mit\Omega})\Big)+\psi^{\rm cov}_{2}\Big(\ell_{2}^{\rm
                    cov}\big((E,{\mit\Omega})\,,\,(e_{1},\om_{1})\big)\,,\,(e_{2},\om_{2})\Big)\\
                  &-\psi^{\rm cov}_{2}\Big(\ell_{2}^{\rm
                    cov}\big((e_{2},\om_{2})\,,\,(E,{\mit\Omega})\big)\,,\,(e_{1},\om_{1})\Big)+\ell_{3}^{\rm
                    cov}\big((e_{1},\om_{1})\,,\,(e_{2},\om_{2})\,,\,(E,{\mit\Omega})\big)
  \\[4pt] 
& \hspace{1cm} =-\ell_{2}\Big(\psi^{\rm cov}_{1}(e_{1},\om_{1})\,,\,\psi^{\rm cov}_{2}\big((e_{2},\om_{2})\,,\,(E,{\mit\Omega})\big)\Big)-\ell_{2}\Big(\psi^{\rm cov}_{1}(e_{2},\om_{2})\,,\,\psi^{\rm cov}_{2}\big((e_{1},\om_{1})\,,\,(E,{\mit\Omega})\big)\Big) \\
&\hspace{2cm}+\ell_{2}\Big(\psi^{\rm cov}_{1}(E,{\mit\Omega})\,,\,\psi^{\rm cov}_{2}\big((e_{1},\om_{1})\,,\,(e_{2},\om_{2})\big)\Big)
  \ .
\end{align*}
Expanding the left-hand side, we have
\begin{align*}
\Big(\dd x^{\mu}&\otimes \Tr \big(\iota_{\mu} \om_{2} \dwedge (E\wedge e_{1} -\om_{1} \wedge {\mit\Omega})\big),0\Big)+\Big(\dd x^{\mu}\otimes \Tr \big(\iota_{\mu} \om_{1} \dwedge (E\wedge e_{2} -\om_{2} \wedge {\mit\Omega})\big), 0 \Big) \\
&- \Big(\dd x^{\mu}\otimes \Tr\big(\iota_{\mu}\om_{1}\dwedge(E\wedge e_{2})+\iota_{\mu} \om_{2} \dwedge (E\wedge e_{1})-\iota_{\mu} \om_{1}\dwedge(\om_{2}\wedge {{\mit\Omega}})
- \iota_{\mu} \om_{2}\dwedge(\om_{1}\wedge {{\mit\Omega}} )\big),0 \Big)\\[4pt]
& \hspace{2cm} = (0,0) \ ,
\end{align*}
while the right-hand side vanishes since all terms vanish individually by definition of $\psi^{\rm cov}_{2}$.

\subsubsection*{$\boldsymbol{n=4}$: Internal degree $\boldsymbol{-2}$}

The $L_\infty$-morphism relations in this case act non-trivially on
fields whose total degrees are $2$, $3$, $4$ and~$5$. We will also
implicitly use the fact that the non-covariant brackets $\ell_{3}$ and
$\ell_{4}$ vanish identically. One may easily
check that all terms in the identity vanish individually for all
possible combinations of fields of total degree~$3$, so the only
non-trivial checks required are in the remaining three total degrees.

\paragraph{$\underline{{\rm Total~degree}~2:}$} 
When acting on three pairs of gauge transformations
$(\xi_{1},\rho_{1})$, $(\xi_{2},\rho_{2})$ and $(\xi_{3},\rho_{3})$,
and a pair of Euler--Lagrange derivatives $(E,{\mit\Omega})$, all
terms vanish individually by definition of $\psi^{\rm cov}_{2}$, $\ell_{3}^{\rm
  cov}$ and $\ell_{4}^{\rm cov}$. We may also act on two pairs of gauge transformations $(\xi_{1},\rho_{1})$
and $(\xi_{2},\rho_{2})$, and two pairs of dynamical fields $(e_{1},\om_{1})$
and $(e_{2},\om_{2})$. We need to check
\begin{align*}
&-\psi^{\rm cov}_{2}\Big(\ell^{\rm cov}_{3}\big((\xi_{1},\rho_{1}),(\xi_{2},\rho_{2}),(e_{1},\om_{1})\big),(e_{2},\om_{2})\Big)-\psi^{\rm cov}_{2}\Big(\ell^{\rm cov}_{3}\big((\xi_{1},\rho_{1}),(\xi_{2},\rho_{2}),(e_{2},\om_{2})\big),(e_{1},\om_{1})\Big)\\
& \hspace{0.5cm} -\psi^{\rm cov}_{2}\Big(\ell^{\rm cov}_{3}\big((\xi_{1},\rho_{1})\,,\,(e_{1},\om_{1})\,,\,(e_{2},\om_{2})\big)\,,\,(\xi_{2},\rho_{2})\Big) \\
& \hspace{1cm} +\psi^{\rm cov}_{2}\Big(\ell^{\rm cov}_{3}\big((\xi_{2},\rho_{2})\,,\,(e_{1},\om_{1})\,,\,(e_{2},\om_{2})\big)\,,\,(\xi_{1},\rho_{1})\Big)+\ell_{4}^{\rm cov}\big((\xi_{1},\rho_{1})\,,\,(\xi_{2},\rho_{2})\,,\,(e_{1},\om_{1})\,,\,(e_{2},\om_{2})\big) \\[4pt]
&\hspace{2cm}
  =-\ell_{2}\Big(\psi^{\rm cov}_{2}\big((\xi_{1},\rho_{1})\,,\,(\xi_{2},\rho_{2})\big)\,,\,\psi^{\rm cov}_{2}\big((e_{1},\om_{1})\,,\,(e_{2},\om_{2})\big)\Big)
  \\ & \hspace{3cm} -\ell_{2}\Big(\psi^{\rm cov}_{2}\big((\xi_{1},\rho_{1})\,,\,(e_{1},\om_{1})\big)\,,\,\psi^{\rm cov}_{2}\big((\xi_{2},\rho_{2})\,,\,(e_{2},\om_{2})\big)\Big)\\
&\hspace{4cm}
  +\ell_{2}\Big(\psi^{\rm cov}_{2}\big((\xi_{2},\rho_{2})\,,\,(e_{1},\om_{1})\big)\,,\,\psi^{\rm cov}_{2}\big((\xi_{1},\rho_{1})\,,\,(e_{2},\om_{2})\big)\Big)
  \ .
\end{align*}
Expanding the left-hand side, we get
\begin{align*}
&-\big(0,-[\iota_{\xi_{1}}\om_{1},\iota_{\xi_{2}}
  \om_{2}]-[\iota_{\xi_{1}}\om_{2},\iota_{\xi_{2}}
  \om_{1}]\big)+\big(0,-[\iota_{\xi_{2}}\om_{1},\iota_{\xi_{1}}
  \om_{2}]-[\iota_{\xi_{2}}\om_{2},\iota_{\xi_{1}}
  \om_{1}]\big)+\big(0,\iota_{\xi_{1}}\iota_{\xi_{2}}[\om_{1},\om_{2}]\big)
  \\[4pt]
& \hspace{3cm} =\big(0,-[\iota_{\xi_{2}}\om_{1},\iota_{\xi_{1}}
  \om_{2}]-[\iota_{\xi_{2}}\om_{2},\iota_{\xi_{1}} \om_{1}]\big) \ ,
\end{align*}
after expanding $\iota_{\xi_{1}}\iota_{\xi_{2}}[\om_{1},\om_{2}]$
using the Leibniz rule and
cancelling terms by antisymmetry. This is easily seen to be the same
as the simple expansion of the right-hand side.

\paragraph{$\underline{{\rm Total~degree}~4:}$} 
The only non-trivial check required in this case, whereby not every term
in the identity vanishes individually, is when acting on two pairs of
gauge transformations $(\xi_{1},\rho_{1})$ and $(\xi_{2},\rho_{2})$, a
pair of dynamical fields $(e,\om)$ and a pair of Noether identities $(\CX,\CP)$. Excluding terms that vanish by definition of $\psi^{\rm cov}_{2}$, it remains to check
\begin{align*}
\ell_{4}^{\rm
  cov}\big((\xi_{1},\rho_{1})\,,\,(\xi_{2},\rho_{2})\,,\,(e,\om)\,,\,(\CX,\CP)\big)&=
                                                                                     -\ell_{2}\Big(\psi^{\rm cov}_{2}\big((\xi_{1},\rho_{1})\,,\,(e,\om)\big)\,,\,\psi^{\rm cov}_{2}\big((\xi_{2},\rho_{2})\,,\,(\CX,\CP)\big)\Big)\\
& \hspace{1cm}
  +\ell_{2}\Big(\psi^{\rm cov}_{2}\big((\xi_{2},\rho_{2})\,,\,(e,\om)\big)\,,\,\psi^{\rm cov}_{2}\big((\xi_{1},\rho_{1})\,,\,(\CX,\CP)\big)\Big)
  \ .
\end{align*}
The left-hand side expands as
\begin{align*}
(0,\om \wedge \iota_{\xi_{1}}\iota_{\xi_{2}}\CP)&=\big(0,-\iota_{\xi_{1}}(\om\wedge \iota_{\xi_{2}}\CP)+\iota_{\xi_{1}}\om \cdot \iota_{\xi_{2}}\CP\big)\\[4pt]
&=\big(0,\iota_{\xi_{1}}\iota_{\xi_{2}}(\om\wedge \CP)-\iota_{\xi_{1}}(\iota_{\xi_{2}}\om\cdot \CP) +\iota_{\xi_{1}}\om \cdot \iota_{\xi_{2}}\CP\big)\\[4pt]
&=\big(0,-\iota_{\xi_{2}}\om\cdot\iota_{\xi_{1}} \CP
  +\iota_{\xi_{1}}\om \cdot \iota_{\xi_{2}}\CP\big) \ ,
\end{align*}
which is precisely the expansion of the right-hand side.

\paragraph{$\underline{{\rm Total~degree}~5:}$} 
Again there is only one non-trivial check required, now when acting on
a pair gauge transformations $(\xi,\rho)$, two pairs of dynamical
fields $(e_{1},\om_{1})$ and $(e_{2},\om_{2})$, and a pair of Noether identities $(\CX,\CP)$. Excluding again terms that vanish by definition of $\psi^{\rm cov}_{2}$, it remains to check
\begin{align*}
&-\psi^{\rm cov}_{2}\Big(\ell_{3}\big((\xi,\rho)\,,\,(e_{1},\om_{1})\,,\,(e_{2},\om_{2})\big)\,,\,(\CX,\CP)\Big)+\ell_{4}^{\rm
           cov}\big((\xi,\rho)\,,\,(e_{1},\om_{1})\,,\,(e_{2},\om_{2})\,,\,(\CX,\CP)\big)
  \\[4pt]
& \hspace{1cm}
  =\ell_{2}\Big(\psi^{\rm cov}_{2}\big((\xi,\rho)\,,\,(e_{1},\om_{1})\big)\,,\,\psi^{\rm cov}_{2}\big((e_{2},\om_{2})\,,\,(\CX,\CP)\big)\Big) \\
& \hspace{2cm} +\ell_{2}\Big(\psi^{\rm cov}_{2}\big((\xi,\rho)\,,\,(e_{2},\om_{2})\big)\,,\,\psi^{\rm cov}_{2}\big((e_{1},\om_{1})\,,\,(\CX,\CP)\big)\Big)
  \ .
\end{align*}
Expanding the left-hand side, we get
\begin{align*}
\psi^{\rm cov}_{2}\big((\iota_{\xi}&\om_{1}\wedge e_{2} +\iota_{\xi} \om_{2}\wedge
e_{1},[\iota_{\xi}\om_{1},\om_{2}]+[\iota_{\xi}\om_{2},\om_{1}])\,,\,(\CX,\CP)\big)
        \\ & \hspace{2cm} +\Big(\dd
  x^{\mu}\otimes\Tr\big((\iota_{\mu}[\iota_{\xi}\om_{1},\om_{2}]+\iota_{\mu}[\iota_{\xi}\om_{2},\om_{1}])\dwedge
  \CP\big)\,,\,0\Big) \\[4pt]
& \hspace{3cm} =-\Big(\dd
  x^{\mu}\otimes\Tr\big((\iota_{\mu}[\iota_{\xi}\om_{1},\om_{2}]+\iota_{\mu}[\iota_{\xi}\om_{2},\om_{1}])\dwedge
  \CP\big)\,,\,0\Big) \\
& \hspace{4cm} +\Big(\dd x^{\mu}\otimes\Tr\big((\iota_{\mu}[\iota_{\xi}\om_{1},\om_{2}]+\iota_{\mu}[\iota_{\xi}\om_{2},\om_{1}])\dwedge \CP\big)\,,\,0\Big)\\[4pt]
& \hspace{5cm} = (0,0) \ ,
\end{align*}
while the right-hand side expands as
\begin{align*}
\ell_{2}\Big(\big(0,-\iota_{\xi}\om_{1}\big)\,,\,\big(-\dd
  x^{\mu}\otimes \Tr(\iota_{\mu}\om_{2}\dwedge \CP),0\big)\Big)
  +_{(1\leftrightarrow 2)}=(0,0) \ .
\end{align*}

This completes the proof that the maps $\{\psi^{\rm cov}_n\}$ from
Section~\ref{sec:ECPiso} indeed do define an $L_{\infty}$-morphism (in the case $d=3$).

\section{Calculations in four dimensions}
\label{app:4dBVBRST}

In this appendix we illustrate the explicit
dualization of the BV--BRST formalism, focusing on the case $d=4$. That is, we will show that the brackets defined in
Section~\ref{sec:Linfty4d} are dual to the non-covariant BV differential of \cite{ECBV}. We have already done this in
Section~\ref{sec:ECPBRST} for the kinematical sector of the
Einstein--Cartan--Palatini theory in any dimension $d$, so we only
need to check that the BV transformations of the antifields, given in
\eqref{eq:QBV}, dualize to the remaining brackets of the dynamical
sector and those on the space of Noether identities. This may be seen as an alternative proof of the $d=4$ homotopy relations by appealing to the duality with the BV--BRST
formalism from Section~\ref{sec:BV-BRST}, where they are automatically
guaranteed to hold by nilpotency of the BV differential $Q_{\textrm{\tiny BV}}^2=0$.  Again we shall
set $\Lambda=0$ for brevity in these calculations.

\subsubsection*{Dynamical brackets}

We start from the first transformation of \eqref{eq:QBV} specialised
to the case $d=4$:
\begin{align} \label{eq:QBVedag4d}
Q_{\textrm{\tiny BV}}{e^{\dagger}}^{\,a_{1}a_{2} a_{3}}_{\mu_{1}\mu_{2}
  \mu_{3}}=-e^{[a_{1}}_{[\mu_{1}}\,R^{a_{2}a_{3}]}_{\mu_{2}\mu_{3}]}+
  4\,\big({e^{\dagger}}^{\,a_{1}a_{2}a_{3}}_{[\mu_{1}\mu_{2}
  \mu_{3}}\,\partial_{\sigma]}\xi^{\sigma} -
  {e^{\dagger}}^{[a_{1}a_{2}a_{3}}_{\mu_{1}\mu_{2}\mu_{3}}\,
  \rho^{d]}{}_{d}\big) +\partial_{\sigma}\big(\xi^{\sigma}\,
  {e^{\dagger}}^{\,a_{1}a_{2} a_{3}}_{\mu_{1}\mu_{2} \mu_{3}}\big) \ .
\end{align}
Dualizing we retrieve the dynamical brackets of our
$L_\infty$-algebra for the coframe field $e$ in four dimensions, as we now demonstrate. 

For $(e,\om) \in  {\scrF_{\textrm{\tiny BV}}}\,_{0}$, using the
natural duality pairing $\langle-|-\rangle$ between
$\text{\Large$\odot$}^\bullet\scrF_{\textrm{\tiny BV}}$ and $\text{\Large$\odot$}_\FR^\bullet\scrF_{\textrm{\tiny BV}}^\star$ we obtain
\begin{align*}
\langle Q_{\textrm{\tiny BV}}{e^{\prime\,\dagger}}^{\,a_{1}a_{2}a_{3}}_{\mu_{1}\mu_{2}\mu_{3}}|e\odot \om\rangle &=\langle- e_{[\mu_{1}}^{\prime\,[a_{1}}\odot \partial_{\mu_{2}}\om^{\prime\,a_{2}a_{3}]}_{\mu_{3}]}|e\odot \om\rangle  \nn\\[4pt] &=-{e}^{[a_{1}}_{[\mu_{1}}\, \partial_{\mu_{2}}{\om}^{a_{2}a_{3}]}_{\mu_{3}]}\nn \\[4pt]
&=\langle{e^{\prime\,\dagger}}^{\,a_{1}a_{2}a_{3}}_{\mu_{1}\mu_{2}\mu_{3}}|-e\dwedge
  \dd \om\rangle \\[4pt]
&=:(-1)^{|Q_{\textrm{\tiny BV}}|\,
  |e^{\prime\,\dagger}|}\,\langle{e^{\prime\,\dagger}}^{\,a_{1}a_{2}a_{3}}_{\mu_{1}\mu_{2}\mu_{3}}|{\DD_{\textrm{\tiny
  BV}}}\,_2(e\odot \om)\rangle \ ,
\end{align*}
where 
$$
\DD_{\textrm{\tiny BV}} := Q^\star_{\textrm{\tiny
    BV}}:\text{\Large$\odot$}^\bullet \scrF_{\textrm{\tiny BV}}
\longrightarrow \text{\Large$\odot$}^\bullet \scrF_{\textrm{\tiny BV}}
$$
is determined by the decomposition
$$
{\rm pr}_{\scrF_{\textrm{\tiny BV}}} \circ \DD_{\textrm{\tiny BV}} = \sum_{n=1}^\infty\, {\DD_{\textrm{\tiny
      BV}}}\,_n
$$
with component maps ${\DD_{\textrm{\tiny BV}}}\,_{n}: \text{\Large$\odot$}^{n} \scrF_{\textrm{\tiny BV}}
\rightarrow \scrF_{\textrm{\tiny BV}}$. 
Thus ${\DD_{\textrm{\tiny BV}}}\,_2(e\odot \om)= e\dwedge \dd\om$, and so
\begin{align*}
\ell_{2}({}^{s^{-1}}e\wedge {}^{s^{-1}}\om)&=s^{-1}\circ {\DD_{\textrm{\tiny
                                   BV}}}\,_2\circ (s\otimes
                                   s)({}^{s^{-1}}e\wedge {}^{s^{-1}}\om) \\[4pt]
  &=(-1)^{|{}^{s^{-1}}e|}\,s^{-1}\circ {\DD_{\textrm{\tiny BV}}}\,_2(e\odot \om)\nn \\[4pt]
&=- {}^{s^{-1}}e\dwedge \dd {}^{s^{-1}}\om
\end{align*}
as required.

Similarly, for $e,\om_{1},\om_{2} \in  {\scrF_{\textrm{\tiny BV}}}\,_{0}$
we get
\begin{align*}
\langle Q_{\textrm{\tiny BV}}{e^{\prime\,\dagger}}^{\,a_{1}a_{2}a_{3}}_{\mu_{1}\mu_{2}\mu_{3}}|e\odot \om_{1}\odot \om_{2}\rangle &=\langle- e^{\prime\,[a_{1}}_{[\mu_{1}}\odot \om^{\prime\,a_{2}}{}_{|c|\mu_{2}}\odot \om^{\prime\,|c|a_{3}]}_{\mu_{3}]}|e\odot \om_{1}\odot \om_{2}\rangle \nn \\[4pt]
&=-\big({e}^{[a_{1}}_{[\mu_{1}}\,
  {\om_{1}}^{a_{2}}{}_{|c|\mu_{2}}\, {\om_{2}}^{|c|a_{3}]}_{\mu_{3}]}+
  {e}^{[a_{1}}_{[\mu_{1}}\, {\om_{2}}^{a_{2}}{}_{|c|\mu_{2}}\,
  {\om_{1}}^{|c|a_{3}]}_{\mu_{3}]}\big)\nn \\[4pt] 
&=\langle {e^{\prime\,\dagger}}^{\,a_{1}a_{2}a_{3}}_{\mu_{1}\mu_{2}\mu_{3}}|- e\dwedge [\om_{1},\om_{2}]\rangle \nn \\[4pt]
&=:\langle
  {e^{\prime\,\dagger}}^{\,a_{1}a_{2}a_{3}}_{\mu_{1}\mu_{2}\mu_{3}}|-(-1)^{|Q_{\textrm{\tiny
  BV}}|\, |e^{\prime\,\dagger}|} \, {\DD_{\textrm{\tiny BV}}}\,_3(e\odot \om_{1}
  \odot \om_{2})\rangle \ .
\end{align*}
Thus ${\DD_{\textrm{\tiny BV}}}\,_3(e\odot \om_{1}\odot \om_{2})= e\dwedge [\om_{1},\om_{2}]$, and so 
\begin{align*}
\ell_{3}({}^{s^{-1}}e\wedge {}^{s^{-1}}\om_{1}\wedge
                        {}^{s^{-1}}\om_{2})&=s^{-1}\circ {\DD_{\textrm{\tiny
                                        BV}}}\,_3\circ (s\otimes s\otimes s)({}^{s^{-1}}e\wedge {}^{s^{-1}}\om_{1}\wedge {}^{s^{-1}}\om_{2})\nn \\[4pt] &= (-1)^{2\,|{}^{s^{-1}}e|+|{}^{s^{-1}}\om_{1}|}\,( {}^{s^{-1}}e\dwedge [ {}^{s^{-1}}\om_{1}, {}^{s^{-1}}\om_{2}])\nn \\[4pt]
&=- {}^{s^{-1}}e\dwedge[{}^{s^{-1}}\om_{1}, {}^{s^{-1}}\om_{2}]
\end{align*}
as required.

Next, we note that the Lie derivative appears explicitly in \eqref{eq:QBVedag4d}, as the
fourth and second terms expand into 
\begin{align*}
&\partial_{\sigma} \xi^{\sigma}\, {e^{\dagger}}^{\,a_{1}a_{2}a_{3}}_{\mu_{1}\mu_{2}\mu_{3}} + \xi^{\sigma}\,\partial_{\sigma}{e^{\dagger}}^{a_{1}a_{2}a_{3}}_{\mu_{1}\mu_{2}\mu_{3}}+{e^{\dagger}}^{\,a_{1}a_{2}a_{3}}_{\mu_{1}\mu_{2}\mu_{3}}\,\partial_{\sigma}\xi^{\sigma}-{e^{\dagger}}^{\,a_{1}a_{2}a_{3}}_{[\mu_{1}\mu_{2}|\sigma|}\,\partial_{\mu_{3}]}\xi^{\sigma}+{e^{\dagger}}^{\,a_{1}a_{2}a_{3}}_{[\mu_{1}|\sigma|\mu_{2}}\,\partial_{\mu_{3}]}\xi^{\sigma}-{e^{\dagger}}^{\,a_{1}a_{2}a_{3}}_{\sigma[\mu_{1}\mu_{2}}\,\partial_{\mu_{3}]}\xi^{\sigma} \nn \\[4pt]
& \hspace{4cm} =\xi^{\sigma}\,\partial_{\sigma}{e^{\dagger}}^{\,a_{1}a_{2}a_{3}}_{\mu_{1}\mu_{2}\mu_{3}} +\partial_{[\mu_{3}}\xi^{\sigma}\,{e^{\dagger}}^{\,a_{1}a_{2}a_{3}}_{\mu_{1}\mu_{2}]\sigma} +\partial_{[\mu_{2}}\xi^{\sigma}\,{e^{\dagger}}^{\,a_{1}a_{2}a_{3}}_{\mu_{1}|\sigma|\mu_{3}]} +\partial_{[\mu_{1}}\xi^{\sigma}\,{e^{\dagger}}^{\,a_{1}a_{2}a_{3}}_{|\sigma|\mu_{2}\mu_{3}]} 
\end{align*}
where we used $|\xi|=1$ and $|e^{\dagger}|=-1$. This expression
extracts the components of the Lie derivative of a three-form by
dualization, as
expected. Explicitly, for $\xi\in {\scrF_{\textrm{\tiny BV}}}\,_{-1}$ and
$e^{\dagger}\in {\scrF_{\textrm{\tiny BV}}}\,_{1}$ we get
\begin{align*}
&\langle Q_{\textrm{\tiny BV}}{e^{\prime\,\dagger}}^{\,a_{1}a_{2}a_{3}}_{\mu_{1}\mu_{2}\mu_{3}}|\xi\odot e^{\dagger}\rangle  \nn \\[4pt]
& \qquad =\langle \xi^{\prime\,\sigma}\odot\partial_{\sigma}{e^{\prime\,\dagger}}^{\,a_{1}a_{2}a_{3}}_{\mu_{1}\mu_{2}\mu_{3}} +\partial_{[\mu_{3}}\xi^{\prime\,\sigma}\odot{e^{\prime\,\dagger}}^{\,a_{1}a_{2}a_{3}}_{\mu_{1}\mu_{2}]\sigma} +\partial_{[\mu_{2}}\xi^{\prime\,\sigma}\odot{e^{\prime\,\dagger}}^{\,a_{1}a_{2}a_{3}}_{\mu_{1}|\sigma|\mu_{3}]} +\partial_{[\mu_{1}}\xi^{\prime\,\sigma}\odot{e^{\prime\,\dagger}}^{\,a_{1}a_{2}a_{3}}_{|\sigma|\mu_{2}\mu_{3}]} | \xi\odot e^{\dagger}\rangle \nn\\[4pt]
& \qquad =(-1)^{|\xi'|\, |e^{\prime\,\dagger}|}\,
  (\LL_{\xi}e^{\dagger})^{a_{1}a_{2}a_{3}}_{\mu_{1}\mu_{2}\mu_{3}}
  \\[4pt]
& \qquad
  =\langle{e^{\prime\,\dagger}}^{\,a_{1}a_{2}a_{3}}_{\mu_{1}\mu_{2}\mu_{3}}|\LL_{\xi}e^{\dagger}\rangle\\[4pt]
& \qquad =: -(-1)^{|Q_{\textrm{\tiny BV}}|\,
  |e^{\prime\,\dagger}|}\,\langle{e^{\prime\,\dagger}}^{\,a_{1}a_{2}a_{3}}_{\mu_{1}\mu_{2}\mu_{3}}|{\DD_{\textrm{\tiny
  BV}}}\,_2(\xi\odot e^{\dagger})\rangle \ .
\end{align*}
Thus ${\DD_{\textrm{\tiny BV}}}\,_2(\xi\odot e^{\dagger})=\LL_{\xi}e^{\dagger}$, and so 
\begin{align*}
\ell_{2}({}^{s^{-1}}\xi\wedge {}^{s^{-1}}e^{\dagger})=s^{-1}\circ {\DD_{\textrm{\tiny BV}}}\,_2\circ
  (s\otimes s)({}^{s^{-1}}\xi\wedge {}^{s^{-1}} e^{\dagger})= \LL_{{}^{s^{-1}}\xi}{}^{s^{-1}}e^{\dagger}
\end{align*}
as required.

Lastly, we note that the part concerning the local Lorentz transformations in
\eqref{eq:QBVedag4d} may be expanded as
\begin{align*}
4 \, \rho^{[d}{}_{d}\, {e^{\dagger}}^{\,a_{1}a_{2}a_{3}]}_{\mu_{1}\mu_{2}\mu_{3}}= -\rho^{[a_{1}}{}_{d}\,
  {e^{\dagger}}^{|d|a_{2}a_{3}]}_{\mu_{1}\mu_{2}\mu_{3}}-\rho^{[a_{2}}{}_{d}\,
  {e^{\dagger}}^{\,a_{1}|d|a_{3}]}_{\mu_{1}\mu_{2}\mu_{3}} -
  \rho^{[a_{3}}{}_{d}\, {e^{\dagger}}^{\,a_{1}a_{2}]d}_{\mu_{1}\mu_{2}\mu_{3}}
\end{align*}
where we used antisymmetry of $\rho^{ab}$. This contains the action of an infinitesimal Lorentz transformation on a
three-vector: Indeed, the corresponding Euler--Lagrange derivative transforms
as
\begin{align*}
\rho \cdot (e\dwedge R) = \rho
\cdot e\dwedge R +e\dwedge [\rho,R] = \big(\rho^{a}{}_{d}\, e^{d}\wedge R^{bc}+ e^{a}\wedge
\rho^{b}{}_{d}\, R^{dc}+ e^{a}\wedge\rho^{c}{}_{d}\,
R^{bd}\big)\,{\tt E}_{a}\wedge {\tt E}_{b} \wedge {\tt E}_{c} \ .
\end{align*}
Dualizing for $\rho \in {\scrF_{\textrm{\tiny BV}}}\,_{-1}$ and
$e^{\dagger}\in {\scrF_{\textrm{\tiny BV}}}\,_{1}$ we obtain
\begin{align*}
\langle Q_{\textrm{\tiny BV}}{e^{\prime\,\dagger}}^{\,a_{1}a_{2}a_{3}}_{\mu_{1}\mu_{2}\mu_{3}}|\rho\odot e^{\dagger}\rangle &=\langle -\rho^{\prime\,[a_{1}}{}_{d}\odot {e^{\prime\,\dagger}}^{|d|a_{2}a_{3}]}_{\mu_{1}\mu_{2}\mu_{3}}-\rho^{\prime\,[a_{2}}{}_{d}\odot {e^{\prime\,\dagger}}^{\,a_{1}|d|a_{3}]}_{\mu_{1}\mu_{2}\mu_{3}} - \rho^{\prime\,[a_{3}}{}_{d}\odot {e^{\prime\,\dagger}}^{\,a_{1}a_{2}]d}_{\mu_{1}\mu_{2}\mu_{3}}|\rho\odot e^{\dagger}\rangle \nn \\[4pt]
&=-(-1)^{|\rho'|\,
  |e^{\prime\,\dagger}|}\,\big({\rho}^{[a_{1}}{}_{d}\,
  {e^{\dagger}}^{|d|a_{2}a_{3}]}_{\mu_{1}\mu_{2}\mu_{3}}-{\rho}^{[a_{2}}{}_{d}\,
  {e^{\dagger}}^{\,a_{1}|d|a_{3}]}_{\mu_{1}\mu_{2}\mu_{3}} -
  {\rho}^{[a_{3}}{}_{d}\,
  {e^{\dagger}}^{\,a_{1}a_{2}]d}_{\mu_{1}\mu_{2}\mu_{3}} \big) \nn \\[4pt]
&=\langle
  {e^{\prime\,\dag}}^{\,a_{1}a_{2}a_{3}}_{\mu_{1}\mu_{2}\mu_{3}}|\rho\cdot
  e^{\dagger}\rangle \\[4pt]
&=:(-1)^{|Q_{\textrm{\tiny BV}}|\,|e^{\prime\,\dagger}|}\,\langle
  {e^{\prime\,\dag}}^{\,a_{1}a_{2}a_{3}}_{\mu_{1}\mu_{2}\mu_{3}}|{\DD_{\textrm{\tiny
  BV}}}\,_2(\rho\odot e^{\dagger})\rangle \ .
\end{align*}
Thus ${\DD_{\textrm{\tiny BV}}}\,_2(\rho\odot e^{\dagger})=-\rho\cdot e^{\dagger}$ and so
\begin{align*}
\ell_{2}({}^{s^{-1}}\rho\wedge {}^{s^{-1}}e^{\dagger})=s^{-1}\circ
  {\DD_{\textrm{\tiny BV}}}\,_2 \circ (s\otimes s)({}^{s^{-1}}\rho\wedge {}^{s^{-1}}e^{\dagger})=- {}^{s^{-1}}\rho\cdot {}^{s^{-1}} e^{\dagger}
\end{align*}
as desired.
Hence we have recovered the full dynamical $L_{\infty}$-algebra for
the coframe field $e$ given in Section~\ref{sec:Linfty4d}. One
similarly obtains the dynamical brackets for the connection $\om$ from the
second transformation in \eqref{eq:QBV}.

\subsubsection*{Noether identity brackets}

Now we consider the third transformation of \eqref{eq:QBV} specialised
to $d=4$ dimensions:
\begin{align}
Q_{\textrm{\tiny BV}}{\rho^{\dagger}}^{\,a_{1}
  a_{2}}_{\mu_{1}\mu_2\mu_3 \mu_{4}}=& \ -\tfrac32\, e_{d[\mu_{1}}\,
                                     {e^{\dagger}}^{\,da_{1}
                                     a_{2}}_{\mu_{2}\mu_3 \mu_{4}]}
                                     - \om^{a_{1}}{}_{d[\mu_{1}}{}\,
                                     {\om^{\dagger}} ^{\,da_{2}
                                     }_{\mu_{2}\mu_3 \mu_{4}]}
                                     + \partial_{[\mu_{1}}
                                     {\omega^{\dagger}}^{\,a_{1}
                                     a_{2}}_{\mu_{2}\mu_3 \mu_{4}]}
                                     \nn \\ & -\rho^{a_{1}}{}_d{}\,{\rho^{\dagger}}^{\,da_{2}
                                     }_{\mu_{1}\mu_2\mu_3 \mu_{4}}
                                     + \partial_{\sigma}\big(\xi^{\sigma}\,
                                       {\rho^{\dagger}}^{\,a_{1}
                                       a_{2}}_{\mu_{1}\mu_2\mu_3
                                       \mu_{4}}\big) \ .
\label{eq:QBVlambdadag4d}\end{align}
The first three terms extract the components of the brackets corresponding to the Noether identity for local $\sSO_+(1,3)$ Lorentz transformations:
\begin{align} \label{eq:Noetherrot4d}
-\dd^\omega\CF_\omega-\tfrac32\,\CF_e\wedge e = 0 \ ,
\end{align}
as we now demonstrate. 

Dualizing as we did previously, for $\om^\dag\in{\scrF_{\textrm{\tiny BV}}}\,_{1}$ we get
\begin{align*}
\langle Q_{\textrm{\tiny BV}}{\rho^{\prime\,\dagger}}^{\,a_{1}
  a_{2}}_{\mu_{1}\mu_2\mu_3 \mu_{4}} | \om^\dag\rangle = \partial_{[\mu_1}\om^\dag_{\mu_2\mu_3\mu_4]} = \langle {\rho^{\prime\,\dagger}}^{\,a_{1}
  a_{2}}_{\mu_{1}\mu_2\mu_3 \mu_{4}}|\dd\om^\dag\rangle =: \langle {\rho^{\prime\,\dagger}}^{\,a_{1}
  a_{2}}_{\mu_{1}\mu_2\mu_3 \mu_{4}}|{\DD_{\textrm{\tiny BV}}}\,_1\om^\dag\rangle
\end{align*}
where we used $|Q_{\textrm{\tiny BV}}|=1$ and $|\om^\dag|=2$. Thus ${\DD_{\textrm{\tiny BV}}}\,_1\om^\dag=\dd\om^\dag$ and so
\begin{align*}
\ell_1({}^{s^{-1}}\om^\dag) = s^{-1}\circ {\DD_{\textrm{\tiny BV}}}\,_1\circ s({}^{s^{-1}}\om^\dag) = \dd{}^{s^{-1}}\om^\dag \ .
\end{align*}
Next, for $e\in {\scrF_{\textrm{\tiny BV}}}\,_{0}$ and $e^\dag\in {\scrF_{\textrm{\tiny BV}}}\,_{1}$ we obtain
\begin{align*}
\langle Q_{\textrm{\tiny BV}}{\rho^{\prime\,\dagger}}^{\,a_{1}
  a_{2}}_{\mu_{1}\mu_2\mu_3 \mu_{4}} | e\odot e^\dag\rangle &= -\tfrac32\,
e_{d[\mu_1}\,e^{\dag\,da_1a_2}_{\mu_2\mu_3\mu_4]} \\[4pt]
&= -\tfrac32\,(e^\dag\wedge e)^{a_1a_2}_{\mu_1\mu_2\mu_3\mu_4} \\[4pt]
&= \langle {\rho^{\prime\,\dagger}}^{\,a_{1}
  a_{2}}_{\mu_{1}\mu_2\mu_3 \mu_{4}}| -\tfrac32\,e^\dag\wedge e\rangle \\[4pt]
  &=: \langle {\rho^{\prime\,\dagger}}^{\,a_{1}
  a_{2}}_{\mu_{1}\mu_2\mu_3 \mu_{4}}| {\DD_{\textrm{\tiny BV}}}\,_2(e\odot e^\dag)\rangle \ .
\end{align*}
Thus ${\DD_{\textrm{\tiny BV}}}\,_2(e\odot e^\dag) = -\frac32\, e^\dag\wedge e$ and so
\begin{align*}
\ell_2({}^{s^{-1}}e\wedge{}^{s^{-1}}e^\dag)=s^{-1}\circ {\DD_{\textrm{\tiny BV}}}\,_2\circ(s\otimes s)({}^{s^{-1}}e\wedge{}^{s^{-1}}e^\dag) = -s^{-1}\circ{\DD_{\textrm{\tiny BV}}}\,_2(e\odot e^\dag) = \tfrac32\, {}^{s^{-1}}e^\dag\wedge{}^{s^{-1}}e \ .
\end{align*}
Continuing in an identical fashion, for $\om\in {\scrF_{\textrm{\tiny BV}}}\,_{0}$ and $\om^\dag\in {\scrF_{\textrm{\tiny BV}}}\,_{1}$ we find
\begin{align*}
\ell_2({}^{s^{-1}}\om\wedge{}^{s^{-1}}\om^\dag) = {}^{s^{-1}}\om\wedge{}^{s^{-1}}\om^\dag \ .
\end{align*}
The Noether identity \eqref{eq:Noetherrot4d} is then encoded in $Q_{\textrm{\tiny BV}}^2\rho^\dag=0$.

It is easy to see that the fourth term in \eqref{eq:QBVlambdadag4d}
dualizes to the action of a local Lorentz transformation on a two-vector: For
$\rho\in {\scrF_{\textrm{\tiny BV}}}\,_{-1}$ and $\rho^\dag\in
{\scrF_{\textrm{\tiny BV}}}\,_{2}$, via similar manipulations we find
$$
\ell_2({}^{s^{-1}}\rho\wedge{}^{s^{-1}}\rho^\dag) =
-{}^{s^{-1}}\rho\cdot{}^{s^{-1}}\rho^\dag \ .
$$
Lastly, the fifth term in \eqref{eq:QBVlambdadag4d} extracts the
components of the Lie derivative of a four-form: For $\xi\in
{\scrF_{\textrm{\tiny BV}}}\,_{-1}$ and $\rho^\dag\in
{\scrF_{\textrm{\tiny BV}}}\,_{2}$ we obtain
\begin{align*}
\langle Q_{\textrm{\tiny BV}}{\rho^{\prime\,\dagger}}^{\,a_{1}
  a_{2}}_{\mu_{1}\mu_2\mu_3 \mu_{4}} |\xi\odot\rho^\dag\rangle
  = \partial_\sigma\big(\xi^\sigma\, {\rho^{\dagger}}^{\,a_{1}
  a_{2}}_{\mu_{1}\mu_2\mu_3 \mu_{4}}\big) = \langle {\rho^{\prime\,\dagger}}^{\,a_{1}
  a_{2}}_{\mu_{1}\mu_2\mu_3 \mu_{4}} | \LL_\xi\rho^\dag\rangle =: \langle {\rho^{\prime\,\dagger}}^{\,a_{1}
  a_{2}}_{\mu_{1}\mu_2\mu_3 \mu_{4}} | {\DD_{\textrm{\tiny
  BV}}}\,_2(\xi\odot\rho^\dag)\rangle 
\end{align*}
and so
$$
\ell_2({}^{s^{-1}}\xi\wedge{}^{s^{-1}}\rho^\dag)=s^{-1}\circ
{\DD_{\textrm{\tiny BV}}}\,_2 \circ(s\otimes
s)({}^{s^{-1}}\xi\wedge{}^{s^{-1}}\rho^\dag) =
\LL_{{}^{s^{-1}}\xi}{}^{s^{-1}}\rho^\dag \ .
$$
These thus recover the actions of local Lorentz transformations and diffeomorphisms
on the Noether identities corresponding to the $\sSO_+(1,3)$ gauge
symmetry. One similarly obtains the brackets for the Noether
identities corresponding to diffeomorphisms, together with the action
of gauge transformations on them, from the fourth transformation in
\eqref{eq:QBV}.

\end{document}